\documentclass[a4paper,11pt]{article}
\pdfoutput=1

\usepackage[margin=0.8in]{geometry}

\usepackage[utf8]{inputenc}
\usepackage[british]{babel}
\usepackage{csquotes} 
\usepackage[T1]{fontenc}

\usepackage{pifont}

\usepackage{amsmath,amssymb,amsthm}
\usepackage{mathtools}

\usepackage[
    labelfont=bf %
]{caption}
\usepackage{float}
\usepackage{subcaption}
\usepackage{tabularx}
\usepackage{tikz}
\usetikzlibrary{decorations.markings,arrows,patterns,patterns.meta,shapes}
\usetikzlibrary{positioning}

\usepackage{pgfplots}
\pgfplotsset{compat=1.18}
\pgfdeclareplotmark{donut}{
    \pgfseteorule
    \pgfpathellipse{\pgfpoint{0cm}{0cm}}
             {\pgfpoint{3pt}{0cm}}
             {\pgfpoint{0cm}{1pt}}
     \pgfpathellipse{\pgfpoint{0cm}{0cm}}
             {\pgfpoint{4pt}{0cm}}
             {\pgfpoint{0cm}{2pt}}
    \pgfusepath{fill}
}
\pgfdeclareplotmark{donut*}{
     \pgfpathellipse{\pgfpoint{0cm}{0cm}}
             {\pgfpoint{4pt}{0cm}}
             {\pgfpoint{0cm}{2pt}}
    \pgfusepath{fill}
}

\usepackage[outline]{contour}
\contourlength{0.2em}

\usepackage[compat=0.6]{yquant}
\useyquantlanguage{groups}

\usepackage{url}

\usepackage{mathrsfs}
\usepackage{dsfont}
\usepackage{setspace}
\usepackage{verbatim}
\usepackage[export]{adjustbox}
\usepackage{tqft}
\usepackage[inline]{enumitem}
\usepackage{multirow}
\usepackage{bigdelim}

\usepackage{quantikz}
\usepackage{graphicx}

\usepackage{xcolor}
\usepackage{bm}
\definecolor{darkblue}{RGB}{0,0,128}
\definecolor{darkgreen}{RGB}{0,150,0}

\definecolor{dark-red}{rgb}{0.4,0.15,0.15}
\definecolor{dark-blue}{rgb}{0.15,0.15,0.4}
\definecolor{dark-green}{rgb}{0.15,0.4,0.15}
\definecolor{medium-blue}{rgb}{0,0,0.5}

\definecolor{drawing-blue}{HTML}{c1ecff}
\definecolor{drawing-yellow}{HTML}{f6dab5}
\definecolor{drawing-red}{HTML}{e5bdd2}

\usepackage[%
	backend=biber,
	style=alphabetic,
	url=false,
	isbn=true,
	maxbibnames=10,
	backref, 
]{biblatex}
\usepackage[pdfusetitle]{hyperref}
\hypersetup{
    breaklinks,
    colorlinks,
	linkcolor={dark-blue},
	citecolor={dark-green},
	urlcolor={medium-blue},
    filecolor=red,
}

\usepackage[nameinlink,capitalise]{cleveref}
\crefname{figure}{Figure}{Figures}

\crefname{conjecture}{Conjecture}{Conjectures}
\Crefname{conjecture}{Conjecture}{Conjectures}

\newtheorem{theorem}{Theorem}
\newtheorem*{theorem*}{Theorem}
\newtheorem{lemma}[theorem]{Lemma}

\newtheorem{claim}[theorem]{Claim}
\newtheorem{remark}[theorem]{Remark}

\newtheorem{definition}[theorem]{Definition}
\newtheorem{example}[theorem]{Example}
\usepackage{thmtools}
\usepackage{thm-restate}

\usepackage{braket}
\usepackage{dsfont}
\usepackage[normalem]{ulem}
\usepackage{algorithm}
\usepackage{algpseudocode}
\usepackage{newfloat}
\makeatletter
\newenvironment{subroutine}[1][tbp]
  {\renewcommand{\ALG@name}{Sub-Routine}\begin{algorithm}[#1]}
  {\end{algorithm}}
\makeatother

\usepackage{stmaryrd}

\def\Tr{\operatorname{Tr}}

\DeclareMathOperator{\GRS}{\mathsf{GRS}}

\DeclareMathOperator{\balpha}{\boldsymbol{\alpha}}
\newcommand{\bv}{\ensuremath{\boldsymbol{v}}}
\newcommand{\bu}{\ensuremath{\boldsymbol{u}}}

\newcommand{\sansserif}[1]{%
  \ifmmode
  \mathsf{#1}%
  \else
   \textsf{#1}%
  \fi
}

\usepackage[mathscr]{euscript}

\usepackage{mathtools}

\DeclarePairedDelimiterXPP\bigo[1]{O}{(}{)}{}{#1}
\DeclarePairedDelimiterXPP\littleo[1]{o}{(}{)}{}{#1}
\DeclarePairedDelimiterXPP\bigomega[1]{$\Omega$}{(}{)}{}{#1}
\DeclarePairedDelimiterXPP\bigtheta[1]{$\Theta$}{(}{)}{}{#1}

\newcommand{\e}{\varepsilon}

\graphicspath{{../}{./}{./tikz}{figs/}}

\renewcommand{\sec}[1]{\hyperref[sec:#1]{Section~\ref*{sec:#1}}}
\newcommand{\app}[1]{\hyperref[app:#1]{Appendix~\ref*{app:#1}}}
\newcommand{\ssec}[1]{\hyperref[ssec:#1]{Subsection~\ref*{ssec:#1}}}
\newcommand{\fig}[1]{\hyperref[fig:#1]{Figure~\ref*{fig:#1}}}
\newcommand{\tab}[1]{\hyperref[tab:#1]{Table~\ref*{tab:#1}}}
\newcommand{\lemm}[1]{\hyperref[lemm:#1]{Lemma~\ref*{lemm:#1}}}
\newcommand{\propos}[1]{\hyperref[propos:#1]{Proposition~\ref*{propos:#1}}}
\newcommand{\thm}[1]{\hyperref[thm:#1]{Theorem~\ref*{thm:#1}}}
\newcommand{\alg}[1]{\hyperref[alg:#1]{Algorithm~\ref*{alg:#1}}}
\newcommand{\defn}[1]{\hyperref[defn:#1]{Definition~\ref*{defn:#1}}}

\newcommand{\be}{\begin{equation}}
\newcommand{\ee}{\end{equation}}
\newcommand{\bea}{\begin{eqnarray}}
\newcommand{\eea}{\end{eqnarray}}

\newcommand{\bem}{\begin{multline}}
\newcommand{\eem}{\end{multline}}

\usepackage{array}
\usepackage{booktabs}
\usepackage{makecell}
\usepackage{diagbox}
\usepackage{multicol}
\usepackage{multirow}

\usepackage{siunitx}

\newcolumntype{?}{!{\vrule width 1pt}}

\usepackage{authblk}

\title{Concatenating Algebraic Codes over High-Rate\\Quantum LDPC Codes}
\author[1,2,3]{Adam Wills\thanks{Email: a\_wills@mit.edu}}
\author[1]{Michael E. Beverland}
\author[1]{Lev S. Bishop}
\author[1]{Jay M. Gambetta}
\author[1]{\linebreak Patrick Rall}
\author[1]{Vikesh Siddhu}
\author[1]{Andrew W. Cross}

\affil[1]{IBM Quantum\vspace*{0.2cm}}
\affil[2]{Center for Theoretical Physics — a Leinweber Institute\linebreak Massachusetts Institute of Technology, Cambridge, MA\vspace*{0.2cm}}

\affil[3]{The NSF AI Institute for \linebreak Artificial Intelligence and Fundamental Interactions}

\date{\today}

\begin{document}
\maketitle

\begin{abstract}
Different quantum error correction schemes trade off overhead, error suppression, and hardware connectivity. 
Code concatenation can relax these tradeoffs by using an outer code whose non-local connectivity is supplied by logical operations of an inner code rather than directly by hardware. 
Prior works~\cite{pattison2025hierarchical,gidney2025yoked} showed that this can reduce memory overhead for local low-rate inner codes such as the surface code. 
Here, we study concatenation over non-local, high-rate inner codes.
Such inner codes experience correlated errors among the many logical qubits in a single codeblock. 
We handle this by treating each block as a single logical Galois qudit, enabling concatenation with algebraic outer codes with excellent parameters and, crucially, list decoders. 
In particular, we consider a memory system formed by concatenating quantum Reed–Solomon outer codes over the gross code. 
For fault-tolerant syndrome extraction, we develop a Galois qudit Shor scheme using “time-like” Reed–Solomon protection against measurement errors.
Interestingly, a lightweight fault tolerance scheme, that would fail for qubits, works well for large-alphabet qudits, suggesting a very different theory of fault tolerance for such qudits.
The whole protocol is optimised via improved bicycle instruction logical error rates, novel compilation strategies, and recent decoder post-selection rules.

At uniform $10^{-3}$ physical noise, the concatenated gross code reaches the teraquop regime, which it previously could not access, with a lower space overhead than the $288$-qubit two-gross code, while offering several advantages from the engineering standpoint.
Beyond our main case study, we believe the core ideas of Galois qudits, quantum Reed–Solomon outer codes, and list decoding, will prove generically powerful and highly-transferrable ideas across high-rate quantum architectures.
\end{abstract}

\clearpage
\small
\tableofcontents
\normalsize

\clearpage

\section{Overview}
 
In order to realise the benefits promised by large-scale quantum computation, some means of quantum error correction will be required to overcome environmental noise, as well as imperfections in the physical hardware. Formalised for codes by the BT~\cite{bravyi2009no} and BPT bounds~\cite{bravyi2010tradeoffs}, the belief exists that there should always be some tradeoff between the performance of one's quantum error-correcting scheme, and the degree of non-locality of connections between qubits. This immediately creates tension between performance and engineering difficulty, since non-locality is often linked with hardware complexity \cite{mathews2026placing} or gate speed \cite{poole2025architecture}, limiting the ability to scale long-range connections indefinitely.

The surface code~\cite{kitaev1997quantum,bravyi1998quantum,kitaev2003fault} is a very well-known quantum error-correcting code, whose implementation only requires the engineering of local interactions. The price, however, is a considerable space overhead: the ratio of physical qubits to logical qubits. In order to construct a useful quantum computer whose physical size is not prohibitive, the exploration of means to reduce this space overhead is well motivated.

Quantum LDPC (qLDPC) codes~\cite{breuckmann2021quantum} are quantum codes that can reduce the space overhead of the surface code, but necessarily require the presence of hard-to-engineer long-range connections between qubits. With this being said, certain families of quantum LDPC codes are still found to have few enough and sparse enough long-range connections, of an amenable structure, to be promising candidates for modular superconducting hardware. In many quantum modalities, non-local connectivity has costs that scale non-uniformly with the type, range, and degree of interconnection, driving further architectural specialisation and modularity \cite{bravyi2022future,yoder2025tourgrossmodularquantum,tripier2026walking}. The bivariate bicycle (BB) codes~\cite{kovalev2013quantum,panteleev2021degenerate} have emerged as particularly exciting candidates. In particular,~\cite{bravyi2024high} studied promising explicit examples, demonstrating their potential suitability for superconducting hardware. The two examples most important to this work will be the $[[144,12,12]]$ ``gross'' code, and the $[[288,12,18]]$ ``two-gross'' code.

This duality between the surface code, and qLDPC codes, is not quite the end of the story, however. Recently, a work by Gidney et al. introduced the ``yoked surface codes''~\cite{gidney2025yoked}, showing that the memory component of a surface code architecture could have a significantly reduced space overhead than had been previously thought. The idea is to ``yoke'' the surface codes, that is, use them as inner codes in a two-level concatenated scheme, where the outer code is some very high rate code with high-weight parity checks. The whole scheme is still implementable with local interactions, because the checks of the outer codes may be fault-tolerantly measured using the logical operations native to the inner surface codes, namely lattice surgery operations. The space overhead is then improved, because a much smaller inner surface code may be used to achieve the same target logical error rate. While the scope of the yoked surface codes work is quite practical, a more theoretical work~\cite{pattison2025hierarchical}, draws a similar conclusion, where the inner code is again a surface code, but there the outer code is itself a qLDPC code.

Aside from the issue of space overhead, we believe that concatenated schemes are a very well-motivated design space to explore, as follows.
\begin{enumerate}
    \item \label{concatenated_benefit_first}\textit{Simpler chip module design}: The use of concatenation can ease the engineering burden by enabling simpler chips (smaller quantum codes) with fewer and shorter long-range connections to reach lower error rate regimes, enhancing the modularity of the architecture;
    \item \textit{Recovery from catastrophic errors}: Rare, but highly damaging events can occur in hardware, including superconducting hardware~\cite{google2025quantum}. The outer code can allow for recovery from such events, as long as they are sufficiently spatially localised.\footnote{Note that the work~\cite{pattison2025hierarchical} makes a similar observation about the possibility of handling these ``burst'' errors.} To that point, there is increasing understanding that these catastrophic errors are predominantly localised to individual chips~\cite{wu2025mitigating};
    \item \textit{Recalibration and replacement of chips}: As the performance of superconducting chips is known to vary over time, modules need to be monitored. The presence of an outer code allows certain modules to be re-calibrated, or even replaced, as the rest of the computation continues to work;
    \item \textit{Smoothing cost discontinuities}: Whereas surface code can be neatly scaled to yield lower logical error rates, qLDPC codes like BB codes tend to come in more discrete families, where a substantial jump must be made in engineering difficulty and module size, in exchange for a lower logical error rate. We find that concatenation can smooth out these discontinuities, allowing for a more tunable logical error rate with the same hardware;
    \item \label{concatenated_benefit_last}\textit{Fallback against decoder failure}: Heuristic decoders are commonplace for qLDPC codes, for example those based on belief propagation~\cite{panteleev2021degenerate}. It is known that these heuristic decoders can achieve very good performance in a low-enough average time to make decoding real time~\cite{muller2025improved}. However, a long tail to the decoding time distribution could result in a broader computation being periodically slowed down. In such a case, falling back to a decoder for the outer code that has performance guarantees may be prudent.
\end{enumerate}

\begin{figure}[t]
    \centering
    \includegraphics[width=\linewidth]{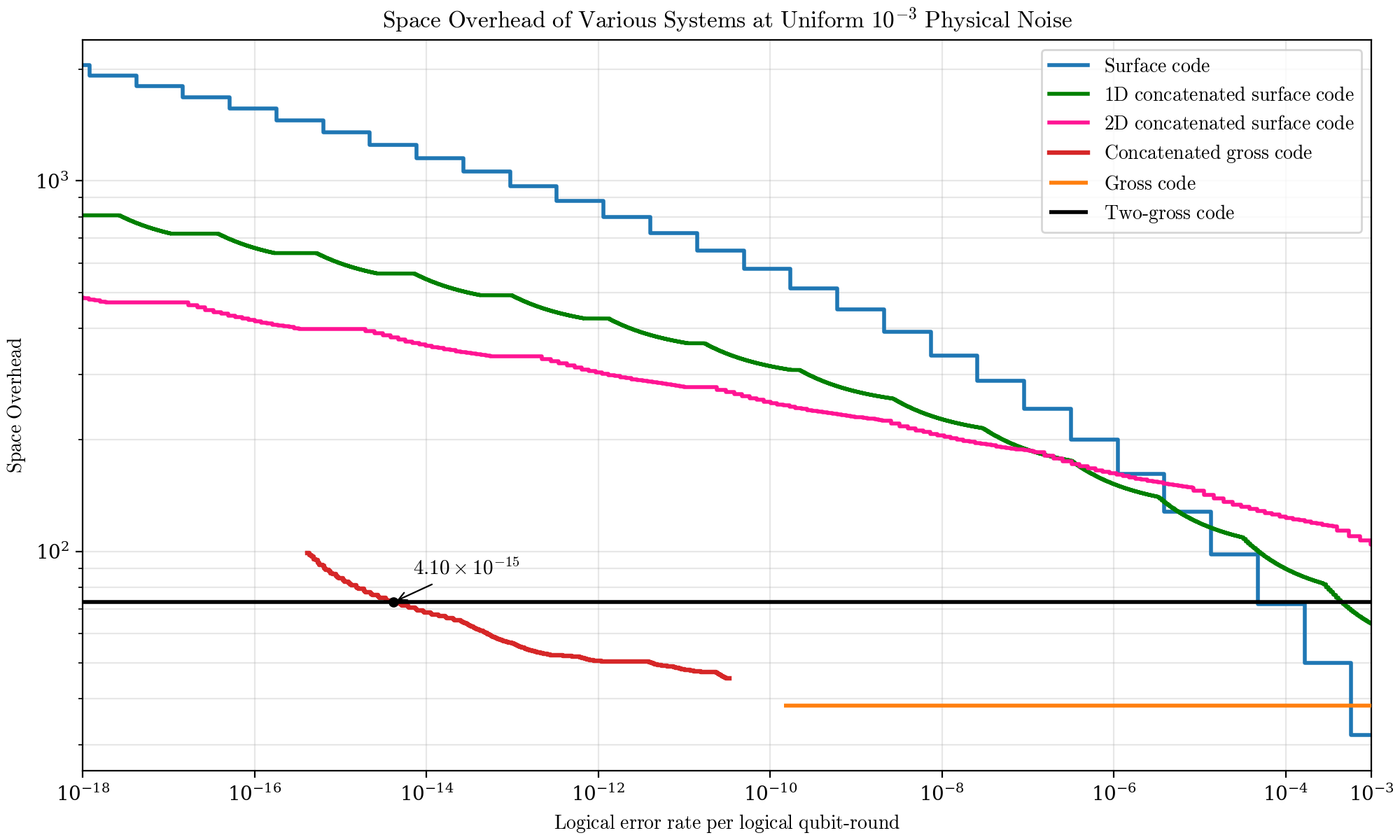}
    \caption{The space overheads of various systems, allowing a maximum size of half a million physical qubits. Concatenation allows the gross code to work in the teraquop regime relevant for the execution of many large-scale quantum algorithms. There, it offers an improved space overhead versus the two-gross code, with several engineering benefits. We show the surface code, as well as the ``yoked'' (concatenated) surface codes from~\cite{gidney2025yoked}. $1\text{D}$ and $2\text{D}$ concatenated surface codes have different layouts. The denser encoding of the latter makes the information harder to access. The concatenated gross codes have the same layout as the $1\text{D}$ concatenated surface codes.}
    \label{fig:main_results}
\end{figure}

With an aim to combine the advantages of high-rate qLDPC codes, as well as the advantages envisioned of concatenation, we construct a concatenated scheme where the inner codes are themselves high-rate qLDPC codes, for the first time. We consider a quantum memory where the inner codes are $144$-qubit ``gross'' BB codes, with logical operations given by the bicycle architecture~\cite{yoder2025tourgrossmodularquantum}. Our results are summarised in Figure~\ref{fig:main_results}, where we see that the concatenated gross code system can operate in the teraquop regime\footnote{The teraquop regime requires logical error rates at most around $10^{-14}$ to $10^{-13}$ per logical qubit-round.} (relevant for the execution of many large-scale quantum algorithms) at uniform $10^{-3}$ physical noise, which it could not previously access. We find that the system offers an improved space overhead versus the $288$-qubit ``two-gross'' code down to a logical error rate of $4.10\times 10^{-15}$ per logical qubit-round, while offering the engineering benefits described. Further improvements are anticipated to the construction; see Appendix~\ref{sec:improved_footprint_morphing} for just one example.

Notice that many of the benefits~\ref{concatenated_benefit_first} to~\ref{concatenated_benefit_last} listed above would only reach their full potential when the entire quantum computer (including the computation area), not only the memory, is formed of a concatenated system. 
Extending the present memory analysis to computation raises additional challenges; for example, in the explicit construction analysed here, the slow outer-code cycle would constrain the rate of computation.
In that direction, considering concatenated quantum memory is an important first step towards designing a fully concatenated quantum architecture. Some of the (many) future directions, as well as opportunities for improving the present construction, are discussed in Subsection~\ref{subsec:discussion}.

\subsection{Methods}

We now present an overview of the techniques used in our constructions and analysis. Figure~\ref{fig:overview_fig} sketches the core ideas graphically.

\begin{figure}[H]
  \centering
  \begin{tikzpicture}[
  xscale=1.00,
  yscale=1.00,
  every node/.append style={transform shape},
  font=\small, >=stealth,
  physbox/.style  ={draw, rounded corners=3pt, fill=blue!10,   minimum width=3.8cm,  minimum height=0.75cm, align=center},
  innerbox/.style ={draw, rounded corners=4pt, fill=orange!15, minimum width=5.8cm,  minimum height=1.5cm,  align=center},
  logbox/.style   ={draw, rounded corners=3pt, fill=orange!20, minimum width=5.8cm,  minimum height=0.65cm, align=center},
  pivotbox/.style ={draw, rounded corners=3pt, fill=red!20,    minimum width=2.8cm,  minimum height=0.85cm, align=center},
  quditbox/.style ={draw, rounded corners=3pt, fill=cyan!20,   minimum width=4.2cm,  minimum height=1.05cm, align=center},
  grossblk/.style ={draw, rounded corners=3pt, fill=orange!15, minimum width=1.55cm, minimum height=0.9cm,  font=\scriptsize, align=center},
  qrsbox/.style   ={draw, rounded corners=4pt, fill=green!10,  minimum width=10.0cm, minimum height=0.75cm, align=center},
  techbox/.style  ={draw, rounded corners=4pt, fill=purple!8},
  resultbox/.style={draw, very thick, rounded corners=6pt, fill=yellow!20,
                    minimum width=11.0cm, minimum height=1.1cm, align=center},
  futurebox/.style={draw, very thick, rounded corners=6pt, fill=teal!15,
                  inner xsep=0.55cm, inner ysep=0.18cm, align=center},
  arrowlabel/.style={font=\scriptsize, fill=white, inner sep=1.5pt},
  seclabel/.style  ={font=\footnotesize\bfseries, align=center, text width=1.6cm},
]

\node[seclabel, anchor=east] at (-0.4, 10.00) {\shortstack[c]{Physical\\Level}};
\node[seclabel, anchor=east] at (-0.4,  8.20) {\shortstack[c]{Inner\\Code}};
\node[seclabel, anchor=east] at (-0.4,  5.65) {\shortstack[c]{Qudit-to-\\Qubit\\Mapping}};
\node[seclabel, anchor=east] at (-0.4,  2.95) {\shortstack[c]{Outer\\Code}};
\node[seclabel, anchor=east] at (-0.4, -0.10) {\shortstack[c]{Fault-\\Tolerant\\Operations\\and\\Analysis}};
\node[seclabel, anchor=east] at (-0.4, -2.65) {Performance};
\node[seclabel, anchor=east] at (-0.4, -4.40)
  {\shortstack[c]{Future\\Improvements}};

\node[draw, very thick, rounded corners=5pt, fill=white,
      font=\normalsize\bfseries, minimum width=12cm, align=center]
  at (7.0, 11.35)
  {Concatenating Quantum Reed-Solomon Codes over the Gross Code};

\node[physbox] (phys) at (7.0, 10.00)
  {$144$ physical qubits +\\
  ancillas};

\node[innerbox] (gross) at (7.0, 8.20)
  {\hyperref[sec:improved_ler_bicycle]{%
   \begin{tabular}{c}
   \textbf{Inner Code --- Gross BB Code}\\[3pt]
   $[[144,\;12,\;12]]$\\[3pt]
   \scriptsize Bicycle instructions with\\\scriptsize improved logical error rates
   \end{tabular}}};

\draw[->, thick] (phys) -- (gross);

\node[logbox] (logq) at (7.0, 6.55) {12 logical qubits};
\draw[->, thick] (gross) -- (logq);

\node[pivotbox] (pivot) at (3.05, 5.05)
  {Pivot qubit\\\scriptsize(logical \#1)};

\node[quditbox] (qudit) at (7.0, 5.05)
  {\href{https://arxiv.org/abs/2605.18981}{%
   \begin{tabular}{c}
   \textbf{Galois qudit}\\
   $q = 2^{11} = 2048$\\
   \scriptsize($s=11$ logical qubits \#2--12)
   \end{tabular}}};

\draw[->, thick] ([xshift=-10ex]logq.south) -- ([xshift=5ex]pivot.north)
  node[arrowlabel, midway, above, xshift=-6ex, yshift=-1ex] {sacrifice};

\draw[->, thick] (logq.south) -- (qudit.north)
  node[arrowlabel, midway, right, xshift=1ex] {\,package};

\node[draw, fill=gray!10, rounded corners=3pt, font=\scriptsize,
      minimum width=2.8cm, align=center, anchor=west] (bmap) at (11.2, 5.05)
  {\href{https://arxiv.org/abs/2605.18981}{%
   \begin{tabular}{c}
   Qudit-to-qubit\\
   mapping $\mathcal{B}=(B_i)_{i=1}^n$\\
   \tiny($\mathbb{F}_q$ basis choice)
   \end{tabular}}};
\draw[->, dashed, gray!70] (bmap.west) -- (qudit.east);

\node[qrsbox] (qrslabel) at (7.0, 3.30)
  {\href{https://arxiv.org/abs/2605.18981}{%
   \textbf{Outer Code --- Quantum Reed-Solomon:}\quad
   $[[n,\;n-2(d-1),\;d]]_q$}};

\draw[->, thick] (qudit.south) -- (qrslabel.north)
  node[arrowlabel, midway, right, xshift=0.08cm, yshift=-0.08cm] {\;$\times n$ blocks};

\node[grossblk] (gb1)  at (2.75, 2.00) {Gross\\block $1$};
\node[grossblk] (gb2)  at (4.45, 2.00) {Gross\\block $2$};
\node[font=\large]     at (5.85, 2.00) {$\cdots$};
\node[grossblk] (gbk)  at (7.00, 2.00) {Gross\\block $i$};
\node[font=\large]     at (8.15, 2.00) {$\cdots$};
\node[grossblk] (gbn1) at (9.55, 2.00) {Gross\\block $n\!-\!1$};
\node[grossblk] (gbn)  at (11.25, 2.00) {Gross\\block $n$};

\coordinate (gbrow-top-left)  at ([xshift=-0.25cm,yshift=0.18cm]gb1.north west);
\coordinate (gbrow-top-right) at ([xshift= 0.25cm,yshift=0.18cm]gbn.north east);
\coordinate (gbrow-bot-left)  at ([xshift=-0.25cm,yshift=-0.18cm]gb1.south west);
\coordinate (gbrow-bot-right) at ([xshift= 0.25cm,yshift=-0.18cm]gbn.south east);

\coordinate (gbrow-top-mid) at ($(gbrow-top-left)!0.5!(gbrow-top-right)$);
\coordinate (gbrow-bot-mid) at ($(gbrow-bot-left)!0.5!(gbrow-bot-right)$);

\draw[thick]
  (gbrow-top-left) -- (gbrow-top-right)
  (gbrow-top-left) -- ++(0,-0.12cm)
  (gbrow-top-right) -- ++(0,-0.12cm);

\draw[thick]
  (gbrow-bot-left) -- (gbrow-bot-right)
  (gbrow-bot-left) -- ++(0,0.12cm)
  (gbrow-bot-right) -- ++(0,0.12cm);

\draw[->, thick] (qrslabel.south) -- (gbrow-top-mid);

\node[techbox, text width=4.3cm, minimum height=1.45cm, align=center] (mix)
  at (2.5, -0.10)
  {\hyperref[sec:ler_calc]{%
   \parbox{4.1cm}{\centering
   \textbf{Spacetime Error Mixing}\\[2pt]
   \scriptsize Random $\mathcal{B}$, $\nu$ $\Rightarrow$\\[1pt]
   $\mathbb{E}_{\mathcal{B},\nu}\!\bigl[\Pr[\text{fail}]\bigr] \leq V$\\
   \scriptsize Qudit fault tolerance}}};

\node[techbox, text width=4.0cm, minimum height=1.45cm, align=center] (shor)
  at (7.0, -0.10)
  {\textbf{Shor EC (qudit)}\\[2pt]
   \scriptsize \hyperref[subsec:ft_qudit_cat_prep]{FT cat state prep}\\\hyperref[sec:compile_Z_meas]{Efficient compilation}\\{\href{https://arxiv.org/abs/2605.20346}{Post-selection strategies}}\\\vspace*{-0.1cm}\hyperref[subsec:time_like_failures]{Time-Like RS Code}};

\node[techbox, text width=3.7cm, minimum height=1.45cm, align=center] (list)
  at (11.5, -0.10)
  {\hyperref[subsubsec:list_decoding]{%
   \shortstack[c]{%
   \textbf{List Decoding (qRS)}\\[2pt]
   \scriptsize $q{=}2048$: decode\\
   beyond $\lfloor(d{-}1)/2\rfloor$}}};

\draw[->, thick] (mix.north  |- gbrow-bot-mid) -- (mix.north);
\draw[->, thick] (shor.north |- gbrow-bot-mid) -- (shor.north);
\draw[->, thick] (list.north |- gbrow-bot-mid) -- (list.north);

\node[resultbox] (result) at (7.0, -2.65)
  {\hyperref[sec:footprint_estimates]{%
   \shortstack[c]{%
   \textbf{Result at $p_{\rm phys} = 10^{-3}$:}\quad
   Concatenated gross code enters the \textbf{teraquop} regime\\
   Improved space overhead over $288$-qubit code down to\\
   logical error rate $4.10\times 10^{-15}$ per logical qubit-round}}};

\draw[->, very thick] (mix.south)  -- (mix.south |- result.north);
\draw[->, very thick] (shor.south) -- (result.north);
\draw[->, very thick] (list.south) -- (list.south |- result.north);

\node[futurebox] (future) at (7.0, -4.70)
  {\shortstack[c]{%
   \textbf{Future Improvements}\\[2pt]
   \hyperref[sec:improved_footprint_morphing]{Morphing Circuit}\\
   \hyperref[subsec:discussion]{Other Directions}}};

\draw[->, very thick] (result.south) -- (future.north);

\foreach \y in {9.4, 7, 4.15, 1.15, -1.5, -3.65}
  \draw[gray!40, dashed, thin] (-0.3,\y) -- (14.5,\y);

\end{tikzpicture}
  \caption{An overview of the moving parts of the construction and analysis. Blue text is hyperlinked.}
  \label{fig:overview_fig}
\end{figure}
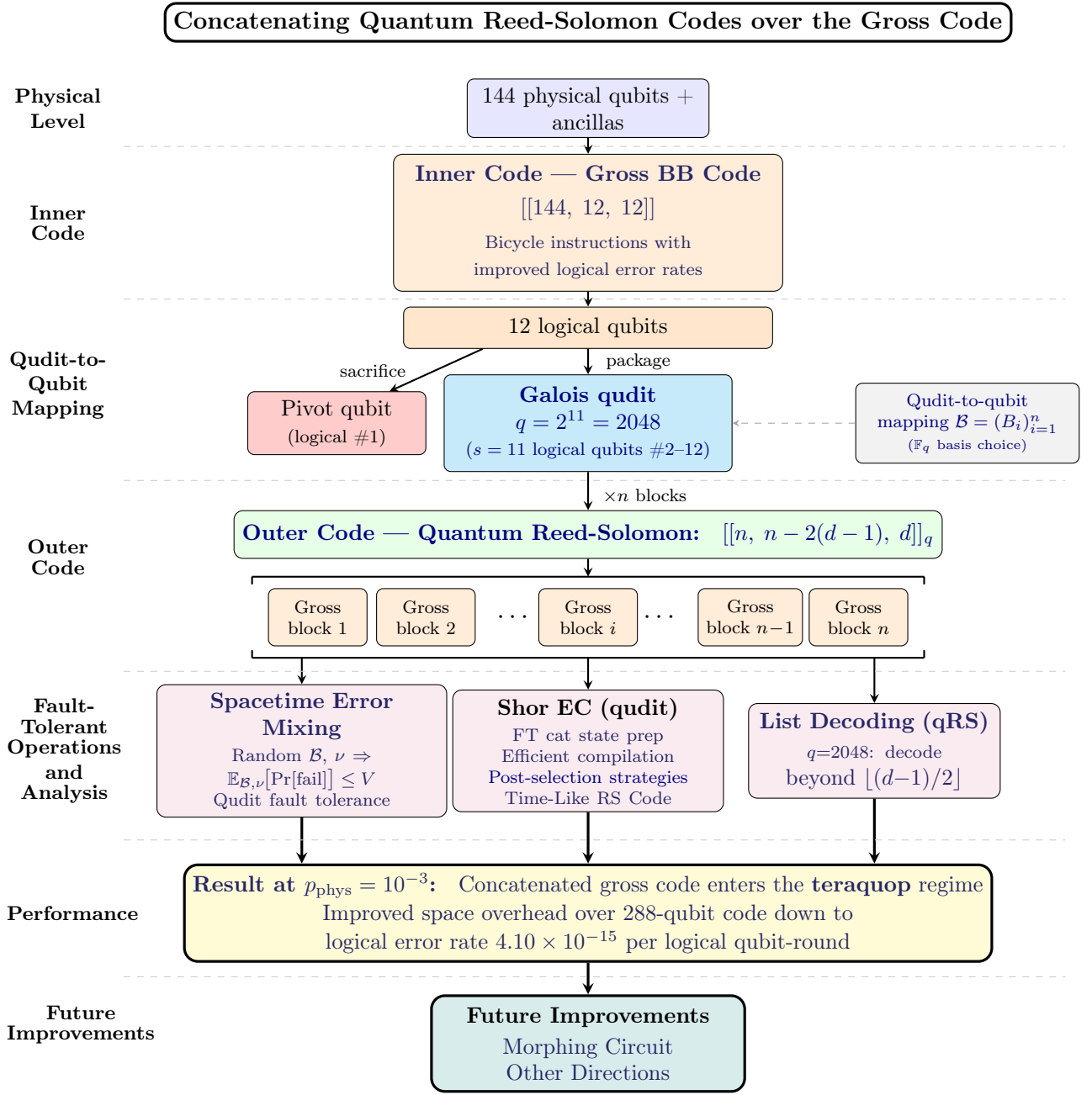

\subsubsection*{Outer Algebraic Codes}

The gross code has $12$ logical qubits. When one considers concatenating an outer code over any high-rate inner code, the first obstruction that one reaches is that the logical qubits in one inner codeblock do not fail independently; they fail in a highly-correlated way. Only different gross code blocks fail independently. If we took any high-rate outer code for qubits with some low distance like $d=4$ or $5$, the concatenated scheme would certainly fail, because a single logical error on a gross code would generically cause an outer code logical error. Of course, one could use an outer code of much higher distance, like $12$ or $24$, but this would lead to an outer code with very poor rate, whereas we want an outer code with very high rate.

The solution is as follows. It has been understood for some time~\cite{ashikhmin2001asymptotically,gottesman2024surviving} that a set of $s$ qubits may be packaged together to form one qudit of dimension $q=2^s$. If one makes the correct choice of Pauli group for the qudit, the set of qubits and the qudit behave the same in every way; their Pauli groups and Clifford groups are isomorphic, and their Clifford hierarchies are in bijection. The present work makes heavy use of this correspondence, and accordingly, we present a companion paper~\cite{wills2026review} that reviews, and builds upon, the existing mathematical formalism. The larger qudits encode the structure of a mathematical object known as a finite field or Galois field, and so have been come to be known as ``Galois qudits''~\cite{eczoo_galois_into_galois}. We will often refer to qudits in this paper; it cannot be over-emphasised that the qudits are themselves simply a mathematical construct, and everything in this paper is ultimately a qubit procedure using the regular qubits and instructions of the gross code.

Having packaged the logical information of each gross code codeblock into one larger qudit, we may go on to make use of outer codes that are specifically designed for these Galois qudits. Namely, we consider quantum Reed-Solomon codes~\cite{abo99,grassl1999quantum}, which are the prototypical algebraic quantum code, for their excellent parameters and structure. These come in an infinite family of increasing distance, where one can construct qudit codes with parameters $[[n,n-2,2]]_q, [[n,n-4,3]]_q, [[n,n-6,4]]_q, [[n,n-8,5]]_q$, and so on, for any $n$ with $n \leq q$, and the subscript $q$ on these parameters indicates that these are codes for qudits of dimension $q$. Note that these are all information-theoretically optimal parameters in that they meet the quantum Singleton bound~\cite{rains2002nonbinary}. It is striking that by using qudits we can scale to higher distances with a minimal loss in the encoding rate; in~\cite{gidney2025yoked}, scaling to distances higher than $2$ requires a much greater loss in the code's encoding rate.

These codes may ultimately be operated as qubit codes. By using qudit-to-qubit mappings~\cite{gottesman2024surviving,wills2026review}, we can convert an $[[n,k,d]]_q$ code into an $[[ns,ks,\geq d]]$ qubit code. The latter code may be viewed as having $n$ sets of $s$ physical qubits, which are then taken to be the sets of $s$ logical qubits in $n$ gross code modules. However, because it is a qubitised qudit code, it treats any qubit error of any weight confined to one set of $s$ qubits (that is, confined to one gross code block) as the same: a weight $1$ error. Throughout, it therefore makes sense to talk about the ``weight'' of an error as its ``block weight'', that is, on how many gross codes it is supported. The outer code having distance $d$ means that more than $\left\lfloor\frac{d-1}{2}\right\rfloor$ groups of $s$ qubits must be damaged in order to create an uncorrectable error. We see that we are precisely treating the distribution of logical errors we want to treat; errors within groups of logical qubits of a codeblock are highly correlated, but between codeblocks are independent, and the outer code counts the weight of errors only in terms of the number of blocks on which they are supported.\footnote{It would be natural for the reader to wonder at this point whether one could concatenate a code over a high-rate inner code by simply using a copy of a qubit outer code on each logical qubit separately. The answer is yes, this is valid, but would lead to a much worse encoding rate on the outer code. As a pertinent example, the outer distance $4$ qubit code considered in~\cite{gidney2025yoked} has parameters $[[n,n-4\sqrt{n}+2,4]]$. Using this separately on each of the $s$ logical qubits in $n$ inner LDPC code modules effectively gives a qubit code with parameters $[[ns,(n-4\sqrt{n}+2)s,4]]$. On the other hand, the qubitised qudit code of distance $4$ has parameters $[[ns,(n-6)s,\geq 4]]$. One checks that the encoding rate of the latter beats the encoding rate of the former (substantially) across all values of $n$.} It is interesting to note that classical coding theory has come to similar conclusions about the utility of packing multiple bits into larger ``dits'', and leveraging algebraic codes, to handle bursty error channels~\cite{forney1971burst}. These ideas have found use in real systems: in communication, storage, and optical media.

\begin{figure}[t]
    \centering
    \includegraphics[width=\linewidth]{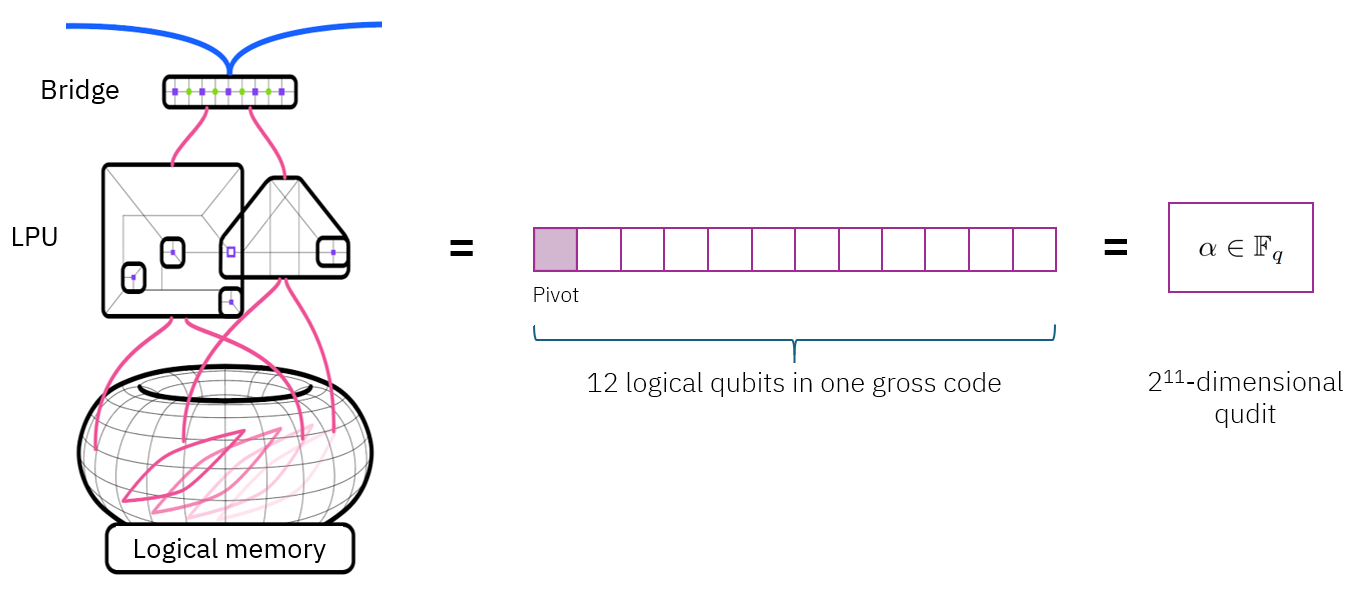}
    \caption{This figure illustrates the key technical idea of packaging the logical information in one gross code into a higher-dimensional qudit. The first logical qubit, the pivot, is sacrificed to enable the logical operations required to operate the memory. The $11$ non-pivot logical qubits become one larger $2^{11}$-dimensional qudit.}
    \label{fig:gross_qubit_to_qudit}
\end{figure}

\subsubsection*{Outer Code Fault-Tolerant Syndrome Extraction}

To operate these outer codes using logical operations on the inner code, we perform a now-standard trick~\cite{cross2024improved,yoder2025tourgrossmodularquantum}, by sacrificing one of the $12$ logical qubits in the gross code (the ``pivot qubit''), which will facilitate logical operations on the other logical qubits. The remaining $s=11$ logical qubits carry data, and are thus packaged into one qudit of dimension $q = 2^{s} = 2048$, as illustrated in Figure~\ref{fig:gross_qubit_to_qudit}. With this, we develop a means to extract the syndrome of our codes fault-tolerantly. We opt for the use of a Shor error correction scheme~\cite{shor1996fault} which we design for Galois qudits. Here, a qudit cat state, which we define, is fault-tolerantly prepared and checked offline, and then consumed to measure one check of the qudit code.\footnote{This could all be described (inconveniently) purely in terms of qubits if one wished. To this end, note that the consumption of each cat state extracts the values of $s$ qubit syndrome bits at once.} In order to protect ourselves against the effects of measurement errors while not introducing an undesirable overhead, we measure an overcomplete basis of the checks rather than simply repeating them. We do this in such a way that our syndromes are protected by a further ``time-like'' Reed-Solomon code.

Our scheme for constructing each cat state is a purely measurement-based one, where our native set of instructions are the bicycle instructions on the gross code~\cite{yoder2025tourgrossmodularquantum}. In order to execute the compilation of our cat states as quickly as possible, we make use of novel compilation tricks (described in Appendix~\ref{sec:compile_Z_meas}) that go beyond the compilation scheme described in~\cite{yoder2025tourgrossmodularquantum}. Specifically, we make the most use of the native gate set by using a single sequence of native rotations to make multiple future measurements native, rather than making one measurement native at a time. Because they are prepared offline, the cat states may be rejected when deemed suspicious and recreated. Aside from their fault-tolerant checking, we also make use of recent post-selection protocols studied for BB codes and their instructions~\cite{wills2026forced}, allowing us to reject suspicious decoding instances. All of this is only possible (at uniform physical noise $10^{-3}$) because we have been able to significantly improve the logical error rates of the bicycle instructions over~\cite{yoder2025tourgrossmodularquantum} by optimising the parameters of the Relay-BP decoder~\cite{muller2025improved} for them; see Appendix~\ref{sec:improved_ler_bicycle} for these improved error rates.

Directly inspired by~\cite{gidney2025yoked}, we use one ancillary system to tend to the fault-tolerant syndrome extraction of many outer code blocks. In particular, our global rectangular grid of gross codes has copies of the outer code arranged in its rows. However, $2$ rows are sacrificed to enable the fault-tolerant generation of cat states, and the consumption of corresponding checks on neighbouring outer codeblocks, as illustrated in Figure~\ref{fig:ft_syndrome_extraction}. These $2$ ancillary rows move through the system over time to fault-tolerantly extract the syndrome of each outer codeblock. It is interesting to note that the layout of our outer codeblocks is most akin to the $1\text{D}$ yoked surface codes in~\cite{gidney2025yoked}, which arguably would make the logical information easier to access than the more dense, $2\text{D}$ setup in~\cite{gidney2025yoked} (of course, the $2\text{D}$ yoked surface codes can offer a better space overhead than the $1\text{D}$ yoked surface codes in the low error regime).

\begin{figure}[t]
    \centering
    \begin{subfigure}[b]{0.4\textwidth}
        \centering
        \includegraphics[width=\textwidth]{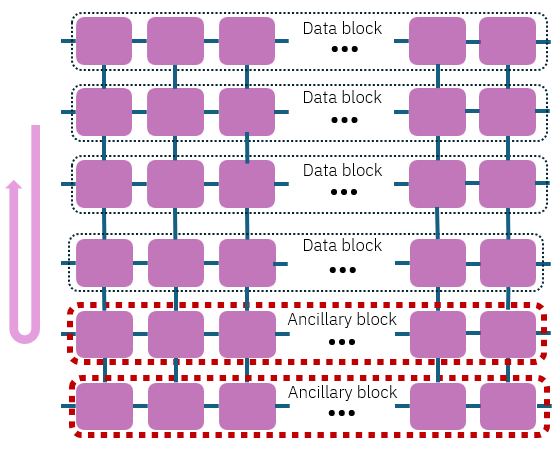}
        
    \end{subfigure}
    \hfill
    \begin{subfigure}[b]{0.55\textwidth}
        \centering
        \includegraphics[width=\textwidth]{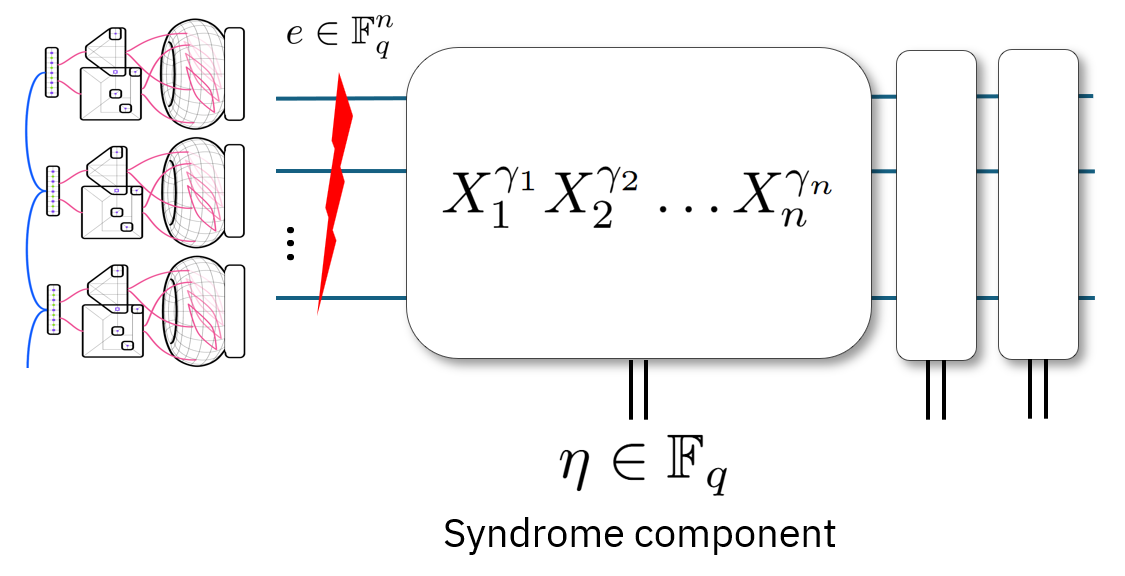}
    \end{subfigure}
    \caption{The left-hand figure illustrates our global memory system. $2$ rows of the system are sacrificed for the fault-tolerant generation of the cat states used for syndrome extraction of neighbouring outer codeblocks. These ancillary blocks move through the system over the course of time to extract the syndrome of other codeblocks. On the right, we illustrate a $Z$ error affecting an outer codeblock, labelled by some vector $e \in \mathbb{F}_q^n$, where $\mathbb{F}_q$ is the finite field of $q$ elements. The syndrome component of the qudit $X$-stabiliser $(\gamma_1, \gamma_2, \ldots, \gamma_n) \in \mathbb{F}_q^n$, which is some value in $\mathbb{F}_q$, is extracted using the appropriate qudit cat state which we will go on to define. Full preliminaries on Galois qudit codes and the measurement of their stabilisers may be found in~\cite{wills2026review}; an abridged description of quantum Reed-Solomon codes specifically may be found in Appendix~\ref{sec:RS_details}.}
    \label{fig:ft_syndrome_extraction}
\end{figure}

\subsubsection*{List Decoding}

Our outer codes themselves need a decoder, but fortunately, the decoding of algebraic codes, including classical Reed-Solomon codes is well-studied~\cite{gorenstein1961class,berlekamp1966nonbinary,welch1983error,massey2003shift,peterson2003encoding}. Because the quantum Reed-Solomon codes are simply formed of two classical Reed-Solomon codes, one for $X$ errors, and one for $Z$ errors, these results transfer immediately. With these algorithms, it is possible to efficiently decode up to half the minimum distance, but one can do more than this using a concept known as \textit{list decoding}~\cite{sudan1997decoding,guruswami1998improved}.\footnote{We recommend~\cite{mcelice2003} for an introduction to the Guruswami-Sudan list decoder for Reed-Solomon codes.} The idea is that (qu)dits of large alphabet size $q$ behave quite differently from (qu)bits, whose alphabet size is $2$. The reason is that, when the alphabet size is large, random errors become very unlikely to ``line up'' with uncorrectable sets of errors.

As an example, consider the distance $[[n,n-10,6]]_q$ quantum Reed-Solomon code. It has many weight-$6$ logical $X$-operators. Considering the support of some weight-$6$ logical $X$-operator, on a qubit code, if $X$ errors happen on $3$ of those coordinates, the error is uncorrectable, as it has the same syndrome as a distinct error of the same weight. However, for qudits, the probability that a random weight-$3$ $X$ error at those $3$ coordinates is uncorrectable is
\begin{equation}\label{eq:weight_3_uncorrectable_d6}
    \mathbb{P}\left[\text{Weight-$3$ $X$ Error is Uncorrectable on the Distance $6$ Code}\right] = \frac{(n-3)(n-4)(n-5)}{6(q-1)^2},
\end{equation}
which is small for $n \ll q$.\footnote{We are indeed working in the $n \ll q$ regime, since we work exclusively with $q = 2048$, and we are interested in outer codes of lengths $n \approx 20$ to $80$.} The effect of working with large-alphabet qudits is that we can generically decode typical errors beyond half the minimum distance. This relies on the \textit{list-decodable structure} of the Reed-Solomon codes, but a \textit{list decoding algorithm} is an efficient algorithm like~\cite{sudan1997decoding,guruswami1998improved} that can list all of the errors agreeing with the inputted syndrome up to some weight.


\subsubsection*{Spacetime Error Mixing}

The formula that was just given, Equation~\eqref{eq:weight_3_uncorrectable_d6}, assumed that the $q-1 = 2047$ $X$ errors that happen on the qudits\footnote{The $2047$ $X$ errors occurring on one qudit correspond to the $2047$ $X$ errors that can occur on the $s=11$ logical non-pivot qubits of one gross code block.} occur with equal probability, given that an $X$ error happens on the qudit. The problem is that we cannot assume that the $q-1$ $X$ errors occurring on the $s$ non-pivot logical qubits (which correspond to the qudit) will occur uniformly, given that an $X$ error occurs. In other words, we cannot assume that the logical $X$ errors taking place on the gross code occur with uniform probability.\footnote{Indeed, in numerical testing, one can see this is not the case.}

Let us elaborate on this issue. We are able to simulate gross codes, and their instructions, numerically. When we do so, typically, we estimate the logical error rate for a particular circuit: the probability that \textit{any} logical Pauli error occurred. This itself is an aggregate statistic; it is the sum of the probabilities that any logical Pauli error occurred. It is feasible to perform simulations to estimate such aggregate statistics numerically, but it is not feasible to numerically estimate the probability of every logical error, since their number grows exponentially with the number of logical qubits. We notice that there is a striking difference between how well-characterised the outer quantum Reed-Solomon codes are, via formulas like Equation~\eqref{eq:weight_3_uncorrectable_d6} which assume equiprobable errors on the inner codes/qudits, and how limited our information is on the \textit{logical error distribution} of the gross codes.

In this paper, we want to be able to make proofs about the performance of our concatenated scheme, even with our limited information. To do this, we introduce a new, and effective, mathematical tool, that we call \textit{spacetime error mixing}. This goes as follows. We introduce some parameter of the protocol, called a \textit{mixing parameter}. If you change the mixing parameter, you get a different protocol, with a different logical error rate. However, one imagines taking a uniformly random choice of the mixing parameter. We show that, in expectation over this random choice, it is possible to upper bound the probability of the logical error rate of the protocol only using knowledge of the aggregate statistics, which we \textit{are able} to estimate numerically.

Let us give a key example. One of our mixing parameters will be the so-called qudit-to-qubit mappings $\mathcal{B}$~\cite{wills2026review}. This is a collection $\mathcal{B} = (B_i)_{i=1}^n$, where $B_i$ is data specifying how the $i$'th qudit in the outer code is broken up into $s$ qubits (or, equivalently, how the $s$ qubits are packaged into the larger qudit). We will show that, in expectation over a random choice of $\mathcal{B}$, the errors sitting on our code at the time of syndrome extraction do behave uniformly. The reason this works is that the uniformly random choice of $\mathcal{B}$ has the effect of ``mixing'' all of the possible $X$-type logical errors, and separately the $Z$-type logical errors, on each gross code.\footnote{To be a bit more concrete, choosing $B_i$ uniformly at random amounts to performing a uniformly random $\mathsf{CNOT}$ circuit amongst the $s$ qubits making up the $i$'th qudit. To be clear, these $\mathsf{CNOT}$ circuits are not being performed. Changing $\mathcal{B}$ changes the specification of the protocol, in a way that is equivalent to acting on each set of $s$ qubits with a random $\mathsf{CNOT}$ circuit.}


Considering some failure mode of our scheme, using only the limited information of aggregate statistics to which we have access, we therefore may calculate some bound
\begin{equation}\label{eq:expectation_mixing_bound}
    \mathbb{E}_{\mathcal{B}}\left\{\mathbb{P}\left[\text{Failure Mode}\right]\right\} \leq V,
\end{equation}
for some number $V$ that we can calculate. This bound necessarily implies that there always exists some choice of $\mathcal{B}$ that makes the $\mathbb{P}\left[\text{Failure Mode}\right]$ small (at most $V$), regardless of what the true logical error distribution of the gross code is. 

In a real system, with real-time decoding, and therefore more information on one's inner code logical error distribution, one may be able to make a wise choice of $\mathcal{B}$. On the other hand, given real-time decoding, one could also literally take $\mathcal{B}$ at random, and simulate to see if the probability of the failure mode is at most $V$. If it is not, the value of $\mathcal{B}$ may be re-randomised, and the result simulated again. What the bound of Equation~\eqref{eq:expectation_mixing_bound} guarantees you is that you will quickly hit a good choice of $\mathcal{B}$, that has a low $\mathbb{P}\left[\text{Failure Mode}\right]$. 

We will go on to show that the same reasoning applies to all mixing parameters of our scheme, and all failure modes of our scheme, simultaneously. Across our memory system, there will be the mixing parameter $\mathcal{B}$, and another mixing parameter called $\nu$ that we will come onto very shortly. The ultimate statement will be that, in expectation over a uniformly random choice of the mixing parameters, the probability of any failure of our scheme may be upper bounded.

\subsubsection*{Qudit Fault Tolerance}

As briefly mentioned above, we protect ourselves against the effects of measurement errors in our Shor error correction scheme by measuring overcomplete bases of checks, rather than simply repeating measurements, to avoid an excessive time overhead in extracting syndromes. Something that we find particularly exciting is that we are able to use a very lightweight scheme for our qudits, that would fail for qubits, because of the unique effects introduced by doing fault tolerance on qudits.

Let us explain this as follows, for the $X$ checks of the outer code (where the discussion for the $Z$ checks is identical). The distance $d$ outer code has $d-1$ qudit $X$ checks (equivalent to $s(d-1)$ qubit $X$ checks). We can imagine the $d-1$ qudit $X$ checks as the rows of some parity-check matrix $H_X \in \mathbb{F}_q^{(d-1)\times n}$, where $\mathbb{F}_q$ is the finite field of $q$ elements. Our protocol for the extraction of the qudit $X$ checks proceeds by first measuring these $d-1$ qudit checks, where each is extracted using one qudit cat state. Next, we measure a number of linear combinations, $M$ of these checks, so that the total set of $d-1+M$ checks forms an overcomplete basis for the space of qudit checks. Concretely, we can consider some matrix $\beta^{(X)} \in \mathbb{F}_q^{M \times (d-1)}$, and the latter $M$ qudit checks extracted are the rows of $\beta^{(X)}\cdot H_X$. Then, if we extract the values $y \in \mathbb{F}_q^{d-1}$ in the first $d-1$ checks, and the values $z \in \mathbb{F}_q^M$ in the latter $M$ checks, we check that $z = \beta^{(X)}\cdot y$. If this equation is not satisfied, we have detected some fault in the round of checks, and we begin the whole round of $d-1+M$ checks again. On the other hand, if the equation is satisfied, we call the round a \textit{self-consistent round of $X$ checks}, and we input $y$ to the decoder of the outer code.

Now, in order to best protect our syndromes, we make a deliberate choice for $\beta^{(X)}$. We consider the classical code with parity-check matrix $H_{X,\text{time}} \coloneq \begin{pmatrix}
        \beta^{(X)} & I
    \end{pmatrix} \in \mathbb{F}_q^{M \times (d-1+M)}$, where $I$ is the $M \times M$ identity matrix. We always want to pick $\beta^{(X)}$ such that this ``time-like'' classical code has distance $M+1$, which is the best distance that can possibly be achieved since it has $M$ parity checks. It turns out that the smoothest way to do this is to consider the parity-check matrix for \textit{another} ``time-like'' classical generalised Reed-Solomon code of distance $M$ and length $d-1+M$, and perform row operations on it to get it into the required form of $H_{X,\text{time}}$. This means that we have a ``time-like'' Reed-Solomon code protecting our syndrome.

When we come onto the analysis, we run into the same problem discussed above in the context of space-time error mixing; how can we calculate the performance of our outer code scheme given our limited information about the failures of the inner code? In the present situation, we call the failure mode the \textit{time-like failure mode}, which is the event that we extract a self-consistent round of checks, but the extracted syndrome $y$ is not the syndrome for the error on the data at any point just before, or during the self-consistent round.\footnote{We will make this precise later; essentially, the time-like failure mode is the event that we got a self-consistent round, but still extracted a bad syndrome.} To bound the probability of time-like failure, the solution is again to introduce a mixing parameter. Here, we consider the mixing parameter $\nu$, which is a vector taking values in $(\mathbb{F}_q^*)^{d-1}$, where $\mathbb{F}_q^*$ is the set of non-zero elements of $\mathbb{F}_q$. Then, we consider replacing the matrix $\beta^{(X)}$ with the matrix $\beta^{(X)}*\nu$, where $\beta^{(X)}*\nu$ is the matrix $\beta^{(X)}$ with all of its columns multiplied by the values in $\nu$.\footnote{Notice that if $\begin{pmatrix}
    \beta^{(X)}&I
\end{pmatrix}$ is the parity-check matrix for a code of distance $M+1$, then so is $\begin{pmatrix}
    \beta^{(X)}*\nu & I
\end{pmatrix}$.}  We can then show, in expectation over a uniformly random choice of $\nu$, that the probability of time-like failure is low. Again, this means that a good choice of $\nu$ exists, regardless of the true inner code logical error distribution.

Most interestingly, we show that we can take a very lightweight  scheme here because of effects that we find to be unique to doing fault tolerance for qudits; \textit{one would simply not use our scheme for qubits}. Specifically, we can take an unusually small value of $M$; indeed, in our analysis, we restrict $M \leq d-2$. Such a scheme would not be used for qubits, because a single data error introduced during a round of checks can keep the round self-consistent, while causing the wrong syndrome to be extracted, that is, the concatenated code fails with the same (or greater) probability as the inner code. However, remarkably, we prove that the probability of such failure modes for our qudits codes (in expectation over the mixing parameters) are heavily suppressed; in fact, we show that they are roughly
\begin{equation}
    \mathbb{P}\left[\text{Syndrome Extraction Fails to One Error}\right]\sim\frac{n(d-1-M)}{(q-1)^M}\cdot P,
\end{equation}
where $P$ is the probability of some gross code failure.\footnote{We are not being precise about this here; we are just trying to capture the order of the quantity.} It is striking to notice that, for qubits, which is the case $q=2$, this formula shows that we are not performing fault tolerance at all --- the probability of outer code failure is (greater than) the probability of inner code failure. However, with large values of $q$, like $q = 2048$, the probability of this failure mode may be made small; in fact, it may be made exponentially small in the code distance by taking larger $M \leq d-2$. We point the reader to Subsection~\ref{subsec:time_like_failures} for the corresponding arguments on time-like failure.

\subsection{Discussion and Outlook}\label{subsec:discussion}

In this work, we have studied the concatenation of algebraic codes over quantum LDPC codes, namely the $144$-qubit ``gross'' bivariate bicycle (BB) code, providing a memory system that can operate in the teraquop regime of one trillion (logical qubits)$\cdot$(logical operations). There, it provides an improved space overhead over the alternative BB code, the $288$-qubit ``two-gross'' code, while offering the engineering benefits discussed. We have introduced a number of new techniques and methods to construct and analyse this protocol. However, we believe there is significant scope for further improvements of this scheme on its own (see Appendices~\ref{sec:improved_footprint_morphing} and~\ref{sec:non_css_outer_codes} for two examples), as well as novel research directions, as follows.
\begin{itemize}
    \item \textit{Soft information decoding}: In~\cite{gidney2025yoked}, the concatenated scheme uses ``soft information'' from the inner surface codes passed to the outer code to maximise performance. This soft information informs the outer decoder which inner codes look most suspicious, allowing the outer code to correct typical errors beyond half the outer code distance. In particular, this enables the use of a distance $2$ outer code in~\cite{gidney2025yoked}, which we do not consider. Note, however, that the use of soft information is likely to be achievable in our case as well. Indeed, \textit{soft information list decoding} is a very well-studied problem for classical Reed-Solomon codes~\cite{koetter2003algebraic,el2004performance,ratnakar2005exponential,el2006iterative,gross2006applications,jiang2008algebraic,zeh2013algebraic,xing2019progressive}, including studies of implementation in classical hardware~\cite{ahmed2004vlsi,scholl2014hardware}. This analysis would then transfer directly to our case because our quantum Reed-Solomon codes are constructed from two classical Reed-Solomon codes. Our inner codes also can provide natural soft information using the same strategies we use for post-selection~\cite{wills2026forced}, which can be used to obtain likelihoods associated to possible logical errors. The analysis of such a system could present further challenges, and so we defer this to future work. However, soft information list decoding can in principle combine the benefits of soft information decoding with the benefits of list decoding for large-alphabet qudits, potentially improving the performance of the scheme considerably. The desire to use soft information in our scheme may also motivate the development of a theory of fault tolerance (not only error correction) in the presence of such soft information/side information, see for example~\cite{wangBrun26}.

    \item \textit{List decoding for fault tolerance}: One of the main messages we wish to convey in the paper is that the ideas of using Galois qudits and algebraic codes are not unique to the gross code; they are transferable to other high-rate architectures. Moreover, the analysis of the fault-tolerant schemes, where we demonstrate fault tolerance effects that are unique to qudits, and absent for qubits, suggests that a theory of qudit fault tolerance should be thought of as genuinely different to the traditional theory of qubit fault tolerance~\cite{gottesman2024surviving}. It is very natural to consider if a theory of list decoding can be developed for fault-tolerant circuits, rather than simply codes. We are not aware of such a theory, even on the classical side.\footnote{Similar effects are indicated by the existence of classical fault tolerance schemes for dits that show a higher threshold than that for bits~\cite{tan2024reliable}, although list decoders for fault-tolerant circuits are not developed.}

    \item \textit{Extractor architectures for concatenated systems}: Recent years have seen considerable development of surgery techniques for high-rate qLDPC codes~\cite{cohen2022low,cross2024improved,cowtan2025fast,cowtan2025parallel,he2025extractors,ide2025fault,zhang2025time,zheng2025high,hillmann2025single,tan2025single,baspin2025fast,yoder2025tourgrossmodularquantum,chang2026constant,yuan2026parsimonious,swaroop2026universal,webster2026pinnacle,williamson2026low,gu2026qgpu,cain2026shor} enabling lower overhead constructions in space and time, and offering greater flexibility. It is quite natural to wonder whether a custom extractor architecture could be engineered that were designed to measure the outer code's checks very efficiently, but not necessarily enabling universal computation, thus giving a very dense, dedicated memory block.
    
    \item \textit{Improved Syndrome Extraction}: As discussed, we use a custom Galois qudit Shor error correction scheme for the fault-tolerant syndrome extraction of our outer codes, which we are able to take to be unusually lightweight because of the fault tolerance effects that are unique to qudits. Even with this, the outer code syndrome extraction is considerably slower than for the yoked surface codes~\cite{gidney2025yoked}. In particular, the relatively limited native gate set of the bicycle instructions on the gross code~\cite{yoder2025tourgrossmodularquantum} renders the generation of the cat states particularly slow, despite our novel compilation tricks (see Appendix~\ref{sec:compile_Z_meas}), presenting one of the key bottlenecks of our scheme. A faster means to fault-tolerantly extract the syndromes of the outer code would benefit the scheme substantially, since one of the dominant noise sources is simply gross codes sitting at idle for a long time in between having their syndromes extracted. This would also have the effect of speeding up the read/write time from and to the memory. Relatedly, a different code or architecture with a broader set of native gates than the current bicycle architecture on the gross code, would enjoy these significant improvements.
    
    \item \textit{Computation on the Concatenated System}: The current analysis is for a memory system only. This means it could serve as a dense memory block of a broader architecture. However, in future, we would like to understand how a concatenated scheme, with high-rate inner codes and algebraic outer codes, could perform in the setting of full computation. This would enable the whole architecture to benefit from the points~\ref{concatenated_benefit_first} to~\ref{concatenated_benefit_last} above.
\end{itemize}

\subsection{Outline of the Paper}

Section~\ref{sec:codes_and_syndrome_extraction} presents a formal description of our protocol, namely the outer codes we use and how they are fault-tolerantly operated. Section~\ref{sec:ler_calc} then upper bounds the logical error rate of one block of our outer code over one of its rounds. Section~\ref{sec:footprint_estimates} then presents the overhead of our memory system against its logical error rate (per logical qubit-round).

Appendix~\ref{sec:compile_Z_meas} presents a means for efficiently compiling certain collections of $Z$ measurements, which is a critical component of the efficient preparation of our cat states. Appendix~\ref{sec:improved_ler_bicycle} presents the improved logical error rates for bicycle instructions on the gross code, and the decoder parameters that achieve these improvements. Appendix~\ref{sec:improved_footprint_morphing} presents the logical error rate of a new idling circuit for the gross code based on morphing~\cite{shaw2025lowering}, and the improved footprint estimates that would be obtained by using it. Appendix~\ref{sec:RS_details} presents a brief description of the outer codes that we use (although full details are found in the companion work~\cite{wills2026review}), and the specific parameter choices. Appendix~\ref{sec:tlf_proofs} contains deferred proofs of lemmas used in the time-like failure analysis in Section~\ref{subsec:time_like_failures}. Appendix~\ref{sec:non_css_outer_codes} contains thoughts on a possible means to improve the construction that makes use of non-CSS outer codes to maximise the potential of list decoding, whose exploration is left to future work. Appendix~\ref{sec:5_qudit_code} presents an explicit example of a $5$-qudit quantum Reed-Solomon code for concreteness, and how it may be expanded to a $55$-qubit code.

\clearpage

\section{Description of the Outer Codes and Syndrome Extraction Circuit}\label{sec:codes_and_syndrome_extraction}

This section will describe the outer codes used in our concatenated approach, and the fault-tolerant syndrome extraction circuit used to operate them. We endeavour to make the following understandable at a high level without mathematical details. For full mathematical details on Reed-Solomon codes and qudit-to-qubit mappings, see the companion paper~\cite{wills2026review}, which is based on prior work~\cite{ashikhmin2001asymptotically,gottesman2024surviving}. For a more brief description of Reed-Solomon codes, see Appendix~\ref{sec:RS_details}.

\subsection{Description of Outer Codes}\label{subsec:code_description}
In this work, we imagine an $m \times n$ grid of $[[144,12,12]]$ quantum bivariate bicycle codes, termed ``gross'' codes~\cite{bravyi2024high,yoder2025tourgrossmodularquantum}, where the grid is imagined to have $m$ rows and $n$ columns. The gross codes have square-grid connectivity to each other, and are taken to have a native set of logical instructions termed ``bicycle'' instructions, as described in~\cite{yoder2025tourgrossmodularquantum}. In particular, ``bridged'' lattice surgery operations between modules are allowed if and only if the codes are adjacent in the grid. 
We think of the $m$ rows as forming codeblocks of our outer code of length $n$, although two rows will be sacrificed to assist the fault-tolerant syndrome extraction of the outer code. Therefore, in total, $m-2$ rows will hold data, and $2$ rows will form an ancillary system. Exactly which two rows form the ancillary rows will change over the course of time. That is, the ancillary rows not carrying data will move through the system over time to assist in the syndrome extraction of the outer codeblocks. See Figure~\ref{fig:ft_syndrome_extraction} for a depiction.

The outer codes into which the gross codes are concatenated will be quantum Reed-Solomon codes~\cite{abo99,grassl1999quantum}. In particular, we take each set of $11$ logical qubits $2, \ldots, 12$ (that is, all logical qubits in the gross code except the first ``pivot'' qubit~\cite{cross2024improved,yoder2025tourgrossmodularquantum}), and package each into one ``Galois qudit'' of dimension $q=2^{s}$, where $s=11$ (and $q = 2^{11} = 2048$) here and throughout the paper. Doing so allows us to concatenate the logical information of each gross code neatly into the outer algebraic qudit code. Let us emphasise that the ``qudits'' here are simply groups of qubits --- indeed, the entire construction may be described without any reference to qudits --- but the qudit description enables a very natural and convenient means to describe concatenations over high-rate codes. The pivot qubit (that is, logical qubit $1$) in each module is simply left as an ancilla that allows the execution of operations required to operate the outer code.

The quantum Reed-Solomon codes that we use as our outer codes have parameters\footnote{Note that these parameters are optimal since they meet the quantum Singleton bound~\cite{rains2002nonbinary}, which states that if an $[[n,k,d]]_q$ quantum code exists, then $n-k \geq 2(d-1)$, regardless of the qudit dimension $q$.}
\begin{equation}
    [[n,n-2(d-1),d]]_q,
\end{equation}
where we may choose any values of $n,d$ satisfying $2(d-1)<n\leq q$. The quantum Reed-Solomon codes we use are defined formally in~\cite{wills2026review}; see also Appendix~\ref{subsec:qrs_def} for a description with only the essential details. Note that the subscript $q$ on these parameters denotes the fact that these are codes over qudits of dimension $q$. In order to concatenate them over our gross codes, they must be converted into \textit{qubit} codes. We note that any $[[n,k,d]]_q$ code is equivalent to an $[[ns,ks,\geq d]]$ code (where $q = 2^s$), where the lack of subscript on these latter parameters indicates that this is a qubit code. Note that this means that a codeblock of the outer code is made up on $n$ physical gross codes, encoding the same logical information of $n-2(d-1)$ gross codes.

In order to convert the qudit codes into qubit codes, a ``qudit-to-qubit mapping'' may be chosen on each of the $n$ physical qudits (gross codes) in the outer code. There are many such qudit-to-qubit mappings, but they may each be specified by a choice of basis $B$ for the field $\mathbb{F}_q$ over $\mathbb{F}_2$. At each coordinate of the code $i \in [n]$, we may choose a different such qudit-to-qubit mapping labelled by the basis $B_i$, and we refer to the collection of bases as $\mathcal{B} = (B_i)_{i=1}^n$.
For the full mathematical details on qudit-to-qubit mappings, the reader may refer to~\cite{wills2026review}. However, for now, we will mention that these mappings give us consistent ways to map all objects of the qudit code to those of a qubit code, including the stabilisers, logical operators, errors, and the code states themselves. For example, a computational basis state on one qudit may be written $\ket{\gamma}$ for $\gamma \in \mathbb{F}_q$, and these are in one-to-one correspondence with $s$-qubit computational basis states $\ket{\mathcal{D}_B(\gamma)}$, where $\mathcal{D}_B:\mathbb{F}_q \to \mathbb{F}_2^s$ is an invertible linear map. If you choose a different mapping between the qudit and the $s$ qubits (encoded in the basis $B$), you get a different linear map $\mathcal{D}_B$. To guide the presentation, we provide an explicit example of a $[[5,1,3]]_{2048}$ code in Appendix~\ref{sec:5_qudit_code}, and show an example of how it can be expanded to a $[[55,11,\geq 3]]_2$ code.

A very important comment to make about such qubitised qudit codes is that, while they have poor distance in the strict sense (for example, this code mentioned has distance $\geq 3$), the notion of distance is a little strange when thought in terms of qubits. Indeed, when imagined as a qubit code, the code has $n$ groups of $s$ physical qubits, where each group is really a qudit. Then, if the qudit code had distance $d$, generically the qubit code will have distance $d$ as well, meaning that there is some way to damage a single qubit in $d$ different groups in order to create a logical error. However, the important thing about these codes is that they treat all errors within each group of $s$ qubits equally, that is, if you damage $\leq d-1$ groups of qubits in any way, it will not be a logical error, regardless of how many qubits in each group you damage. We can see, therefore, that these qudit codes are exactly the natural object when talking about concatenating over high-rate quantum codes such as the gross code. Indeed, when the gross code fails, there is no guarantee that its logical qubits fail independently, and indeed they do not; logical qubits in one gross code experience highly correlated errors. However, separate gross codes do fail independently. This makes the use of these qudit codes extremely natural for our purposes.

Note that all copies of our outer code, on all $m-2$ rows of our gross code grid that carry data, are identical. That is, the qudit codes are identical, and the same set of qudit-to-qubit mappings $\mathcal{B}$ are chosen on every codeblock. We emphasise this latter point; every row of length $n$ has the same set of qudit-to-qubit mappings $\mathcal{B}$.

\subsection{Qudit Cat State Preparation}\label{subsec:ft_qudit_cat_prep}

Given our memory block of $m$ gross code rows, $m-2$ of which form active blocks of the outer code, we need some way to fault-tolerantly extract their syndromes. We will proceed with a version of Shor-style error correction~\cite{shor1996fault} for qudits. To do this, the $2$ rows that do not carry data will be used to fault-tolerantly prepare qudit cat states that can be used to measure one qudit stabiliser of the outer code at a time. It is useful to note that measuring one qudit stabiliser is equivalent to measuring $s$ qubit stabilisers.
This subsection will focus on the fault-tolerant preparation of one qudit cat state.

For the duration of Subsection~\ref{subsec:ft_qudit_cat_prep}, we focus entirely on the two ancilla rows, where one row is described as the ``cat state row'', and one is described as the ``ancilla row''. We aim to generate the qudit cat state on the cat state row, where the ancilla row is only present to enable fault-tolerant checking of the cat state.

To orient the reader, we comment again that the following procedure could be described purely in terms of qubits. In fact, consuming one qudit cat state to measure one qudit stabiliser is the same thing as making one large qubit ancilla state and consuming it to measure $s$ qubit stabilisers in one go. However, we find the description in terms of qudits extremely convenient. We endeavour to make the following understandable at a high level without intimate knowledge of the qudit language, however, for full details, the reader should consult~\cite{wills2026review}.

Note that in this subsection, we are imagining that we are preparing a single qudit cat state to extract the syndrome of an $X$-type qudit stabiliser, but the measurement of a $Z$-type qudit stabiliser proceeds identically, with $X$ and $Z$ swapped everywhere.

\subsubsection{Qudit Cat State Preparation and Verification}

Throughout subsection~\ref{subsec:ft_qudit_cat_prep}, we imagine we are wanting to extract the syndrome of the $X$-type qudit stabiliser $X_1^{\gamma_1}X_2^{\gamma_2}\ldots X_n^{\gamma_n}$ on some block of the outer code, where $\gamma_1, \ldots, \gamma_n \in \mathbb{F}_q$ and $\gamma_i \neq 0$. In order to do this, we wish to prepare the following state on the cat state row, defined in terms of the qudit language~\cite{wills2026review}. Later, we will talk about how this state can be consumed to measure the stabiliser; for now, we are just focused on fault-tolerantly preparing and verifying this state.

\begin{definition}[Qudit Cat State] Given Galois qudits of dimension $q = 2^s$, and given non-zero field elements $\gamma_1, \gamma_2, \ldots, \gamma_n \in \mathbb{F}_q$, the qudit cat state $\mathsf{Cat}(\gamma_1, \gamma_2, \ldots, \gamma_n)$ is defined to be the (unique) $n$-qudit state with the following set of qudit stabilisers:
\begin{itemize}
    \item $X_1^{\gamma_1}X_2^{\gamma_2}\ldots X_n^{\gamma_n}$;
    \item $Z_{i-1}^{\gamma_i}Z_i^{\gamma_{i-1}}$ for each $i = 2, \ldots, n$.
\end{itemize}    
\end{definition}
One may refer to~\cite{wills2026review} for details on what it means to define a qudit stabiliser state, but we comment here that a single qudit stabiliser is really the same thing as $s$ qubit stabilisers.

Note that the state $\mathsf{Cat}(\gamma_1, \gamma_2, \ldots, \gamma_n)$ is the natural generalisation of the usual $n$-qubit cat state, because the usual $n$-qubit cat state is the (unique) qubit state with (qubit) stabilisers $X^{\otimes n}$ and $Z_{i-1}Z_i$ for each $i = 2, \ldots, n$. Additionally, note that the $n$-qubit cat state is the appropriate ancillary state to measure the stabiliser $X^{\otimes n}$. Analogously, the state $\mathsf{Cat}(\gamma_1, \gamma_2, \ldots, \gamma_n)$ will be the appropriate ancillary state to measure the qudit stabiliser $X_1^{\gamma_1}X_2^{\gamma_2}\ldots X_n^{\gamma_n}$.\footnote{A different strategy would be to convert all qudit stabilisers to $s$ different qubit stabilisers and measure them all separately with qubit cat states, but this would be very inefficient because it would require approximately $s$-times longer to perform syndrome extraction.}

On the other hand, one interesting difference with the case of the regular qubit cat state is that the qubit cat state always has the same form, but its size changes. For the qudit Reed-Solomon codes, the qudit stabilisers are essentially always supported on every qudit (they are length $n$), but their form can change. In the rare case that a stabiliser is not supported on every qudit, one can perform a simple modification of this procedure that uses strictly fewer gates, and so we make the pessimistic assumption that every qudit stabiliser is supported one every qudit; we will come back to this again later.

Before we say how we prepare this state fault-tolerantly, let us make clear what the aim of the fault-tolerant preparation is, that is, the notion of ``$t$-fault tolerance'' we are working with. We will show in Subsection~\ref{subsec:ler_single_cat} that our qudit cat state preparation satisfies this notion of $t$-fault tolerance, for a certain $t$.

\begin{definition}\label{def:ft_qudit_cat_property}
    The procedure for preparing the qudit cat state $\mathsf{Cat}(\gamma_1, \gamma_2, \ldots, \gamma_n)$ is $t$-fault-tolerant if, assuming there are some number $s \leq t$ faulty bicycle instructions in the generation of the cat state, there are errors on $\leq s$ qudits of the accepted qudit cat state.
\end{definition}
Let us now describe how to fault-tolerantly prepare and verify this state. Note that the low-level compilation into bicycle instructions will be described in the next subsubsection; for now, we work in terms of a higher-level gate set.

The high-level picture for fault-tolerantly preparing and verifying the qudit cat state will be to start with all our qudits in the qudit $\ket{+}$ state, and then to repeatedly measure and check the qudit stabilisers $Z_{i-1}^{\gamma_i}Z_i^{\gamma_{i-1}}$. To make progress, we thus need a fault-tolerant gadget to measure a weight-two qudit $Z$ operator. The gadget that fault-tolerantly measures the weight-two qudit $Z$ operator $Z^{\alpha}Z^{\beta}$, for non-zero $\alpha, \beta \in \mathbb{F}_q$, on qudits $i-1$ and $i$ is denoted $\left(Z_{i-1}^\alpha Z_i^\beta\right)_{\mathrm{FT}}$. We will describe this gadget shortly, but at this stage we can easily describe the whole qudit cat state preparation and verification.
\begin{enumerate}
    \item\label{step:qudit_cat_first_step} \textit{Initial $X$-Preparation:} Prepare all qudits in the cat state row in the qudit $\ket{+}$ state. This is the same thing as preparing all $ns$ non-pivot logical qubits of the cat state row gross codes in the qubit $\ket{+}$ state;
    \item\label{step:qudit_ZZ_meas_round} \textit{Non-Fault-Tolerant State Preparation:} For each $i = 2, \ldots, n$, measure $Z_{i-1}^{\gamma_i}Z_i^{\gamma_{i-1}}$. This need not be performed fault-tolerantly. Perform the appropriate frame updates so that (in the absence of faults) we hold the state $\mathsf{Cat}(\gamma_1, \gamma_2, \ldots, \gamma_n)$;
    \item\label{step:initial_cat_state_explanation_logical_post_select} \textit{Fault-Tolerant $Z$-Stabiliser Checking:} For each $i = 2, \ldots, n$, use the gadget $\left(Z_{i-1}^{\gamma_i}Z_i^{\gamma_{i-1}}\right)_{\mathsf{FT}}$ to fault-tolerantly measure the qudit operator $Z_{i-1}^{\gamma_i}Z_i^{\gamma_{i-1}}$ on the cat state row. If any non-trivial measurement outcome is obtained (suggesting the presence of some $X$ error), the cat state is rejected. This step is repeated, so that it is performed a total of $R$ times.
\end{enumerate}
Note that this fault-tolerant preparation and verification procedure for the qudit cat state is a natural analogue of something that one might do for a qubit cat state (it is similar to~\cite{delfosse2022fault}, with some differences): non-fault-tolerant preparation via measurements, followed by checking to detect the presence of (possibly high-weight) $X$ errors. The need to detect high-weight $X$ errors is that these can spread to high-weight errors in the code upon consumption of the cat state; we don't attempt to catch $Z$ errors since these do not spread to the data, but only cause measurement errors. While we will formally argue about the effect and spread of faults later, it is instructive for now to mention that if we only did steps~\ref{step:qudit_cat_first_step} and~\ref{step:qudit_ZZ_meas_round}, the cat state preparation would be non-fault-tolerant in a very bad sense, that is, one fault could result in an $X$ error of arbitrary weight. The reason is that that first round of measurements is non-deterministic even in the absence of faults, since we have no prior expectation of what the outcomes ``should'' be. Just as when doing the same scheme for preparing an $n$-qubit cat state, a single fault in the outcome of this measurement can correspond to an arbitrarily high-weight $X$ error, essentially because we perform the wrong frame update in software.

The whole procedure thus looks like preparation in the qudit $\ket{+}$ state, a non-fault-tolerant measurement of the cat state's $Z$ stabilisers, followed by $R$ rounds of fault-tolerant checking of the cat state's $Z$ stabilisers. Each of these $1+R$ rounds of measurement may be performed in two ``layers'', where the first layer simultaneously measures $Z_1^{\gamma_2}Z_2^{\gamma_1}, Z_3^{\gamma_4}Z_4^{\gamma_3}, \ldots$ and the second layer measures $Z_2^{\gamma_3}Z_3^{\gamma_2}, Z_4^{\gamma_5}Z_5^{\gamma_4}, \ldots$. An example of the fault-tolerant qudit cat state preparation and verification circuit is shown in Figure~\ref{fig:ft_qudit_cat_state_prep}.

\begin{figure}[ht]
    \centering
    \includegraphics[width=0.95\linewidth]{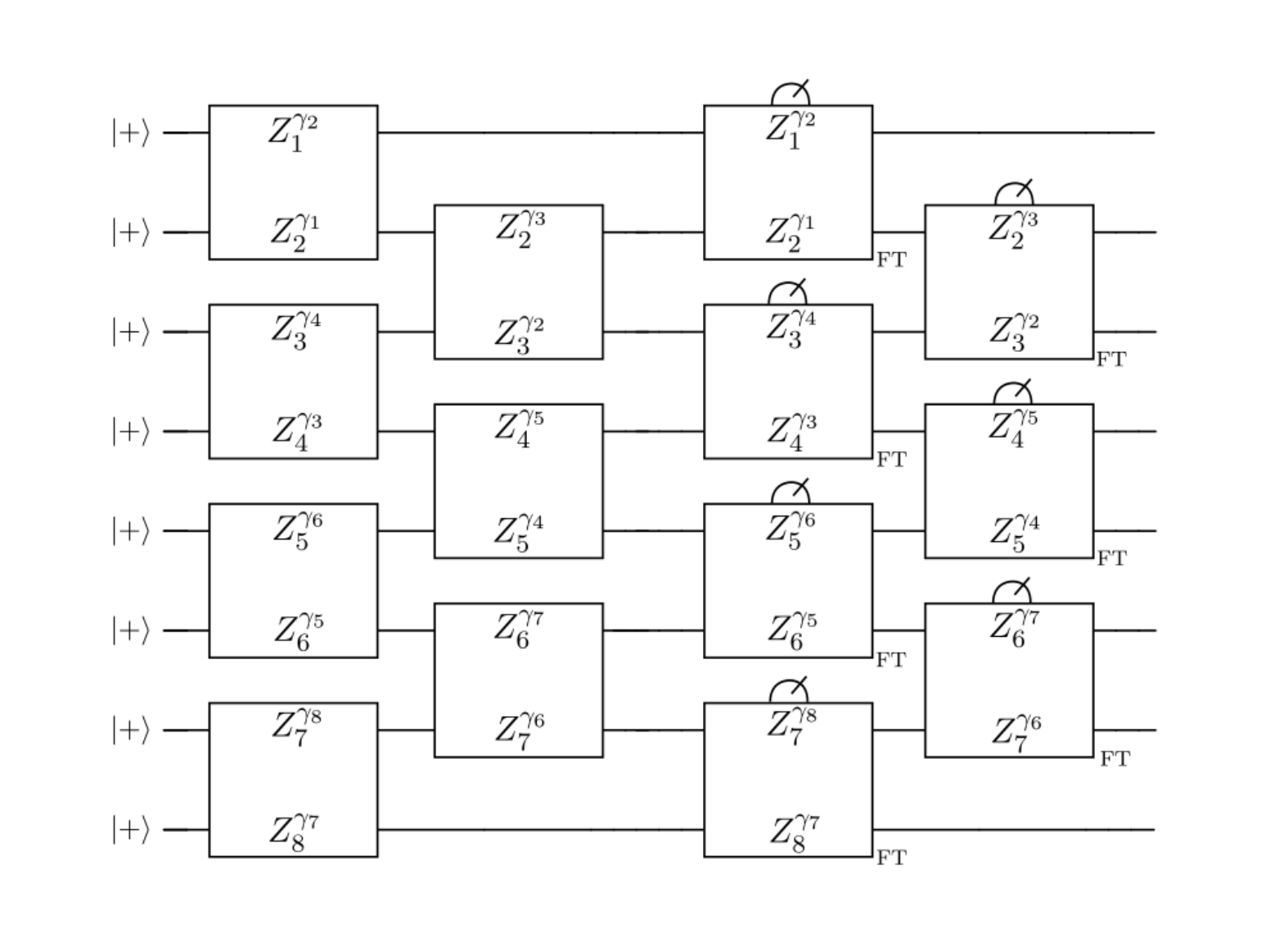}

    \caption{Depiction of the fault-tolerant preparation and verification of the qudit cat state for $n=8$ and $R=1$. Every wire denotes a qudit, which may be imagined as the $s$ non-pivot qubits of a cat state row gross code. The circuit proceeds via an initial non-fault-tolerant preparation, and then $R$ rounds of fault-tolerant checking. Each set of $n-1$ measurements can be executed in two ``layers''. In this diagram, all measurements are denoted by boxes. Those without meter symbols denote ``measurement projections'', that is, measurements followed by the appropriate frame updates to ensure a trivial outcome. Those with meter symbols are not accompanied by frame updates following them. In this cat state preparation, any non-trivial outcome on measurements with meter symbols leads to cat state rejection.}
    \label{fig:ft_qudit_cat_state_prep}
\end{figure}
The non-fault-tolerant qudit measurement $Z_{i-1}^\alpha Z_i^\beta$ may be directly performed between two gross code modules. However, the fault-tolerant qudit measurement gadget will require four qudits/gross codes to ensure fault tolerance. In particular, when the gadget $\left(Z_{i-1}^\alpha Z_i^\beta\right)_{\mathsf{FT}}$ is executed on two cat state row qudits, it will use the two adjacent qudits/gross codes in the ancilla row also. The protocol proceeds as follows, and depicted in Figure~\ref{fig:ft_qudit_ZZ_meas}.
\begin{enumerate}
    \item Prepare the two qudits in the ancilla row in the qudit $\ket{+}$ state;
    \item Measure the qudit operator $Z^\alpha Z^\beta$ (non-fault-tolerantly) on the two ancilla row qudits. Apply the appropriate frame updates so that the two ancilla row qudits are in the state $\mathsf{Cat}(\beta, \alpha)$ (in the absence of faults);
    \item Measure the qudit operator $ZZ$ on the two pairs of adjacent gross codes between the ancilla row and cat state row. Let the outcome on the pair labelled $i-1$ be $\eta_1 \in \mathbb{F}_q$ and the outcome on the pair labelled $i$ be $\eta_2 \in \mathbb{F}_q$. The sub-routine's measurement outcome is then $\alpha\eta_1 + \beta\eta_2$;
    \item Measure out the ancilla row gross codes in the $X$ basis and apply the appropriate frame updates to the cat state row.
\end{enumerate}
As before, the preparation of a qudit in the qudit $\ket{+}$ state is achieved by preparing all its constituent qubits in the qubit $\ket{+}$ state. Similarly, the measurement of a qudit in the $X$ basis may be achieved by measuring all that qudit's constituent qubits in the $X$ basis, see~\cite{wills2026review}. In~\cite{wills2026review}, it is also shown that the measurement of the qudit $ZZ$ operator between qudits with the same qudit-to-qubit mapping may be achieved by measuring $ZZ$ on each of their constituent qubits; note that because we pick the same qudit-to-qubit mappings across columns of our gross code grid, this is the case here. Finally, how to perform the measurement of the qudit $Z^{\alpha}Z^{\beta}$ operator will be discussed in the next subsubsection.

The fact that the fault-tolerant qudit measurement sub-routine does indeed perform the claimed measurement in the absence of faults can be checked with the use of qudit stabiliser tableaux~\cite{wills2026review}; we will also perform a very similar calculation later when we consume the cat state, and so we omit a similar calculation here.

\begin{figure}[ht]
    \centering

    \includegraphics[width=0.6\linewidth]{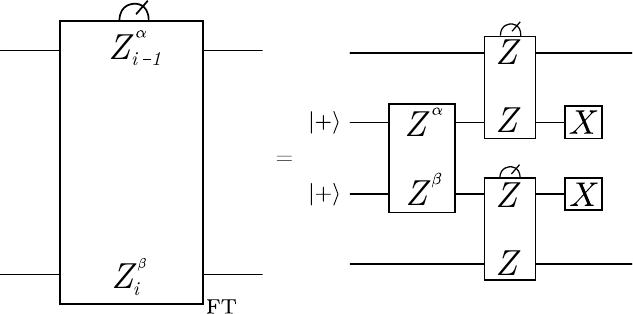}
    
    \caption{The fault-tolerant $Z^\alpha Z^\beta$ qudit measurement sub-routine acts on two cat state row qudits/gross codes (the first and fourth wires on the right-hand side), using the adjacent two ancilla row qudits/gross codes (the second and third wires) as ancillas. All operations depicted are qudit operations, in particular, boxes denote the measurement of the qudit operator by which they are labelled. Again, boxes without meter symbols denote measurement projections, meaning that the measurement is succeeded by the appropriate frame updates to ensure the outcome is trivial (in the absence of faults). 
    The outcome of this routine is taken to be $\alpha \eta_1 + \beta \eta_2$, where $\eta_1, \eta_2 \in \mathbb{F}_q$ are the outcomes of the upper and lower qudit $ZZ$ measurements on the right-hand side, respectively.}
    \label{fig:ft_qudit_ZZ_meas}
\end{figure}

\subsubsection{Compilation into Bicycle Instructions}

Let us now describe how to compile all of the above into the bicycle instructions~\cite{yoder2025tourgrossmodularquantum}. We begin with the easiest elements. First, as mentioned, the preparation/destructive measurement of a qudit in the $X$ basis may be executed via the preparation/destructive measurement of all its qubits in the $X$ basis~\cite{wills2026review}. In turn, because the gross code is a CSS code, all of its logical qubits may be prepared/destructively measured in the $X$ basis by preparing/destructively measuring all of its physical qubits in the $X$ basis~\cite{gottesman2024surviving}, and performing syndrome extraction in the former case. Of course, this will have the effect of preparing/destructively measuring the pivot qubit in the $X$ basis, as well as the $s$ non-pivot qubits forming the qudit, but this is not a problem.

The next easiest thing to compile into the bicycle instructions is the qudit $ZZ$ measurement between adjacent pairs of qudits in the ancilla row and cat state row. As mentioned, and shown in~\cite{wills2026review}, the measurement of $ZZ$ on two qudits with the same qudit-to-qubit mapping may be achieved by measuring $ZZ$ on all their constituent qubits pairwise (note this is $s$ qubit measurements). All qudits in the $i$'th column of our global gross code grid have the same qudit-to-qubit mapping, and so wherever we do the qudit $ZZ$ measurement, we are doing it between qudits with the same qudit-to-qubit mapping. 

The question then exactly becomes, how may one measure the $s$ qubit operators $Z_iZ_i$, for $i \in [s]$, on the $s$ non-pivot logical qubits in two adjacent gross codes? Well, the first thing to notice is that this is equivalent to measuring the $s$ qubit operators $Z^{v_i}Z^{v_i}$, where $(v_i)_{i=1}^s$ are $s$ linearly independent bit strings of length $s$, on this set of $2s$ qubits. The reason we are going to do this is that the bicycle instructions place different costs on different in-module measurements~\cite{yoder2025tourgrossmodularquantum}, and we want to only use the cheapest ``native'' in-module measurements, as we will describe shortly.

For a particular vector $v \in \mathbb{F}_2^s$, the compilation of the measurement $Z^vZ^v$ on two groups of $s$ non-pivot qubits in adjacent gross code modules is depicted in Figure~\ref{fig:qubit_meas}. This uses the pivot qubit in each gross code module as an ancilla. A single measurement of $Z^vZ^v$ between adjacent groups of $s$ qubits is compiled via a measurement projection of the respective pivot qubits in the $X$ basis, in-module measurements $ZZ^v$ on the $s+1$ qubits of both gross codes,\footnote{To be clear, given a bit string $v \in \mathbb{F}_2^s$, the operator $ZZ^v$ denotes the $(s+1)$-qubit operator with support $Z$ on the first qubit, and support $Z^v$ on the latter $s$ qubits.} followed by an inter-module $ZZ$ measurement between the pivot qubits of the gross code modules (which is a native bicycle instruction), and finally followed by a measurement projection of the pivot qubits in the $X$ basis. The desired outcome of the measurement may be found via the $\mathsf{XOR}$ of the three $Z$-type measurements in the compilation.
\begin{figure}[ht]
        \centering
        \includegraphics[width=0.6\linewidth]{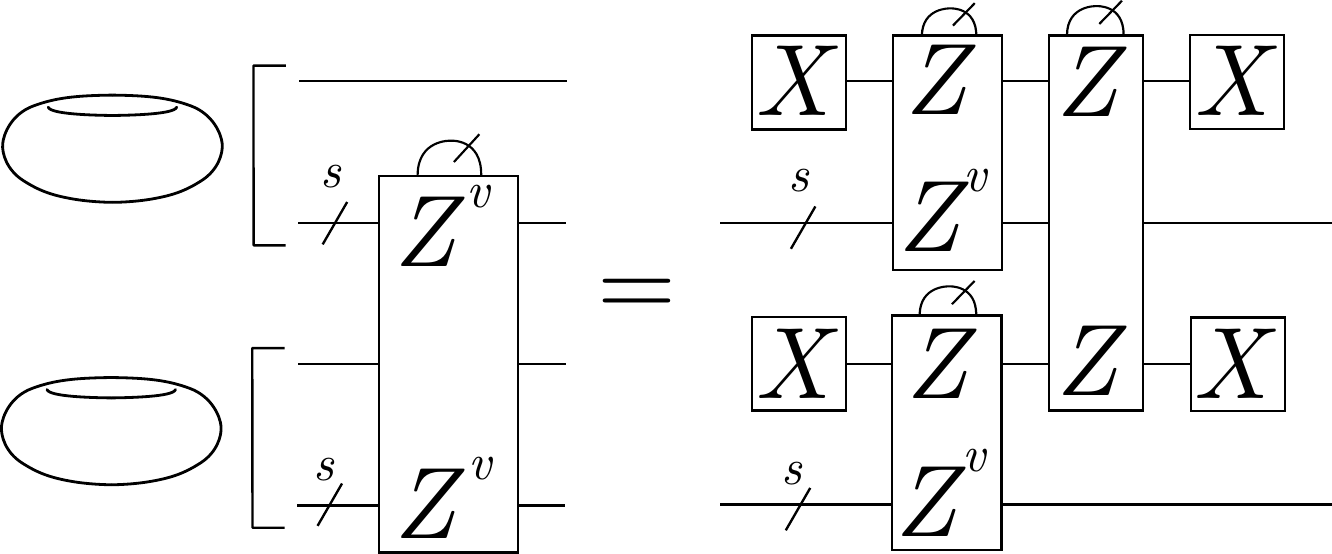}
        \caption{The compilation of the measurement $Z^vZ^v$, where $v \in \mathbb{F}_2^s$, on the $s$ non-pivot qubits in adjacent gross code modules. The pivot qubits in each module are used as ancillas. The measurement of $Z^vZ^v$ on the two groups of $s$ qubits may be inferred via the $\mathsf{XOR}$ of the three $Z$-type measurements on the right-hand side. Note that the $X$-type measurements are \textit{measurement projections}, that is, measurements followed by the appropriate frame updates to ensure a $+1$ outcome. In this step, we will exclusively use bit strings $v \in \mathbb{F}_2^s$ such that $ZZ^v$ is a native in-module measurement on the $s+1$ qubits of the gross code. Note that when performing several of these in sequence, which we want to, the measurement projection of the pivot qubit in the $X$ basis does not need to be performed twice.}
    \label{fig:qubit_meas}
\end{figure}

The only part of this compilation which may not necessarily be a native bicycle instruction is the in-module measurement of $ZZ^v$ on the two gross codes. However, we find that it is possible to find $s$ linearly independent bit strings of length $s$, $(v_i)_{i=1}^s$, such that the measurements $ZZ^{v_i}$ are all native, meaning that they can be executed with shift automorphisms (which are much cheaper than measurement instructions in terms of time and error rate), and one in-module measurement instruction. Even beyond this, some in-module measurement instructions in the bicycle architecture are noisier than others. Referring to the discussion in Appendix~\ref{sec:improved_ler_bicycle}, we divide the in-module instructions into ``half-LPU'' in-module measurements, and ``whole-LPU'' in-module measurements, where the latter are approximately an order of magnitude noisier than the former. We wish to perform our qudit $ZZ$ measurement not only using native instructions, but maximising the use of the half-LPU in-module instructions. Using the tools in~\cite{github_bicycle_architecture_compiler}, we find that the best one can do is to obtain $s$ linearly independent vectors $(v_i)_{i=1}^s$ such that $ZZ^{v_i}$ are all native measurements on the gross code, and such that $5$ of the corresponding in-module measurement instructions are half-LPU, and $6$ are whole-LPU; one example is
\begin{equation}\label{eq:native_Z_meas}
\setcounter{MaxMatrixCols}{11}
\left[
\begin{array}{ccccccccccc}
0&0&0&0&0&1&0&0&0&0&0\\
0&0&0&0&1&1&0&0&0&0&1\\
0&0&1&0&1&1&0&0&1&0&1\\
0&0&1&1&1&1&0&0&1&1&1\\
0&1&0&0&1&1&0&1&0&0&1\\
1&0&0&1&0&1&1&0&0&1&0\\
0&0&0&0&1&0&0&0&0&0&0\\
0&0&1&0&1&0&0&0&0&0&0\\
0&0&1&1&1&0&0&0&0&0&0\\
0&1&0&0&1&0&0&0&0&0&0\\
1&0&0&1&0&0&0&0&0&0&0
\end{array}
\right]
\begin{array}{l}
\rdelim\}{6}{*}[\text{ Whole-LPU}]\\[-0.2ex]
\\
\\
\\
\\
\\
\rdelim\}{5}{*}[\text{ Half-LPU}]\\[-0.2ex]
\\
\\
\\
\\
\end{array}
.
\end{equation}
Here, rows correspond to bit strings $v_i$ such that $ZZ^{v_i}$ is a native measurement in the bicycle instructions. The first $6$ rows correspond to native measurements for which the in-module instruction is a whole-LPU measurement, whereas the latter $5$ correspond to half-LPU measurements.

To assess the cost of the qudit $ZZ$ measurements, we find that $s$ measurements are made between gross codes of the form $Z^vZ^v$, each of which requires the pivot qubit to be reset and recycled in the $X$ basis. However, we note that the pivot qubit does not need to be set in the $X$ basis if it has already been done. Thus, the $X$ measurement does not need to be performed twice between $Z^vZ^v$ measurements. Noting that $X_1$ is a half-LPU in-module instruction, we find that the total number of gates required for a qudit $ZZ$ measurement is $11$ inter-module measurement instructions, $34$ half-LPU in-module measurement instructions, and $12$ whole-LPU in-module measurement instructions. Moreover, the time taken to perform a qudit $ZZ$ measurement is that of $34$ measurement instructions (noting that all measurement instructions take the same time, namely $120$ timesteps, see Table 2 of~\cite{yoder2025tourgrossmodularquantum}).\footnote{Parallelising between the two gross codes as in Figure~\ref{fig:qubit_meas}, a qudit $ZZ$ measurement takes the time to perform $11$ inter-module measurements, $12$ measurement projections of the pivot in the $X$ basis, and $11$ in-module measurements of the form $ZZ^v$.} Note that one may reach the exact same conclusions for the qudit $XX$ measurement which is needed during the extraction of a syndrome of a qudit $Z$ check. That is, the qudit $ZZ$ and qudit $XX$ measurements behave exactly the same.

Finally, let us consider the (non-fault-tolerant) qudit $Z^\alpha Z^\beta$ measurement on adjacent qudits/gross codes. Given a particular $\alpha$ and $\beta$, and a pair of qudit-to-qubit mappings, this qudit measurement may be achieved on the two groups of $s$ qubits as some collection of $s$ measurements $Z^{x_i}Z^{y_i}$, where $(x_i)_{i=1}^s$ and $(y_i)_{i=1}^s$ are each collections of $s$ linearly independent bit strings in $\mathbb{F}_2^s$ (one may check this using the methods in~\cite{wills2026review}). We collect these two sets of vectors into two $s \times s$ (invertible) binary matrices $X$ and $Y$, with rows $x_i$ and $y_i$, respectively. Then, in a similar way to how we treated the qudit $ZZ$ measurement, instead of measuring the $Z$ operators corresponding to the rows of $\begin{pmatrix}
    X & Y
\end{pmatrix}$, it is equivalent to measure the $Z$ operators corresponding to the rows of $\begin{pmatrix}
    AX & AY
\end{pmatrix}$, where $A$ is any invertible $s \times s$ matrix over $\mathbb{F}_2$. Indeed, given the previous discussion, we can always choose a matrix $A$ that makes $AX$ equal to the matrix shown in Equation~\eqref{eq:native_Z_meas}. This will allow us to maximise the use of native operations on the first gross code. However, the measurements on the second gross code will require further compilation. In Appendix~\ref{sec:compile_Z_meas}, we will describe an efficient method to compile these operations, and we describe how we estimate the cost of these sequences of measurements.

\subsubsection{Post-Selected Algorithm for Fault-Tolerant Qudit Cat State Preparation}

We are close to being able to state our algorithm for fault-tolerant qudit cat state preparation and verification formally. Before we can do so, one more ingredient must be stated: that of post-selection.

Our cat state is an ancillary state that is being prepared offline from the data block. Thus, at any point, if we are suspicious of it, we can throw it away (post-select it) and attempt to remake it, with the aim of having an increased error rate on surviving operations, with a manageable post-selection overhead. Notice that we have already done some post-selection, in that when we check the qudit measurement $Z_{i-1}^{\gamma_i}Z_i^{\gamma_{i-1}}$, obtaining any non-trivial measurement outcome causes us to reject/post-select the cat state, and re-attempt its creation (see Step~\ref{step:initial_cat_state_explanation_logical_post_select} previously). We are going to employ some further post-selection, however. That is, the decoder for our inner code (the gross code) attempts to decode every instruction (every in-module measurement instruction, and so on). One can develop ``rules'', called \textit{post-selection strategies}, such that, if a particular decoding instance looks suspicious, we can post-select that operation. The hope is that, with a relatively small amount of post-selection, a substantial improvement in the logical error rate of surviving shots may be obtained.

Developing such post-selection strategies for high-rate qLDPC codes such as the gross code has been an active area recently~\cite{lee2025efficient,xie2026simple}, although we choose to use the strategies from the companion paper~\cite{wills2026forced}. We will specify our exact post-selection strategies used later on (see Subsection~\ref{subsec:ler_single_cat}); for now, we may simply leave the post-selection strategies as a parameter of the protocol, $\mathcal{C}$. We note that we will actually use three such strategies: $\mathcal{C}_1$, $\mathcal{C}_2$ and $\mathcal{C}_3$. The reason for this is that, generically, with such post-selection strategies, a more stringent rule (rejecting more shots) leads to a lower logical error rate of surviving shots, but a greater post-selection overhead. On the other hand, more lenient rules (rejecting fewer shots) lead to a smaller decrease in the logical error rate of surviving shots, but with a more manageable post-selection overhead. In our cat state preparation protocol, there are some operations which, if post-selected, the entire cat state must be recreated. On the other hand, there are some operations that occur offline, even from the cat state, and so if they are post-selected, they may be repeated in place without affecting much else. A key example is the non-fault-tolerant qudit $Z^\alpha Z^\beta$ measurement occurring on the ancilla row as part of the fault-tolerant $Z^\alpha Z^\beta$ measurement sub-routine. If this is post-selected, we don't need to throw away the whole cat state, because post-selection here does not affect the whole cat state; the operation may be retried in place. On the other hand, if one of the qudit $ZZ$ measurements in the fault-tolerant measurement sub-routine is post-selected, the whole cat state is broken, and we must retry the creation of the cat state from the start. In order to get the best error rate within a reasonable post-selection overhead, we therefore pick more stringent post-selection rules for operations whose post-selection does not affect much, and more lenient rules for operations whose post-selection leads to the re-attempting of the entire cat state. In particular, we think of the post-selection condition $\mathcal{C}_2$ as more stringent, and $\mathcal{C}_1$ and $\mathcal{C}_3$ as less so.

We conclude this subsection by stating our fault-tolerant qudit cat state preparation and verification protocol formally in Sub-Routine~\ref{alg:qudit_cat_FT_prep}, which further uses Sub-Routines~\ref{alg:ft_qudit_ZZ_meas} and~\ref{alg:qudit_ZZ_meas}.

\begin{subroutine}[H]
\caption{Fault-Tolerant Qudit Cat State Preparation and Verification}
\label{alg:qudit_cat_FT_prep}
\begin{algorithmic}[1]

\Require Rectangular grid of $2 \times n$ gross codes formed of a ``cat state row'' and ``ancilla row''; non-zero field elements $\gamma_1, \gamma_2, \ldots, \gamma_n \in \mathbb{F}_q$; $R \in \mathbb{N}$ specifying the desired degree of fault tolerance; post-selection conditions $\mathcal{C}_1$, $\mathcal{C}_2$ and $\mathcal{C}_3$.
\Ensure $\mathsf{Cat}(\gamma_1, \gamma_2, \ldots, \gamma_n)$ prepared fault-tolerantly on the cat state row.

\State\label{step:initial+prep} Prepare all logical qubits in the cat state row gross codes in $\ket{+}$, by preparing all the gross codes' physical qubits in $\ket{+}$ and performing syndrome extraction. This has the effect of preparing each of the qudits in that row in the qudit $\ket{+}$ state, as well as the pivot qubits in that row in the qubit $\ket{+}$ state

\For{each even $i \in \{1, \ldots, n\}$}
    \State\label{step:first_meas_layer} Measure $Z_{i-1}^{\gamma_i}Z_i^{\gamma_{i-1}}$ non-fault-tolerantly on the cat state row gross codes $i-1$ and $i$ with post-selection condition $\mathcal{C}_1$. This is achieved using the prescription of Appendix~\ref{sec:compile_Z_meas}
    \If{post-selection occurs}
        \State Repeat steps~\ref{step:initial+prep} and~\ref{step:first_meas_layer} on this pair
    \EndIf
\EndFor

\For{each odd $i \in \{3, \ldots, n-1\}$}
    \State Measure $Z_{i-1}^{\gamma_i}Z_i^{\gamma_{i-1}}$ non-fault-tolerantly on the cat state row gross codes $i-1$ and $i$ with post-selection condition $\mathcal{C}_3$. This is achieved using the prescription of Appendix~\ref{sec:compile_Z_meas}
    \If{post-selection occurs}
        \State Repeat the entire procedure
    \EndIf
\EndFor
\State Apply the appropriate frame updates so that we hold the state $\mathsf{Cat}(\gamma_1, \gamma_2, \ldots, \gamma_n)$ (in the absence of faults)
\For{each even $i \in \{1, \ldots, n\}$}\label{step:first_FT_meas_layer}
    \State Measure $Z_{i-1}^{\gamma_i}Z_i^{\gamma_{i-1}}$ fault-tolerantly on the cat state row gross codes $i-1$ and $i$ using Sub-Routine~\ref{alg:ft_qudit_ZZ_meas} with post-selection conditions $\mathcal{C}_1$ and $\mathcal{C}_2$
    \If{a post-selection occurs under condition $\mathcal{C}_2$, or any non-trivial outcome is obtained}
        \State Repeat the entire procedure
    \EndIf
\EndFor

\For{each odd $i \in \{3, \ldots, n-1\}$}\label{step:second_FT_Meas_layer}
    \State Measure $Z_{i-1}^{\gamma_i}Z_i^{\gamma_{i-1}}$ fault-tolerantly on the cat state row gross codes $i-1$ and $i$ using Sub-Routine~\ref{alg:ft_qudit_ZZ_meas} with post-selection conditions $\mathcal{C}_1$ and $\mathcal{C}_2$
    \If{a post-selection occurs under condition $\mathcal{C}_2$}
        \State Repeat the entire procedure
    \EndIf
\EndFor

\For{$r = 1$ to $R-1$}
    \State Repeat steps~\ref{step:first_FT_meas_layer} and~\ref{step:second_FT_Meas_layer}
\EndFor

\end{algorithmic}
\end{subroutine}

\begin{subroutine}[H]
\caption{Fault-Tolerant Qudit $Z^\alpha Z^\beta$ Measurement in the Bicycle Instructions}
\label{alg:ft_qudit_ZZ_meas}
\begin{algorithmic}[1]
\Require Four adjacent gross codes in a square connectivity, two adjacent modules called ``cat state row gross codes'', the other two called ``ancilla row gross codes''; post-selection conditions $\mathcal{C}_1$ and $\mathcal{C}_2$.
\Ensure The measurement $Z^\alpha Z^\beta$ performed fault-tolerantly on the cat state row gross codes; failure if post-selection occurs under the condition $\mathcal{C}_2$.

\State Initialise all logical qubits in the ancilla row gross codes to $\ket{+}$ by setting all physical qubits to $\ket{+}$ and performing syndrome extraction\label{step:ft_meas_initialisation}
\State Measure the qudit operator $Z^\alpha Z^\beta$ using the prescription in Appendix~\ref{sec:compile_Z_meas} with post-selection condition $\mathcal{C}_1$\label{step:ft_meas_first_meas}
\While{a post-selection occurs}
\State Repeat Steps~\ref{step:ft_meas_initialisation} and~\ref{step:ft_meas_first_meas}
\EndWhile
\State Apply the appropriate frame updates so that the two ancilla row gross codes/qudits are in the state $\mathsf{Cat}(\beta, \alpha)$ (in the absence of faults)
\State\label{step:FT_meas_on_cat} Measure the qudit $ZZ$ operators between  both adjacent pairs of cat state row and ancilla row gross codes using Sub-Routine~\ref{alg:qudit_ZZ_meas} with post-selection condition $\mathcal{C}_2$.
\If{either measurement fails with post-selection condition $\mathcal{C}_2$}
\State Declare failure of the measurement and terminate
\EndIf
\State If the outcomes of the latter two measurements on the first and second pair are $\eta_1, \eta_2 \in \mathbb{F}_q$, respectively, sub-routine measurement outcome is $\alpha \eta_1 + \beta\eta_2$
\State Measure all logical qubits in the ancilla row gross codes in the $X$ basis by measuring all their physical qubits in the $X$ basis. Apply the appropriate frame updates to the cat state row gross codes

\end{algorithmic}
\end{subroutine}

\begin{subroutine}[H]
\caption{Qudit $ZZ$ Measurement in the Bicycle Instructions}
\label{alg:qudit_ZZ_meas}
\begin{algorithmic}[1]
\Require Two adjacent gross code modules; gross code modules have the same qudit-to-qubit mapping; post-selection condition $\mathcal{C}$; $s$ linearly independent vectors $v^{(j)} \in \mathbb{F}_2^s$ such that $ZZ^{v^{(j)}}$ are native measurements in the bicycle instructions on the gross code.
\Ensure The measurement of the qudit $ZZ$ operator performed on the gross codes non-fault-tolerantly; failure if post-selection occurs.
\State Reset pivot qubits to $\ket{+}$
\For{$j = 1$ to $s$}
    \State Perform the native in-module measurement $ZZ^{v^{(j)}}$ on the first gross code module\label{step:first_Zvj_meas}
    \State Perform the native in-module measurement $ZZ^{v^{(j)}}$ on the second gross code module
    \State Measure $ZZ$ between the pivot qubits of the two gross code modules\label{step:inter_mod_ZZ_meas}
    \State Perform the in-module measurement of $X$ on the pivot qubits, and update the Pauli frame
    \If{any outcome in this round satisfies the post-selection condition $\mathcal{C}$}
        \State Declare failure of the measurement and terminate
    \EndIf
\EndFor

\State Calculate the measurement outcome in terms of a single $\mathbb{F}_q$ symbol from the obtained $s$ bits (see~\cite{wills2026review} for the appropriate formalism)

\end{algorithmic}
\end{subroutine}

\subsection{Syndrome Extraction Protocol for Multiple Outer Code Blocks}

Given the above, we now move towards describing how we operate the entire memory system. At a high level, the two ancilla rows will cycle through the $m$ rows, performing syndrome extraction on blocks of the outer code as they do (the idea to use one ancilla system tending to the syndrome extraction of multiple blocks of the outer code is directly inspired by~\cite{gidney2025yoked}). In this subsection, we will go on to describe how one qudit cat state may be consumed to measure one qudit stabiliser, how we ensure resistance against measurement errors by measuring an overcomplete basis of the $X$ stabilisers (the same is done for the $Z$ stabilisers), and finally the exact means by which we can cycle the ancilla rows through the outer code blocks to enable the syndrome extraction.

\subsubsection{Measuring Individual Qudit Stabilisers with Qudit Cat States}

For the purposes of this subsubsection, we assume that we have a cat state prepared (fault-tolerantly) using the methods described in the previous subsection, and it is adjacent to a row containing a block of the outer code. We assume that the prepared cat state is $\mathsf{Cat}(\gamma_1, \gamma_2, \ldots, \gamma_n)$, where $\gamma_i \in \mathbb{F}_q$, $\gamma_i \neq 0$ and, accordingly, the qudit stabiliser whose syndrome we wish to extract is $X_1^{\gamma_1}X_2^{\gamma_2}\ldots X_n^{\gamma_n}$. Note that, in the qubit language, this will extract the syndrome of $s$ $X$-type qubit stabilisers simultaneously.

Before we begin, let us comment on one detail. By taking $\gamma_i \neq 0$, we have implicitly assumed that the $X$-type stabiliser we are measuring has non-trivial support on every qudit, that is, every set of non-pivot qubits in the gross codes in the codeblock. It is possible that we want to measure a check that does not have this property: that it is supported on strictly fewer than $n$ qudits. In this case, we can simply prepare some state $\mathsf{Cat}(\gamma_1, \gamma_2, \ldots, \gamma_C)$ for some $C < n$, and where we still have $\gamma_i \neq 0$, on the cat state row gross codes adjacent to the appropriate gross codes of the data row. This may be handled by a simple modification of the procedure described in Section~\ref{subsec:ft_qudit_cat_prep}. Because this would use strictly fewer gates than the case of a check supported on every qudit, our treating only the case of checks supported on every qudit constitutes a conservative assumption that keeps the presentation as simple as possible. Moreover, making this simplification essentially does not lose us anything, because the checks of a quantum Reed-Solomon code are essentially maximum weight.

In order to consume the state $\mathsf{Cat}(\gamma_1, \gamma_2, \ldots, \gamma_n)$ to measure the qudit stabiliser $X_1^{\gamma_1}X_2^{\gamma_2}\ldots X_n^{\gamma_n}$ on the data block, the procedure is entirely analogous to the way one would do this using measurements on qubits (where the regular $n$-qubit cat state may be consumed to measure the stabiliser $X^{\otimes n}$). Indeed, in that regular qubit scenario, one may consume the cat state by measuring the qubit $XX$ operator $n$ times between each adjacent pair of qubits on the cat state and data block, and obtain the outcome of the $X^{\otimes n}$ measurement by adding up these $n$ measurement outcomes in binary (that is, over the field $\mathbb{F}_2$). The $n$ qubits that were formerly in the cat state may then be unentangled from the data block by measuring them out in the $Z$ basis, and depending on these measurement outcomes we may apply the appropriate frame update to the data block. In doing this, any $X$ errors that were on the cat state spread to the data block, and any $Z$ errors on the cat state induce measurement errors.

The procedure for qudits is entirely analogous. We begin by measuring the qudit $XX$ operator between adjacent qudits in the cat state and the data block. Each of these $n$ measurements outputs some $\eta_i \in \mathbb{F}_q$; the outcome of the measurement of the stabiliser $X_1^{\gamma_1}X_2^{\gamma_2}\ldots X_n^{\gamma_n}$ is then calculated as $\sum_{i=1}^n\gamma_i\eta_i$. The qudits that were formerly in the qudit cat state are then unentangled from the data block by measuring them out in the $Z$ basis, and the appropriate frame updates are applied to the data block. A proof that this indeed measures the qudit stabiliser $X_1^{\gamma_1}X_2^{\gamma_2}\ldots X_n^{\gamma_n}$ is shown shortly, but deferred until after a discussion of the compilation of this stage into the bicycle instructions.

Compilation of this stage into the bicycle instructions is straightforward. Indeed, the qudit $XX$ measurements occur between gross codes with the same qudit-to-qubit mapping, because the qudit-to-qubit mappings are chosen to be the same along all columns of our global $m \times n$ grid of gross codes. Thus, the qudit $XX$ measurements may be executed with Sub-Routine~\ref{alg:qudit_ZZ_meas}, except for an $XX$ measurement rather than a $ZZ$ measurement. These measurements are executed with no post-selection condition, because there is direct interaction with the data block, and so operations cannot be post-selected.

Now, the final operation to compile here into the bicycle instructions is that of measuring out the qudits in the cat state row in the $Z$ basis. Measuring a qudit in the $Z$ basis may be achieved by measuring out all of its constituent qubits in the $Z$ basis (for any qudit-to-qubit mapping)~\cite{wills2026review}, that is, we wish to measure all of the non-pivot logical qubits in that gross code module. Again, since the gross code is a CSS code, that may in turn be achieved by measuring out all of its physical qubits in the $Z$ basis. This will also measure out the pivot qubit in the $Z$ basis, but we do not mind.

Now that we have compiled these steps into the bicycle instructions, let us now prove that our measurement procedure consuming the state $\mathsf{Cat}(\gamma_1, \gamma_2, \ldots, \gamma_n)$ does indeed measure the qudit stabiliser $X_1^{\gamma_1}X_2^{\gamma_2}\ldots X_n^{\gamma_n}$. Doing so is done most conveniently using qudit stabiliser tableaux. This explanation will assume familiarity with~\cite{wills2026review}. For now, we describe the procedure in the absence of faults; the propagation of faults will be considered in Section~\ref{sec:ler_calc}. We simply assume here that some $Z$ error has occurred on the data block such that the syndrome component of the qudit stabiliser $X_1^{\gamma_1}X_2^{\gamma_2}\ldots X_n^{\gamma_n}$ on the data block is some $\eta \in \mathbb{F}_q$; we are going to show how to extract the value $\eta$. At the beginning of this measurement, the system of two rows has stabiliser tableau
\begin{equation}
\left(
\begin{array}{cccccc|cccccc}
X^{\gamma_1} & X^{\gamma_2} & X^{\gamma_3} & \cdots & X^{\gamma_{n-1}} & X^{\gamma_n}
  & & & & & \\
Z^{\gamma_2} & Z^{\gamma_1} &  &  &  &
  & & & & & \\
& Z^{\gamma_3} & Z^{\gamma_2} &  &  &
  & & & & & \\
& & & \ddots &  &
  & & & & & \\
& & & & Z^{\gamma_n} & Z^{\gamma_{n-1}} & 
  & & & & & \\ \hline
& & & & & &
  X^{\gamma_1} & X^{\gamma_2} & X^{\gamma_3} & \cdots & X^{\gamma_{n-1}} & X^{\gamma_n} \\
& & & & & &
  \multicolumn{6}{c}{\hat{\mathcal{L}}_X} \\
& & & & & &
  \multicolumn{6}{c}{\hat{\mathcal{L}}_Z}
\end{array}
\right), 
\begin{pmatrix}
    0\\0\\0\\\vdots\\0\\\eta\\0\\0
\end{pmatrix},
\end{equation}
where this matrix and vector shows the stabiliser operators and their syndrome on this state, as in~\cite{wills2026review}. In the matrix, the left portion corresponds to the cat state row, whereas the right portion corresponds to the data row. There, the characters $\hat{\mathcal{L}}_X$ and $\hat{\mathcal{L}}_Z$ are used to simply denote bases of the $X$-type qudit stabilisers (other than $X_1^{\gamma_1}X_2^{\gamma_2}\ldots X_n^{\gamma_n}$) and of the $Z$-type qudit stabilisers. One can also keep track of the evolution of the logical information with bases of qudit $X$-type and $Z$-type logical operators: $\hat{\mathcal{L}}_{X,\text{ Log}}$ and $\hat{\mathcal{L}}_{Z,\text{ Log}}$.

When the qudit $XX$ measurements are performed on every qudit between the data block and the cat state, supposing the outcomes are $(\eta_i)_{i=1}^n$, the stabiliser tableau becomes

\begin{equation}
\left(
\begin{array}{cccccc|cccccc}
X^{\gamma_1} & X^{\gamma_2} & X^{\gamma_3} & \cdots & X^{\gamma_{n-1}} & X^{\gamma_n}
  & & & & & \\
X & & & & &
& X & & & & & \\
& X & &  &  &
  & &X & & & \\
& &  & \ddots &  &
  & & & \ddots& & \\
& & & & X &  & 
  & & & &X & \\ \hline
& & & & & &
  X^{\gamma_1} & X^{\gamma_2} & X^{\gamma_3} & \cdots & X^{\gamma_{n-1}} & X^{\gamma_n} \\
& & & & & &
  \multicolumn{6}{c}{\hat{\mathcal{L}}_X} \\
\multicolumn{6}{c|}{\hat{\mathcal{L}}'_Z}&
  \multicolumn{6}{c}{\hat{\mathcal{L}}_Z}
\end{array}
\right), 
\begin{pmatrix}
    0\\\eta_1\\\eta_2\\\vdots\\\eta_{n-1}\\\eta\\0\\0
\end{pmatrix},
\end{equation}
where one out of the $n$ measurements is deterministic, such that
\begin{equation}\label{eq:syndrome_extraction_formula}
    \eta = \sum_{i=1}^n\gamma_i\eta_i.
\end{equation}
One can check that this equation is true using qudit stabiliser tableaux update rules, but the example is also developed explicitly in~\cite{wills2026review}.
Notice that the qudit $Z$-type stabilisers have deformed to have some support on the cat state row, such that all entries in the stabiliser tableau always commute; this support is denoted by some rows $\hat{\mathcal{L}}'_Z$. The $Z$-type logical operators will also deform to the cat state row similarly; one could denote the support of the $Z$-type logical operators on the cat state row by some $\hat{\mathcal{L}}'_{Z,\text{ Log}}$.

Now, in the absence of faults, we have extracted the value $\eta$ by the formula in Equation~\eqref{eq:syndrome_extraction_formula}, and it remains to unentangle the cat state row from the data row. Measuring out all of the cat state row qudits in the $Z$ basis gives us
\begin{equation}
\left(
\begin{array}{ccccc|cccccc}
Z &  &  &  &   
  & & & & & \\
 & Z& & & 
&  & & & & & \\
&  & Z&  &  
  & & & & & \\
& &  & \ddots  &
  & & & & & \\
& & & &   Z& 
  & & & & & \\ \hline
& & & & & 
  X^{\gamma_1} & X^{\gamma_2} & X^{\gamma_3} & \cdots & X^{\gamma_{n-1}} & X^{\gamma_n} \\
& & & & &
  \multicolumn{6}{c}{\hat{\mathcal{L}}_X} \\
\multicolumn{5}{c|}{\hat{\mathcal{L}}'_Z}&
  \multicolumn{6}{c}{\hat{\mathcal{L}}_Z}
\end{array}
\right), 
\begin{pmatrix}
    \nu_1\\\nu_2\\\nu_3\\\vdots\\\nu_n\\\eta\\0\\0
\end{pmatrix},
\end{equation}
for some values $(\nu_i)_{i=1}^n$, for $\nu_i \in \mathbb{F}_q$. Note that the single qudit $Z$ stabilisers on the cat state row (in the upper left quadrant) may be used to clean the support of the $Z$-type qudit stabilisers and logical operators from the cat state row. Doing so gives them some generally non-trivial syndrome depending on the values $(\nu_i)_{i=1}^n$, but we can correct this with a frame update. Note that in performing this frame update, $X$ errors that were on the cat state spread to the data block; we handle the spreading of faults in Section~\ref{sec:ler_calc}.

To conclude this subsubsection, we formally state our protocol for consuming a (fault-tolerantly) prepared qudit cat state to measure an individual $X$-type qudit stabiliser in Sub-Routine~\ref{alg:single_qudit_stab_meas}, where the protocol for a $Z$-type qudit stabiliser is entirely analogous.

\begin{subroutine}[H]
\caption{Measurement of an Individual Qudit Stabiliser Using a Cat State}
\label{alg:single_qudit_stab_meas}
\begin{algorithmic}[1]
\Require A data row adjacent to a cat state row; cat state row has been prepared in the state $\mathsf{Cat}(\gamma_1, \gamma_2, \ldots, \gamma_n)$ fault-tolerantly using, for example, Sub-Routine~\ref{alg:qudit_cat_FT_prep}.
\Ensure The data row with the stabiliser $X_1^{\gamma_1}X_2^{\gamma_2}\ldots X_n^{\gamma_n}$ measured.
\State Measure the qudit $XX$ operators between each adjacent qudit in the data row and the cat state row, recording outcomes $(\eta_i)_{i=1}^n$. This may be achieved using (the dual of) Sub-Routine~\ref{alg:qudit_ZZ_meas} with no post-selection condition
\State Stabiliser measurement outcome is obtained by calculating $\sum_{i=1}^n\gamma_i\eta_i$
\State Measure out every qudit in the cat state row in the $Z$ basis; apply appropriate frame updates to the data block

\end{algorithmic}
\end{subroutine}

\subsubsection{Measuring an Overcomplete Basis of Stabilisers of One Outer Code Block}\label{subsubsec:meas_overcomplete}

To ensure resistance against measurement errors, we will measure more than the minimum number of $d-1$ $X$ checks, and $d-1$ $Z$ checks, of the quantum Reed-Solomon code. We wish to measure an overcomplete basis of each space of checks. Note that, here, we only describe the full set of stabilisers measured on one data block, and what we will do when a measurement error is detected. The actual analysis of the effect of faults takes place in Section~\ref{sec:ler_calc}.

Generally, in fault tolerance, a canonical way to handle the possibility of measurement errors is to repeat the measurement of each syndrome to gain confidence in its outcome. To save on the number of qudit measurements we must make, we instead measure an \textit{overcomplete basis} of the qudit stabiliser of $X$ type, and of $Z$ type, and check that our measurement outcomes satisfy the linear combinations that they should.

More explicitly, we do the following. Let $H_X \in \mathbb{F}_q^{(d-1)\times n}$ be a parity-check matrix, whose rows span the space of $X$-checks of our code. We also pick some matrix of values $\beta^{(X)} \in \mathbb{F}_q^{M \times (d-1)}$ (we will describe how to make a good choice in Section~\ref{sec:ler_calc}), where $M$ is some integer at most $d-2$. Then, to extract the syndrome corresponding to the $X$ checks of the code, we start by measuring the qudit $X$ checks corresponding to the rows of $H_X$, from top to bottom, and then measure the $X$ checks corresponding to the rows of $\beta^{(X)}\cdot H_X$, from top to bottom. This means that, in extracting the syndrome of one outer code block, we measure $d-1+M$ $X$ checks, assuming nothing goes wrong.

When we measure these $d-1+M$ checks, we first extract some syndrome vector $y \in \mathbb{F}_q^{d-1}$, and then we extract some syndrome vector $z \in \mathbb{F}_q^M$. In the absence of faults, these would satisfy $z = \beta^{(X)}\cdot y$. As we measure the latter $M$ checks, we check that the required equations hold, and if we detect any violation of them, we begin all $d-1+M$ checks again. If we complete the $d-1+M$ checks with all the required equations satisfied, that round of $X$ checks is called a \textit{self-consistent round of $X$ checks}, and the corresponding syndrome $y \in \mathbb{F}_q^{d-1}$ gets passed to the outer code decoder. Subsection~\ref{subsec:time_like_failures} will upper bound the probability of \textit{time-like failure}, which is the failure mode of our outer code which occurs when (roughly --- we will be more precise about this later), the outer code decoder was handed the wrong syndrome, even though the round was self-consistent.

The whole procedure is performed in the same way for the $Z$ checks. That is, we measure the $d-1$ usual $Z$ checks specified by the rows of some parity-check matrix $H_Z \in \mathbb{F}_q^{(d-1)\times n}$, followed by $M$ linear combinations given by the rows of $\beta^{(Z)}\cdot H_Z$, where $\beta^{(Z)} \in \mathbb{F}_q^{M \times (d-1)}$. Consistency is checked on each of these linear combinations as we go. Any violation detected leads to a re-attempt of the round of $d-1+M$ $Z$ checks. Getting to the end of the $d-1+M$ $Z$ checks with all equations satisfied allows us to call those $Z$ checks a self-consistent round. We will describe in Subsection~\ref{subsec:time_like_failures} how good choices of $\beta^{(X)}$ and $\beta^{(Z)}$ can be made.

We formally state our algorithm for the extraction of the syndrome of one outer code block in Algorithm~\ref{alg:ft_extraction_one_block}.

\begin{algorithm}[ht]
\caption{Fault-Tolerant Extraction of Stabilisers on One Block of the Outer Code}
\label{alg:ft_extraction_one_block}
\begin{algorithmic}[1]
\Require A data row, cat state row, and ancilla row (in this order) adjacent to each other; $X$ and $Z$ parity-check matrix $H_X, H_Z \in \mathbb{F}_q^{(d-1)\times n}$, matrices $\beta^{(X)} \in \mathbb{F}_q^{M \times (d-1)}$ and $\beta^{(Z)} \in \mathbb{F}_q^{M \times (d-1)}$; post-selection conditions $\mathcal{C}_1$, $\mathcal{C}_2$ and $\mathcal{C}_3$.
\Ensure The data block with all stabilisers fault-tolerantly measured.
\For{$i \in [d-1]$}\label{step:ft_X_extraction_start}
\State Let $\boldsymbol{w} \in \mathbb{F}_q^n$ be the $i$'th row of $H_X$.
\State Use Sub-Routine~\ref{alg:qudit_cat_FT_prep} with post-selection conditions $\mathcal{C}_1$, $\mathcal{C}_2$ and $\mathcal{C}_3$ to prepare the cat state corresponding to $X^{\boldsymbol{w}}$
\State Use Sub-Routine~\ref{alg:single_qudit_stab_meas} to measure the stabiliser $X^{\boldsymbol{w}}$. Record the outcome as $y_i \in \mathbb{F}_q$
\EndFor
\For{$i \in [M]$}
\State Let $\boldsymbol{w} \in \mathbb{F}_q^n$ be the $i$'th row of $\beta^{(X)}\cdot H_X$.
\State Use Sub-Routine~\ref{alg:qudit_cat_FT_prep} with post-selection conditions $\mathcal{C}_1$, $\mathcal{C}_2$ and $\mathcal{C}_3$ to prepare the cat state corresponding to $X^{\boldsymbol{w}}$
\State Use Sub-Routine~\ref{alg:single_qudit_stab_meas} to measure the stabiliser $X^{\boldsymbol{w}}$. Record the outcome as $z_i \in \mathbb{F}_q$
\If{$z_i \neq (\beta^{(X)}\cdot y)_i$}
\State Begin again from Step~\ref{step:ft_X_extraction_start}
\EndIf
\EndFor
\State The most recent round of $d-1+M$ $X$ checks completed is called ``self-consistent''. The outcome of the $X$ syndrome extraction is deemed to be $y \in \mathbb{F}_q^{d-1}$, which is handed to the decoder of the outer code.
\For{$i \in [d-1]$}\label{step:ft_Z_extraction_start}
\State Let $\boldsymbol{w} \in \mathbb{F}_q^n$ be the $i$'th row of $H_Z$.
\State Use (the dual of) Sub-Routine~\ref{alg:qudit_cat_FT_prep} with post-selection conditions $\mathcal{C}_1$, $\mathcal{C}_2$ and $\mathcal{C}_3$ to prepare the cat state corresponding to $Z^{\boldsymbol{w}}$
\State Use (the dual of) Sub-Routine~\ref{alg:single_qudit_stab_meas} to measure the stabiliser $Z^{\boldsymbol{w}}$. Record the outcome as $y_i \in \mathbb{F}_q$
\EndFor
\For{$i \in [M]$}
\State Let $\boldsymbol{w} \in \mathbb{F}_q^n$ be the $i$'th row of $\beta^{(Z)}\cdot H_Z$.
\State Use (the dual of) Sub-Routine~\ref{alg:qudit_cat_FT_prep} with post-selection conditions $\mathcal{C}_1$, $\mathcal{C}_2$ and $\mathcal{C}_3$ to prepare the cat state corresponding to $Z^{\boldsymbol{w}}$
\State Use (the dual of) Sub-Routine~\ref{alg:single_qudit_stab_meas} to measure the stabiliser $Z^{\boldsymbol{w}}$. Record the outcome as $z_i \in \mathbb{F}_q$
\If{$z_i \neq (\beta^{(Z)}\cdot y)_i$}
\State Begin again from Step~\ref{step:ft_Z_extraction_start}
\EndIf
\EndFor
\State The most recent round of $d-1+M$ $Z$ checks completed is called ``self-consistent''. The outcome of the $Z$ syndrome extraction is deemed to be $y \in \mathbb{F}_q^{d-1}$, which is handed to the decoder of the outer code.
\end{algorithmic}
\end{algorithm}

\subsubsection{Sharing Ancilla and Cat Rows Across Outer Code Blocks}

As mentioned previously, we imagine the entire memory block as occupying the space of a rectangular grid of gross codes of $m$ rows and $n$ columns. At any one time, $m-2$ of these rows are storing logical information, as they are blocks of the outer code (of length $n$), and the remaining $2$ rows are being used to assist the fault-tolerant syndrome extraction of one code block, as we have just described. These $2$ rows are imagined as moving through the system so that they can tend to the syndrome extraction of each block of the outer code. The idea to have one ancillary system moving through the blocks of outer codes and tending to each of their syndrome extractions is directly inspired by~\cite{gidney2025yoked}.

It remains to specify how we can have the ancillary system move through the data blocks. Imagine we have three rows, one data row, and two ancillary rows; we call the ancillary row next to the data row the ``first ancillary row'' and the other one the ``second ancillary row''. The aim is to teleport the state of the data row into the gross codes forming the second ancillary row. This may be achieved with the usual teleportation procedure~\cite{bennett1993teleporting}, but for Galois qudits.\footnote{Note, however, that this procedure is entirely equivalent to teleporting groups of $s$ qubits.}

The teleportation begins by preparing all qudits in both ancilla rows in the $\ket{+}$ state. This may be achieved by preparing all the logical qubits in the rows in the state $\ket{+}$ (the pivot qubits as well as the non-pivot qubits making up each qudit), which may in turn be accomplished by preparing all physical qubits in all of their gross codes in the state $\ket{+}$ and performing syndrome extraction. We then measure the qudit $ZZ$ operator pairwise between all adjacent gross code pairs spanning the two ancillary row. Because we pick the same qudit-to-qubit mappings in all columns of our $m \times n$ gross code grid, this is the same thing as measuring the qubit $ZZ$ operator pairwise on every non-pivot logical qubit in the adjacent pairs of gross codes. In the bicycle instructions, one may perform this operation using Sub-Routine~\ref{alg:qudit_ZZ_meas}.\footnote{Notice that because this first qudit $ZZ$ measurement occurs offline, it may be post-selected with some post-selection condition, although we omit this for simplicity.} We may then measure the qudit $ZZ$ operator, as well as the qudit $XX$ operator on all pairs of adjacent gross codes in the data row and the first ancillary row. This can again be accomplished in the bicycle instructions using Sub-Routine~\ref{alg:qudit_ZZ_meas}, as well as its dual $XX$ measurement version. Note also that these two operations cannot use any post-selection condition because they directly affect a data block. The fact that this implements teleportation of the Galois qudits may be shown using qudit stabiliser tableaux calculations~\cite{wills2026review} entirely analogously to the qubit case. While we describe fault tolerance in detail later, here we comment that the teleportation is fault-tolerant because it is effectively a transversal operation on the outer code (errors are not spread between gross codes in the teleported outer code block).

We conclude this section with a formal description of the teleportation procedure; see Sub-Routine~\ref{alg:row_teleportation}. Then, we give our top algorithm for the operation of the quantum memory in Algorithm~\ref{alg:top_quantum_alg}.

\begin{subroutine}[ht]
\caption{Teleportation of One Code Block across Ancillary Rows}
\label{alg:row_teleportation}
\begin{algorithmic}[1]
\Require One code block of the outer code adjacent to two ancilla rows.
\Ensure The code block in the row that was formerly the second ancillary row. 
\State\label{step:teleportation_start} Set every logical qubit in the two ancillary rows to $\ket{+}$. This may be achieved by preparing all physical qubits in the gross codes in these rows to $\ket{+}$ and performing syndrome extraction
\State\label{step:teleportation_first_meas} Measure the qudit $ZZ$ operator between each adjacent pair spanning the two ancillary rows. This may be accomplished using Sub-Routine~\ref{alg:qudit_ZZ_meas} with no post-selection condition
\If{any of these measurement is post-selected}
\State Repeat Steps~\ref{step:teleportation_start} and~\ref{step:teleportation_first_meas} for the corresponding pair
\EndIf
\State Measure the qudit $ZZ$ operator between each gross code pair spanning the data row and the first ancillary row. This may be accomplished using Sub-Routine~\ref{alg:qudit_ZZ_meas} with no post-selection condition
\State Measure the qudit $XX$ operator between each gross code pair spanning the data row and the first ancillary row. This may be accomplished using the dual of Sub-Routine~\ref{alg:qudit_ZZ_meas} with no post-selection condition
\State The outer code block now lives in the space of the second ancillary row, after the appropriate frame updates are applied
\end{algorithmic}
\end{subroutine}

\begin{algorithm}[ht]
\caption{Global Operation of the Memory System}
\label{alg:top_quantum_alg}
\begin{algorithmic}[1]
\Require Two ancillary rows as the bottom two rows of the grid; $m-2$ upper rows constitute blocks of the outer code; matrices $\beta^{(X)} \in \mathbb{F}_q^{M \times (d-1)}$ and $\beta^{(Z)} \in \mathbb{F}_q^{M \times (d-1)}$; post-selection conditions $\mathcal{C}_1, \mathcal{C}_2$ and $\mathcal{C}_3$.
\While{Memory is running}
\For{$i = 1$ to $m-2$}
\State Use Algorithm~\ref{alg:ft_extraction_one_block} with matrices $\beta^{(X)}, \beta^{(Z)}$ and post-selection conditions $\mathcal{C}_1, \mathcal{C}_2$ and $\mathcal{C}_3$ to fault-tolerantly extract the syndrome of the outer code block currently in row $i+2$ \State Given some decoding algorithm for the outer code, compute a correction and perform the required frame update
\State Use Sub-Routine~\ref{alg:row_teleportation} to teleport the code block in the $(i+2)$'th row to the $i$'th
\EndFor
\For{$i=m-2$ to $1$}
\State Use Sub-Routine~\ref{alg:row_teleportation} to teleport the codeblock in the $i$'th row to row $i+2$
\EndFor
\EndWhile
\end{algorithmic}
\end{algorithm}

\clearpage

\section{Bounding the Logical Error Rate}\label{sec:ler_calc}

In this section, we will bound the logical error rate of one copy of the outer code over one of its rounds.

\subsection{Overview and Definitions}

Given that there are several moving parts, we will begin by providing a high-level overview of the essential parts of this calculation.

For one copy of the outer code carrying logical information, it experiences errors (corresponding to logical errors on the inner gross codes) from numerous sources. First, the gross codes constituting a block of the outer code sit in memory for a long time. Therefore, while an idling gross code has an error rate far below that of its logical instructions~\cite{yoder2025tourgrossmodularquantum}, a particular module of the outer code experiences many more idling operations than logical instructions, and so the idling must be considered as one of the main sources of error. At the other extreme, each qudit does have some contact with logical instructions through the teleportations, cat state consumption during syndrome extraction, and may have errors transferred to it from faults in the cat state preparation itself.

We will treat the $X$ errors and $Z$ errors as entirely decoupled for simplicity, and provide analysis for $Z$ errors (caught by $X$ checks), where the discussion for $X$ errors is identical. 

We recall that, in an attempt to mitigate the effects of faults during syndrome extraction, we measured an overcomplete basis of $X$ checks, checking that the measurement outcomes are self-consistent, beginning all $X$ checks again if any violation is detected. If we reach the end of all $d-1+M$ $X$ checks with all the required equations satisfied, the round is deemed a ``self-consistent round of $X$ checks''. 

The first main failure mode of our code that we consider is called \textit{time-like failure}, and it is the event that we observe a self-consistent round, but the syndrome extracted is not the syndrome for the $Z$ error on the data just before, or during the self-consistent round. That is, during a self-consistent round, we consider the $Z$ error on the data at the moment just before the self-consistent round began, or just before any of its checks (there are $d-1+M$ such moments). If the syndrome of the self-consistent round is not the syndrome for any of those $Z$ errors, then we say that time-like failure has occurred. The probability of time-like failure will be bounded in Subsection~\ref{subsec:time_like_failures}. Of course, we must consider time-like failure for $X$ errors as well as $Z$ errors.

The other failure mode for our outer code will be \textit{space-like failure}. Supposing that time-like failure has not occurred, space-like failure is the event that the decoder for our outer code does not return the correct error for the syndrome it was handed. That is, if time-like failure has not occurred, the decoder was handed the correct syndrome for the error at some moment just before or during the corresponding self-consistent round. Space-like failure is the event that that exact error is not returned.\footnote{Note that this definition contains a little ambiguity. Indeed, it is possible that time-like failure does not occur, and the extracted syndrome corresponds to the error at multiple distinct moments during the self-consistent round. We would count this event as space-like failure, although we neglect it in our analysis, since it is very much less likely than other events that we count as failure. That is, if time-like failure does not occur, we assume that the given syndrome corresponds to exactly one error on the data in the corresponding self-consistent round.} The point is that, when calculating the space-like failure probability on one outer code block, at the time of one self-consistent round of $X$ checks, we may simply look back to the previous self-consistent round of $X$ checks. Assuming that time-like failure and space-like failure have not happened in the past, the $Z$ error on that code block may be thought of as completely removed from some instant within that previous self-consistent round, and we may consider all of the error that accumulates on the outer code block between that last self-consistent round, and the present one.

When returning a correction corresponding to a given syndrome, our outer code decoder will simply find the minimum weight error for the syndrome it is handed, which may be achieved efficiently using known list decoders for classical Reed-Solomon codes~\cite{sudan1997decoding,guruswami1998improved}. If there is any ambiguity between two errors of equal weight, it chooses randomly amongst the possibilities. In subsection~\ref{subsec:spacelike_failures}, we will upper bound the probability of space-like failure, including a discussion of the list decodable structure of the codes at hand, which allows us to correct more errors than one might naively expect (for example, most weight-$2$ errors are correctable on a distance $4$ code). 

When computing the logical error rate of the outer code, we will take the error rates of instructions of the gross code (without any post-selection) as described in Appendix~\ref{sec:improved_ler_bicycle} (for logical measurement instructions, see the values with $S = 100$). There are many different types of in-module and inter-module gates in the gross code bicycle instructions~\cite{yoder2025tourgrossmodularquantum}, however, as we state in that Appendix, we make the following simplifications for choosing error rates to be representative of the various instructions:
\begin{itemize}
    \item The in-module $X_1$ measurement is taken to be representative of all ``half-LPU'' in-module instructions, namely, the in-module measurements of $\langle X_1, Z_7 \rangle \cup \langle X_7, Z_1\rangle$;
    \item The in-module $Y_1$ measurement is taken to be representative of all ``whole-LPU'' in-module instructions, namely, the other $9$ in-module measurement instructions;
    \item The inter-module $X_1X_1$ measurement is taken to be representative of all inter-module measurement instructions.
\end{itemize}
For further discussion on these errors rates, in particular why we take these specific simplifications, see Appendix~\ref{sec:improved_ler_bicycle}.

Estimating the logical error rates of bicycle instructions under certain post-selection conditions (that is, the logical error rate of the instructions that have survived that post-selection condition) is harder. We are interested in a physical noise rate $p = 10^{-3}$ and $10^{-4}$, and direct testing of our post-selection strategies at these noise rates is too computationally expensive to be feasible (at least with the current testing of the decoder which takes place in CPU).\footnote{We estimate that hundreds of millions of core-hours would be required to directly test the post-selection strategies on each instruction at $p = 10^{-3}$.} In~\cite{wills2026forced}, the logical error rates of bicycle instructions under various post-selection strategies are tested at higher physical noise rates like $p = 2.5\times 10^{-3}$. In order to estimate the logical error rate of bicycle instructions under post-selection conditions at lower physical noise rate, \textit{we assume that the factor of reduction in error rate achieved by the post-selection strategy is the same (or greater) at the lower noise rates as to the noise rates we test}. For example, suppose that we find that at physical noise $p = 2.5\times 10^{-3}$, a $1\%$ post-selection rate under our post-selection strategy improves the logical error rate of the instruction by two orders of magnitude; we assume that a $1\%$ post-selection rate improves the logical error rate of the instruction by at least two orders of magnitude at $p = 10^{-3}$ or $10^{-4}$. The justification is as follows. The post-selection strategy that we use (based on the forced gap, see~\cite{wills2026forced}) is marking instances as suspicious for which corrections giving logically distinct outcomes have similar likelihoods (this is just like the regular complementary gap for the surface code~\cite{gidney2025yoked}). The post-selected error rate would therefore be expected to fall at a \textit{steeper} rate than the un-postselected error rate as the physical noise decreases,\footnote{See for example Figure 5b of~\cite{gu2026scalable} for this behaviour on an alternative post-selection strategy.} and so assuming that the factor of logical error rate reduction is the same (or larger) at the lower noise rate is likely to be a pessimistic assumption.

\subsubsection{Notation}

For the coming subsections, it will be helpful to define some notation for various errors. We will also define further notation as we go along, but it should be helpful to start with an overview here.

At the lowest level, our gate set is based on the above-stated in-module ``half LPU'' and ``whole LPU'' measurement instructions, as well as the inter-module measurement instructions. We denote the probability that each of these fail in any way as, respectively, $\mathbb{P}\left[\text{half}\right]$, $\mathbb{P}\left[\text{whole}\right]$ and $\mathbb{P}\left[\text{inter}\right]$. These error rates can be found in Appendix~\ref{sec:improved_ler_bicycle} (with $S = 100$). Next, we will also need to consider the effects of post-selection conditions $\mathcal{C}$. For example, $\mathbb{P}\left[\text{half}, \mathcal{C}\right]$ denotes the probability that a half-LPU in-module measurement instruction failed in some way, despite surviving post-selection according to the condition $\mathcal{C}$. Finally, for gross code idling, the probability of error is denoted $\mathbb{P}\left[\text{idle}\right]$, where we note that no post-selection can be considered here because all the idling noise we consider will be directly on codeblocks.

We can also consider error rates for higher-level gate sets. For example, considering the fault-tolerant qudit $Z^{\alpha}Z^{\beta}$ measurement in Sub-Routine~\ref{alg:ft_qudit_ZZ_meas} with post-selection conditions $\mathcal{C}_1$ and $\mathcal{C}_2$, then $\mathbb{P}\left[\left(Z^{\alpha}Z^{\beta}\right)_{\mathrm{FT}}, \mathcal{C}_1, \mathcal{C}_2\right]$ denotes the probability of some error, despite the protocol not declaring failure (due to post-selection condition $\mathcal{C}_2$). We can also consider the non-fault-tolerant qudit $Z^\alpha Z^\beta$ measurement, for which $\mathbb{P}\left[Z^\alpha Z^\beta, \mathcal{C}_1\right]$ denotes the probability of some error, despite not being post-selected under condition $\mathcal{C}_1$. Frequently, we will also use $\mathbb{P}\left[ZZ\right]$ and $\mathbb{P}\left[XX\right]$ to denote the probability of any error in the qudit $ZZ$ or $XX$ measurement. These are executed according to Sub-Routine~\ref{alg:qudit_ZZ_meas}, up to swapping $Z$ and $X$. We can similarly write, for example, $\mathbb{P}\left[ZZ, \mathcal{C}_2\right]$ for the probability of a fault in the qudit $ZZ$ measurement, despite the operation not being post-selected under $\mathcal{C}_2$.

It will be useful to think about higher-level routines still. For example, we will use the notation $\mathbb{P}\left[\text{Cat}\right]$ to denote the probability that there is any error on an accepted qudit cat state, or any error in its consumption.\footnote{Note that, strictly speaking, this probability does depend on the cat state in question, that is, on which measurement is taking place. Specifically, the $Z^\alpha Z^\beta$ measurement routines will differ between cat states. In the coming calculations, probabilities like $\mathbb{P}\left[\text{Cat}\right]$ will be estimated using the average behaviour. While this is an approximation, we note that in reality we have the power to re-randomise individual qudit-to-qubit mappings $B_i$ in any particularly problematic case; this latter comment will become more clear as we go on.} A formal expression for this quantity will be provided in Subsection~\ref{subsec:time_like_failures}.

Finally, it will be important to consider the time to execute various operations, where time will be counted in terms of physical timesteps, see Table 2 of~\cite{yoder2025tourgrossmodularquantum}. We note that all half-LPU in-module measurement instructions, whole-LPU in-module measurement instructions and inter-module measurement instructions all take the same amount of time. Therefore, we can denote $\tau_{\text{Meas}}$ as the time, in terms of number of physical timesteps, to execute any one of these operations. In particular, we have $\tau_{\text{Meas}} = 120$~\cite{yoder2025tourgrossmodularquantum}.

In Subsection~\ref{subsec:ler_single_cat}, we are going to study the creation and verification of a single cat state, in particular the expected time to create a cat state, and the distribution of errors on an accepted cat state. Then, in Subsection~\ref{subsec:time_like_failures}, we are going to treat time-like failures, and finally in Subsection~\ref{subsec:spacelike_failures}, we are going to treat space-like failures. We will conclude in Subsection~\ref{subsec:all_together}, where we will show how to handle all failure modes together at once.

\subsection{Creation of a Single Cat State}\label{subsec:ler_single_cat}

In this subsection, we will consider the expected time to create a cat state, and describe the effect of faults in their creation. Cat states are created according to Sub-Routine~\ref{alg:qudit_cat_FT_prep}. Note that the discussions for cat states for $X$-type syndrome extraction and $Z$-type syndrome extraction are the same; we consider the extraction of the syndrome of some $X$-type qudit stabiliser $X_1^{\gamma_1}X_2^{\gamma_2}\ldots X_n^{\gamma_n}$. For the post-selection conditions in our protocol, we make the following choices throughout:
\begin{itemize}
    \item $\mathcal{C}_1$ is a post-selection condition that makes use of the ``Forced Gap'' post-selection strategy given in~\cite{wills2026forced}. In particular, we let half-LPU in-module instructions, whole-LPU in-module instructions, and inter-module instructions be post-selected at rates $(r_{1,\text{half}}, r_{1,\text{whole}}, r_{1,\text{inter}})$, which are left as parameters for now;\footnote{A post-selection rate is the fraction of shots that are post-selected under the post-selection rule. A higher post-selection rate rejects more shots, but generically leads to a lower logical error rate on surviving shots.}
    \item $\mathcal{C}_2$ is taken to be the same as $\mathcal{C}_1$ but with post-selection rates $(r_{2,\text{half}}, r_{2,\text{whole}}, r_{2, \text{inter}})$;
    \item $\mathcal{C}_3$ is taken to be the same as $\mathcal{C}_1$ but with post-selection rates $(r_{3,\text{half}}, r_{3,\text{whole}}, r_{3, \text{inter}})$.
\end{itemize}

\subsubsection{Expected Time to Create a Cat State}\label{subsubsec:expected_time_create_cat}

Let us start by calculating the expected time to create one qudit cat state. This is important because the cat states take a long time to prepare, and while this is occurring, errors are accumulating on all the idling data blocks. For simplicity, we will only use the expected time to create a cat state to calculate the time between two syndrome extractions of one block of the outer code, thus neglecting the tail of the distribution of the time to create one cat state. We do this because the number of cat states created between two rounds of syndrome extraction of one block of the outer code is at least $2(m-2)(d-1+M)$,\footnote{It will in general be larger than this due to non-self-consistent rounds getting rejected.} which is large enough that the time to create a cat state will be concentrated on its mean.\footnote{Again, in reality, given an especially slow cat state, resulting from particularly difficult $Z^\alpha Z^\beta$ measurements, one may re-randomise certain qudit-to-qubit mappings $B_i$.}

To calculate the expected time to create one cat state, we must consider the time taken to run all the operations used in its preparation and verification, as well as the possibility of various post-selections. While being created and verified, two types of post-selections can occur, which we call ``physical-level post-selection'' and ``logical-level post-selection'', described as follows.
\begin{enumerate}
    \item \textit{Physical-level post-selection}: some post-selection occurring because the gross code decoder detected a suspicious decoding instance on some operation. The post-selection strategies are the ``Forced Gap'' strategies given in~\cite{wills2026forced}. Depending on which operations are post-selected, some of the cat state creation procedure may have to be reattempted, or all of it;
    \item \textit{Logical-level post-selection}: one of the fault-tolerant $Z^\alpha Z^\beta$ measurements in the $R$ rounds of checking gave a non-trivial measurement outcome.
\end{enumerate}
We will start by considering the possibility of physical-level post-selection, and how it affects the time to create the cat state, and consider logical-level post-selection later on.

We begin by thinking about one non-fault-tolerant qudit $Z^\alpha Z^\beta$ measurement, which forms a major part of the cat state preparation. Appendix~\ref{sec:compile_Z_meas} presents an efficient algorithm for compiling these non-fault-tolerant qudit measurements. There, we describe how any such measurement\footnote{That is, given any $\alpha$ and $\beta$, and any qudit-to-qubit mappings (which change how the qudit measuremnt is realised as a set of qubit measurements).} may always be executed using $11$ inter-module operations, and the number of half and whole LPU instructions on the first gross code is fixed (in fact, it is fixed to $17$ half-LPU instructions and $6$ whole-LPU instructions). The compilation algorithm then attempts to reduce the total number of in-module instructions used on the second gross code in performing the measurement.

The next thing to note is that, as we discuss more later, the qudit-to-qubit mappings $B_i$ at each coordinate of our outer code may be chosen randomly. Therefore, the qubit manifestations of these $Z^\alpha Z^\beta$ measurements will behave randomly over these qudit-to-qubit mappings.\footnote{Choosing a random qudit-to-qubit mapping on each qudit amounts to performing a random $\mathsf{CNOT}$ circuit on each group of $s$ qubits making up the qudit.} As explained in Appendix~\ref{sec:compile_Z_meas}, a brute force enumeration of all possible qubit manifestations of these measurements is completely out of bounds, but it is possible to estimate their cost over this random choice of qudit-to-qubit mappings. Sampling many possible such measurements, we let our compilation algorithm find the quickest execution it can. We find a mean number of half-LPU and whole-LPU in-module instructions on the second gross code of $56.05$ and $50.53$, respectively, giving a mean total number of in-module instructions across both gross codes of $73.05$ and $56.53$, respectively. As well as $11$ inter-module instructions always being used, we will use these numbers to calculate the error rate of a non-fault-tolerant qudit $Z^\alpha Z^\beta$ measurement.

In terms of the time taken to execute such a qudit measurement, we also note in Appendix~\ref{sec:compile_Z_meas} that the time to perform the measurement is always given by the time to perform the operations on the second gross code. Therefore, the mean time to perform one of these measurements is $(56.05+50.53+11)\cdot\tau_{\text{Meas}} = 117.58\cdot\tau_{\text{Meas}}$.\footnote{Recall that $\tau_{\text{Meas}} = 120$ is the number of physical timesteps to perform any measurement instruction on the gross code, see Table $2$ of~\cite{yoder2025tourgrossmodularquantum}.} However, about $n/2$ of these qudit $Z^\alpha Z^\beta$ measurements are performed at once (recall that $n$ is the length of an outer code block), and we must wait for the slowest one to finish. To estimate the mean time to perform one of these ``layers'' of $n/2$ measurements, we estimate the mean of the maximum time length over $n/2$ sampled measurements. The data for this is depicted in Figure~\ref{fig:sampled_meas_cost}. Taking a large value of $n = 70$,\footnote{We found only a very weak dependence of $n$ in this estimation, and so stuck with a pessimistic large value for simplicity.} we estimate the time for one $Z^\alpha Z^\beta$ measurement as $125\cdot\tau_{\text{Meas}}$.
\begin{table}[ht]
\centering
\renewcommand{\arraystretch}{1.2}
\begin{tabular}{c c r}
\toprule
\textbf{Code} & \textbf{Half/Whole/Inter} & \textbf{Count} \\
\midrule
\multirow{2}{*}{\makecell{First gross code}}
  & Half-LPU  & 17 \\
  & Whole-LPU & 6 \\\hline
\addlinespace
\multirow{2}{*}{\makecell{Second gross code}}
  & Half-LPU  & 56.05 \\
  & Whole-LPU & 50.53 \\\hline
\addlinespace
\multirow{3}{*}{\makecell{Total across\\both codes}}
  & Half-LPU & 73.05 \\
  & Whole-LPU & 56.53 \\
  & Inter-module & 11 \\
\bottomrule
\end{tabular}
\caption{Counts of instructions taken for non-fault-tolerant qudit $Z^\alpha Z^\beta$ measurements between two gross codes. The execution time is always determined by inter-module operations and operations on the second code. Therefore, the mean time to execute one $Z^\alpha Z^\beta$ measurement is $117.58\cdot\tau_{\text{Meas}}$. To estimate the time to execute a layer of $n/2$ $Z^\alpha Z^\beta$ measurements, we take a mean of the worst-case execution time over $n/2$ randomly sampled $Z^\alpha Z^\beta$ measurements at a large value of $n = 70$, giving a time $\approx125\cdot\tau_{\text{Meas}}$.}
\label{tab:qudit_ZZ_meas_instr_count}
\end{table}

The next thing we note is that these non-fault-tolerant qudit $Z^\alpha Z^\beta$ measurements are always post-selected under some condition (either $\mathcal{C}_1$ or $\mathcal{C}_3$). We recall that the strategy $\mathcal{C}_i$ post-selects the half-LPU in-module instructions, whole-LPU in-module instructions, and inter-module instructions, at rates $r_{i,\text{half}}, r_{i,\text{whole}}$, and $r_{i,\text{inter}}$, respectively. The probability that a non-fault-tolerant qudit $Z^\alpha Z^\beta$ measurement is accepted under the post-selection condition $\mathcal{C}_i$ is therefore
\begin{equation}
    \mathbb{P}\left[Z^\alpha Z^\beta \text{ Success}, \mathcal{C}_i\right]\coloneq(1-r_{i,\text{half}})^{73.05}(1-r_{i,\text{whole}})^{56.53}(1-r_{i,\text{inter}})^{11}.
\end{equation}
In order to calculate the expected time to create and verify a cat state, we begin with the non-fault-tolerant preparation, which proceeds via two layers of non-fault-tolerant $Z^\alpha Z^\beta$ measurement directly on the cat state row. The first layer can be repeated in place, and we can move on when the last one is done. When a qudit $Z^\alpha Z^\beta$ measurement fails (is post-selected), its expected time is taken to be half that of a successful measurement, since we can restart as soon as one of its instructions is post-selected, and on average the post-selected operation will be in the middle of the routine. 

Now, given $N$ attempts at a task in parallel, where the task succeeds with probability $p \in (0,1]$, a successful task takes unit time, and a failed task takes half-unit time, and we can re-attempt each task without affecting others, one checks that the expected time to complete all tasks is
\begin{equation}
    1 + \frac{1}{2}\sum_{j=1}^N(-1)^{j+1}\begin{pmatrix}
        N\\j
    \end{pmatrix}\frac{(1-p)^j}{1-(1-p)^j}.
\end{equation}
We therefore have the expected time for $\leq n/2$ $Z^\alpha Z^\beta$ measurements to all succeed post-selection under $\mathcal{C}_1$ as (at most)
\begin{equation}
    \tau_{\text{One Layer }Z^\alpha Z^\beta, \;\mathcal{C}_1} \coloneq 125\cdot \tau_{\text{Meas}}\left(1+\frac{1}{2}\sum_{j=1}^{\left\lfloor n/2\right\rfloor}(-1)^{j+1}\begin{pmatrix}
        \left\lfloor n/2\right\rfloor\\j
    \end{pmatrix}\frac{(1-p)^j}{1-(1-p)^j}\right),
\end{equation}
where $p \coloneq \mathbb{P}\left[Z^\alpha Z^\beta \text{ Success}, \mathcal{C}_1\right]$.
Now, the next layer in the non-fault-tolerant preparation is another layer of $Z^\alpha Z^\beta$ measurements that occur directly on the cat state, so their post-selection (under the condition $\mathcal{C}_3$) entails the rejection of the entire cat state. The probability that this second layer of measurements is successful is (at least)
\begin{equation}
    \mathbb{P}\left[\text{Second Layer Success}\right]\coloneq \mathbb{P}\left[Z^\alpha Z^\beta \text{ Success}, \mathcal{C}_3\right]^{n/2}.
\end{equation}
Again, when the second layer is successful, it takes time $125\cdot\tau_{\text{Meas}}$, but when unsuccessful, it takes an expected time of half that. Therefore, the expected time to complete the non-fault-tolerant cat stat preparation is
\begin{equation}
    \tau_{\text{Non FT Cat}}\coloneq \frac{\tau_{\text{One Layer }Z^\alpha Z^\beta, \;\mathcal{C}_1}}{q}+\frac{1+q}{2q}\cdot 125\tau_{\text{Meas}}, \;q\coloneq \mathbb{P}\left[\text{Second Layer Success}\right].
\end{equation}
We can then consider the expected time to execute the $R$ rounds of fault-tolerant checking, where we recall that we are ignoring logical level post-selection for now. This constitutes $2R$ layers of $\leq n/2$ non-fault-tolerant $Z^\alpha Z^\beta$ measurements in parallel, post-selected under $\mathcal{C}_1$, which can be re-tried without affecting anything else, each followed by $\leq n$ qudit $ZZ$ measurements in parallel (interacting the ancilla row with the cat state row). If any of these qudit $ZZ$ measurements fail, we have to begin the cat state routine from the beginning. We already know that the expected time to finish all the offline non-fault-tolerant $Z^\alpha Z^\beta$ measurements is
\begin{equation}
    \tau_{\text{One Layer }Z^\alpha Z^\beta, \;\mathcal{C}_1}.
\end{equation}
We can also see that the probability of succeeding a layer of qudit $ZZ$ measurements is (at least)
\begin{equation}
    \mathbb{P}\left[ZZ\text{ Layer Success}, \mathcal{C}_2\right] \coloneq (1-r_{2,\text{half}})^{34n}(1-r_{2,\text{whole}})^{12n}(1-r_{2,\text{inter}})^{11n},
\end{equation}
recalling that a qudit $ZZ$ measurement always uses $34$, $12$, and $11$ half-LPU in-module, whole-LPU in-module, and $11$ inter-module instructions, respectively. We recall that the time taken to execute a qudit $ZZ$ measurement is
\begin{equation}
    \tau_{ZZ} \coloneq 34\cdot \tau_{\text{Meas}},
\end{equation}
and again when one of these fails, we take its expected time as half of this, since when any one of its instructions fail, we may restart immediately, and on average the failing operation is half way through the routine. From here, we can calculate the expected time to create a cat state, assuming no logical-level post-selection:
\begin{equation}
    \tau_{\text{Cat, No LLPS}} = \begin{cases}\frac{\tau_{\text{Non FT Cat}} + \left(\tau_{\text{One Layer }Z^\alpha Z^\beta, \;\mathcal{C}_1}+\frac{1+p_2}{2}\cdot 34\tau_{\text{Meas}}\right)\frac{1-p_2^{2R}}{1-p_2}}{p_2^{2R}}&\text{ if }p_2 \neq 1\\
    \tau_{\text{Non FT Cat}}+2R\left(\tau_{\text{One Layer }Z^\alpha Z^\beta, \;\mathcal{C}_1} + 34\cdot\tau_{\text{Meas}}\right)&\text{ if } p_2 = 1\end{cases},
\end{equation}
where
\begin{equation}
    p_2 \coloneq \mathbb{P}\left[\text{ZZ Layer Success}, \;\mathcal{C}_2\right].
\end{equation}
Finally, we note that in terms of the expected time to create a cat state, logical level post-selection is negligible, because the probability of any fault (not post-selected at a physical level) in the creation of any cat state is low.\footnote{Note that this actually has to be the case if we are ever likely to complete a self-consistent round of checks; we will come onto this momentarily.} We thus take the expected time to create one cat state as
\begin{equation}
    \tau_{\text{Cat}} = \tau_{\text{Cat, No LLPS}}.
\end{equation}
On the other hand, the total probability of any fault across a whole round of $d-1+M$ checks does start to become just high enough that we should not neglect it. Letting $\mathbb{P}\left[\text{Cat}\right]$ be the probability of any error on an accepted qudit cat state, or any error in its consumption,\footnote{A full expression for $\mathbb{P}\left[\text{Cat}\right]$ will be provided in Subsection~\ref{subsec:time_like_failures}.} we have the probability of any error in an attempted round of $d-1+M$ checks as
\begin{equation}
    (d-1+M)\cdot\mathbb{P}\left[\text{Cat}\right].
\end{equation}
The expected number of attempts at a single round of $d-1+M$ checks is then
\begin{equation}\label{eq:expected_round_attempts}
    \frac{1}{1-(d-1+M)\cdot\mathbb{P}\left[\text{Cat}\right]},
\end{equation}
and we will take the pessimistic assumption that any rejected round of $d-1+M$ checks is rejected at the end.

We conclude this subsubsection by writing down an expression for the expected number of gross code idling operations (syndrome cycles) making up one round of the outer code. This is
\begin{equation}
    N_{\text{Outer Round}} = \frac{(m-2)}{\tau_{\text{Idle}}}\left(6\cdot\tau_{ZZ}+\frac{2(d-1+M)\cdot\tau_{ZZ}+2(d-1+M)\cdot\tau_{\text{Cat}}}{1-(d-1+M)\cdot\mathbb{P}\left[\text{Cat}\right]}\right).
\end{equation}
Here, $\tau_{\text{Idle}} = 8$ is the number of physical timesteps in one idling operation (see Table $2$ of~\cite{yoder2025tourgrossmodularquantum}), $6\cdot\tau_{ZZ}$ is the time taken to perform the teleportations on one codeblock,\footnote{Referring to Algorithm~\ref{alg:top_quantum_alg}, each outer code block experiences two instances of teleportation between two rounds of outer code syndrome extraction. Referring to Sub-Routine~\ref{alg:row_teleportation}, each teleportation operation takes a time $3\cdot\tau_{ZZ}$ to execute, recalling that $\tau_{ZZ} = 34\cdot\tau_{\text{Meas}}$ is the time to execute a qudit $ZZ$ or qudit $XX$ measurement.} $2(d-1+M)\cdot\tau_{ZZ}$ is the time taken to consume the cat states on one outer codeblock, and $2(d-1+M)\cdot\tau_{\text{Cat}}$ is the time taken to prepare the necessary cat states. The latter two terms are multiplied by Equation~\eqref{eq:expected_round_attempts} to account for rejected rounds of checks.

\subsubsection{High-Level Description of Accepted Cat State Error Distribution}

We now want to study the distribution of errors on an accepted cat state. That is, we want to be able to answer questions like: given an accepted cat state, what is the probability it has some $Z$ error on it? Here, we consider a cat state used for the extraction of an $X$-type qudit syndrome. The discussion for the cat states used for the extraction of $Z$-type qudit syndromes is entirely dual to this one. Let us start by simplifying the problem, before diving into the details. 

We want to think about the distribution of Pauli errors on an accepted cat state $\mathsf{Cat}(\gamma_1, \gamma_2, \ldots, \gamma_n)$. We start by noting that there is not necessarily a notion of ``weight'' of a $Z$ error on an accepted cat state; a $Z$ error is either present or it is not. This is because any $Z$ error on such a cat state is equivalent to a $Z$ error of weight at most one, as shown in the following claim.\footnote{It is instructive to recall that the same is true of a regular $n$-qubit cat state.}
\begin{claim}\label{claim:Z_error_size}
    Suppose we have a qudit cat state $\mathsf{Cat}(\gamma_1, \gamma_2, \ldots, \gamma_n)$ for $\gamma_i \in \mathbb{F}_q$, $\gamma_i \neq 0$ which has been acted on by a Pauli error. The Pauli error is equivalent to a Pauli error whose $Z$ support has weight at most one.
\end{claim}
\begin{proof}
    The qudit cat state has qudit stabilisers $Z_{i-1}^{\gamma_i}Z_i^{\gamma_{i-1}}$ for each $i = 2, \ldots, n$, and so by definition it is stabilised by operators $Z_{i-1}^{A\gamma_{i}}Z_i^{A\gamma_{i-1}}$ for each $i = 2, \ldots, n$ for every $A \in \mathbb{F}_q^*$ (see~\cite{wills2026review}). Suppose the Pauli error has $Z$ support $Z^\eta$ on the first qudit. Considering the qudit stabiliser $Z_1^{\gamma_2}Z_2^{\gamma_1}$ and $A = \eta\gamma_2^{-1}$, we can see that it is possible to clean the $Z$ support of the error off of the first qudit, at the expense of changing the support of the error on the second qudit. Continuing in this way, the $Z$ support of the error may be reduced to the size of at most one qudit.
\end{proof}
We will therefore only be worried about the actual weight of $X$ errors on the accepted cat state; for the $Z$ errors we will say that they are either present or not. This is also the natural thing to do because $Z$ errors will just show up as measurement errors when the cat state is consumed in syndrome extraction; $X$ errors on the qudit cat state will actually spread to the data block when the cat state is consumed, so we have to worry about their weight.

\subsubsection{Spread of Errors in the Fault-Tolerant Qudit \texorpdfstring{$Z^\alpha Z^\beta$}{} Measurement Sub-Routine}\label{subsubsec:error_spread_ft_qudit_ZZ}

In order to study the distribution of Pauli errors on an accepted qudit cat state, we will begin by studying the spread of errors in the fault-tolerant qudit $Z^\alpha Z^\beta$ measurement sub-routine. In particular, it is called a fault-tolerant measurement gadget, but we haven't actually said yet what that means precisely. It means that one fault in the measurement sub-routine can increase the weight of the error on the cat state by at most one, and arbitrarily many faults can increase the weight of the error on the cat state by at most two, as we will go on to show.

Note that, in the initial non-fault-tolerant cat state preparation, a single fault, for example a fault giving us the wrong outcome of a $Z^\alpha Z^\beta$ measurement, can cause us to perform the wrong frame update, corresponding to an arbitrarily high-weight $X$ error. This is because these initial measurements are non-deterministic, even in the absence of faults, and so we have no prior expectation of what these initial measurements ``should'' be; there is no logical-level post-selection in this first round.\footnote{This is the exact same effect as can be seen for qubits. Suppose, for example, we try and make a regular $7$-qubit cat state by starting with $7$ qubits in the state $\ket{+}$ and measuring $ZZ$ on qubits $(1,2), (3,4), (5,6)$ in one layer, and then on $(2,3), (4,5), (6,7)$ in the next layer, and performing frame updates depending the outcomes. If we get the wrong measurement outcome on the $ZZ$ measurement on $(4,5)$, then that is equivalent to having a cat state with $X$ error $X_5X_6X_7$, or equivalently, $X_1X_2X_3X_4$.} Thus, one fault in non-fault-tolerant preparation can arbitrarily damage the cat state with a high-weight $X$ error, and/or some $Z$ error.

However, we can now think about the effects of faults in the subsequent fault-tolerant qudit $Z^\alpha Z^\beta$ measurements. Referring back to the description of this protocol in Sub-Routine~\ref{alg:ft_qudit_ZZ_meas}, and depicted in Figure~\ref{fig:ft_qudit_ZZ_meas}, we will only consider the effect of faults in the non-fault-tolerant qudit $Z^\alpha Z^\beta$ measurement, and the two qudit $ZZ$ measurements. As always, because the preparation and measurement of the qudits in the $X$ (or $Z$) basis can be accomplished via transversal gross code operations which will error at a far lower rate than the qudit measurements, and because any faults they do have could have also been created by the qudit measurement operations, they are neglected. One can then check that the following are true, assuming the operations survived physical-level post-selection:\footnote{These can be checked directly using the Galois qudit stabiliser formalism~\cite{wills2026review}. Again, it is instructive to note that the same fault-spreading behaviour holds for the analogous qubit protocol (which is recovered by simply setting $q=2$).}
\begin{enumerate}
    \item\label{ft_qudit_ZZ_fault_case_first_meas} Consider the non-fault-tolerant qudit $Z^\alpha Z^\beta$ measurement between ancilla row qudits. One can consider this failing arbitrarily: producing the wrong measurement outcome, and depositing some Pauli error on the ancilla row gross codes. This can create a measurement error of the whole fault-tolerant qudit $Z^\alpha Z^\beta$ measurement sub-routine, and can create a $Z$ error on the cat state, but it cannot spread $X$ errors to the cat state;
    \item\label{ft_qudit_ZZ_fault_case_latter_two_meas} Consider one of the two latter qudit $ZZ$ measurements (between one ancilla row gross code and one cat state row gross code). Again, suppose it fails arbitrarily. This can create a measurement error of the whole fault-tolerant measurement sub-routine, and it can create a $Z$ error on the cat state, and it can create an $X$ error on the cat state qudit that it touches, but it cannot create an $X$ error on the cat state qudit that it does not touch.
\end{enumerate}
We see, therefore, that a single fault in a fault-tolerant qudit $Z^\alpha Z^\beta$ measurement sub-routine cannot increase the weight of the $X$ error on the cat state by more than one. It \textit{is} possible that, because of a qudit $ZZ$ measurement failure, a fault-tolerant qudit $Z^\alpha Z^\beta$ measurement sub-routine with one fault fails to catch a given $X$ error, and actually increases the weight of the $X$ error,\footnote{In reality, this is a very specific failure, and can be considered a very rare event. The qudit $ZZ$ measurement can have $2^s-1 = 2047$ measurement errors, and exactly one is required to miss the existing $X$ error.} but even then the weight of the $X$ error has only increased by one. Again, this differs from the initial non-fault-tolerant preparation because there we attempt to perform frame updates to obtain the desired cat state, meaning that a single fault can create an arbitrary-weight $X$ error. By contrast, in the latter $R$ rounds, any non-trivial measurement outcome leads to cat state rejection.

\subsubsection{Fault Tolerance of the Qudit Cat State}\label{subsubsec:ft_of_qudit_cat}

Now we have studied an individual fault-tolerant qudit $Z^\alpha Z^\beta$ measurement sub-routine, we can now consider the entire fault-tolerant preparation of the qudit cat state. Recall that the whole qudit cat state preparation follows by an initial non-fault-tolerant round of measurements of the operators $Z_{i-1}^{\gamma_i}Z_i^{\gamma_{i-1}}$, followed by $R$ rounds of measuring these operators fault-tolerantly. Every round of $n-1$ measurements (fault-tolerant or non-fault-tolerant) is executed in two layers.

First, it is important to appreciate that any of the latter $R$ rounds of fault-tolerant $Z^\alpha Z^\beta$ measurements will always catch any existing $X$ errors on the cat state, assuming that that round does not suffer a fault. This is because, given any non-trivial $X$ error on the cat state, that $X$ error does not commute with some operator $Z_{i-1}^{\gamma_i}Z_i^{\gamma_{i-1}}$, and so a non-trivial measurement outcome will be obtained by a faultless round of measurements. One immediate corollary of this fact is that, in order for there to be any $X$ error on the final qudit cat state, there must be a fault in the final round of fault-tolerant measurement. Another corollary is that, while the very first round of fault-tolerant measurement can create a high-weight $X$ error on the cat state, even resulting from a single fault, that $X$ error will be caught unless at least $R$ further faults occur.

Moreover, we can in fact see that our qudit cat state preparation and verification routine is $R$-fault-tolerant in the sense of Definition~\ref{def:ft_qudit_cat_property}.
\begin{claim}\label{claim:cat_is_ft}
    The preparation and verification of the qudit cat state is $R$-fault-tolerant. This means that if $s \leq R$ faults occur during its preparation and verification, and if it is accepted, it has at most $s$ errors on it.
\end{claim}
\begin{proof}
    We do not worry about the size of the $Z$ error, because Claim~\ref{claim:Z_error_size} tells us any $Z$ error is equivalent to a $Z$ error of weight at most one, and there is no $Z$ error on the cat state in the absence of any faults. It remains to handle the $X$ errors.
    
    If there is any $X$ error on the cat state after the initial round of measurements (the round attempting to initially prepare the cat state), it will necessarily be caught unless at least $R$ additional faults occur in the verification (in principle, it could also be cancelled out by future faults, in which case there is no $X$ error at all). 
    
    In the case that no $X$ error results from the initial non-fault-tolerant preparation, one fault in each fault-tolerant $Z_{i-1}^{\gamma_i}Z_i^{\gamma_{i-1}}$ measurement sub-routine can increase the size of the $X$ error by at most one, and an arbitrary number of faults in said sub-routine can increase the size of the $X$ error by at most two. The result follows.
\end{proof}
We have established that $s$ faults leave $s$ $X$ errors on the qudit cat state, if $s \leq R$. To complete the picture of the $X$-error distribution on an accepted qudit cat state, we need to understand the \textit{set of ways} that various numbers of $X$ errors can end up on the qudit cat state. To start with, we have that the probability of having $\leq 1$ $X$ error on an accepted qudit cat state is at most
\begin{equation}
    n\cdot \mathbb{P}\left[ZZ, \mathcal{C}_2\right],
\end{equation}
where $\mathbb{P}\left[ZZ, \mathcal{C}_2\right]$ is the probability of any fault in the qudit $ZZ$ measurement, despite the operation not being post-selected under the condition $\mathcal{C}_2$.\footnote{Note that here, we only keep leading order terms, that is, we only consider events where $1$ fault leads to the weight-one $X$ error, not examples where $2$ (or more) faults leads to a weight-one $X$ error.} The reason for this is as follows. Recalling the Rules~\ref{ft_qudit_ZZ_fault_case_first_meas} and~\ref{ft_qudit_ZZ_fault_case_latter_two_meas} for error spreading in the fault-tolerant $Z^\alpha Z^\beta$ measurement, we see that, in the $R$ rounds of fault-tolerant checking, $X$ errors can only be \textit{created} by the failure of the qudit $ZZ$ measurements. Moreover, because we are considering a cat state preparation routine with only $1$ fault, any $X$ error on a qudit will be caught if checked a further time by a fault-tolerant $Z^\alpha Z^\beta$ measurement routine. Thus, a single $X$ error on an accepted cat state requires that one of the last qudit $ZZ$ measurements to touch any of the $n$ qudits fails and leaves an $X$ error on that qudit, thus giving this formula.

What is the probability of having $2$ $X$ errors on an accepted qudit cat state? We take the probability of having $2$ $X$ errors on an accepted qudit cat state to be at most
\begin{equation}
    \begin{pmatrix}
        n\\2
    \end{pmatrix}\cdot\mathbb{P}\left[ZZ, \mathcal{C}_2\right]^2,
\end{equation}
because it requires $2$ out of the $n$ last qudit $ZZ$ measurements to touch any qudit to fail. Notice that this is a slight approximation, although in a negligible sense. Indeed, there are sets of $2$ faults that can lead to earlier $X$ faults not being caught. However, similar to what was discussed above, this requires the corresponding fault-tolerant $Z^\alpha Z^\beta$ measurements in the last layer to experience the correct one out of its $2^s-1 = 2047$ measurement errors, and then also deposit another $X$ error on the other qudit. In addition, note that there will only be $\sim n$ of these fault paths, rather than $\begin{pmatrix}
    n\\2
\end{pmatrix}$. These fault paths are therefore sub-leading in probability and count, and thus neglected. For the same reasons, we take the probability of $3$ $X$ errors on the accepted qudit cat state to be at most
\begin{equation}
    \begin{pmatrix}
        n\\3
    \end{pmatrix}\cdot\mathbb{P}\left[ZZ, \mathcal{C}_2\right]^3,
\end{equation}
and so on.

We now want to consider the true failure of the cat state preparation, which is the case that more than $R$ $X$ errors end up on the accepted qudit cat state,\footnote{This will be counted later as a failure of our outer code. In particular, it will be counted as part of the space-like failure mode.} which, by Claim~\ref{claim:cat_is_ft}, requires more than $R$ faults in the whole fault-tolerant cat state preparation. The probability of having $>R$ $X$ errors on the accepted qudit cat state is thus at most
\begin{equation}\label{eq:cat_state_failure_prob}
    \frac{n}{2}\cdot \mathbb{P}\left[Z^\alpha Z^\beta, \mathcal{C}_1\right]\cdot \mathbb{P}\left[\left(Z^\alpha Z^\beta\right)_{\mathsf{FT}}, \mathcal{C}_1, \mathcal{C}_2\right]^R + \frac{n}{2}\cdot \mathbb{P}\left[Z^\alpha Z^\beta, \mathcal{C}_3\right]\cdot \mathbb{P}\left[\left(Z^\alpha Z^\beta\right)_{\mathsf{FT}}, \mathcal{C}_1, \mathcal{C}_2\right]^R 
\end{equation}
where $\mathbb{P}\left[\left(Z^\alpha Z^\beta\right)_{\mathsf{FT}}, \mathcal{C}_1, \mathcal{C}_2\right]$ is the probability of any un-post-selected fault in the fault-tolerant measurement sub-routine, and $\mathbb{P}\left[Z^\alpha Z^\beta, \mathcal{C}_i\right]$ is the probability of any fault in the non-fault-tolerant $Z^\alpha Z^\beta$ measurement, despite surviving post-selection under the condition $\mathcal{C}_i$. Why is this the case? We know that, at leading order, we must have $1+R$ faults, and they must be distributed equally across the initial non-fault-tolerant preparation, and the $R$ fault-tolerant rounds of checking. Depending on whether the fault happens in the first or second layer, the fault in the non-fault-tolerant preparation occurs with probability $\mathbb{P}\left[Z^\alpha Z^\beta, \mathcal{C}_1\right]$, or $\mathbb{P}\left[Z^\alpha Z^\beta, \mathcal{C}_3\right]$, respectively. Then, if one of these high-weight $X$ error is approaching one of the $R$ fault-tolerant rounds of checking, at least one of the fault-tolerant measurement routines has to fail to not detect the measurement. Note that this is quite pessimistic, because it is also the case here that the fault-tolerant measurement routine can experience many different measurement errors, and in order not to catch a given $X$ error, it must experience a particular one.\footnote{Note, however, that we do not make any approximation here, as we did above. The reason is as follows. Suppose that a particular non-fault-tolerant $Z^\alpha Z^\beta$ measurement had a distribution of measurement errors very heavily weighted on one particular failure; suppose that almost all of its measurement errors were some $\eta \in \mathbb{F}_q^*$. Then, it would only require $1+R$ of these to cause the catastrophic cat state failure, which is what our formula captures. On the other hand, we made the approximation we did above on the probability of, say, $2$ $X$ errors on the accepted cat state, because the case we are neglecting requires the qudit $ZZ$ measurement to fail in two distinct and specific ways: first, to deposit an $X$ error without having a measurement error (or the cat state would be post-selected), and then to experience the exact corresponding measurement error, while also depositing another $X$ error.}

We now conclude the description of the distribution of $X$ errors on an accepted qudit cat state by writing down the following expressions for completeness:
\begin{align}
    \mathbb{P}\left[\left(Z^\alpha Z^\beta\right)_{\mathrm{FT}}, \mathcal{C}_1, \mathcal{C}_2\right] &= \mathbb{P}\left[Z^\alpha Z^\beta, \mathcal{C}_1\right] + 2\cdot \mathbb{P}\left[ZZ, \mathcal{C}_2\right],\\
    \mathbb{P}\left[Z^\alpha Z^\beta, \mathcal{C}_i\right] &= 73.05\cdot \mathbb{P}\left[\text{half}, \mathcal{C}_i\right] + 56.53\cdot\mathbb{P}\left[\text{whole}, \mathcal{C}_i\right] + 11\cdot \mathbb{P}\left[\text{inter}, \mathcal{C}_i\right],\label{eq:non_ft_alpha_beta_meas_C1_prob}\\
    \mathbb{P}\left[ZZ, \mathcal{C}_2\right] &= 34\cdot \mathbb{P}\left[\text{half}, \mathcal{C}_2\right] + 12\cdot \mathbb{P}\left[\text{whole}, \mathcal{C}_2\right] + 11\cdot \mathbb{P}\left[\text{inter}, \mathcal{C}_2\right].
\end{align}
As for the $Z$ errors, as discussed, there is no notion of ``weight'' of $Z$ error; a $Z$ error is either present or not. Since no $Z$ error is present if no un-post-selected fault occurs in the cat state creation, the probability of having a $Z$ error may be upper bounded by the probability of any un-post-selected fault in the cat state creation, i.e.,
\begin{equation}\label{eq:cat_creation_error_prob}
    \mathbb{P}\left[\text{Cat Creation}\right]\coloneq \frac{n}{2}\cdot\mathbb{P}\left[Z^\alpha Z^\beta, \mathcal{C}_1\right] + \frac{n}{2}\cdot\mathbb{P}\left[Z^\alpha Z^\beta , \mathcal{C}_3\right] + R\cdot(n-1)\cdot\mathbb{P}\left[\left(Z^\alpha Z^\beta\right)_{\mathrm{FT}}, \mathcal{C}_1, \mathcal{C}_2\right].
\end{equation}

\subsection{Time-Like Failures}\label{subsec:time_like_failures}

In this subsection, we will consider the probability of time-like failure, one of the two failure modes of our outer code (time-like failure and space-like failure).

Let us recall the setting for the $X$ checks, where the discussion for the $Z$ checks is identical. For the $X$ checks, we measure $d-1$ ``regular'' qudit checks which span the space of stabilisers, followed by $M$ linear combinations of those checks, where $M$ is some integer at most $d-2$. The presence of the latter is to protect ourselves against faults in the syndrome extraction; see Algorithm~\ref{alg:ft_extraction_one_block}. As we measure the checks, we check that the extracted syndrome satisfies the linear combinations they are expected to, rejecting the whole round of $d-1+M$ $X$ checks if any violation is detected. Getting to the end of the round of $X$ checks with all the required linear combinations satisfied causes the round to be labelled ``self-consistent'', and we extract the syndrome and hand it to the outer decoder. The outer decoder will take the syndrome and decide on a most likely correction, which is then implemented. Considering all of the instants just before each of the $d-1+M$ $X$ checks in the self-consistent round, we say that time-like failure has occurred if we obtain a self-consistent round, but we do not have the actual syndrome for the error on the data at one of those instants. That is, a time-like failure occurs when we get a self-consistent round, but the syndrome was not the syndrome for the error on the data immediately before or during the round at any point. A time-like failure is deemed a logical error of the outer code. We will handle our other failure, space-like failure, in the next subsection (space-like failure occurs when time-like failure did not occur, but the decoder does not return the correct error, see the next subsection). The aim of the present subsection is to bound the probability of time-like failure.

One thing that makes the probability of time-like failure (which will also be a problem for bounding the probability of space-like failure) hard to bound is that we cannot assume perfect knowledge of the logical error distributions of our gross codes and their operations. For example, an idling gross code can fail in $4^{12}-1 = 16777215$ ways. It is possible using simulations to estimate aggregate statistics, like the total probability of some failure (the sum of the probabilities of these $16777215$ failures) but it is not computationally feasible to estimate the probability of every one of these errors. We cannot even get an especially good idea of the whole distribution; the bottleneck at the moment is that simulations are performed using CPUs. In future, when decoding takes place in FPGAs or ASICs and becomes real-time~\cite{maurer2025real,maurya2025fpga}, one would be able to get a good idea of the full logical error distribution, either for a simulator with some noise model, or on a real device. Note that one could just make an \textit{assumption} that all logical errors occur uniformly, and indeed this makes the following calculations very easy. However, in the simulations that we can do, this clearly appears not to be the case.\footnote{Indeed, it likely should not be true considering the differing weights of logical operators in the circuits.}

Now, for this paper, we would like to be able to make proofs, even with our limited information about these logical error distributions. In order to make progress, we introduce a technique called spacetime error mixing. In one sentence, this is a mathematical technique that will make these calculations possible, using only the knowledge of aggregate statistics (like the total failure probability of an idling gross code), instead of requiring knowledge of the probabilities of individual logical errors. This technique will crop up for the time-like failures and for the space-like failures. It is actually easier to introduce this notion intuitively for space-like failures, and so we will do that now, for the sake of guiding the reader.

\begin{example}[Intuitive Introduction to Error Mixing for Space-Like Failures] Consider the case of space-like failure. The decoder for the $X$ checks is attempting to decode a classical Reed-Solomon code, with only the input of the syndrome, in order to determine the $Z$ error to which it corresponded. It trusts that time-like failure has not occurred, and so it will simply return the minimum-weight correction consistent with the syndrome. When multiple errors of the same minimum weight have the same syndrome, it chooses at random from them. Of course, some patterns of $Z$ errors on the data are correctable in the strict sense (meaning they are the unique minimum-weight error with the given syndrome), whereas others are uncorrectable (they are not the unique minimum-weight error with the given syndrome).

Meanwhile, we do not know the exact distribution of logical $Z$ errors on gross codes or their instructions. The only thing we can reliably estimate are aggregate statistics like the total probability of some failure; in particular, we cannot assume that all logical $Z$ errors occur uniformly. Initially, this seems worrying. How can we ensure that one of the outer code uncorrectable error patterns is not ``lining up'' with some of the more likely $Z$ errors, thus making the probability of space-like failure undesirably high? Note that if we knew this were happening, we would change the qudit-to-qubit mappings $\mathcal{B}$, since these determined to which qudit errors certain qubit errors correspond.

If one had a good knowledge of one's error distribution, say if one had real-time decoding on a  real device or simulator, one could perhaps make a wise choice of $\mathcal{B}$. In the absence of this information, we introduce a new notion that allows us to make proofs, proving that the probability of the failure mode can be kept desirably low, even with the limited information that we have.\footnote{One could take a very naive upper bound where the probability of every different logical error is upper bounded by the aggregate probability, but this leads to a far too loose bound on the overall logical error rate. For example, suppose an idling gross code fails with any logical error with probability $\mathbb{P}\left[\text{idle}\right]$. If we assume all its logical errors are uniform, then each logical error occurs with probability $\frac{\mathbb{P}\left[\text{idle}\right]}{16777215}$, over $7$ orders of magnitude lower. We don't want to have to assume uniformity of logical errors, not least because it doesn't appear to be true, but note that we \textit{can} say that the average of the probabilities of all those logical errors is $\frac{\mathbb{P}\left[\text{idle}\right]}{16777215}$. This is the sort of thing we will use in our mixing arguments.} The idea is to consider a parameter of the protocol, called a ``mixing parameter''. Changing the mixing parameter changes the exact specification of the protocol. We can imagine taking a uniformly random choice of the mixing parameter. One can then more easily calculate, in expectation over this random choice, the probability of the given failure mode. This is a calculation that becomes feasible, and relies only on the knowledge of certain aggregate statistics. Moreover, if we show that the probability of failure in expectation over the mixing parameter is some calculated value, it follows that there exists some choice of the mixing parameter such that the probability of the failure mode is at most the calculated value. The idea is that, given a real device, with more knowledge of one's logical error distribution, one could make such a choice of the mixing parameter. In this paper, we show that there always exists such a choice, regardless of the underlying logical error distributions.\footnote{In practice, given real-time decoding, what one would do would be to take the mixing parameter at random, and simulate to see if the probability of the failure mode is low enough. If it is not, one can continue taking another random value. Since we know that, in expectation over the random choice, the probability of the failure mode has to be below some value, one will quickly hit a good choice.}

For the space-like failure, the only mixing parameter to consider will be $\mathcal{B}$, the qudit-to-qubit mappings specifying how each of the qudits are decomposed into qubits~\cite{wills2026review}. In particular, $\mathcal{B} = (B_i)_{i=1}^n$, where $B_i$ is a basis for $\mathbb{F}_q$ over $\mathbb{F}_2$. Changing $B_i$ changes the qudit error to which a given qubit error corresponds.    
\end{example}

As we go through the time-like failure mode, we will see that there is one more mixing parameter that we will use in addition to $\mathcal{B}$. This will be $\nu \in (\mathbb{F}_q^*)^{d-1}$, where $\mathbb{F}_q^*$ is the set of non-zero elements of $\mathbb{F}_q$. We will see how $\nu$ comes in later, but roughly, $\nu$ will be used to randomise the linear combinations of checks that are measured in the latter set of $M$ checks. Across all these arguments, the high-level idea will always be the same: we show in expectation over some mixing parameter(s) that the probability of a given failure mode is at most some desirable value, and so there exists some choice of the mixing parameter(s) that leads to the probability of the failure mode being at most the desirable value. We will tie all the parameter choices and failure modes together in Subsection~\ref{subsec:all_together}, where in particular, we will see how choices of the mixing parameters can be made that are good choices for all failure modes combined.

\subsubsection{Overview of Time-Like Failures}

We recall that a self-consistent round of $X$ checks begins by measuring $d-1$ regular qudit checks spanning the space of checks, followed by measuring $M$ linear combinations of those checks. Now, let us specify this further. We consider some parity-check matrix $H_X \in \mathbb{F}_q^{(d-1)\times n}$ which will specify the first $d-1$ checks measured, where we measure the checks in the rows from top to bottom. Moreover, we recall (see Appendix~\ref{sec:RS_details} or~\cite{wills2026review}), that $H_X$ is a generator matrix for some classical code $\mathsf{GRS}_{d-1}(\balpha,\bu)$, where $\balpha = (\alpha_1, \ldots, \alpha_n)$, where $\alpha_i$ are distinct points in $\mathbb{F}_q$, and $\bu = (u_1, \ldots, u_n)$, and $u_i$ are non-zero values in $\mathbb{F}_q$. Explicitly, we have $(H_X)_{ab} = u_b\alpha_b^{a-1}$.\footnote{We will not make explicit proofs of the $Z$ checks, since they behave the same. To this end, note that one may choose $(H_Z)_{ab} = v_b\alpha_b^{a-1}$ for some $\bv = (v_1, \ldots, v_n)$, where $v_i \in \mathbb{F}_q$ are all non-zero.} It is important to note that every entry of $H_X$ is not zero, because we choose all $\alpha_i$ to be non-zero evaluation points.

Next, we recall from Algorithm~\ref{alg:ft_extraction_one_block} that the $M$ linear combinations of checks correspond to rows of the matrix $\beta^{(X)}\cdot H_X$, where $\beta^{(X)} \in \mathbb{F}_q^{M \times (d-1)}$ is some full-rank matrix. Again, we measure these rows from top to bottom. As we extract the $d-1+M$ syndrome components, we are checking that they satisfy the linear combinations that they should. This is exactly the same thing as checking that our extracted syndrome is a codeword of the classical code over $\mathbb{F}_q$ with parity-check matrix
\begin{equation}
    H_{X,\text{time}}\coloneq\begin{pmatrix}
        \beta^{(X)}&I
    \end{pmatrix}\in \mathbb{F}_q^{M\times(d-1+M)},
\end{equation}
where $I$ is an $M\times M$ identity matrix. We think of this as a ``time-like'' classical code protecting the syndromes of our outer code. We will describe later how a good choice of $\beta^{(X)}$ is made.

From the point of view of time-like failure on a self-consistent round of $X$ checks, when we consume a cat state, two types of faults can occur (or their combination):
\begin{itemize}
    \item \textit{Measurement error}: When we consume a cat state, thus making a measurement of a qudit $X$ check, a measurement error by some value $\eta \in \mathbb{F}_q$ is defined to be the event that the extracted syndrome component is not the syndrome component for the $Z$ error on the data immediately before this cat state consumption, and it was incorrect by an amount $\eta$. $\eta = 0$ is the case of no measurement error. Most commonly, this will occur because there was some $Z$ error on the accepted cat state, but it can also occur because there were faults in the operations consuming the cat state.
    \item \textit{Cat data error}: A cat data error is defined to be the event that the consumption of the cat state introduces a $Z$ error to the data block, that is, the $Z$ error on the data block after the cat state consumption is not the same as the $Z$ error on the data block before the cat state consumption. Notice that this can only occur because of failures in the operations consuming the cat state, because $Z$ errors on the accepted cat state cannot spread to the data (because we are extracting $X$ check syndromes).
    \item \textit{Mixed error}: Both a measurement error and a cat data error occur, that is, a failure in the operations consuming the cat state introduces a $Z$ error to the data, while also giving us the wrong syndrome component for the error on the data before the syndrome component was extracted.
\end{itemize}
We note that cat data errors (and data errors forming part of mixed errors) are always local, because the data errors can only occur due to failures in the cat state consumption, which act transversally on the data block. That is, $s$ faults occurring during the consumption of the cat state always result in at most $s$ new $Z$ errors on the data block.

In order to best protect ourselves against the effects of measurement errors and cat data errors/mixed errors, we always choose the time-like code (the classical code with parity-check matrix $H_{X,\text{time}}$) to be a classical code with distance $M+1$.\footnote{Note that by choosing $H_{X,\text{time}}$ to be a parity-check matrix for a classical code of distance $M+1$, it follows that $\beta^{(X)}$ is full-rank.} This can always be achieved by, for example, choosing $H_{X,\text{time}}$ to be a parity-check matrix for a classical Reed-Solomon code of length $d-1+M$ and distance $M+1$ over $\mathbb{F}_q$, and performing the necessary row operations to get it into the required form.\footnote{We emphasise that this time-like Reed-Solomon code has nothing to do with the ``space-like'' Reed-Solomon codes of length $n$ defining our quantum code.} We say again that time-like failure is defined as the event that we get a self-consistent round, but the syndrome extracted is not the syndrome for the $Z$ error on the data just before or during the round (that is, not the syndrome for the error just before any of the $d-1+M$ checks). To better understand time-like failure, let's start with the following as a warm-up.
\begin{claim}\label{claim:tlf_regular_check_fault}
    Time-like failure cannot occur on a self-consistent round unless a fault occurred in one of the $d-1$ regular checks.
\end{claim}
\begin{proof}
    If no fault occurred in any of the $d-1$ regular checks, the syndrome $y \in \mathbb{F}_q^{d-1}$ extracted by those checks is the exact syndrome for the error on the data during all $d-1$ regular checks. Since the round is self-consistent, the latter $M$ checks give no ambiguity that the syndrome is indeed $y$.
\end{proof}
Notice that at this stage we have not used the fact that $H_{X,\text{time}}$ defines a classical code with distance $M+1$. That point can be used when treating cases with only measurement errors, as follows.
\begin{claim}\label{claim:tlf_meas_only}
    Given a self-consistent round, suppose that no cat data (or mixed) errors were experienced; only measurement errors possibly occurred. The minimum number of measurement errors required to cause time-like failure is $M+1$.
\end{claim}
\begin{proof}
    Given no cat data errors, the error on the data, call it $e \in \mathbb{F}_q^n$, is unchanged throughout the self-consistent round. Suppose that the values measured by the first $d-1$ checks were $y \in \mathbb{F}_q^{d-1}$, and by the second set of checks were $z \in \mathbb{F}_q^{M}$. Self-consistency means that $z = \beta^{(X)}\cdot y$. The true syndrome for the error at all points in the round is $\hat{y} \coloneq H_X\cdot e$ for the $d-1$ regular checks and $\hat{z} \coloneq \beta^{(X)}\cdot \hat{y}$ for the linear combinations. $y$ and $z$ differ from $\hat{y}$ and $\hat{z}$, respectively, by the vectors of measurement errors that took place, and since $z = \beta^{(X)}\cdot y$, these errors must form a codeword of the code with parity-check matrix $H_{X,\text{time}}$. To cause time-like failure, there must be at least one non-trivial measurement error. Since the code with parity-check matrix $H_{X,\text{time}}$ has distance $M+1$, the result follows.
\end{proof}
Next, it is natural to ask how many data errors (or mixed errors) are required to cause time-like failure. It turns out that, given a bad choice of $\beta^{(X)}$, it is possible for a single fault to cause time-like failure.
\begin{example}[Time-Like Failure with a Single Cat Data Error]\label{example:tlf_single_cat_data_error}
    Suppose that, entering the self-consistent round, the error on the data is trivial. Suppose that during one of the regular $d-1$ checks, say the $k$'th one, a single data error occurs, making the error on the data some $\lambda e_i$ after that point, where $\lambda \in \mathbb{F}_q^*$, $i \in [n]$ and $e_i$ is the vector with a $1$ in the $i$'th position and $0$ elsewhere. One can then see that the syndrome extracted by the regular $d-1$ checks is $\lambda \Pi_{>k}h_i$, where $h_i$ is the $i$'th column of $H_X$, and $\Pi_{>k}$ is a $(d-1)\times(d-1)$ diagonal matrix projecting onto all coordinates greater than $k$ (which is the zeros matrix for $k = d-1$). 
    
    After the first $d-1$ checks, the error on the data has syndrome $\lambda h_i$. Then, since the latter $M$ checks are without fault, they measure the values $\lambda \beta^{(X)}\cdot h_i$. The round being self-consistent then requires that $\lambda \beta^{(X)} \Pi_{>k} h_i = \lambda \beta^{(X)} h_i$, which is if and only if $\Pi_{\leq k}h_i \in \ker(\beta^{(X)})$, where $\Pi_{\leq k}$ is the projector onto the first $k$ coordinates.

    Now, the self-consistent round being a time-like failure requires that the extracted syndrome is not the syndrome for the error on the data at any point just before or during the round. The syndromes for the two errors on the data just before and during the round were $0$ and $\lambda h_i$. The extracted syndrome was $\lambda \Pi_{>k}h_i$, where $\lambda \neq 0$, and so given that the round is self-consistent, we have time-like failure if and only if $k \neq d-1$ (since every entry of $h_i$ is non-zero).
\end{example}
At first, the above seems very strange. We are \textit{not} going to preclude the possibility that a single fault can cause a failure of our whole protocol. However, having read the above example, one may have the feeling that such faults would have to be very specific, and even that they need not exist at all. For example, there is nothing forcing the existence of an $i$ and $k$ such that $\Pi_{\leq k}h_i \in \ker(\beta^{(X)})$. What our mixing arguments will do in the following subsections will be to handle all of this together neatly, and show that a low probability of time-like failure may be ensured. We would stress at this point that the arguments on time-like failure are probably not best understood at the intuitive level from the point of view of traditional ``$t$-fault tolerance'' arguments~\cite{gottesman2024surviving}, because of the effect of doing fault tolerance over qudits with a large alphabet. More generally, we feel that one of the interesting contributions of this paper is that the theory of fault tolerance on qudits should be considered quite differently to that on qubits.

We want to calculate the probability of time-like failure in expectation over a uniformly random choice of mixing parameters. To do so, we will split time-like failure into various cases. This is necessary because we can see from the above discussion that measurement errors and cat data errors/mixed errors must be treated quite differently. Indeed, from Claim~\ref{claim:tlf_meas_only} and Example~\ref{example:tlf_single_cat_data_error}, we can see that at least $M+1$ measurement errors are required for time-like failure, but a single very specific cat data error can cause time-like failure, given a bad choice of $\beta^{(X)}$. This is convenient in the sense that measurement errors are very much more likely in our setup than a cat data error/mixed error. This is the case because a measurement error can result from any $Z$ error on the accepted cat state, or faults in the cat state consumption, whereas a cat data error/mixed error can result only from faults in the consumption of the cat state. In order to calculate the leading order contributions to time-like failure, we consider the following three cases, where we calculate the probability of each type of time-like failure in expectation over the mixing parameters.
\begin{enumerate}
    \item Case $1$: One cat data error;
    \item Case $2$: One mixed error;
    \item Case $3$: $M+1$ measurement errors.
\end{enumerate}
One can also produce formulas for higher-order failure modes involving, say, one cat data error and one measurement error, but any will be sub-leading compared to these (and can also be handled using more mixing arguments).
\subsubsection{Case 1: One Cat Data Error}\label{subsubsec:tlf_case1}

We start by computing the probability of case $1$ time-like failure (in expectation over some mixing parameters) in one particular self-consistent round of $X$ checks. We begin with a claim characterising these failures. Note that, by Claim~\ref{claim:tlf_regular_check_fault}, the cat data error must occur on one of the regular $d-1$ checks.

\begin{claim}\label{claim:tlf_case_1_count}
    In a particular self-consistent round of $X$ checks, suppose one cat data error occurs in the first $d-1$ checks. Such a fault is indexed by
    \begin{itemize}
        \item The index $k \in [d-1]$ of the regular check that experiences the cat data error;
        \item The outer code coordinate $i \in [n]$ to which the data error is introduced;
        \item The value $\lambda \in \mathbb{F}_q^*$ of the $Z$ error introduced to the $i$'th coordinate.
    \end{itemize}
    The number of these faults causing time-like failure is
    \begin{equation}
        (q-1)\cdot N_1(n,d,\balpha,\beta^{(X)}),
    \end{equation}
    where
    \begin{equation}
        N_1(n,d,\balpha,\beta^{(X)}) = \left|\left\{(k,i) \in [d-2]\times[n]:\Pi_{\leq k}h_i \in \ker(\beta^{(X)})\right\}\right|,
    \end{equation}
    where $\Pi_{\leq k}$ is a $(d-1)\times (d-1)$ matrix projecting onto the first $k$ coordinates, and $h_i$ is the $i$'th column of $H_X$.
\end{claim}
\begin{remark}\label{rmk:tlf_case_1_count_monomial_invariance}
    Notice that $N_1(n,d,\balpha,\beta^{(X)})$ does not depend on the values $\bu$ that come into the definition of $H_X$ (recall that $(H_X)_{ab} = u_b\alpha_b^{a-1}$), because $\Pi_{\leq k}h_i \in \ker(\beta^{(X)})$ if and only if $u\Pi_{\leq k}h_i \in \ker(\beta^{(X)})$ for any $u \in \mathbb{F}_q^*$. One consequence of this is that these results apply equally to the $Z$ checks as to the $X$ checks, because we can choose $H_Z$ to be the same as $H_X$ but with different non-zero values multiplying the columns. We will not always remark on it explicitly going forward, but all results apply equally well for the $Z$ checks as for the $X$ checks.
\end{remark}
\begin{proof}[Proof of Claim~\ref{claim:tlf_case_1_count}]
    Without loss of generality, the error on the data entering the self-consistent round is trivial. Suppose that for any $k \in [d-1]$, the $k$'th check introduces the weight-one error $\lambda e_i \in \mathbb{F}_q^n$. The remaining checks measure the syndrome of $\lambda e_i$ without fault, meaning that the full syndrome measured by the regular $d-1$ checks is $y = \lambda \Pi_{>k}h_i$, where $h_i$ is the $i$'th column of the parity-check matrix $H_X$, and $\Pi_{>k}$ is the $(d-1)\times(d-1)$ matrix projecting onto coordinates $(k+1, \ldots, d-1)$, which is the zeros matrix for $k = d-1$.

    The syndrome for the error on the data after the regular $d-1$ checks is $\lambda h_i$, and the latter $M$ checks measure the values $z = \lambda \beta^{(X)}h_i$ without fault. Self-consistency occurs if and only if $z = \beta^{(X)}\cdot y$, which is if and only if $\Pi_{\leq k}h_i \in \ker(\beta^{(X)})$.

    Given self-consistency, time-like failure occurs if and only if the measured syndrome $y = \lambda \Pi_{>k}h_i$ is not the syndrome for the error on the data at any point just before or during the self-consistent round. These syndromes are exactly $0$ and $\lambda h_i$. Thus, given self-consistency, and because every entry of $h_i$ is non-zero, time-like failure occurs if and only if $k \neq d-1$.

    Given a value of $k$ and $i$ causing time-like failure, the value of $\lambda \in \mathbb{F}_q^*$ is unconstrained, giving the result.
\end{proof}

So far, we have seen that the number of case $1$ time-like failures can be characterised by properties of the matrix $\beta^{(X)}$ and the matrix $H_X$ (in particular, the values $\balpha$, not $\bu$). The value $N_1(n,d,\balpha,\beta^{(X)})$ can actually change if one varies $\balpha$ and $\beta^{(X)}$.\footnote{To orient the reader, it is worth noting that one could simply choose some $\balpha$ and $\beta^{(X)}$ such that the corresponding $N_1$ is zero, and then case $1$ time-like failure never occurs. The problem is then that with such a fixed choice, it would be hard to go on and argue about our other time-like failure cases. What we do to handle all the cases together neatly is to upper bound the expectation over some mixing parameters as a certain value, and then draw the conclusion that there is always some choice of the mixing parameter such that the probability of that failure mode is at most that value. Then, all the cases of failure will be handled together in Subsection~\ref{subsec:all_together} where, in particular, we will show that a single choice of mixing parameter is possible such that the probability of all failure modes together is desirably low.} To handle this, we introduce a mixing parameter $\nu$. $\nu$ is a parameter of the protocol which takes values in $(\mathbb{F}_q^*)^{d-1}$, where we emphasise that $\mathbb{F}_q^*$ is the set of non-zero values of $\mathbb{F}_q$. $\nu$ is going to have the effect of randomising the $M$ linear combinations of checks that we measure. Indeed, given any choice of $\beta^{(X)}$, call it $\hat{\beta}^{(X)}$, define the matrix $\hat{\beta}^{(X)}*\nu$ to be the matrix $\hat{\beta}^{(X)}$ with all its columns multiplied by the entries of $\nu$. 

We will imagine fixing some choice of $\hat{\beta}^{(X)}$ (such that the code with parity-check matrix $\begin{pmatrix}
    \hat{\beta}^{(X)}&I
\end{pmatrix}$ has distance $M+1$), and then picking a uniformly random choice of the mixing parameter $\nu$ over its possible $(q-1)^{d-1}$ values. The linear combinations measured are then determined by the matrix $\hat{\beta}^{(X)}*\nu$. We will calculate the probability of case $1$ time-like failure in expectation over this choice. Our final results do not depend on the choice of $\hat{\beta}^{(X)}$, as long as the code with parity-check matrix $\begin{pmatrix}
    \hat{\beta}^{(X)}&I
\end{pmatrix}$ has distance $M+1$.\footnote{Notice that if the code with parity-check matrix $\begin{pmatrix}
    \hat{\beta}^{(X)}&I
\end{pmatrix}$ has distance $M+1$, then so does the code with parity-check matrix $\begin{pmatrix}
    \hat{\beta}^{(X)}*\nu & I
\end{pmatrix}$.} In the coming argument, we will also consider the mixing parameter $\mathcal{B}$, which we recall are the qudit-to-qubit mappings determining how each qudit is decomposed into qubits~\cite{wills2026review}.
\begin{claim}\label{claim:tlf_case_1_prob}
    In expectation over uniformly random qudit-to-qubit mappings $\mathcal{B}$, and vectors $\nu \in (\mathbb{F}_q^*)^{d-1}$, the probability of case $1$ time-like failure on one self-consistent round of $X$ checks is
    \begin{equation}
        \mathbb{E}_{\mathcal{B},\nu}\left\{\mathbb{P}\left[\mathsf{TLF}_1^{(X)}\;|\;\mathcal{B},\nu\right]\right\}\leq \mathbb{P}\left[XX\right]\cdot n\sum_{k=M+1}^{d-2}\frac{D_k}{(q-1)^k},
    \end{equation}
    where
    \begin{equation}
        D_k = \begin{cases}0 &\text{ if } k \leq M\\
        \sum_{t=0}^{k-M-1}(-1)^t\begin{pmatrix}
            k\\t
        \end{pmatrix}\left(q^{k-M-t}-1\right)&\text{ if } k \geq M+1\end{cases}.
    \end{equation}
    Here, $\mathbb{P}\left[XX\right]$ is the probability of any fault in the qudit $XX$ measurement, which is a fixed operation executed according to (the dual of) Sub-Routine~\ref{alg:qudit_ZZ_meas}.
\end{claim}
\begin{remark}\label{rmk:tlf_case_1_prob_leading_order}
    While the formula is a little complicated, since we think of $q$ as large and $d$ small, one can see that it behaves like
    \begin{equation}
        \sim\mathbb{P}\left[XX\right]\cdot \frac{n(d-2-M)}{(q-1)^M}
    \end{equation}
    at leading order,
    since we care about $M \leq d-2$.
    It is very striking to notice that, while this failure mode happens due to only a single fault, its contribution to the probability of overall failure may be exponentially suppressed (since we may choose any $M$ up to $d-2$), and indeed increasingly strongly for larger $q$. It is also very important to notice that this effect would not be observed for the qubit case of $q=2$. Because of this, we say again that the intuition developed for qubit fault tolerance may not transfer well to fault tolerance on large qudits such as these. Qudit fault tolerance brings in genuinely new effects, and we believe that computing these probabilities in expectation over mixing parameters is the natural way to talk about these things.
\end{remark}
\begin{remark}
    If one were to consider higher-order failure modes involving, say, multiple cat data errors, one would obtain formulas involving multiple factors of $\mathbb{P}\left[\text{Cat}\right]$, and possibly higher-order factors of $q-1$. These will be sub-leading, therefore, and so are neglected.
\end{remark}
\begin{proof}[Proof of Claim~\ref{claim:tlf_case_1_prob}]
    Referring back to Claim~\ref{claim:tlf_case_1_count}, the case $1$ time-like failures are indexed by some $(k,i,\lambda) \in [d-2]\times[n]\times\mathbb{F}_q^*$. Let us denote the set of the $(k,i,\lambda)$ that cause time-like failure for a given choice of $\beta^{(X)}$ by $\mathsf{TLF}_1^{(X)}(\beta^{(X)})$. We are going to fix some choice of $\hat{\beta}^{(X)}$ (the choice is arbitrary as long as the code with parity-check matrix $\begin{pmatrix}
    \hat{\beta}^{(X)}&I
\end{pmatrix}$ has distance $M+1$), and consider $\beta^{(X)} = \hat{\beta}^{(X)}*\nu$.

    Now, given some $(k,i,\lambda) \in \mathsf{TLF}_1^{(X)}(\hat{\beta}^{(X)}*\nu)$, the probability of the event to which it corresponds is denoted
    \begin{equation}
        \mathbb{P}\left[\text{Data}^{(k)},i,Z^\lambda\;|\;\mathcal{B}\right].
    \end{equation}
    This is the probability that the $k$'th check introduces a data error $Z^\lambda$ on the $i$'th coordinate. Also, notice that this probability depends on the qudit-to-qubit mappings $\mathcal{B}$; strictly speaking, it only depends on $B_i$, the qudit-to-qubit mapping for the $i$'th qudit.\footnote{This is because the bicycle instructions fail with a fixed distribution of logical qubit errors, but then the qudit-to-qubit mapping $B_i$ determines the qubit error that must take place for there to be the qudit error $Z^\lambda$.}

    Given a choice of $\mathcal{B}$ and $\nu$, the probability of case $1$ time-like failure is
    \begin{equation}
        \mathbb{P}\left[\mathsf{TLF}_1^{(X)}\;|\;\mathcal{B},\nu\right] = \sum_{(k,i,\lambda) \in \mathsf{TLF}_1^{(X)}(\hat{\beta}^{(X)}*\nu)}\mathbb{P}\left[\text{Data}^{(k)},i,Z^\lambda\;|\;\mathcal{B}\right].
    \end{equation}
    We start by taking the expectation over $\mathcal{B}$, so that we obtain
    \begin{align}
        \mathbb{E}_{\mathcal{B}}\left\{\mathbb{P}\left[\mathsf{TLF}_1^{(X)}\;|\;\mathcal{B},\nu\right]\right\} &= \sum_{(k,i,\lambda) \in \mathsf{TLF}_1^{(X)}(\hat{\beta}^{(X)}*\nu)}\mathbb{E}_{\mathcal{B}}\left\{\mathbb{P}\left[\text{Data}^{(k)},i,Z^\lambda\;|\;\mathcal{B}\right]\right\}\\
        &= \sum_{(k,i,\lambda) \in \mathsf{TLF}_1^{(X)}(\hat{\beta}^{(X)}*\nu)}\mathbb{E}_{B_i}\left\{\mathbb{P}\left[\text{Data}^{(k)},i,Z^\lambda\;|\;\mathcal{B}\right]\right\},
    \end{align}
    where the expectation on the right-hand side is really only over $B_i$, the qudit-to-qubit mapping on the $i$'th qudit, hence the second line. Again, changing $B_i$ changes the qubit error that must occur in order for the qudit error produced to be $Z^\lambda$. Moreover, taking a uniformly random choice of $B_i$, the qubit error that must occur on the $i$'th coordinate in order for the corresponding qudit error to be $\lambda$ is uniformly randomised over its $q-1$ possible values.\footnote{One way to see this is to note that changing the qudit-to-qubit mapping $B_i$ is equivalent to performing a $\mathsf{CNOT}$ circuit on the $s$ qubits making up the $i$'th qudit. Choosing a uniformly random $B_i$ corresponds to performing a uniformly random $\mathsf{CNOT}$ circuit on those qubits.} Note also that in order for a $Z$ error to be created in a cat data error, the operation consuming the cat state (a qudit $XX$ measurement, which is executed according to (the dual of) Sub-Routine~\ref{alg:qudit_ZZ_meas}) must fail. We thus have
    \begin{equation}
        \mathbb{E}_{B_i}\left\{\mathbb{P}\left[\text{Data}^{(k)},i,Z^\lambda\;|\;\mathcal{B}\right]\right\} = \frac{1}{q-1}\sum_{v \in \mathbb{F}_2^s\setminus\{0\}}\mathbb{P}\left[XX,Z^v\right],
    \end{equation}
    where $\mathbb{P}\left[XX,Z^v\right]$ is the probability that the qudit $XX$ measurement routine (which, again, is a fixed routine executed according to (the dual of) Sub-Routine~\ref{alg:qudit_ZZ_meas}) fails and leaves the $Z$-type error $Z^v$ on the data qudit. Each $\mathbb{P}\left[XX,Z^v\right]$ represents a distinct failure of the qudit $XX$ measurement routine, and so we have
    \begin{equation}
        \sum_{v \in \mathbb{F}_2^s\setminus\{0\}}\mathbb{P}\left[XX,Z^v\right] \leq \mathbb{P}\left[XX\right],
    \end{equation}
    where $\mathbb{P}\left[XX\right]$ is the probability of any error in the qudit $XX$ measurement routine. We have thus found that
    \begin{align}
        \mathbb{E}_{\mathcal{B}}\left\{\mathbb{P}\left[\mathsf{TLF}_1^{(X)}\;|\;\mathcal{B},\nu\right]\right\} &\leq \frac{1}{q-1}\cdot\mathbb{P}\left[XX\right]\cdot\sum_{(k,i,\lambda) \in \mathsf{TLF}_1^{(X)}(\hat{\beta}^{(X)}*\nu)}\\
        &=\mathbb{P}\left[XX\right]\cdot N_1(n,d,\balpha,\hat{\beta}^{(X)}*\nu),
    \end{align}
    where the second line follows by Claim~\ref{claim:tlf_case_1_count}. We then simply have
    \begin{equation}
        \mathbb{E}_{\mathcal{B},\nu}\left\{\mathbb{P}\left[\mathsf{TLF}_1^{(X)}\;|\;\mathcal{B},\nu\right]\right\} \leq \mathbb{P}\left[XX\right]\cdot \mathbb{E}_{\nu}\left\{N_1(n,d,\balpha,\hat{\beta}^{(X)}*\nu)\right\}
        .
    \end{equation}
    The result follows from the following lemma, whose proof is deferred to Appendix~\ref{sec:tlf_proofs}.
    \begin{lemma}\label{lem:tlf_case_1}
        \begin{equation}
            \mathbb{E}_{\nu}\left\{N_1(n,d,\balpha,\hat{\beta}^{(X)}*\nu)\right\} = n\sum_{k=M+1}^{d-2}\frac{D_k}{(q-1)^k},
        \end{equation}
        where
        \begin{equation}
        D_k = \begin{cases}0 &\text{ if } k \leq M\\
        \sum_{t=0}^{k-M-1}(-1)^t\begin{pmatrix}
            k\\t
        \end{pmatrix}\left(q^{k-M-t}-1\right)&\text{ if } k \geq M+1\end{cases}.
    \end{equation}
    \end{lemma}
\end{proof}

\subsubsection{Case 2: One Mixed Error}\label{subsubsec:tlf_case2}

Let us now consider the probability of case $2$ time-like failure, which is time-like failure resulting from a single mixed error. One may wonder why this is considered separately to the case $1$ time-like failure. The reason is that the mixing arguments are a little different between them (in particular, we will use below the fact that the measurement error occurring is non-zero). They must both be considered, however, because they give the same order contribution to the overall probability of time-like failure.

We again begin by characterising these failures. By Claim~\ref{claim:tlf_regular_check_fault}, the mixed error must occur in one of the $d-1$ regular checks.
\begin{claim}\label{claim:tlf_case_2_count}
    In a particular self-consistent round of $X$ checks, suppose one mixed error occurs in the first $d-1$ checks. Such a fault is indexed by
    \begin{itemize}
        \item The index $k \in [d-1]$ of the regular check that experiences the fault;
        \item The outer code coordinate $i \in [n]$ to which the data error is introduced;
        \item The value $\lambda \in \mathbb{F}_q^*$ of the $Z$ error introduced to the $i$'th coordinate;
        \item The value $\eta \in \mathbb{F}_q^*$ of the measurement error.
    \end{itemize}
    The number of these faults causing time-like failure is at most
    \begin{equation}
        (q-1)\cdot N_2(n,d,\balpha,\beta^{(X)}),
    \end{equation}
    where
    \begin{equation}
        N_2(n,d,\balpha,\beta^{(X)}) = \left|\left\{(k,i,\mu) \in [d-1]\times[n]\times\mathbb{F}_q^*:\Pi_{\leq k}h_i +\mu e_k \in \ker(\beta^{(X)})\setminus\{0\}\right\}\right|,
    \end{equation}
    where $\Pi_{\leq k}$ is a $(d-1)\times (d-1)$ matrix projecting onto the first $k$ coordinates, $h_i$ is the $i$'th column of $H_X$, and $e_k$ is the vector with a $1$ in the $k$'th position, and $0$ elsewhere.
\end{claim}
\begin{remark}
    One can again check that the result is invariant of rescaling columns of $H_X$ by non-zero values, and so again the result applies equally well to the $Z$ checks as to the $X$ checks.
\end{remark}
\begin{proof}[Proof of Claim~\ref{claim:tlf_case_2_count}]
   Without loss of generality, the error on the data entering the self-consistent round is trivial. Suppose that for any $k \in [d-1]$, the $k$'th check introduces the $Z$ error $\lambda e_i \in \mathbb{F}_q^n$, and experiences the (non-trivial) measurement error $\eta$. Then, the syndrome extracted by the $d-1$ regular checks is $y = \lambda \Pi_{>k}h_i + \eta e_k$.

   After the $d-1$ regular checks, the syndrome for the error on the data is $\lambda h_i$, and since the $M$ linear combinations of checks are extracted without fault, they measure values $z  = \lambda \beta^{(X)}\cdot h_i$. Self-consistency then occurs if and only if 
   \begin{equation}
       z = \beta^{(X)}\cdot y \iff \lambda\Pi_{\leq k}h_i+\eta e_k \in \ker(\beta^{(X)}) \iff \Pi_{\leq k}h_i + \eta\lambda^{-1}e_k \in \ker(\beta^{(X)}).
   \end{equation}
   Given a self-consistent round, time-like failure occurs if and only if $y$ is not the syndrome for the error on the data at any point just before or during the round, which is if and only if 
   \begin{equation}
       \lambda\Pi_{>k}h_i+\eta e_k \notin \{0, \lambda h_i\}\iff \lambda \Pi_{\leq k}h_i + \eta e_k \notin \{0,\lambda h_i\} \iff \Pi_{\leq k}h_i + \eta \lambda^{-1}e_k \notin \{0,h_i\}.
   \end{equation}
   Given a self-consistent round, time-like failure thus does not occur if $\Pi_{\leq k}h_i + \eta\lambda^{-1}e_k = 0$ or $h_i$. However, it turns out to be most simple for later arguments to count the case of a self-consistent round with $\Pi_{\leq k}h_i + \eta \lambda^{-1}e_k = h_i$ as time-like failure.\footnote{This is why the statement of Claim~\ref{claim:tlf_case_2_count} has the phrase ``at most'', whereas the statement of Claim~\ref{claim:tlf_case_1_count} did not.} Given a self-consistent round, we have that time-like failure occurs only if
   \begin{equation}
       \Pi_{\leq k}h_i + \eta\lambda^{-1}e_k \neq 0.
   \end{equation}
   Finally, given $(k,i,\mu) \in [d-1]\times[n]\times\mathbb{F}_q^*$ such that $\Pi_{\leq k}h_i+\mu e_k \in \ker(\beta^{(X)})\setminus\{0\}$, there are exactly $(q-1)$ choices of $\eta,\lambda$ such that $\eta\lambda^{-1} = \mu$, and the result follows.   
\end{proof}
We will now make a similar mixing argument as in case $1$ time-like failure to compute the probability of case $2$ time-like failure in expectation over a uniformly random choice of qudit-to-qubit mappings $\mathcal{B}$, and vectors $\nu \in (\mathbb{F}_q^*)^{d-1}$ determining the $M$ linear combinations of checks.
\begin{claim}\label{claim:tlf_case_2_prob}
    In expectation over uniformly random qudit-to-qubit mappings $\mathcal{B}$, and vectors $\nu \in (\mathbb{F}_q^*)^{d-1}$, the probability of case $2$ time-like failure on one self-consistent round of $X$ checks is
    \begin{align}
        \mathbb{E}_{\mathcal{B},\nu}\left\{\mathbb{P}\left[\mathsf{TLF}_2^{(X)}\;|\;\mathcal{B},\nu\right]\right\} &\leq \mathbb{P}\left[XX\right]\cdot n\sum_{k=M+1}^{d-1}\left(\frac{D_k}{(q-1)^k}+\frac{1}{(q-1)^{k-1}}\begin{pmatrix}
            k-2\\M-1
        \end{pmatrix}\right),
    \end{align}
    where
        \begin{equation}
        D_k = \begin{cases}0 &\text{ if } k \leq M\\
        \sum_{t=0}^{k-M-1}(-1)^t\begin{pmatrix}
            k\\t
        \end{pmatrix}\left(q^{k-M-t}-1\right)&\text{ if } k \geq M+1\end{cases}.
    \end{equation}
\end{claim}
\begin{remark}\label{rmk:tlf_case_2_prob_leading_order}
    Similarly to as in the case of Remark~\ref{rmk:tlf_case_1_prob_leading_order}, the formula behaves like
    \begin{equation}
        \sim2\cdot\mathbb{P}\left[XX\right]\cdot \frac{n(d-1-M)}{(q-1)^M}.
    \end{equation}
\end{remark}
\begin{proof}[Proof of Lemma~\ref{claim:tlf_case_2_prob}]
    Referring back to Claim~\ref{claim:tlf_case_2_count}, case $2$ time-like failures are indexed by some $(k,i,\lambda,\eta) \in [d-1]\times[n]\times\mathbb{F}_q^*\times\mathbb{F}_q^*$. For a given choice of $\beta^{(X)}$, we denote the set of these causing time-like failure as $\mathsf{TLF}_2^{(X)}(\beta^{(X)})$. We fix some valid choice of $\hat{\beta}^{(X)}$, and consider $\beta^{(X)} = \hat{\beta}^{(X)}*\nu$.

    Given some $(k,i,\lambda,\eta) \in \mathsf{TLF}_1^{(X)}(\hat{\beta}^{(X)}*\nu)$, the probability of the time-like failure to which it corresponds is denoted
    \begin{equation}
        \mathbb{P}\left[\text{Mixed}^{(k)},i,Z^\lambda, \eta\;|\;\mathcal{B}\right], 
    \end{equation}
    which is the probability that the $k$'th check of the round experiences a measurement error $\eta$, while also introducing a data error $\lambda e_i$.

    Now, for a given $(k,i) \in [d-1]\times [n]$, define the set
    \begin{equation}
        \mathsf{TLF}_2^{(X)}[k,i](\beta^{(X)}) \coloneq \left\{\mu \in \mathbb{F}_q^*:\Pi_{\leq k}h_i + \mu e_k \in \ker(\beta^{(X)})\setminus\{0\}\right\}.
    \end{equation}
    Referring back to Claim~\ref{claim:tlf_case_2_count}, we see that the full set $\mathsf{TLF}_2^{(X)}(\beta^{(X)})$ may be constructed by iterating over all $(k,i) \in [d-1]\times[n]$, then iterating over all $\mu \in \mathsf{TLF}_2^{(X)}[k,i](\beta^{(X)})$, and then setting $\lambda,\eta$ to be each of the $q-1$ choices for which $\eta\lambda^{-1} = \mu$. Given a choice of $\mathcal{B}$ and $\nu$, the probability of case $2$ time-like failure is therefore
    \begin{align}
        \mathbb{P}\left[\mathsf{TLF}_2^{(X)}\;|\;\mathcal{B},\nu\right] &= \sum_{(k,i) \in [d-1]\times[n]}\sum_{\mu \in \mathsf{TLF}_2^{(X)}[k,i](\hat{\beta}^{(X)}*\nu)}\sum_{\lambda \in \mathbb{F}_q^*}\mathbb{P}\left[\text{Mixed}^{(k)},i,Z^\lambda,\mu\lambda\;|\;\mathcal{B}\right]\\
        &= \sum_{(k,i) \in [d-1]\times[n]}\sum_{\lambda \in \mathbb{F}_q^*}\sum_{\mu \in \mathsf{TLF}_2^{(X)}[k,i](\hat{\beta}^{(X)}*\nu)}\mathbb{P}\left[\text{Mixed}^{(k)},i,Z^\lambda,\mu\lambda\;|\;\mathcal{B}\right].
    \end{align}
    Taking the expectation over a uniformly random choice of qudit-to-qubit mappings $\mathcal{B}$, and vectors $\nu \in (\mathbb{F}_q^*)^{d-1}$, we have
    \begin{equation}
        \mathbb{E}_{\mathcal{B},\nu}\left\{\mathbb{P}\left[\mathsf{TLF}_2^{(X)}\;|\;\mathcal{B},\nu\right]\right\} = \sum_{(k,i) \in [d-1]\times[n]}\sum_{\lambda \in \mathbb{F}_q^*}\mathbb{E}_{\mathcal{B},\nu}\left\{\sum_{\mu \in \mathsf{TLF}_2^{(X)}[k,i](\hat{\beta}^{(X)}*\nu)}\mathbb{P}\left[\text{Mixed}^{(k)},i,Z^\lambda,\mu\lambda\;|\;\mathcal{B}\right]\right\}.
    \end{equation}
    The calculation of the expectation on the right-hand side is deferred to the proof of the following lemma in Appendix~\ref{sec:tlf_proofs}.
    \begin{lemma}\label{lem:tlf_case_2}
    For every $(k,i,\lambda) \in [d-1]\times[n]\times\mathbb{F}_q^*$, we have
        \begin{equation}
            \mathbb{E}_{\mathcal{B},\nu}\left\{\sum_{\mu \in \mathsf{TLF}_2^{(X)}[k,i](\hat{\beta}^{(X)}*\nu)}\mathbb{P}\left[\text{Mixed}^{(k)},i,Z^\lambda,\mu\lambda\;|\;\mathcal{B}\right]\right\} \leq\left(\frac{D_k}{(q-1)^{k+1}}+\frac{1}{(q-1)^k}\begin{pmatrix}
            k-2\\M-1
        \end{pmatrix}\right)\mathbb{P}\left[XX\right], 
        \end{equation}
        where
        \begin{equation}
        D_k = \begin{cases}0 &\text{ if } k \leq M\\
        \sum_{t=0}^{k-M-1}(-1)^t\begin{pmatrix}
            k\\t
        \end{pmatrix}\left(q^{k-M-t}-1\right)&\text{ if } k \geq M+1\end{cases}.
    \end{equation}
    \end{lemma}
    \noindent This lemma gives us
    \begin{align}
        \mathbb{E}_{\mathcal{B},\nu}\left\{\mathbb{P}\left[\mathsf{TLF}_2^{(X)}\;|\;\mathcal{B},\nu\right]\right\} &\leq \mathbb{P}\left[XX\right]\sum_{(k,i) \in [d-1]\times[n]}\sum_{\lambda\in \mathbb{F}_q^*}\left(\frac{D_k}{(q-1)^{k+1}}+\frac{1}{(q-1)^k}\begin{pmatrix}
            k-2\\M-1
        \end{pmatrix}\right)\\
        &=\mathbb{P}\left[XX\right]\cdot n\sum_{k=M+1}^{d-1}\left(\frac{D_k}{(q-1)^k}+\frac{1}{(q-1)^{k-1}}\begin{pmatrix}
            k-2\\M-1
        \end{pmatrix}\right),
    \end{align}
    as required.
\end{proof}

\subsubsection{Case 3: \texorpdfstring{$M+1$}{} Measurement Errors}\label{subsubsec:tlf_case3}
We now conclude with the case $3$ time-like failure: those caused by measurements only. From Claim~\ref{claim:tlf_meas_only}, we know that $M+1$ measurement errors are required for time-like failure. Again, the reason that this one must be considered is because, while more measurement errors are required, measurement errors are more likely than cat data errors/mixed errors.
\begin{claim}\label{claim:tlf_case_3_prob}
    In expectation over vectors $\nu \in (\mathbb{F}_q^*)^{d-1}$, the probability of case $3$ time-like failure on one self-consistent round of $X$ checks is
    \begin{equation}
        \frac{d-1}{q-1}\cdot\mathbb{P}\left[\text{Cat}\right]^{M+1},
    \end{equation}
    up to lower order terms, unless $M = 1$, in which case it is at most
    \begin{equation}
        \frac{1}{q-1}\begin{pmatrix}
            d\\2
        \end{pmatrix}\cdot\mathbb{P}\left[\text{Cat}\right]^{2},
    \end{equation}
    up to lower order terms.
    Here, $\mathbb{P}\left[\text{Cat}\right]$ is the probability of any error on an accepted cat state, or any fault in its consumption. That is,
    \begin{equation}
        \mathbb{P}\left[\text{Cat}\right] = \mathbb{P}\left[\text{Cat Creation}\right] + \mathbb{P}\left[\text{Cat Consumption}\right]
    \end{equation}
    where the former term was defined in Equation~\eqref{eq:cat_creation_error_prob}, and $\mathbb{P}\left[\text{Cat Consumption}\right]$ is the probability of any error in the consumption of a cat state. This is
    \begin{equation}
        \mathbb{P}\left[\text{Cat Consumption}\right] = n\cdot \mathbb{P}\left[XX\right],
    \end{equation}
    where $\mathbb{P}\left[XX\right]$ is the probability of any error in the qudit $XX$ measurement, which we recall is a fixed operation executed according to (the dual of) Sub-Routine~\ref{alg:qudit_ZZ_meas}.
\end{claim}
\begin{proof}
    We consider again the ``time-like'' code with parity-check matrix $\begin{pmatrix}
        \hat{\beta}^{(X)}*\nu & I
    \end{pmatrix}$. We note that, since we choose $\hat{\beta}^{(X)}$ such that the code with parity-check matrix $\begin{pmatrix}
        \hat{\beta}^{(X)}&I
    \end{pmatrix}$ has distance $M+1$, this code has distance $M+1$ also, for any $\nu$. We know from Claim~\ref{claim:tlf_meas_only} that case $3$ time-like failures correspond to non-trivial sequences of measurement errors forming codewords of this code, which are then necessarily formed of at least $M+1$ measurement errors. To obtain the leading order contribution to the failure probability, we consider patterns of $M+1$ measurement errors only. 

    The code with parity-check matrix $\begin{pmatrix}\hat{\beta}^{(X)}*\nu & I\end{pmatrix}$ is a code with $M$ (linearly independent) parity checks, and distance $M+1$, and is therefore a maximum distance separable (MDS) code. The weight distribution of an MDS code is exactly characterised~\cite{macwilliams1977theory}. In particular, given any set of $M+1$ coordinates, there are exactly $q-1$ codewords of weight $M+1$ supported at those coordinates, and they may be obtained from each other by multiplications by scalars in $\mathbb{F}_q^*$.
    
    Consider any set $S \subseteq [d-1+M]$, where $|S| = M+1$. Let $S_1 \subseteq [d-1]$ denote the support of $S$ in the first $d-1$ elements, and let $S_2 \subseteq [M]$ denote the support of $S$ in the latter $M$ elements. Given a choice of $S$, denote an arbitrary codeword of the code with parity-check matrix $\begin{pmatrix}
        \hat{\beta}^{(X)}&I
    \end{pmatrix}$ that is non-zero exactly at $S$ via 
    \begin{equation}
        \begin{pmatrix}
        x(S_1) & y(S_2)
    \end{pmatrix},
    \end{equation}
    where $x(S_1) \in \mathbb{F}_q^{d-1}$ and $y(S_2) \in \mathbb{F}_q^M$. Then all of the codewords of that code supported exactly on $S$ are 
    \begin{equation}
        \eta\begin{pmatrix}
        x(S_1) & y(S_2)
    \end{pmatrix}\text{ for } \eta \in \mathbb{F}_q^*.
    \end{equation}
    Moreover, all of the codewords of the code with parity-check matrix $\begin{pmatrix}
        \hat{\beta}^{(X)}*\nu & I
    \end{pmatrix}$ supported exactly at $S$ are then
    \begin{equation}
        \eta\begin{pmatrix}
        x(S_1)*\nu^{-1} & y(S_2)
    \end{pmatrix}\text{ for } \eta \in \mathbb{F}_q^*,
    \end{equation}
    where given vectors $v,w$, we have $v*w$ denoting the componentwise multiplication of $v$ and $w$, and given the vector $\nu$ (which has all non-zero entries), $\nu^{-1}$ denotes its componentwise inversion.

    Given a particular choice of $\nu$, the probability of case $3$ time-like failure is
    \begin{equation}\label{eq:case_3_tlf_prob_given_nu}
        \mathbb{P}\left[\mathsf{TLF}_3^{(X)}\;|\;\nu\right] = \sum_{\substack{S \subseteq [d-1+M]\\ |S| = M+1}}\sum_{\eta \in \mathbb{F}_q^*}\left(\prod_{j \in S_1}\mathbb{P}\left[\text{Meas}^{(j,1)}, \eta x(S_1)_j\nu_j^{-1}\right]\right)\left(\prod_{k \in S_2}\mathbb{P}\left[\text{Meas}^{(k,2)}, \eta y(S_2)_k\right]\right).\footnote{We take the convention that an empty product is equal to $1$.}
    \end{equation}
    Here, $\mathbb{P}\left[\text{Meas}^{(j,1)},\mu\right]$ denotes the probability of the $j$'th check in the first set of $d-1$ checks experiencing a measurement error $\mu \in \mathbb{F}_q^*$, and similarly for $\mathbb{P}\left[\text{Meas}^{(k,2)},\mu\right]$. Taking the expectation over $\nu$, we get
    \begin{multline}\label{eq:case_3_tlf_prob_expectation_nu}
        \mathbb{E}_{\nu}\left\{\mathbb{P}\left[\mathsf{TLF}_3^{(X)}\;|\;\nu\right]\right\} = \\\sum_{\substack{S \subseteq [d-1+M]\\ |S| = M+1}}\sum_{\eta \in \mathbb{F}_q^*}\left(\prod_{j \in S_1}\mathbb{E}_{\nu_j}\left\{\mathbb{P}\left[\text{Meas}^{(j,1)}, \eta x(S_1)_j\nu_j^{-1}\right]\right\}\right)\left(\prod_{k \in S_2}\mathbb{P}\left[\text{Meas}^{(k,2)}, \eta y(S_2)_k\right]\right).
    \end{multline}
We then have
\begin{equation}
    \mathbb{E}_{\nu_j}\left\{\mathbb{P}\left[\text{Meas}^{(j,1)},\eta x(S_1)_j\nu_j^{-1}\right]\right\} = \frac{1}{q-1}\cdot\mathbb{P}\left[\text{Meas}^{(j,1)},*\right],
\end{equation}
where the right-hand side denotes the probability of any non-trivial measurement error in the $j$'th check of the first set of $d-1$ checks. This is true because, choosing a uniformly random $\nu_j \in \mathbb{F}_q^*$ yields a uniformly random $\eta x(S_1)_j\nu_j^{-1} \in \mathbb{F}_q^*$, noting that $\eta x(S_1)_j \neq 0$. In turn, we have
\begin{equation}
    \mathbb{P}\left[\text{Meas}^{(j,1)},*\right] \leq \mathbb{P}\left[\text{Cat}\right].
\end{equation}This gives us
\begin{equation}\label{eq:tlf_case_3_break_sum}
    \mathbb{E}_\nu\left\{\mathbb{P}\left[\mathsf{TLF}_3^{(X)}\;|\;\nu\right]\right\} \leq \sum_{\substack{S \subseteq [d-1+M]\\|S|=M+1}}\left(\frac{\mathbb{P}\left[\text{Cat}\right]}{q-1}\right)^{|S_1|}\sum_{\eta \in \mathbb{F}_q^*}\left(\prod_{k \in S_2}\mathbb{P}\left[\text{Meas}^{(k,2)},\eta y(S_2)_k\right]\right).
\end{equation}
We are now going to break the sum on the right-hand side over $S$ into parts where $S_2$ is empty, and $S_2$ is non-empty. We first handle cases of $S$ where $S_2$ is non-empty, which is equivalent to $|S_1| \neq M+1$. For each such $S$, pick any coordinate in $S_2$, say $s \in S_2$. We take the naive upper bound
\begin{equation}\label{eq:pessimistic}
    \mathbb{P}\left[\text{Meas}^{(k,2)},\eta y(S_2)_k\right] \leq \mathbb{P}\left[\text{Cat}\right]
\end{equation}
for all $k \neq s$, which holds since any measurement error on any check requires an error on the corresponding cat state or an error in its consumption. The contribution from these terms is
\begin{equation}
    \sum_{\substack{S \subseteq [d-1+M]\\|S| = M+1\\|S_1| \neq M+1}}\left(\frac{\mathbb{P}\left[\text{Cat}\right]}{q-1}\right)^{|S_1|} \mathbb{P}\left[\text{Cat}\right]^{|S_2|-1}\sum_{\eta \in \mathbb{F}_q^*}\mathbb{P}\left[\text{Meas}^{(s,2)},\eta y(S_2)_s\right].
\end{equation}
Then, $\sum_{\eta \in \mathbb{F}_q^*}\mathbb{P}\left[\text{Meas}^{(s,2)}, \eta y(S_2)_s\right] \leq \mathbb{P}\left[\text{Cat}\right]$, since the left-hand side is a sum over probabilities of distinct measurement errors. These terms are then at most
\begin{align}
    \sum_{\substack{S \subseteq [d-1+M]\\|S| = M+1\\|S_1|\neq M+1}}\left(\frac{\mathbb{P}\left[\text{Cat}\right]}{q-1}\right)^{|S_1|}\mathbb{P}\left[\text{Cat}\right]^{|S_2|}
    &= \mathbb{P}\left[\text{Cat}\right]^{M+1}\cdot\sum_{\substack{S \subseteq [d-1+M]\\|S| = M+1\\|S_1|\neq M+1}}\frac{1}{(q-1)^{|S_1|}}.
\end{align}
For these terms we keep only the leading order term given by $|S_1| = 1$, which is
\begin{equation}
    \mathbb{P}\left[\text{Cat}\right]^{M+1}\cdot\frac{d-1}{q-1}.
\end{equation}
We now return to Equation~\eqref{eq:tlf_case_3_break_sum}, and consider the contribution to the right-hand side from terms for which $|S_1| = M+1$, so that $S_2$ is empty. These terms are simply
\begin{equation}
    \sum_{\substack{S \subseteq [d-1+M]\\|S| = M+1\\|S_1| = M+1}}\left(\frac{\mathbb{P}\left[\text{Cat}\right]}{q-1}\right)^{M+1}\sum_{\eta \in \mathbb{F}_q^*} = \begin{pmatrix}
        d-1\\M+1
    \end{pmatrix}\frac{\mathbb{P}\left[\text{Cat}\right]^{M+1}}{(q-1)^M},
\end{equation}
which are sub-leading unless $M = 1$. If $M=1$, then combining the terms resulting from $S_2$ non-empty and $S_2$ empty gives the claimed value because $\begin{pmatrix}
    d-1\\2
\end{pmatrix} + d-1 = \begin{pmatrix}
    d\\2
\end{pmatrix}$.
\end{proof}
\begin{remark}\label{rmk:case_3_tlf_prob_pessimistic}
    This formula is likely very pessimistic, since we lose a lot in our upper bound in Equation~\eqref{eq:pessimistic}. Under uniform errors, for example, we would have instead obtained the formula for case $3$ time-like failure
    \begin{equation}
        \frac{1}{(q-1)^M}\begin{pmatrix}
            d-1+M\\M+1
        \end{pmatrix}\mathbb{P}\left[\text{Cat}\right]^{M+1},
    \end{equation}
    which is much stronger for $M > 1$. However, what we have will suffice for our purposes.
\end{remark}
\begin{remark}
    There is one subtlety which we have neglected to treat, namely that the cat states in the latter $M$ checks themselves change when $\nu$ changes, and so strictly speaking the distribution of measurement errors on the latter set of $M$ checks depends explicitly on $\nu$. This is because, when we change $\nu$, the compilation of the cat state can change, with slightly more or fewer whole in-module or half in-module measurements in each $Z^\alpha Z^\beta$ measurement, in a different order. This means that, strictly speaking, in Equation~\eqref{eq:case_3_tlf_prob_given_nu}, the values $\mathbb{P}\left[\text{Meas}^{(k,2)},\eta y(S_2)_k\right]$ should be $\mathbb{P}\left[\text{Meas}^{(k,2)},\eta y(S_2)_k\;|\;\nu\right]$. When we take the expectation over $\nu$ in Equation~\eqref{eq:case_3_tlf_prob_expectation_nu}, we cannot make progress simply by calculating $\mathbb{E}_{\nu_j}\left\{\mathbb{P}\left[\text{Meas}^{(j,1)},\eta x(S_1)_j\nu_j^{-1}\right]\right\}$ as we do in the subsequent step. We neglect this point because, in practice, one can further use the randomisation of $\mathcal{B}$ to change the compilation of the $Z^\alpha Z^\beta$ measurements, given any problematic cases. In addition, the whole estimate is likely to be very pessimistic; see Remark~\ref{rmk:case_3_tlf_prob_pessimistic}.
\end{remark}

\subsection{Space-like Failures}\label{subsec:spacelike_failures}

\subsubsection{Overview of Space-Like Failures}\label{subsubsec:overview_slf}

In this subsection, we are going to upper bound the probability of space-like failure, one of the two failure modes of our outer code, aside from time-like failure (handled in Subsection~\ref{subsec:time_like_failures}). The discussion here will be focused on space-like failure for $Z$ errors (detected by $X$ checks), where the discussion for $X$ errors will be identical. In places where it is not immediately clear why the discussion for the $X$ errors is the same, it will be noted.

Let us now recall the setting. Focusing on one outer code block, periodically, its qudit $X$ checks are measured. This proceeds by measuring $d-1$ ``regular'' checks, spanning the space of checks, and then $M$ linear combinations of them. A self-consistent round of $X$ checks is defined as one for which the $d-1+M$ syndrome components satisfy the linear combinations that they should. Immediately retrying all $d-1+M$ checks from the start when an inconsistency is detected, a self-consistent round is eventually achieved. When a self-consistent round of $X$ checks is achieved, we said that time-like failure occurs if the syndrome extracted is not the syndrome for the $Z$ error on the data immediately before or during the round (that is, at any of the $d-1+M$ moments immediately before any of the $d-1+M$ checks were measured). In Subsection~\ref{subsec:time_like_failures}, the probability of time-like failure was upper bounded. 

In this subsection, we upper bound the other failure mode of our code, space-like failure. Now, if time-like failure does not occur, the outer code decoder holds the syndrome for the $Z$ error on the code block at some point just before or during the self-consistent round of $X$ checks. Space-like failure is the event that time-like failure does not occur, but the outer code decoder does not return the exact error for the syndrome it was handed. Note that, on its own, this definition is slightly ambiguous. Indeed, it is possible in principle that the syndrome extracted is the syndrome for multiple distinct $Z$ errors on the data just before or during the self-consistent round. Such an event would be considered a space-like failure of the outer code, although it is neglected in our analysis since this event is much less likely than other events that we count as failures.\footnote{Indeed, this event requires an entire logical operator of the outer code to appear during the self-consistent round, while also keeping the round self-consistent.} We therefore assume that, if time-like failure does not occur, the syndrome handed to the outer code decoder is the exact syndrome for a unique $Z$ error on the outer code block just before or during the self-consistent round. Space-like failure is then the event that the outer code decoder does not return that exact $Z$ error. The aim of this subsection is to upper bound the probability of space-like failure, focusing on one outer code block, and one self-consistent round of $X$ checks.

Given a syndrome, the outer code decoder will return the minimum-weight error with that syndrome: a task that may be achieved efficiently using list decoders for classical Reed-Solomon codes~\cite{sudan1997decoding,guruswami1998improved}. If there are multiple distinct errors of the same minimum weight with the given syndrome, the outer code decoder picks from them randomly.

In order to upper bound the probability of space-like failure on one outer code block, on one self-consistent round of $X$ checks, there are two main components to the calculation. Broadly speaking, these are ``What are the dangerous patterns of $Z$ errors?'' and ``How many $Z$ errors are there?''. More concretely for the former, in Subsubsection~\ref{subsubsec:list_decoding}, we will consider the structure of the classical Reed-Solomon codes. This subsubsection allow us to answer questions like ``Given the distance $7$ code, what fraction of weight-$4$ $Z$ errors are uncorrectable?'', where the answers are usually lower than one naively expects, because we use large-alphabet qudits, and because of the list-decodable structure of our codes.

On the latter, ``How many $Z$ errors are there?'', in Subsubsection~\ref{subsubsec:prob_weight_e_errors}, we consider the $Z$ error during the self-consistent round of $X$ checks, whose syndrome is handed to the outer code decoder. We wish to be able to answer questions like ``What is the probability that that $Z$ error has weight $4$?'' There, the point will be that, assuming time-like failure and space-like failure have not occurred in the past on that outer code block, the $Z$ error may be considered to have been removed from the code block at some time during the previous self-consistent round of $X$ checks, and all the $Z$ errors have accumulated on this outer code block since then.

Finally, Subsubsection~\ref{subsubsec:bounding_prob_slf} will draw Subsubsection~\ref{subsubsec:list_decoding} and Subsubsection~\ref{subsubsec:prob_weight_e_errors} together, enabling us to calculate an upper bound on the probability of space-like failure on this particular outer code block, on this particular self-consistent round of $X$ checks. We will make a mixing argument there, for which the punchline will be that, in expectation over mixing parameters, errors may be treated uniformly from the point of view of space-like failure, and so the calculation of the upper bounding of space-like failure becomes relatively straightforward.

\subsubsection{List Decoding}\label{subsubsec:list_decoding}

The decoder for $Z$ errors on our outer code takes the syndrome as input, which is some vector in $\mathbb{F}_q^{d-1}$, and simply finds the minimum-weight correction of the corresponding classical Reed-Solomon code with that syndrome. If there are multiple minimum-weight errors with that syndrome, it chooses one at random. Now, the classical code that the outer code is trying to decode has length $n$ and distance $d$, and so any error of weight at most $\left\lfloor\frac{d-1}{2}\right\rfloor$ can be reliably corrected, and in fact efficient decoders achieving decoding of any such error have been known for some time~\cite{gorenstein1961class,berlekamp1966nonbinary,welch1983error,massey2003shift,peterson2003encoding}. The striking feature about these codes, however, is that they have a structure enabling a large fraction of errors beyond the usual $\left\lfloor\frac{d-1}{2}\right\rfloor$ radius to be corrected.

This concept first emerged in the work of Sudan in 1997~\cite{sudan1997decoding}, and was broadly improved by the work of Guruswami and Sudan~\cite{guruswami1998improved} (for an overview of the Guruswami-Sudan list decoding algorithm for Reed-Solomon codes, see the review by McEliece~\cite{mcelice2003}). 
The idea is as follows. Given a distance $d$ classical Reed-Solomon code, there is always some way to place errors at $\left\lfloor\frac{d-1}{2}\right\rfloor + 1$  locations that will fool a minium-weight decoder. Indeed, an adversary may place $\left\lfloor\frac{d-1}{2}\right\rfloor + 1$ errors that form a restriction of a weight-$d$ codeword, and a minimum-weight decoder may then mistake this error for its logical complement in that codeword.

However, it turns out that if errors are occurring randomly, the probability that $\left\lfloor\frac{d-1}{2}\right\rfloor+1$ errors line up with some restriction of a codeword is very low, and this effect becomes more noticeable when $n$ (the length of the code) is much less than $q$ (the alphabet size).\footnote{Note that we are in this regime because we consider $q = 2^{11} = 2048$ and $n$ of the order of about $30$ to $80$.} To give one explicit example, consider the scenario of a length $n$ and distance $6$ classical Reed-Solomon code over $\mathbb{F}_q$, on which errors are occurring uniformly at random (i.e., if an error occurs on a coordinate, it is uniformly randomly one of the $q-1$ errors $\eta \in \mathbb{F}_q^*$). It turns out that the probability that a random weight-$3$ error has the same syndrome as a distinct weight-$3$ error is $\frac{1}{(q-1)^2}\begin{pmatrix}
    n-3\\3
\end{pmatrix}\sim\frac{n^3}{6q^2}$, which is very small for $n \ll q$. These effects go further. For example, if $\frac{n}{q}$ is \textit{very} small, the majority of weight-$4$ errors are correctable on the distance $6$ code.\footnote{In the present context, we are not necessarily in the regime that $\frac{n}{q}$ is small enough that the majority of weight-$4$ errors are correctable on the distance $6$ code.}

\cite{sudan1997decoding} and~\cite{guruswami1998improved} initiated the study of efficient algorithms that can provide the minimum-weight correction for a classical Reed-Solomon code consistent with the syndrome, where the minimum-weight correction consistent with the syndrome in general has weight greater than $\left\lfloor\frac{d-1}{2}\right\rfloor$. These algorithms are called list decoding algorithms because they output a (short) list of possible corrections consistent with the inputted syndrome, often coming with guarantees about the elements that are certain to appear in the list. For example, we are guaranteed that all possible corrections consistent with the syndrome below a certain weight appear in the list. Because of the structure of the code, one expects that in a typical case there would be a unique minimum-weight correction in the list; in the context of the above example, when decoding a distance $6$ code of length $n$, if a random weight-$3$ error occurs, there would typically only be one weight-$3$ correction in the list (the actual error), assuming that $\frac{n^3}{6q^2} \ll 1$. In this work, we can just assume that one of the very well-studied and efficient list decoders is used, like the Guruswami-Sudan list decoder~\cite{guruswami1998improved}, so that the minimum weight correction is always applied (and if there are multiple we can choose at random). Let us emphasise, however, that the term ``list decoding'' is not just talking about the properties of an algorithm, it is also referring to the structure of the code that makes it ``list decodable''.

For the remainder of this subsubsection, we are going to consider weight $e$ errors on distance $d$ codes, and consider situations in which they are correctable or uncorrectable, for $e > \left\lfloor\frac{d-1}{2}\right\rfloor$. Note that an error is called uncorrectable if it is not the minimum-weight error with its given syndrome. However, an uncorrectable error may still be corrected by the outer code decoder if the outer code decoder correctly guesses the error (at random) from a list of minimum-weight errors with the same syndrome.

\paragraph{Distance 3 Code}

Unfortunately, list decoding does not have much effect at distance $d=3$, and we take all weight-$2$ errors to be space-like failures.\footnote{The story would be different if one considered soft information decoding, but in this paper we are just consider an outer code decoder whose input is the syndrome.}

\paragraph{Distance 4 Code}

List decoding starts to become more interesting at distance $d=4$. We take all weight-$3$ errors as space-like failures, but not all weight-$2$ errors. Recalling that an error is called correctable if it is the unique minimum-weight error with its syndrome, we have the following.
\begin{claim}\label{claim:d=4_list_decoding}
    On a length $n$ and distance $4$ classical Reed-Solomon code over $\mathbb{F}_q$, the fraction of weight $2$ errors that are uncorrectable is at most
    \begin{equation}
        \frac{(n-2)(n-3)}{2(q-1)}.
    \end{equation}
\end{claim}
\begin{proof}
    In the distance $4$ code, a weight-$2$ error is uncorrectable only if it is mistaken for another weight-$2$ error, which is only if it is the restriction of some weight-$4$ codeword. Being a maximum distance separable (MDS) code, the weight distribution of the code is exactly characterised~\cite{macwilliams1977theory}. In particular, its number of weight-$4$ codewords is
    \begin{equation}
        (q-1)\begin{pmatrix}
            n\\4
        \end{pmatrix}.
    \end{equation}
    The number of weight-$2$ restrictions of weight-$4$ codewords is therefore
    \begin{equation}
        6(q-1)\begin{pmatrix}
            n\\4
        \end{pmatrix}.
    \end{equation}
    The total number of weight-$2$ errors is
    \begin{equation}
        (q-1)^2\begin{pmatrix}
            n\\2
        \end{pmatrix},
    \end{equation}
    and so the fraction of weight-$2$ errors that are uncorrectable is at most
    \begin{equation}
        \frac{6\begin{pmatrix}
            n\\4
        \end{pmatrix}}{(q-1)\begin{pmatrix}
            n\\2
        \end{pmatrix}} = \frac{(n-2)(n-3)}{2(q-1)}.
    \end{equation}
\end{proof}
For the errors that are uncorrectable, we need more fine-grained information. That is, when a weight-$2$ error shares a syndrome with a distinct weight-$2$ error, we need to know with how many distinct weight-$2$ errors it shares its syndrome. This information is needed because, in a situation where a weight-$2$ error shares a syndrome with $k$ distinct weight-$2$ errors, the outer decoder picks at random from those $k+1$ choices. The following quantity gives the fraction of weight-$2$ errors that share a syndrome with exactly $k$ distinct weight-$2$ errors:
\begin{multline}
    F_{2\to k\times2}(n,d=4,\balpha) \coloneq\\ \frac{\left|\left\{ e \in \mathbb{F}_q^n:\; |e| = 2, \;H_Xe = H_Xe', \text{ for exactly } k\text{ distinct } e' \in \mathbb{F}_q^n\text{ with } \;|e'| = 2, \;e' \neq e\right\}\right|}{(q-1)^2\begin{pmatrix}
        n\\2
    \end{pmatrix}}.
\end{multline}
Notice that this does not depend on the values $\bu$, and so in particular is the same for the $X$ checks as for the $Z$ checks. This quantity \textit{can} depend on the evaluation points $\balpha$. We are going to go on to make explicit choices of the evaluation points $\balpha$, and we will estimate these fractions in a computer for the choices of $\balpha$; see Subsection~\ref{subsec:evaluation_points}.

\paragraph{Distance 5 Code}

For the distance $5$ codes, all weight-$1$ and $2$ errors are trivially correctable. However, the situation for weight-$3$ errors is a little more complicated. We have the following.
\begin{claim}\label{claim:d=5_list_decoding}
    On a length $n$ and distance $5$ classical Reed-Solomon code over $\mathbb{F}_q$, the fraction of weight-$3$ errors that have the same syndrome as a weight-$2$ error is at most
    \begin{equation}
        \frac{(n-3)(n-4)}{2(q-1)^2}.
    \end{equation}
\end{claim}
\begin{proof}
    This follows by the same logic as Claim~\ref{claim:d=4_list_decoding}, noting that the number of weight-$3$ restrictions of weight-$5$ codewords is
    \begin{equation}
        10(q-1)\begin{pmatrix}
            n\\5
        \end{pmatrix},
    \end{equation}
    and that the number of weight-$3$ errors is $(q-1)^3\begin{pmatrix}
        n\\3
    \end{pmatrix}$.
\end{proof}
If a weight-$3$ error shares a syndrome with a weight-$2$ error, that weight-$3$ error is considered a space-like failure, because the outer code decoder will always give the correction as the weight-$2$ error.

On the other hand, on the distance $5$ code, a weight-$3$ error may also share its syndromes with a distinct weight-$3$ error. Like for the weight-$2$ errors on the distance $4$ code, we wish to consider a quantity capturing the fraction of weight-$3$ errors that share a syndrome with exactly $k$ distinct weight-$3$ errors. That is, we define
\begin{multline}
    F_{3\to k\times3}(n,d=5,\balpha) \coloneq \\\frac{\left|\left\{ e \in \mathbb{F}_q^n:\; |e| = 3, \;H_Xe = H_Xe', \text{ for exactly } k\text{ distinct } e' \in \mathbb{F}_q^n\text{ with } \;|e'| = 3, \;e' \neq e\right\}\right|}{(q-1)^3\begin{pmatrix}
        n\\3
    \end{pmatrix}}.
\end{multline}
Again, we choose explicit $\balpha$ and estimate these fractions in a computer for the choice of $\balpha$; see Subsection~\ref{subsec:evaluation_points}.

Weight-$4$ errors are taken as space-like failures on the distance $5$ code.

\paragraph{Distance 6 Code}
For the distance $6$ code, via essentially the same argument as that in Claims~\ref{claim:d=4_list_decoding} and~\ref{claim:d=5_list_decoding}, we have that the fraction of weight-$3$ errors that are uncorrectable (have the same syndrome as a distinct weight-$3$ error) is at most
\begin{equation}
    \frac{(n-3)(n-4)(n-5)}{6(q-1)^2}.
\end{equation}
We similarly define
\begin{multline}
    F_{3\to k\times3}(n,d=6,\balpha) \coloneq \\\frac{\left|\left\{ e \in \mathbb{F}_q^n:\; |e| = 3, \;H_Xe = H_Xe', \text{ for exactly } k\text{ distinct } e' \in \mathbb{F}_q^n\text{ with } \;|e'| = 3, \;e' \neq e\right\}\right|}{(q-1)^3\begin{pmatrix}
        n\\3
    \end{pmatrix}}.
\end{multline}
At our alphabet size, and the lengths we are interested in, we simply take a weight-$4$ error to be a space-like failure on the distance $6$ code, because such an error generically shares a syndrome with many distinct weight-$4$ errors.

\paragraph{Distance 7 Code} For the distance $7$ code, via similar reasoning to the above, the fraction of weight-$4$ errors that share a syndrome with a weight-$3$ error (and so are certainly uncorrectable) is at most
\begin{equation}
    \frac{(n-4)(n-5)(n-6)}{6(q-1)^3}.
\end{equation}
We similarly define the quantity
\begin{multline}
    F_{4\to k\times 4}(n,d=7,\balpha) \coloneq \\\frac{\left|\left\{ e \in \mathbb{F}_q^n:\; |e| = 4, \;H_Xe = H_Xe', \text{ for exactly } k\text{ distinct } e' \in \mathbb{F}_q^n\text{ with } \;|e'| = 4, \;e' \neq 4\right\}\right|}{(q-1)^4\begin{pmatrix}
        n\\4
    \end{pmatrix}}.
\end{multline}
Weight-$5$ errors are taken as space-like failures on the distance $7$ code.

\paragraph{Distance 8 Code} Again, via the same reasoning, for the distance $8$ code, the fraction of weight-$4$ errors that share a syndrome with another weight-$4$ error is
\begin{equation}
    \frac{(n-4)(n-5)(n-6)(n-7)}{24(q-1)^3}.
\end{equation}
We do not attempt to estimate a quantity $F_{4\to k \times 4}(n,d=8,\balpha)$, since reliably estimating these quantities becomes very computationally expensive. We therefore take the pessimistic assumption that any weight-$4$ error sharing a syndrome with a distinct weight-$4$ error on the distance $8$ code is a space-like failure. 

Similarly, estimating a quantity $F_{5\to k \times 5}(n,d=8,\balpha)$ would be very computationally expensive, and so we again take the pessimistic assumption that any weight-$5$ error definitely causes space-like failure if it shares a syndrome with any distinct error with weight at most $5$.\footnote{This is not a pessimistic assumption if it shares a syndrome with an error of weight $3$ or $4$, but it is if it shares a syndrome with a distinct error of weight $5$, since in this case the correct weight-$5$ error might be correctly picked from the list.} With this, we can upper bound the fraction of uncorrectable weight-$5$ errors by considering the number of weight-$8$, $9$ and $10$ codewords in a classical Reed-Solomon code of length $n$ and distance $8$, which are
\begin{align}
    W_{8,n,d=8} &= (q-1)\begin{pmatrix}
        n\\8
    \end{pmatrix}\\
    W_{9,n,d=8} &= (q-1)(q-8)\begin{pmatrix}
        n\\9
    \end{pmatrix}\\
    W_{10,n,d=8} &= (q-1)(q^2-9q+36)\begin{pmatrix}
        n\\10
    \end{pmatrix},
\end{align}
since Reed-Solomon codes are MDS codes, which have fully determined weight distributions~\cite{macwilliams1977theory}. The number of restrictions of such codewords to five coordinates is then at most
\begin{equation}
    56\cdot W_{8,n,d=8} + 126\cdot W_{9,n,d=8} + 252\cdot W_{10,n,d=8},
\end{equation}
since $\begin{pmatrix}
    8\\5
\end{pmatrix} = 56$, and so on. The fraction of uncorrectable weight-$5$ errors is therefore at most
\begin{equation}
    \frac{56\cdot W_{8,n,d=8} + 126\cdot  W_{9,n,d=8} + 252\cdot W_{10,n,d=8}}{(q-1)^5\begin{pmatrix}
        n\\5
    \end{pmatrix}}.
\end{equation}
We always take a weight-$6$ error to be a space-like failure.

\paragraph{Distance 9 Code} On the distance $9$ code, we take any weight-$6$ error to be a space-like failure, and a weight-$5$ error sharing its syndrome with any distinct error of weight at most $5$ to be a space-like failure. For the latter, via similar reasoning to the above, the fraction of uncorrectable weight-$5$ errors is at most
\begin{equation}
    \frac{126\cdot W_{9,n,d=9} + 252\cdot W_{10,n,d=9}}{(q-1)^5\begin{pmatrix}
        n\\5
    \end{pmatrix}},
\end{equation}
where
\begin{align}
    W_{9,n,d=9} &= (q-1)\begin{pmatrix}
        n\\9
    \end{pmatrix}\\
    W_{10,n,d=9} &= (q-1)(q-9)\begin{pmatrix}
        n\\10
    \end{pmatrix}.
\end{align}
 
\subsubsection{The Probability of Weight-\texorpdfstring{$e$}{} Errors}\label{subsubsec:prob_weight_e_errors}    

In order to upper bound the probability of space-like failure on one outer code block, on one self-consistent round of $X$ checks, the other major component of our calculation is to upper bound the probability that the $Z$ error during the self-consistent round has weight $e$, for various integers $e$.\footnote{As always, the notion of weight means ``block weight'' or ``qudit weight'', meaning, on how many gross codes the $Z$ error is supported.} 

In more detail, we recall that when treating space-like failure, we assume time-like failure has not occurred (we have bounded its probability in Subsection~\ref{subsec:time_like_failures}, and so can assume it has not occurred up to that probability), and therefore from this self-consistent round of $X$ checks we have extracted the actual syndrome for a $Z$ error on the data just before or during the given self-consistent round of $X$ checks. As justified above (see Subsubsection~\ref{subsubsec:overview_slf}), we make the simplification that the syndrome corresponds to a unique error on the data just before or during the self-consistent round. We want to calculate the probability that this error has weight $e$.

To do this, we must consider all the contributions to the given $Z$ error over time. Now, assuming that time-like failure and space-like failure have not occurred on this outer code block in the past, the $Z$ error on this outer code block may be considered to have been removed from that code block at some instant during the previous self-consistent round of $X$ checks, and all of the $Z$ errors have accumulated since then. 

Notationally, it is convenient to give a name to the moment just before, or during, the current self-consistent round of $X$ checks where the $Z$ error on the data has had its syndrome extracted by the self-consistent round. That is, we say that the syndrome extracted by the present self-consistent round of $X$ checks is the syndrome for the $Z$ error on the data at time $t_i$. This is one of the times just before one of the $d-1+M$ $X$ checks in the self-consistent round. Such a time exists, because we assume time-like failure has not occurred.\footnote{There will be multiple such moments; one may pick $t_i$ to be one of these moments arbitrarily, say, the first one.} Similarly, we say that the syndrome extracted by the previous self-consistent round of $X$ checks is the syndrome for the $Z$ error on the data at time $t_{i-1}$.\footnote{Again, we can pick an arbitrary suitable time, say, the last one.} The point is that all the $Z$ errors that contribute to the syndrome extracted at time $t_i$ come from operations between time $t_{i-1}$ and time $t_i$.

The $Z$ errors contributing to the syndrome extracted at time $t_i$ come from the following places:
\begin{itemize}
    \item Operations in the consumptions of cat states during the previous self-consistent round, but after time $t_{i-1}$;
    \item Gross codes idling for a long time between syndrome extraction rounds;
    \item Errors in teleportations as the ancillary system moves through the codeblock;
    \item All measurements of $Z$ checks, including those that failed to produce a self-consistent round. This includes $Z$ errors in those cat states that could have spread to the data;\footnote{We emphasise that $Z$ errors created in cat states for $X$ checks cannot spread to the data. Similarly, $X$ errors created in cat states for $Z$ checks cannot spread to the data. These both only create measurement errors. On the converse, $Z$ errors on accepted cat states for the measurement of $Z$ checks spread to the data, and $X$ errors on accepted cat states for the measurement of $X$ checks spread to the data.
    }
    \item Operations in the consumption of cat states that failed to produce a self-consistent round of $X$ checks;
    \item Operations in the consumption of cat states in the self-consistent round of $X$ checks corresponding to time $t_i$, but before time $t_i$.
\end{itemize}
We neglect the first and last points, except for the very last cat state consumption in the first point (that is, the very last cat state consumption in the previous self-consistent round). The reason for this is that $Z$ errors created during a self-consistent round of $X$ checks almost always make the round not self-consistent at all. Interestingly, this argument can be formalised in almost exactly the same way we formalise our Claims~\ref{claim:tlf_case_1_prob} and~\ref{claim:tlf_case_2_prob}, but we omit such a formalisation for brevity. In words, it is very unlikely that a $Z$ error occurring on anything but the last cat state consumption in a round of $X$ checks can keep the round self-consistent. More qualitatively, the probabilities of these events are suppressed in factors of $\frac{1}{q-1}$, like the probabilities of Case $1$ and $2$ time-like failure, so their contribution to the weight of the $Z$ error at time $t_i$ is sub-leading.

Next, it is useful to identify all noise sources that are always the same between times $t_{i-1}$ and $t_i$, that is, they do not change adaptively. These are:
\begin{itemize}
    \item The final cat state consumption of the previous self-consistent round;
    \item The idling operations;
    \item Teleportation procedures.
\end{itemize}
These noise sources are always the same between time $t_{i-1}$ and $t_i$, and so are called the ``static'' noise sources. They all contribute noise transversally to the codeblock, and so they create a total probability of $Z$ error on each qudit of at most
\begin{equation}
    p_{\text{Static}} \coloneq \mathbb{P}\left[ZZ\right]\footnote{Recall that the qudit $XX$ and qudit $ZZ$ measurements may be treated in the same way. The probability of one of them having any fault is denoted $\mathbb{P}\left[ZZ\right]$.} + N_{\text{Outer Round}}\cdot \mathbb{P}\left[\text{idle}\right]\footnote{As mentioned in Subsubsection~\ref{subsubsec:expected_time_create_cat}, the time for one code cycle when gross codes are sitting at idle is always taken to be the same number of rounds, and so is counted as a non-adaptive/static noise source here. Strictly speaking, this noise source is adaptive because the amount of time at idle depends on how many other cat states and rounds get rejected. However, as discussed there, we take the time between syndrome extraction rounds on one block to be fixed at its mean for simplicity, which is a justified approximation given the large number of cat states and rounds between two outer syndrome extractions.} + 6\cdot\mathbb{P}\left[ZZ\right].
\end{equation}
Here, the three terms correspond to the final cat state consumption of the previous self-consistent round, idling, and teleportations, respectively. We have that $\mathbb{P}\left[\text{idle}\right]$ is the probability of error in a gross code idling, and we recall $N_{\text{Outer Round}}$ from Subsubsection~\ref{subsubsec:expected_time_create_cat}. 

Now, consider some set of qudits labelled by $A \subseteq [n]$ of size $|A| = a$. The probability that the static noise sources create $Z$ errors exactly on $A$ is at most
\begin{equation}
    \left(p_{\text{Static}}\right)^a.
\end{equation}
However, we also need to consider the non-static or ``adaptive'' noise sources, that is, those that can change based on how many other errors occur. Concretely, these are the measurements of $Z$ checks, as well as rounds of $X$ checks that failed to be self-consistent. 

We begin with the measurement of $Z$ checks. In analysing these, note we do not consider attempted rounds of $Z$ checks that were rejected for inconsistency while not creating a $Z$ error on the data block, since these do not create a larger $Z$ error, but have sub-leading probability.\footnote{Notice that rounds that get rejected while not creating more errors on the data do still affect the analysis in that they increase the mean time to achieve a self-consistent round of checks, thereby increasing $N_{\text{Outer Round}}$, see Equation~\eqref{eq:expected_round_attempts} and the surrounding analysis.} Further, we make the pessimistic assumption for simplicity that whenever an error is created on the data by an attempted round of $Z$ checks, that round is always declared inconsistent, and thus must be retried.\footnote{This is strictly pessimistic because a pure $Z$ error created on the data would not cause the round to be inconsistent.} Finally, we make the pessimistic assumption that rounds are always found to be inconsistent at the very end (in reality they could be exited early).

We can start by recalling Subsection~\ref{subsec:ler_single_cat}, in particular Subsubsection~\ref{subsubsec:ft_of_qudit_cat}. We recall that an accepted qudit cat state (for measuring $Z$ checks\footnote{This is analogous to the analysis done there for $X$ errors on cat states used for measuring $X$ checks. $X$ errors on cat states used for measuring $X$ checks spread to the data during cat state consumption. Similarly, $Z$ errors on cat states used for measuring $Z$ checks spread to the data during cat state consumption.}) has a $Z$ error on each of its qudits with probability at most $\mathbb{P}\left[ZZ, \mathcal{C}_2\right]$,\footnote{Recall that the story is different when thinking about very large errors on the accepted qudit cat state, that is, errors of weight $>R$ (which requires a fault in the initial non-fault-tolerant preparation of the cat state to go uncaught). However, these are considered failures of the outer code on their own, and are accounted for separately in Subsubsection~\ref{subsubsec:bounding_prob_slf}. Specifically, Equation~\eqref{eq:p_cat_state_failure} accounts for all events where a fault in the initial non-fault-tolerant preparation goes uncaught.} which are then transferred to the code block upon the consumption of the cat state. In addition, when the cat state is consumed, $Z$ errors are created on each qudit with probability at most $\mathbb{P}\left[ZZ\right]$, so that for a single $Z$ check measurement, a $Z$ error is created on each qudit with probability at most
\begin{equation}
    p_{Z, \text{Meas}}\coloneq \mathbb{P}\left[ZZ, \mathcal{C}_2\right] + \mathbb{P}\left[ZZ\right].
\end{equation}
Then, over an attempted round of $d-1+M$ checks, a $Z$ error is created on each qudit with probability at most
\begin{equation}
    p_{Z, \text{Round}} \coloneq (d-1+M)\cdot p_{Z, \text{Meas}}.
\end{equation}
Now fix some particular subset of the qudits labelled by some subset $B \subseteq [n]$ of size $|B| = b$.\footnote{There are $\begin{pmatrix}
    n\\b
\end{pmatrix}$ such subsets, but just pick one for now.} We aim to upper bound the probability that, in attempting to obtain a self-consistent round of $Z$ checks, errors are created exactly on $B$. To leading order, this requires $\leq b$ attempted rounds of $Z$ checks to be declared inconsistent, and exactly one $Z$ error to be created on each qudit of $B$ over the attempted rounds. We must account for all ways that errors can be created on the qudits of $B$ over the attempted rounds of $Z$ checks. Thus, the probability that $Z$ errors are created exactly on $B$ while we attempt to measure $Z$ checks is at most
\begin{equation}
    B_b(p_{Z,\text{Round}})^b,
\end{equation}
where $B_b$ is the number of ways that the set $\{1, \ldots, b\}$ can be partitioned into ordered subsets. $B_b$ are called the ordered Bell numbers (or Fubini numbers)~\cite{oeisA000670},\footnote{$B_b$ admits an alternative, but equivalent definition as the number of weak orderings on a set of $b$ elements~\cite{wiki_ordered_bell_number}, but the definition in terms of partitions is more appropriate for the present purpose.} and the first few are shown in Table~\ref{tab:ordered_bell}.
\begin{table}[ht]
\centering
\begin{tabular}{c|rrrrrrrr}
\hline
$b$ & $0$ & $1$ & $2$ & $3$ & $4$ & $5$ & $6$ & $7$ \\
\hline
$B_b$ &$1$ & $1$ & $3$ & $13$ & $75$ & $541$ & $4683$& $47293$ \\
\hline
\end{tabular}
\caption{The first few ordered Bell numbers.}
\label{tab:ordered_bell}
\end{table}

For some more explanation, consider the case of $b=2$, where we want to upper bound the probability that we create a $Z$ error on two particular qudits of the codeblock, call them qudits $i$ and $j$, while attempting to measure the $Z$ checks. We can either create $Z$ errors on $i$ in the first attempted round and $j$ in the second attempted round, $j$ in the first attempted round and $i$ in the second attempted round, or both $i$ and $j$ in the first attempted round, with other events being sub-leading, giving the upper bound $3(p_{Z,\text{Round}})^2$. Appropriately, $B_3 = 2$ because the number of ways to partition the set $\{1, 2\}$ into ordered subsets is $3$, namely $(\{1\}, \{2\}), (\{2\}, \{1\})$, and $(\{1,2\})$.

The same logic allows us to treat the $Z$ errors created by attempted rounds of $X$ checks that did not end up being self-consistent. Noting that $Z$ errors on accepted cat states used for the measurement of $X$ checks do not spread to the data, upon the measurement of a single $X$ check, a $Z$ error is created on each qudit with probability at most
\begin{equation}
    p_{X,\text{Meas}} \coloneq \mathbb{P}\left[ZZ\right],
\end{equation}
arising from the consumption of the cat state only. Over an attempted round of $d-1+M$ checks, a $Z$ error is created on each qudit with probability at most
\begin{equation}
    p_{X,\text{Round}} \coloneq (d-1+M)\cdot p_{X,\text{Meas}},
\end{equation}
and we go on to find that the probability that $Z$ errors are created exactly on some subset $C \subseteq [n]$ of size $|C| = c$, due to inconsistent rounds of $X$ checks, is at most
\begin{equation}
    B_c\left(p_{X,\text{Round}}\right)^c.
\end{equation}
It only remains to combine the contributions from the static noise sources, the rounds of $Z$ checks, and the inconsistent rounds of $X$ checks. Consider some subset of qudits $T \subseteq [n]$ of size $|T| = e$. What is the probability there are $Z$ errors exactly on $T$? We can consider non-negative integers $a,b,c$ such that $a+b+c = e$. $a$ will be the number of qudits with $Z$ errors created by static sources, $b$ the number of qudits with $Z$ errors created by rounds of $Z$ checks, and $c$ the number of qudits with $Z$ errors created by failed rounds of $X$ checks. There are $\frac{e!}{a!b!c!}$ ways to assign the qudits in $T$ to each of these sources, and given such an assignment, the qudits in $T$ get $Z$ errors on them with probability at most $\left(p_{\text{Static}}\right)^aB_b\left(p_{Z,\text{Round}}\right)^bB_c\left(p_{X,\text{Round}}\right)^c$. Thus, the probability that $Z$ errors are created exactly on the subset $T$ is at most
\begin{equation}
    \sum_{\substack{a+b+c=e\\a,b,c\geq 0}}\frac{e!}{a!b!c!}\left(p_{\text{Static}}\right)^aB_b\left(p_{Z,\text{Round}}\right)^bB_c\left(p_{X,\text{Round}}\right)^c.
\end{equation}
By a union bound, over all subsets $T$ of size $e$, the probability of a weight-$e$ $Z$ error at time $t_i$ is found to be at most
\begin{equation}
    \begin{pmatrix}
        n\\e
    \end{pmatrix}\sum_{\substack{a+b+c=e\\a,b,c\geq 0}}\frac{e!}{a!b!c!}\left(p_{\text{Static}}\right)^aB_b\left(p_{Z,\text{Round}}\right)^bB_c\left(p_{X,\text{Round}}\right)^c.
\end{equation}

\subsubsection{Bounding the Probability of Space-Like Failure}\label{subsubsec:bounding_prob_slf}

In the previous two subsections, we have argued about the effects of list decoding, which tells us about the fraction of each error weight that we are able to correct, and then we argued about the probability that we end up with an error of each weight on the data.

To upper bound the probability of space-like failure, we start by considering cases where $>R$ errors get spread to the data in one cat state consumption. Recalling Subsubsection~\ref{subsubsec:ft_of_qudit_cat}, for one cat state, we recall that we found the probability that $>R$ errors spread to the data in Equation~\eqref{eq:cat_state_failure_prob}, and let us call that expression $p_{\text{Cat State Failure}}$. This was the probability that any fault in the initial non-fault-tolerant checking goes uncaught. Between time $t_{i-1}$ and $t_i$, the probability that any fault in the initial non-fault-tolerant preparation of a cat state goes uncaught is at most
\begin{equation}\label{eq:p_cat_state_failure}
    2(d-1+M)\cdot p_{\text{Cat State Failure}},
\end{equation}
to leading order. This expression is then taken as a contribution to the overall space-like failure probability, but in the remainder of the space-like failure calculation, we may neglect cases where a fault in the non-fault-tolerant cat state preparation goes uncaught.

We want to make a mixing argument now to calculate the probability of space-like failure not resulting from high-weight errors spreading to the data in a cat state consumption. It turns out that, with our mixing arguments, we may treat the qudit errors as uniform from the point of view of space-like failure.
\begin{claim}\label{claim:slf_uniform_errors}
    In expectation over a uniformly random choice of qudit-to-qubit mappings, the probability of any particular $Z$ error at time $t_i$ is what one would obtain if qudit $Z$ errors occurred uniformly.
\end{claim}
\begin{proof}
    Consider any of our codes, and the probability that the error $Z^{\vec{\eta}}$ is the error on the data at time $t_i$, for some $\vec{\eta}\in \mathbb{F}_q^n$, where $\vec{\eta} = (\eta_1, \ldots, \eta_n)$. Given a choice of qudit-to-qubit mappings $\mathcal{B}$, the probability of this may be written
    \begin{equation}
        \mathbb{P}\left[\vec{\eta}\;|\;\mathcal{B}\right].
    \end{equation}
    Note that this is meaningful because the distribution of $Z$ errors on the data at time $t_i$ is always the same.\footnote{The distribution is always the same assuming a time-like failure/space-like failure has not occurred in the past on the given outer codeblock, since in that case all $Z$ errors accumulated on the outer codeblock over one of its rounds.} Similarly to what was discussed for the time-like failures in Subsection~\ref{subsec:time_like_failures}, the qudit-to-qubit mappings $B_i$ governs which qubit error may take place in order for the corresponding qudit error to be $\eta_i$. Taking the expectation over $\mathcal{B}$ gives us
    \begin{equation}
        \mathbb{E}_{\mathcal{B}}\left\{\mathbb{P}\left[\vec{\eta}\;|\;\mathcal{B}\right]\right\} = \frac{1}{(q-1)^{|\vec{\eta}|}}\mathbb{P}\left[|\vec{\eta}|\right],
    \end{equation}
    since if $\eta_i \neq 0$, taking the expectation over $B_i$ mixes $\eta_i$ uniformly over its $q-1$ possibilities. Note that, here, $\mathbb{P}\left[|\vec{\eta}|\right]$ is the probability of having any weight $|\vec{\eta}|$ $Z$ error on the data at time $t_i$. This is the same expression that would be obtained for having the error $\vec{\eta}$ on the data at time $t_i$ if logical errors were uniform.
\end{proof}
In expectation over the choice of $\mathcal{B}$, the calculation of the space-like failure probability becomes straightforward with the ingredients already calculated, because logical errors on gross codes and their operations may be treated as uniform under $\mathbb{E}_{\mathcal{B}}$. For example, consider weight-$3$ errors on the distance $5$ code. If a weight-$3$ error has the same syndrome as a weight-$2$ error, space-like failure certainly occurs, because the weight-$3$ error is always mistaken for the weight-$2$ error. Let the set of weight-$3$ qudit errors sharing a syndrome with a weight-$2$ error be $E$. Then, the probability of having any error in $E$ on the code block at time $t_i$ is at most
\begin{equation}
    \sum_{\vec{\eta} \in E}\mathbb{P}\left[\vec{\eta}\;|\;\mathcal{B}\right].
\end{equation}
Taking the expectation over $\mathcal{B}$ then makes this equal to (using linearity of expectation)
\begin{equation}
    \frac{|E|}{(q-1)^3}\cdot\mathbb{P}\left[3\right] \leq \frac{(n-3)(n-4)}{2(q-1)^2}\cdot\mathbb{P}\left[3\right],
\end{equation}
recalling Claim~\ref{claim:d=5_list_decoding}. This is the contribution to the overall probability of space-like failure of weight-$3$ errors that share a syndrome with a weight-$2$ error, in expectation over the uniformly random choice of $\mathcal{B}$. More generally, we find that the probabilities of various events causing space-like failure, such as the occurrence of a weight-$3$ $Z$ error that has the same syndrome as a weight-$2$ error on the distance $5$ code, as characterised by the set of errors $E$, may be calculated assuming a uniform logical error distribution under $\mathbb{E}_{\mathcal{B}}$, by simply applying Claim~\ref{claim:slf_uniform_errors} and the linearity of expectation.

For another example, weight-$3$ errors can also be mistaken for distinct weight-$3$ errors. If the weight-$3$ error has the same syndrome as one or more distinct weight-$3$ errors, the decoder chooses one at random. If the weight-$3$ error shares its syndrome with $k$ distinct weight-$3$ errors, the probability of a failure in this case is then $\frac{k}{k+1}$. Using a similar argument, in expectation over $\mathcal{B}$, the total contribution to the space-like failure probability from weight-$3$ errors being mistaken for other weight-$3$ errors is at most
\begin{equation}
    \sum_{k=1}^\infty \frac{k}{k+1}\cdot F_{3\to k\times 3}(n,d=5,\balpha)\cdot\mathbb{P}\left[3\right].
\end{equation}
The same arguments apply for other outer code distances for calculating the probability of space-like failure in expectation over the mixing parameter $\mathcal{B}$. In all cases, calculations become straightforward since qudit $Z$ errors become uniform under $\mathbb{E}_{\mathcal{B}}$.

\subsection{Handling All Parameter Choices and Failure Modes Together}\label{subsec:all_together}

It remains to handle all of our failure modes and parameter choices together. Let us be exact about the order of parameter choices for the protocol. First, $n$ and $d$ are chosen. Then, evaluation points $\balpha$ are chosen; see Subsection~\ref{subsec:evaluation_points}. Then, a choice of mixing parameters can be made, namely the qudit-to-qubit mappings $\mathcal{B}$, and choices of $\nu$ for the $X$ checks and $Z$ checks, determining the linear combinations measured in the latter set of $M$ checks in a round of $d-1+M$ checks.\footnote{Strictly speaking, there is a mixing parameter $\nu$ for time-like failure on $X$ checks, and another $\nu$ for time-like failure on $Z$ checks, although the discussion for the two is identical.}

Our protocol can fail due to time-like failure on the $X$ or $Z$ checks, as well as space-like failure on $X$ or $Z$ checks. Consider any two events that cause failure of one block of the outer code during one of its rounds (that is, space-like failure or time-like failure on an adjacent set of $X$ and $Z$ checks). The following argument simply applies a union bound, followed by linearity of expectation, and may be extended to any number of events:
\begin{multline}
    \mathbb{E}_{\text{Mixing Parameters}}\left\{\mathbb{P}\left[\text{Failure Event 1 Or Failure Event 2}\right]\right\}\\ \leq \mathbb{E}_{\text{Mixing Parameters}}\left\{\mathbb{P}\left[\text{Failure Event 1}\right] + \mathbb{P}\left[\text{Failure Event 2}\right]\right\}\\
    = \mathbb{E}_{\text{Mixing Parameters}}\left\{\mathbb{P}\left[\text{Failure Event 1}\right]\right\} + \mathbb{E}_{\text{Mixing Parameters}}\left\{\mathbb{P}\left[\text{Failure Event 2}\right]\right\}.
\end{multline}
In previous subsections, we upper bounded these separate quantities on the right-hand side, that is, the probabilities of various events causing outer code failure, in expectation over mixing parameters. By the union bound, and the linearity of expectation, it follows that there is a single choice of mixing parameters such that the probability of all failure modes together (the left-hand side) is at most the sum of all these individually upper-bounded values on the right-hand side.

\clearpage

\section{Footprint Estimates}\label{sec:footprint_estimates}

We now wish to consider the overhead of our protocols, and in particular the minimum overhead that can be obtained for a target logical error rate (per logical qubit-round). The parameters of our protocol are
\begin{itemize}
    \item The number of columns in the gross code grid $n$ (this is the length of our codes);
    \item The number of rows in our gross code grid $m$ ($m-2$ of these rows carry a copy of our outer code, and $2$ are used for syndrome extraction);
    \item The distance of our code $d$ (we consider $d = 3, 4, 5, 6, 7, 8, 9$);
    \item The evaluation points of our codes $\balpha$ (these are fixed to values listed in Appendix~\ref{subsec:evaluation_points});
    \item $M$: the number of linear combinations of $X$ checks, or $Z$ checks, that are measured in a round of $X$ checks, or $Z$ checks, after the $d-1$ regular checks;
    \item $R$: the number of times the stabilisers of the cat states are fault-tolerantly checked;
    \item Post-selection strategies $\mathcal{C}_i$ for $i = 1, 2, 3$. Each of these is determined by three post-selection rates $r_{i,\text{half}}, r_{i,\text{whole}}$, and $r_{i,\text{inter}}$ for half-LPU in-module instructions, whole-LPU in-module instructions, and inter-module instructions, respectively.
\end{itemize}
Our aim is to find the minimum-space overhead system achieving a target logical error rate per logical qubit-round which uses at most $500,000$ physical qubits. In order to do this, we start by considering no post-selection, that is, all $\mathcal{C}_i$ are trivial, and consider all $n \in \{20, 21, \ldots, 80\}$, $m \in \{3, 4, 5, \ldots, 50\}$, $d \in \{3,4,5,6,7,8,9\}$, and $M,R \in \{1,2\}$. Then, for each one of these configurations, we attempt to optimise the post-selection strategies to achieve the lowest logical error rate. Note that making the post-selection strategies more stringent is a tradeoff. It decreases the noise spread to the data block through syndrome extraction and can decrease the probability of time-like failure. On the other hand, post-selecting cat states more heavily leads to gross codes sitting in memory for longer, and thus accumulating more error.

Since we use the ancilla system constructions and logical instructions of~\cite{yoder2025tourgrossmodularquantum}, the number of logical qubits in our memory system is  $11\cdot(m-2)\cdot\left[n-2(d-1)\right]$, and the number of physical qubits is $378nm + 44(n-1)(m-1) + 22(n-1) + 22(m-1)$. The sizes of un-concatenated gross and two-gross systems are taken again with the ancilla systems in~\cite{yoder2025tourgrossmodularquantum} assuming a large rectangular-grid connectivity. Below we report the minimum overhead of a system achieving a target logical error rate per logical qubit-round. Note that here the term ``round'' means $8$ physical timesteps, the time to extract a full syndrome (see Table 2 of~\cite{yoder2025tourgrossmodularquantum}). For some context, we also include the minimum overhead of concatenated surface code systems, using the extrapolations in~\cite{gidney2025factor}, and the size of the systems in~\cite{gidney2025yoked}.\footnote{Note that we are using a uniform physical noise model so we do not compare to the extrapolations in the main text of~\cite{gidney2025yoked} which use a different noise model. For the uniform noise model, we consider the most recent extrapolations in~\cite{gidney2025factor}.} The results are shown in Figure~\ref{fig:results_1e-3}.

\begin{figure}[t]
    \centering
    \includegraphics[width=\linewidth]{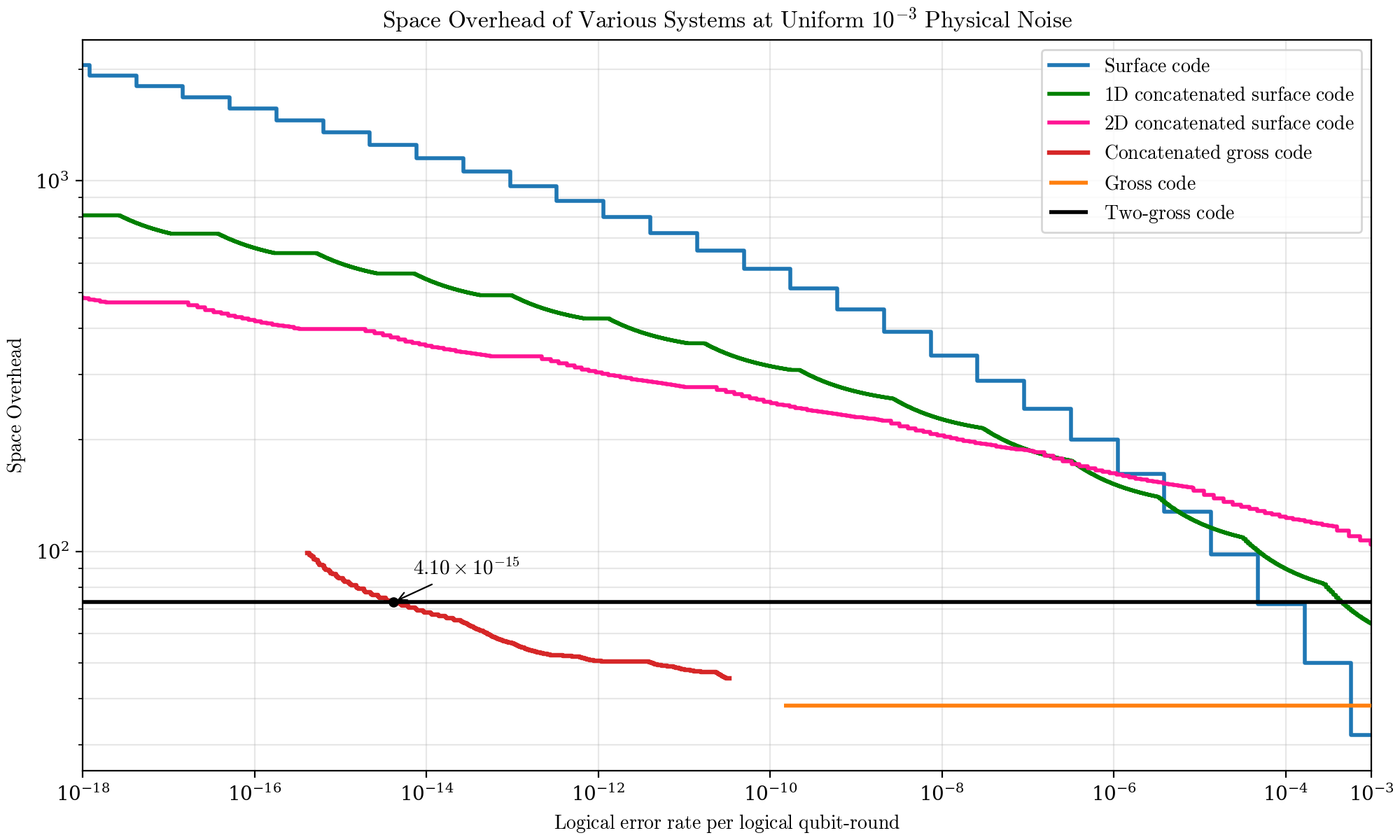}
    \caption{The space overheads of various systems at $10^{-3}$ physical noise. All systems are chosen with the optimal configuration while keeping their size below $500,000$ physical qubits. The space overhead of the concatenated gross code memory crosses over that of the two-gross code at a logical error rate of $4.10\times 10^{-15}$ per logical qubit-round.}
    \label{fig:results_1e-3}
\end{figure}

We find that the concatenated gross system does indeed smooth out the cost discontinuity between the gross and two-gross codes, and provides a better space overhead than the two-gross code down to $4.10\times 10^{-15}$ per logical qubit-round, while offering the numerous engineering benefits discussed in the Overview. This is significant because it can reach the teraquop regime relevant for many large-scale quantum algorithms, which starts around a logical error rate of $10^{-14}$ to $10^{-13}$ per logical qubit-round. This was previously inaccessible to the gross code at this physical noise strength.

\begin{figure}[ht]
    \centering
    \includegraphics[width=\linewidth]{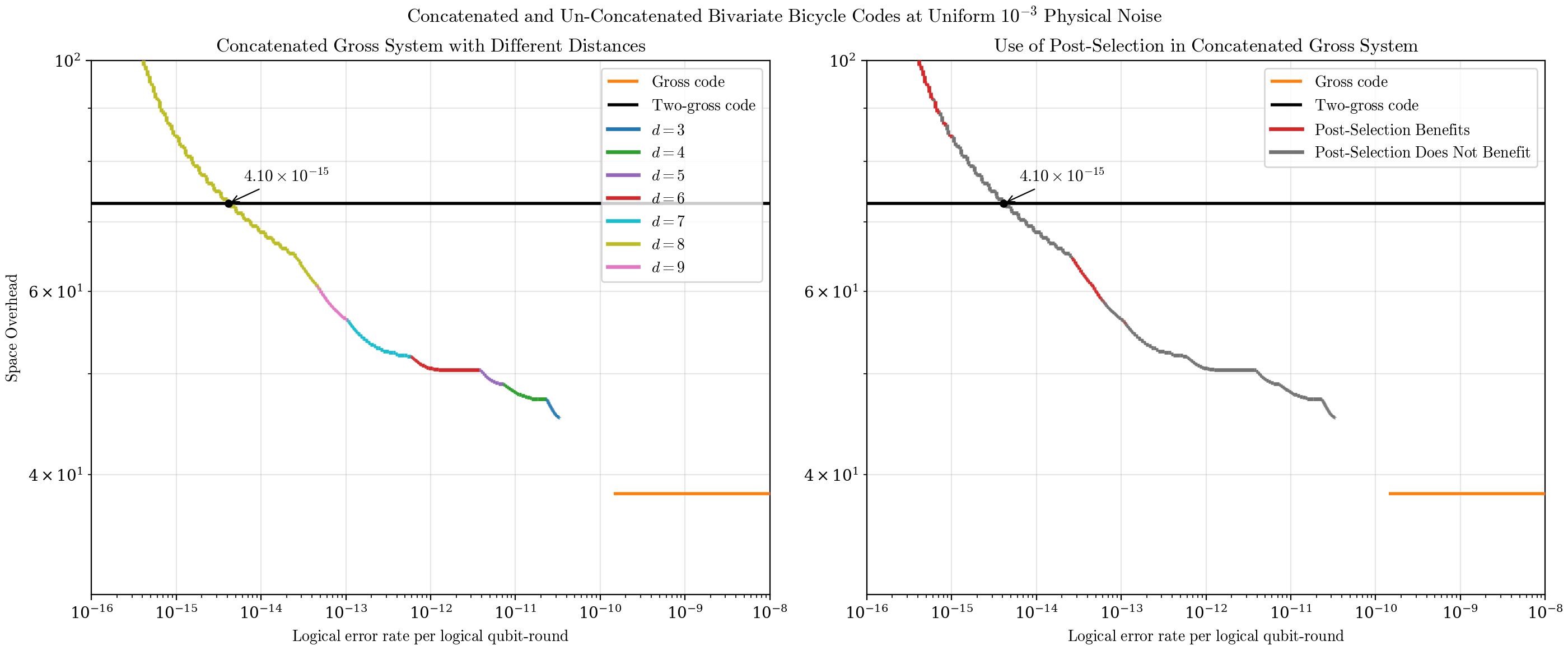}
    \caption{Further information on the performance of the concatenated gross code system at $10^{-3}$ physical noise. The left-hand panel shows the places where outer codes of various distances are used in the optimal system, and the right-hand panel shows the places where the optimal system does (red), and does not (grey), make use of some non-trivial post-selection strategy.}
    \label{fig:1e-3_BB_subplots}
\end{figure}
For additional information just on the BB codes, we also provide Figure~\ref{fig:1e-3_BB_subplots}. Specifically, we show where the optimal concatenated gross system uses an outer code of each distance $d = 3,4,5,6,7,8,9$. We also show where the optimal system does, and does not use some amount of post-selection. On the former, we find that each distance has a regime in which it is optimal; interestingly, the $d=9$ system comes ``before'' the $d=8$ system, which is likely a result of the tradeoff between higher distances giving higher space-like failure error suppression, but also requiring the measurement of more checks. Higher outer code distances than $d=9$ could become useful if one allowed larger system sizes than this. As for the post-selection, at $10^{-3}$ physical noise, we find that the strategies are only used infrequently, and we found that the improvements are modest. The reason for this is clear. At $10^{-3}$ physical noise, the dominant failure mode arises from the fact that the cat states are very slow to prepare, and gross codes must sit at idle for a long time between rounds of syndrome extraction. Post-selection strategies slow down the time to create cat states, and it is usually not worth the improvement in the cat state noise that they yield.

\begin{figure}[ht]
    \centering
    \includegraphics[width=\linewidth]{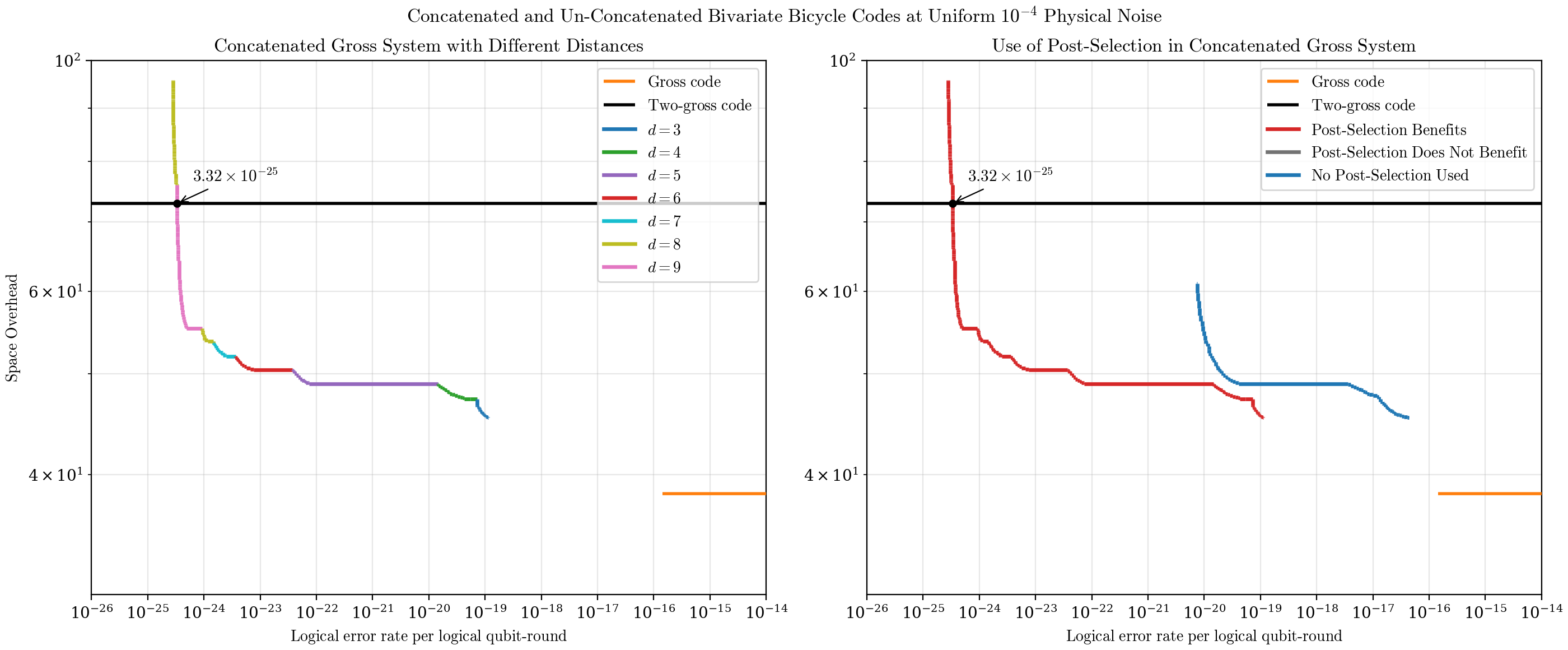}
    \caption{The performance of our scheme at a uniform physical noise strength $10^{-4}$. We show the space overhead of the optimal system at each target logical error rate per logical qubit-round. In the left-hand panel, we show where different outer code distances are used in the optimal system, and in the right-hand panel, we show that some post-selection is used in all optimal systems. We also plot the performance of the scheme if no post-selection were used at all (in blue).}
    \label{fig:1e-4_BB_subplots}
\end{figure}
On the other hand, we can also consider the performance of our scheme at a uniform physical noise of $10^{-4}$, as shown in Figure~\ref{fig:1e-4_BB_subplots}. Here, post-selection is seen to be of considerable benefit, where the right-hand panel shows that the optimal system is always using some post-selection strategy. In that panel, we also show the curve that would be obtained if no post-selection were used; we see that the improvements rendered by the post-selection strategies are large: from $2$ to over $4$ orders of magnitude of logical error rate per logical qubit-round at the same space overhead. As for the left-hand panel, similarly to $10^{-3}$ noise, each outer code distance has a regime in which it is optimal.

Note that in all figures, the space overhead blows up at low-enough logical error rate. Indeed, for a fixed physical error rate, it must do so at some point regardless of the outer code distance used, because at some point the post-selection of the cat states and/or of rounds of checks becomes prohibitive. To go lower, one could use a higher-distance inner code, or improved circuits, instructions or compilation for the inner code (we show such a potential improvement, deferred to future work, in Appendix~\ref{sec:improved_footprint_morphing}).

\section*{Acknowledgements}\label{acknowledgements}
\addcontentsline{toc}{section}{Acknowledgements}

The authors acknowledge enlightening conversations with Isaac Chuang, Aram Harrow, Zhiyang He, Anirudh Krishna, Christopher A. Pattison, Emily Pritchett, John A. Smolin, Theodore J. Yoder, Hengyun Zhou, and Guanyu Zhu. We thank Isaac Chuang for suggesting Figure~\ref{fig:overview_fig}. AW is particularly grateful for the support throughout of fellow IBM interns Kathleen Chang, Louis Golowich, Anasuya Lyons, and Mackenzie H. Shaw.

AW acknowledges funding from NSF grant PHY-2325080. AW acknowledges funding from the MIT-IBM Watson AI Lab. AW acknowledges that this work is supported by the National Science Foundation under Cooperative Agreement PHY-2019786. This pre-print is assigned number MIT-CTP/6017.

This work benefitted from computational resources. AW is grateful to Isaac Chuang for helping secure access to these resources. This research used resources of the National Energy Research Scientific Computing Center, a DOE Office of Science User Facility supported by the Office of Science of the U.S. Department of Energy under Contract No. DE-AC02-05CH11231 using NERSC award NERSC DDR-ERCAP0038585. This work used CPU resources at PSC Bridges-2~\cite{brown2021bridges} through allocation CIS260050 from the Advanced Cyberinfrastructure Coordination Ecosystem: Services \& Support (ACCESS) program~\cite{boerner2023access}, which is supported by U.S. National Science Foundation grants \#2138259, \#2138286, \#2138307, \#2137603, and \#2138296. We acknowledge the MIT Office of Research Computing and Data for providing high-performance computing resources that have contributed to the research results reported within this paper. We acknowledge the MIT SuperCloud~\cite{reuther2018interactive} and Lincoln Laboratory Supercomputing Center for providing HPC resources that have contributed to the research results reported within this paper/report. The computations in this paper were partially run on the FASRC Cannon cluster supported by the FAS Division of Science Research Computing Group at Harvard University. This work is supported by the National Science Foundation under Cooperative Agreement PHY-2019786 (The NSF AI Institute for Artificial Intelligence and Fundamental Interactions, http://iaifi.org/). This research was done in part using services provided by the OSG Consortium~\cite{osg_ospool_2006,ruth2007open,sfiligoi2009pilot,osg_osdf_2015}, which is supported by the National Science Foundation awards \#2030508 and \#2323298. AW is particularly grateful for the expertise and guidance of OSPool staff.

\printbibliography

\appendix

\section{Efficient Compilation of Qudit \texorpdfstring{$Z^\alpha Z^\beta$}{} Measurements on Two Adjacent Gross Codes}\label{sec:compile_Z_meas}

In our fault-tolerant preparation and verification of the qudit cat state, we require the (non-fault-tolerant) measurement of qudit $Z^\alpha Z^\beta$ operators on the two qudits made up by two adjacent gross codes. In turn, this corresponds to the measurement of $s$ qubit operators $Z^{x_i}Z^{y_i}$, where $(x_i)_{i=1}^s$ and $(y_i)_{i=1}^s$ are each collections of $s$ linearly independent bit strings of length $s$, on the two sets of $s$ non-pivot logical qubits on adjacent gross codes, as may be checked using the techniques in~\cite{wills2026review}. This qudit measurement operation is the most time-consuming part of our cat state preparation and verification, and so compiling this operation efficiently to the bicycle instructions is of great importance.

The first thing we do is to consider matrices $X$ and $Y$ whose rows are $x_i$ and $y_i$, respectively. These are $s \times s$ invertible binary matrices. We can notice that measuring the operators corresponding to the rows of $\begin{pmatrix}
    X & Y
\end{pmatrix}$ is equivalent to measuring the operators corresponding to the rows of $\begin{pmatrix}
    AX & AY
\end{pmatrix}$ for any invertible $s \times s$ binary matrix $A$. What we always do is to choose the matrix $A$ such that $AX = V$, where
\begin{equation}\label{eq:V_matrix}
\setcounter{MaxMatrixCols}{11}
    V = 
\left[
\begin{array}{ccccccccccc}
0&0&0&0&0&1&0&0&0&0&0\\
0&0&0&0&1&1&0&0&0&0&1\\
0&0&1&0&1&1&0&0&1&0&1\\
0&0&1&1&1&1&0&0&1&1&1\\
0&1&0&0&1&1&0&1&0&0&1\\
1&0&0&1&0&1&1&0&0&1&0\\
0&0&0&0&1&0&0&0&0&0&0\\
0&0&1&0&1&0&0&0&0&0&0\\
0&0&1&1&1&0&0&0&0&0&0\\
0&1&0&0&1&0&0&0&0&0&0\\
1&0&0&1&0&0&0&0&0&0&0
\end{array}
\right]
\begin{array}{l}
\rdelim\}{6}{*}[\text{ Whole-LPU}]\\[-0.2ex]
\\
\\
\\
\\
\\
\rdelim\}{5}{*}[\text{ Half-LPU}]\\[-0.2ex]
\\
\\
\\
\\
\end{array}
.
\end{equation}
Writing the rows of this matrix as $(v_i)_{i=1}^s$, this matrix has the property that these rows are linearly independent, and moreover that $ZZ^{v_i}$ are all native measurements on the gross code~\cite{yoder2025tourgrossmodularquantum}, as may be checked using~\cite{github_bicycle_architecture_compiler}. In fact, the first $6$ of these rows correspond to native measurements derived from whole-LPU instructions, whereas the latter $5$ of these rows correspond to native measurements derived from half-LPU instructions (see Appendix~\ref{sec:improved_ler_bicycle} for a discussion of half-LPU and whole-LPU in-module instructions), which we find is the best that one can do.

So far, what we have done is redefined the problem such that we want to measure the $s$ $Z$ operators corresponding to the rows of $\begin{pmatrix}
    V & W
\end{pmatrix}$ (where we set $W = AY$); we denote the rows of $W$ as $(w_i)_{i=1}^s$. In terms of a gate set higher than the bicycle instructions, one individual measurement $Z^{v_i}Z^{w_i}$ may be compiled as shown in Figure~\ref{fig:Z_mixed_meas}.
\begin{figure}[ht]
    \centering
        \includegraphics[width=0.6\linewidth]{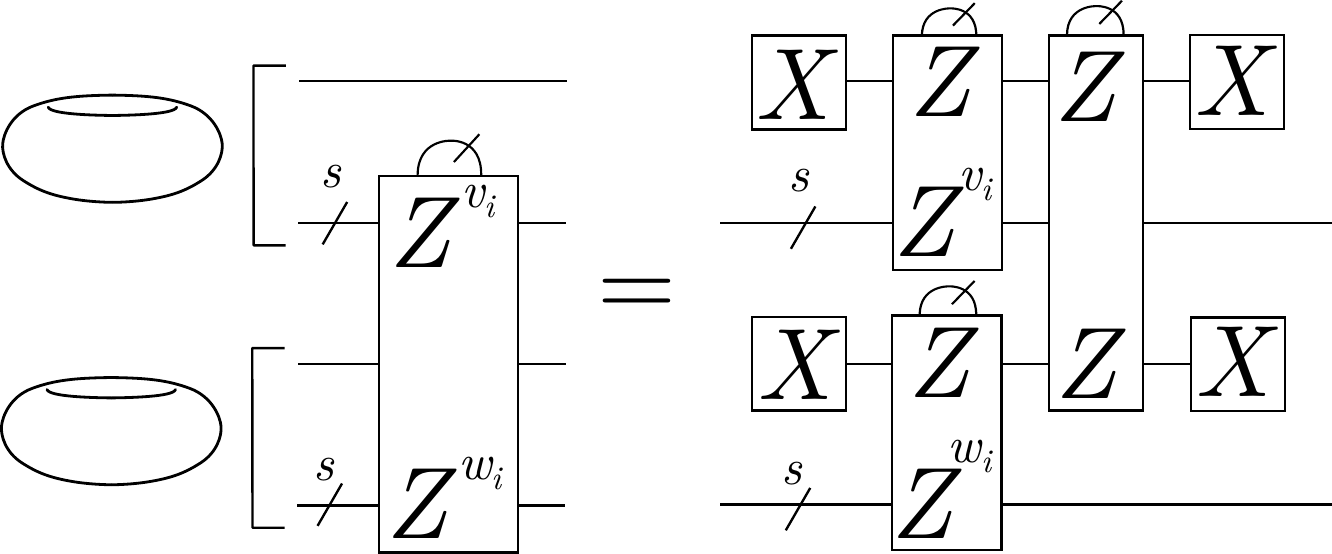}
        \caption{The compilation of a $Z^{v_i}Z^{w_i}$ measurement on two adjacent gross codes into a gate set higher than the bicycle instructions. All boxes denote qubit measurements; those without meter symbols denote ``measurement projections''~\cite{yoder2025tourgrossmodularquantum}, meaning measurements followed by the appropriate frame updates to ensure we recover the $+1$-eigenstate (in the absence of faults). The outcome of the whole measurement is taken as the $\mathsf{XOR}$ of the three $Z$-type measurements on the right-hand side.}
    \label{fig:Z_mixed_meas}
\end{figure}

Of course, the vectors $v_i$ were deliberately chosen such that $ZZ^{v_i}$ are all native in-module measurements on the gross code, thus executable via automorphisms and a single in-module measurement instructions. On the other hand, the operations $ZZ^{w_i}$ are generically not native measurements, and so further compilation is required to obtain the operation fully in terms of bicycle instructions. The way this may be handled is by using the gross code's so-called \textit{native rotations}~\cite{cross2024improved,yoder2025tourgrossmodularquantum}. The idea is as follows. Suppose that $R$ is an $11$-qubit Pauli on the gross code such that $BR$ is a native measurement, where $B$ is some element of $\{X,Y,Z\}$. Then, the $11$-qubit Clifford rotation $R(\pi/2) = e^{i\frac{\pi}{4}R}$ may be executed on the non-pivot qubits of the gross code via a measurement projection of the pivot qubit in the $A$-basis, a measurement projection of $BR$, and finally a measurement projection of the pivot qubit in the $C$-basis, where $A, B, C \in \{X,Y,Z\}$ are all distinct; see~\cite{yoder2025tourgrossmodularquantum}. Such operations $R(\pi/2)$ are called the native rotations of the gross code, and it is known that they generate the Clifford group on the $11$ non-pivot qubits. It is possible, therefore, to measure any desired operator on the gross code by conjugating native measurements by native rotations. The approach that is taken in~\cite{yoder2025tourgrossmodularquantum} is to measure any needed in-module operator via a sequence of native rotations $R_1\ldots R_k$, performing a native measurement, and then performing the reverse sequence of native rotations $(R_1\ldots R_k)^\dagger$. However, this does introduce a somewhat undesirable compilation overhead.

In this work, we make two improvements to this. The first may be done very easily, but gives only small savings, whereas the second requires a little more effort, but leads to considerable savings. For the first, one observes that if the pivot qubit is ever measured in the $A$ basis, for any $A \in \{X,Y,Z\}$, measuring the pivot qubit in the $A$ basis twice in a row is redundant. That is, if the compilation ever calls for the same measurement on the pivot qubit twice in a row, one may be ``absorbed'' into the other. Such optimisation can show up in, for example, sequences of native rotations; we can actually often force such absorptions to happen by, for example, swapping the bases $A$ and $C$ in the implementation of a native rotation (described above) to maximise the number of absorptions.

The second observation that we make is as follows. The gross code with the bicycle instructions has $540$ native measurements~\cite{yoder2025tourgrossmodularquantum}. For the compilation in~\cite{yoder2025tourgrossmodularquantum}, sequences of native rotations are being used to rotate one desired operator to measure to one native rotation, before rotating back. It is natural to wonder whether greater efficiency can be achieved by rotating multiple measurements into native measurements, rather than just one at a time. For example, we could perform a single sequence of native rotations to make the next two (or three or four) measurements native.

In this work we find that it is indeed more efficient, at least for the $Z$ operators that we study. Consider three linearly independent vectors $w_1, w_2, w_3 \in \mathbb{F}_2^s$. Using the tools in~\cite{github_bicycle_architecture_compiler}, it is possible by brute force to find the shortest sequence of native rotations $R_1, \ldots, R_k$ on the $11$ non-pivot qubits such that $(R_1\ldots R_k)ZZ^{w_1}(R_1\ldots R_k)^\dagger$, $(R_1\ldots R_k)ZZ^{w_2}(R_1\ldots R_k)^\dagger$ and $(R_1\ldots R_k)ZZ^{w_3}(R_1\ldots R_k)^\dagger$ are all native measurements. We also optimise our rotation sequence length using the absorptions previously described (repeated measurements of the pivot qubit in the same basis are redundant), and where we prefer the use of half-LPU measurement instructions to whole-LPU measurement instructions when we have a choice. One can also do the same for pairs of vectors $w_1, w_2 \in \mathbb{F}_2^s$.

With all the ingredients introduced, it is possible to state our whole compilation algorithm for these sequences of $11$ measurements. We perform the $11$ measurements in $3$ batches of $3$, and one batch of $2$. To execute one batch on the second gross code, rotations are performed on the non-pivot qubits to make the whole batch native, and then the appropriate number (either three or two) routines are performed as in Figure~\ref{fig:Z_mixed_meas}, where now every operation shown is native. Then, the native rotation sequence is undone. Notice that because the rotation sequence consumes the pivot qubit, these rotation sequences must be performed outside of the routines shown in Figure~\ref{fig:Z_mixed_meas}. There is one more comment to make in terms of optimisation, and that is that there are many ways to break up the $11$ desired measurements into $3$ batches of $3$, and one batch of $2$. For each set of $11$ measurements, one can try to break up the set in different ways at random, and pick one's preferred compilation.

\subsection{The Cost of the Qudit \texorpdfstring{$Z^\alpha Z^\beta$}{} Measurements}

How can we estimate the cost of one non-fault-tolerant qudit $Z^\alpha Z^\beta$ measurement? There are only a small number of such qudit measurements, and so enumerating each of their costs by brute force would be in reach if we only cared about qudits. However, the trouble is that there are very many qudit-to-qubit mappings, and the manifestation of each of these measurements in terms of qubits depends on the qudit-to-qubit mappings. Because there are over $10^{31}$ qudit-to-qubit mappings on each qudit, enumerating all possible measurements and evaluating their cost is out of bounds. However, in the main text, in particular in Section~\ref{sec:ler_calc}, we discuss how we take our qudit-to-qubit mappings $\mathcal{B}$ on each qudit randomly rather than making an explicit choice. As such, our aim is to benchmark the typical cost of one qudit $Z^\alpha Z^\beta$ measurement over this random choice. 

Now, the random qudit-to-qubit mappings $\mathcal{B}$ causes the matrix $W$, discussed earlier in this appendix, to behave randomly. To estimate the cost of a typical qudit $Z^\alpha Z^\beta$ measurement, we thus sample an invertible $s \times s$ matrix $W$ at random, and consider how efficiently we can measure the $Z$ operators corresponding to the rows of $\begin{pmatrix}
    V & W
\end{pmatrix}$ on the two sets of non-pivot logical qubits in the two adjacent gross codes, where $V$ is the fixed matrix shown in Equation~\eqref{eq:V_matrix}.

\begin{figure}[t]
    \centering
    \includegraphics[width=\linewidth]{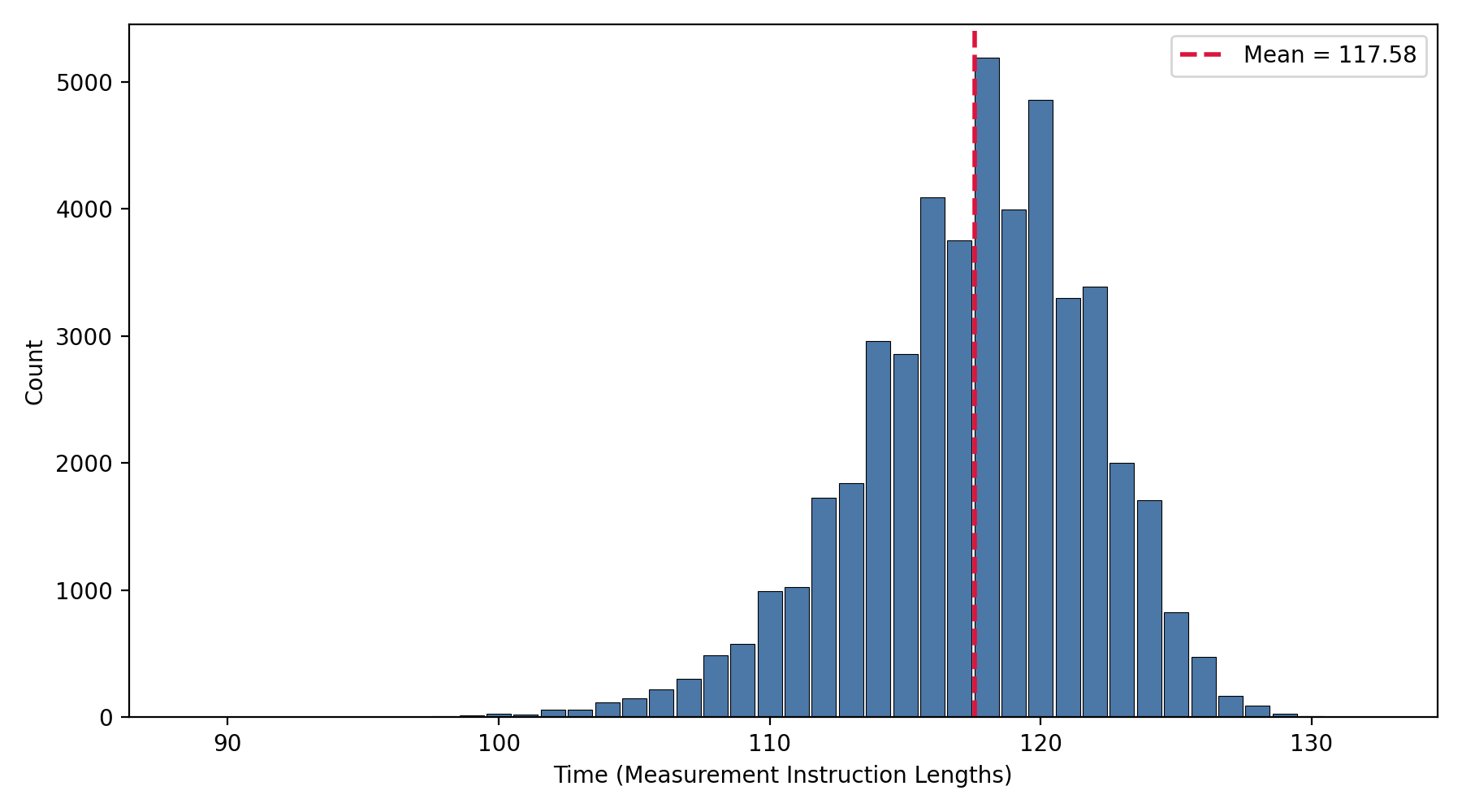}
    \caption{Histogram of sampled measurement sequences compiled to the bicycle instructions using our compilation algorithm with \texttt{num-decomposition-attempts} = 100. This histogram shows the time to perform the full qudit $Z^\alpha Z^\beta$ measurement, in units of measurement lengths, noting that all in and inter-module measurement instructions take the same time: $\tau_{\text{Meas}} = 120$ timesteps (see Table 2 of~\cite{yoder2025tourgrossmodularquantum}). On the first gross code, the numbers of half-LPU and whole-LPU in-module instructions used are always $17$ and $6$, whereas on the second gross code the mean numbers used are found to be $56.05$ and $50.53$, respectively. $11$ inter-module measurements are always used. The time is determined by the in-module operations on the second code and the inter-module operations. We see that the mean total numbers of instructions used (across both codes) are $73.05, 56.53$ and $11$ half-LPU in-module, whole-LPU in-module, and inter-module, respectively.}
    \label{fig:sampled_meas_cost}
\end{figure}

Our algorithm breaks up the $11$ resulting $Z$ measurements at random into $3$ batches of $3$ and one batch of $2$, and evaluates the resultant number of half-LPU and whole-LPU in-module instructions required on the second gross code. The in-module instructions required on the first gross code, as well as the inter-module instructions, are fixed. The algorithm performs the described absorption of repeated pivot measurements where possible, meaning that the same measurement is not repeated multiple times in a row on the pivot qubit. The algorithm decomposes the $11$ measurements into $3$ groups of $3$ and $1$ group of $2$ some tunable number \texttt{num-decomposition-attempts} times, picking the best one according to some objective function that we choose. Here, we simply aim to minimise the objective function
\begin{equation}
    \text{Number of half-LPU instructions} + \text{Number of whole-LPU instructions},
\end{equation}
that is, we are aiming to perform the measurement as quickly as possible.\footnote{Because the half-LPU in-module instruction are about an order of magnitude less noisy than the whole-LPU in-module instructions, if we were mainly concerned with error rate of the routine, not time, then we would choose an objective function $(0.1)*\text{Number of half-LPU instructions} + \text{Number of whole-LPU instructions}$.}
We note that the in-module instructions on the first gross code are always fixed to $6$ whole-LPU instructions, and $17$ half-LPU instructions. $11$ inter-module instructions are always required. 

We sample about $50,000$ matrices $W$ at random and plot the time taken to perform the corresponding sequence of measurements in Figure~\ref{fig:sampled_meas_cost}, in units of measurement lengths (noting that all in-module and inter-module measurement instructions take the same time: $120$ timesteps~\cite{yoder2025tourgrossmodularquantum}). The mean time for one of these measurements is found to be $117.58$. All of these use $11$ inter-module measurements, and a mean is found for the total count of half-LPU and whole-LPU in-module measurements of $73.05$ and $56.53$, respectively, across both gross codes, which is relevant for calculating error rates. Note that, in terms of the time, there are about $n/2$ of these measurements performed in parallel, and we must wait for the slowest one to finish. We wish to take a time corresponding to the slowest one, which then depends on $n$. We take the time for the (non-fault-tolerant) $Z^\alpha Z^\beta$ measurement to be the slowest one at $n = 70$ (a value larger than we care about), which we estimate from the histogram to be about about $125$ measurement time-lengths. Note that this is still likely pessimistic, since in reality, given a particularly slow case, one would simply re-randomise the appropriate qudit-to-qubit mappings.

\section{Improved Logical Error Rates for Bicycle Instructions}\label{sec:improved_ler_bicycle}
In~\cite{yoder2025tourgrossmodularquantum}, logical error rates were presented for operations on the gross code, including logical measurement operations using lattice surgery. One of the reasons we are able to improve the performance of our concatenated memory system on the bicycle architecture is as a result of improvements that have been made to these logical error rates in this work. The improvements that we make follow from optimising the parameters of the Relay-BP decoder~\cite{muller2025improved}, in particular the so-called ``gamma parameters'', that is, the range $[\gamma_{\text{min}}, \gamma_{\text{max}}]$ from which the memory strengths $\gamma_j$ are uniformly and independently drawn, for each leg, and for each error node $j$. For example,~\cite{muller2025improved} makes use of the range $[\gamma_{\text{min}}, \gamma_{\text{max}}] = [-0.24, 0.66]$ for the idling gross code. It is also found in that work that optimising the gamma parameters is a crucial step to obtaining optimal performance of the decoder for a given circuit. As part of this work, we perform optimisations of the gamma parameters for circuits whose logical error rate is critical to us.

We also attempted to make improvements to the actual circuits for the logical measurement instructions. In particular, the lattice surgery logical measurement operations follow via a merging of the code block(s) with the ancillary system, some number of rounds in the combined system, and then a splitting of the two. In principle, the optimal number of rounds in the merged system should be around $d_{\text{Circ}}$, the distance of the circuit in question, where the intuition is that this should balance the effect of time-like errors (benefitting from a larger number of merged rounds) and space-like errors (benefitting from a smaller number of merged rounds). It is possible, however, that an optimal number of merged measurement rounds would be slightly above or below $d_{\text{Circ}}$. We attempted such optimisations, but found the potential improvements to be mild, and therefore opt to stick with $d_{\text{Circ}}$ rounds in the merged system.

By optimising the gamma parameters, we demonstrate significant improvements at uniform physical noise $p = 10^{-3}$ to the $X_1$ in-module measurement (the measurement of $X$ on the pivot qubit), the $Y_1$ in-module measurement (the measurement of $Y$ on the pivot qubit) and the $X_1X_1$ inter-module measurement (the measurement of the $XX$ operator on two pivot qubits in adjacent gross codes) on the gross code. The results are shown in Table~\ref{tab:logical_error_rates_1e-3}.
\renewcommand{\arraystretch}{1.5}
\begin{table}[ht]
\centering
\begin{tabular}{|c|ccc|ccc|}
\hline
\multirow{2}{*}{\textbf{Circuit}} 
& \multirow{2}{*}{$\mathbf{S}$} 
& \multirow{2}{*}{\raisebox{0.45ex}{$\gamma_{\text{min}}$}} 
& \multirow{2}{*}{\raisebox{0.45ex}{$\gamma_{\text{max}}$}}
& \multicolumn{3}{c|}{\textbf{Logical Error Rate} ($p = 10^{-3}$)} \\
\cline{5-7}
 &  &  &  & \textbf{Low} & \textbf{Mean} & \textbf{High} \\
\hline\hline
\multirow{2}{*}{\makecell{Half-LPU\\In-Module Measurement}}
  & 20  & \multirow{2}{*}{-0.13} & \multirow{2}{*}{0.85} 
  & $2.3 \times 10^{-7}$ & $2.9\times 10^{-7}$ & $3.6\times 10^{-7}$ \\
  & 100 &                         &                        
  & $1.2\times 10^{-7}$ & $1.7\times 10^{-7}$ & $2.3\times 10^{-7}$ \\
\hline
\multirow{2}{*}{\makecell{Whole-LPU\\In-Module Measurement}}
  & 20  & \multirow{2}{*}{-0.14} & \multirow{2}{*}{0.76} 
  & $2.0\times 10^{-6}$ & $2.2\times 10^{-6}$ & $2.5\times 10^{-6}$ \\
  & 100 &                         &                        
  & $1.3\times 10^{-6}$ & $1.6\times 10^{-6}$ & $2.0\times 10^{-6}$ \\
\hline
\multirow{2}{*}{\makecell{Inter-Module\\Measurement}}
  & 20  & \multirow{2}{*}{-0.19} & \multirow{2}{*}{-0.90} 
  & $2.9\times 10^{-6}$ & $3.5\times 10^{-6}$ & $4.1\times 10^{-6}$ \\
  & 100 &                         &                        
  & $2.3\times 10^{-6}$ & $2.8\times 10^{-6}$ & $3.2\times 10^{-6}$ \\
\hline
\end{tabular}
\caption{Improved logical error rates for logical measurement bicycle instructions on the gross code at uniform physical noise $p = 10^{-3}$, as well as the Relay-BP memory parameters $\gamma_{\min}$ and $\gamma_{\max}$ used to obtain them. The exact measurements chosen to represent the half-LPU and whole-LPU in-module measurements, as well as the inter-module measurement, are $X_1$, $Y_1$, and $X_1X_1$, respectively, see~\cite{yoder2025tourgrossmodularquantum} for these definitions. We report the mean value, as well as (Low, High), forming the ends of a 95\% confidence interval on the estimated mean value. Our optimisations provide substantial improvements over the logical error rates in~\cite{yoder2025tourgrossmodularquantum} (see Table 2 therein): by nearly three orders of magnitude in the case of the inter-module measurement. These numbers are the results of direct Monte Carlo simulation at uniform $10^{-3}$ noise. These simulations were performed on the MIT Engaging cluster, the MIT SuperCloud, the Harvard Cannon, the Open Science Grid, PSC Bridges-2, and Perlmutter; see \hyperref[acknowledgements]{Acknowledgements}.}
\label{tab:logical_error_rates_1e-3}
\end{table}

Throughout the work, we split the in-module instructions used in our scheme into so-called ``half LPU'' in-module instructions and ``whole LPU'' in-module instructions. To explain this, note that of the $15$ in-module instructions on the gross code, some make use of half of the logical processing unit (LPU), whereas some make use of the whole LPU. The half-LPU instructions are the $6$ in-module measurements of $\langle X_1, Z_7 \rangle \cup \langle X_7, Z_1\rangle$ in the language of~\cite{yoder2025tourgrossmodularquantum}, whereas the whole-LPU in-module instructions are the remaining $9$. In general, one finds that the whole-LPU instructions error far higher (indeed, approximately an order of magnitude higher) than the half-LPU instructions in simulation, likely because the presence of more LPU qubits is causing more minimum-weight fault paths. All the half-LPU instruction circuits, and all the whole-LPU instruction circuits have very similar (albeit not identical) structure, and so we take the error rate of the $X_1$ and $Y_1$ measurement to be representative of the half-LPU and whole-LPU instructions throughout our work. Note that doing this is slightly optimistic, because some of the half-LPU and whole-LPU measurements do use additional edges to these two. On the other hand, these are very small differences compared to the size of the LPU, and the most common in-module instructions used across our scheme are the simplest cases of $X_1, Y_1$ and $Z_1$ measurements anyway. On the other hand, we take the error rate of the inter-module $X_1X_1$ to be representative of all inter-module instructions used in our scheme. Indeed, note that throughout this work, the only inter-module instructions used are the $X_1X_1$ and $Z_1Z_1$ measurements, and since these have an extremely similar structure, we take the logical error rate of the $X_1X_1$ inter-module instruction as representative of all our inter-module instructions.

In terms of obtaining these numbers, note that, in Table~\ref{tab:logical_error_rates_1e-3}, we show the result of direct Monte Carlo simulations of logical error rates of various circuits; these numbers are not extrapolated. The displayed numbers under (Low, High) represent a $95\%$ confidence interval on the sampled mean logical error rate. These simulations were performed using Stim~\cite{gidney2021stim}, as well as the open source implementation of the Relay-BP decoder~\cite{trmue_relay}. The details on the noise model, and methodology behind the circuit constructions, are available in~\cite{yoder2025tourgrossmodularquantum}. Note that in this table, we have considered the use of the Relay-BP decoder using $S = 20$ and $S=100$ for the two logical measurements; this variable is the number of solutions to the decoding problem that the decoder seeks, see~\cite{muller2025improved}, and is called $\mathtt{stop\_nconv}$ in the implementation~\cite{trmue_relay}. We find a modest improvement from $S = 20$ to $S = 100$. The logical error rates for $S = 100$ will be used in analysing our concatenated memory systems. Note that with $S = 100$, the decoding would be outside the scope of real-time decoding in FPGAs~\cite{maurer2025real,maurya2025fpga}, but would likely be feasible with ASICs. In addition, since we only consider a memory system in this paper, it is not essential that decoding strictly keeps up with the extraction of syndromes, since one is able to correct past mistakes in software.

We are also interested in the performance of these gadgets at uniform physical noise $p = 10^{-4}$. There, we simulate only the inter-module $XX$ instruction, and only for one set of gamma parameters (the gamma parameters that were optimised at $10^{-3}$ noise). We did not observe a failure, and so we take the logical error rate to be the upper bound of a $95\%$ confidence interval on the given statistic. Even so, this still represents an improvement over the corresponding logical error rate in~\cite{yoder2025tourgrossmodularquantum}. This is presented in Table~\ref{tab:logical_error_rates_1e-4}. We take the remaining error rates for the operations at $10^{-4}$ noise from~\cite{yoder2025tourgrossmodularquantum}, where both whole LPU and half LPU instructions are given the error rate of the same in-module instruction, which is a worst-case (whole-LPU) example. 

\renewcommand{\arraystretch}{1.5}
\begin{table}[ht]
\centering
\begin{tabularx}{0.85\linewidth}{|c|ccc|>{\centering\arraybackslash}X
                                     >{\centering\arraybackslash}X
                                     >{\centering\arraybackslash}X|}
\hline
\multirow{2}{*}{\textbf{Circuit}} 
& \multirow{2}{*}{$\mathbf{S}$} 
& \multirow{2}{*}{\raisebox{0.45ex}{$\gamma_{\text{min}}$}} 
& \multirow{2}{*}{\raisebox{0.45ex}{$\gamma_{\text{max}}$}}
& \multicolumn{3}{c|}{\textbf{Logical Error Rate} ($p = 10^{-4}$)} \\
\cline{5-7}
 &  &  &  & \textbf{Low} & \textbf{Mean} & \textbf{High} \\
\hline\hline
\multirow{2}{*}{\makecell{Inter-Module\\Measurement}}
  & \multirow{2}{*}{20}  & \multirow{2}{*}{-0.19} & \multirow{2}{*}{0.90} 
  & \multirow{2}{*}{$0$} & \multirow{2}{*}{$0$} & \multirow{2}{*}{$1.5\times 10^{-8}$} \\
  & & & & & & \\
\hline
\end{tabularx}
\caption{Monte Carlo simulation of the inter-module $X_1X_1$ measurement at uniform physical noise $p = 10^{-4}$. Here, no failures were observed at all, and so we take the logical error rate to be the ``High'' value listed, i.e., the top of the 95\% confidence interval.}
\label{tab:logical_error_rates_1e-4}
\end{table}

We list the remaining parameters used in our simulations with the Relay-BP decoder, where we use variable names as in the implementation~\cite{trmue_relay}:
\begin{center}
\verb|"gamma0": 0.1,|\\
\verb|"num_sets": 1201,|\\
\verb|"pre_iter": 80,|\\
\verb|"set_max_iter": 60.|
\end{center}

Finally, for the idling operation, we consider the regular gross code circuit as in~\cite{yoder2025tourgrossmodularquantum}, with a logical error rate of $10^{-8.8}$. However, in Appendix~\ref{sec:improved_footprint_morphing}, we consider a new idling circuit for the gross code that makes use of morphing~\cite{shaw2025lowering}, and is believed to have a circuit distance matching that of the code ($12$), unlike the regular circuit which has circuit distance $10$. Via direct simulation, we demonstrate a considerable improvement over the regular circuit. In Appendix~\ref{sec:improved_footprint_morphing}, we explain why we stick with the regular idling circuit for now, but that significant improvements to the footprint estimates are likely in future using these kinds of advances.

\section{Improved Footprint Estimates via Morphing Idling Circuit}\label{sec:improved_footprint_morphing}

One appealing feature of the concatenation construction is that it can be improved via changes to the inner code and its instructions. At physical noise $10^{-3}$, the bottleneck in the construction comes from the fact that the native gateset of the current bicycle architecture~\cite{yoder2025tourgrossmodularquantum} is relatively limited, introducing a significant compilation overhead. Even with the tricks of Appendix~\ref{sec:compile_Z_meas}, cat states take a long time to prepare, and the bottleneck to the logical error rate arises from gross codes sitting at idle for a long time between rounds of syndrome extraction. With this said, improvements to either the native gateset, or the gross code idling circuit, can quickly improve the footprint estimates substantially.

We demonstrate why we believe this to be likely in future by considering a new circuit for the idling gross code. Unlike the regular circuit which has circuit distance $10$, this circuit is believed to have circuit distance $12$, matching that of the code, and is constructed using Shaw-Terhal morphing~\cite{shaw2025lowering}. We simulate the performance of the circuit using direct Monte Carlo simulation over $12$ rounds. Normalising the logical error rate per round, we obtain the results in Table~\ref{tab:logical_error_rate_morphing}.
\renewcommand{\arraystretch}{1.5}
\begin{table}[ht]
\centering
\begin{tabular}{|c|ccc|ccc|}
\hline
\multirow{2}{*}{\textbf{Circuit}} 
& \multirow{2}{*}{$\mathbf{S}$} 
& \multirow{2}{*}{\raisebox{0.45ex}{$\gamma_{\text{min}}$}} 
& \multirow{2}{*}{\raisebox{0.45ex}{$\gamma_{\text{max}}$}}
& \multicolumn{3}{c|}{\textbf{Logical Error Rate} ($p = 10^{-3}$)} \\
\cline{5-7}
 &  &  &  & \textbf{Low} & \textbf{Mean} & \textbf{High} \\\hline
\multirow{2}{*}{\makecell{Morphing Idling Circuit\\(Per Round)}}
  & \multirow{2}{*}{20} & \multirow{2}{*}{$-0.09$} & \multirow{2}{*}{$0.76$}  & \multirow{2}{*}{$9.2\times 10^{-12}$} & \multirow{2}{*}{$4.4\times 10^{-11}$} & \multirow{2}{*}{$1.3\times 10^{-10}$} \\
  &                     &  &  &  &  &  \\
\hline
\end{tabular}
\caption{Simulated logical error rate per round for the new idling circuit for the gross code based on morphing, parameters used, as well as a $95\%$ confidence interval. See Appendix~\ref{sec:improved_ler_bicycle} for a discussion of the gamma parameters. This is the result of direct Monte Carlo simulation.}
\label{tab:logical_error_rate_morphing}
\end{table}

Owing to the low logical error rate,\footnote{As a point of interest, the logical error rate is too low here to attempt a direct sweep of the gamma parameters to optimise for logical error rate (see Appendix~\ref{sec:improved_ler_bicycle} for a discussion on gamma parameters). In lieu of this, we chose gamma parameters according to the fastest average decoding time per shot (since this is seen to correlate with logical error rate somewhat), without observing any fails, and then simulated the logical error rate at the single point $(\gamma_{\min}, \gamma_{\max}) = (-0.09, 0.76)$.} and the fact that we perform direct Monte Carlo simulation, without relying on extrapolation, the current confidence interval is very wide. Nevertheless, a significant improvement over the regular gross code idling circuit at $10^{-8.8}$ logical error rate per round~\cite{yoder2025tourgrossmodularquantum} is strongly suggested.

In Figures~\ref{fig:results_with_morphing_1e-3} and~\ref{fig:1e-3_with_morphing_BB_subplots}, we show what is obtained for the footprint estimates of the whole system if one uses $4.4\times 10^{-11}$ for the logical error rate of idling per round, keeping the error rates of other operations as in Appendix~\ref{sec:improved_ler_bicycle}. Significant improvements are rendered by doing this in the footprint estimates, as shown in Figure~\ref{fig:results_with_morphing_1e-3}. In addition, Figure~\ref{fig:1e-3_with_morphing_BB_subplots} shows that the the system with the morphing circuit logical error rate for idling always benefits from post-selection, suggesting that the bottleneck from idling noise has indeed been alleviated, and the error contributions from idling and from operations in cat state generation/consumption are now much more balanced.

\begin{figure}[ht]
    \centering
    \includegraphics[width=\linewidth]{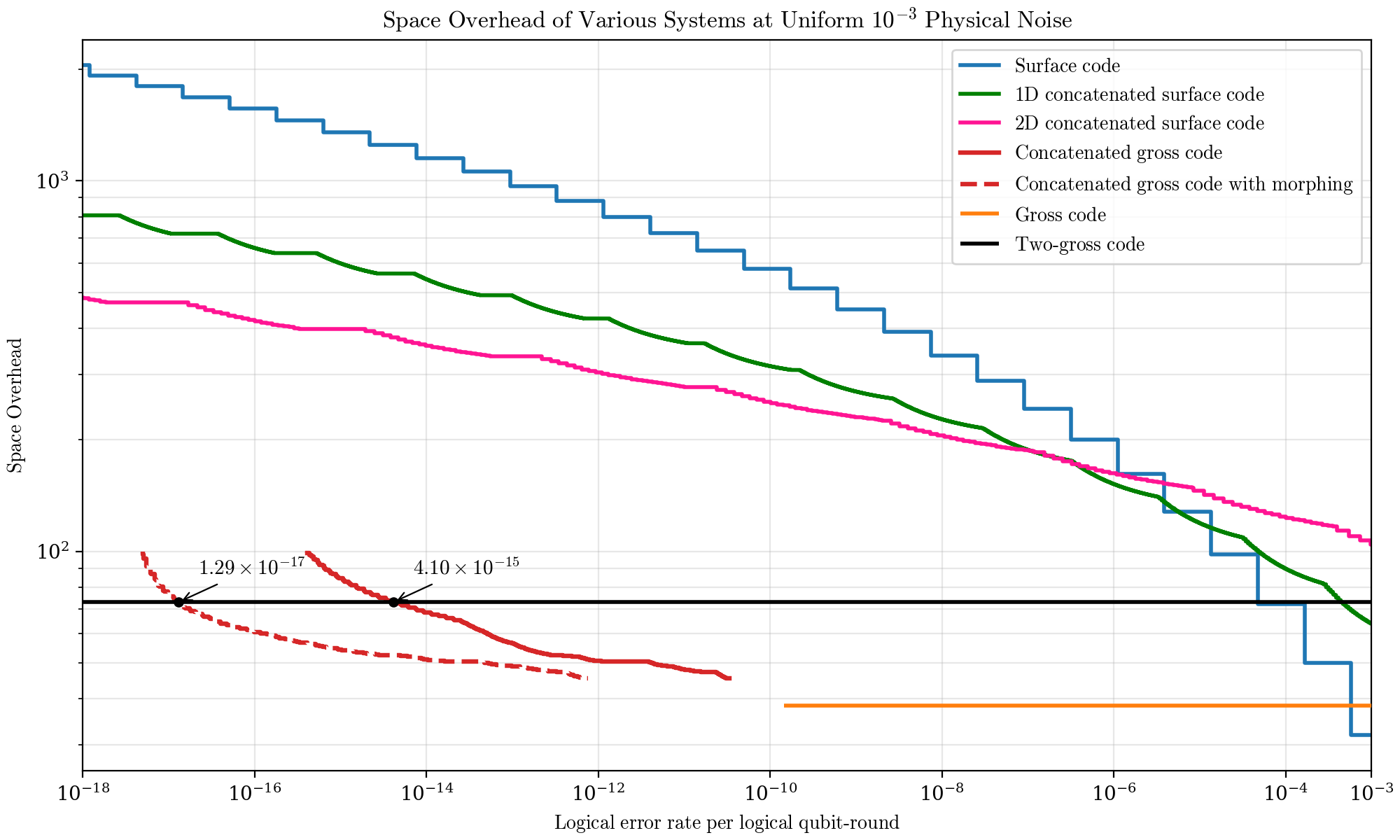}
    \caption{The same plot as Figure~\ref{fig:results_1e-3}, where an additional dashed line shows what is obtained with the logical error rate of the morphing circuit used for the idling operation.}
    \label{fig:results_with_morphing_1e-3}
\end{figure}

\begin{figure}[ht]
    \centering
    \includegraphics[width=\linewidth]{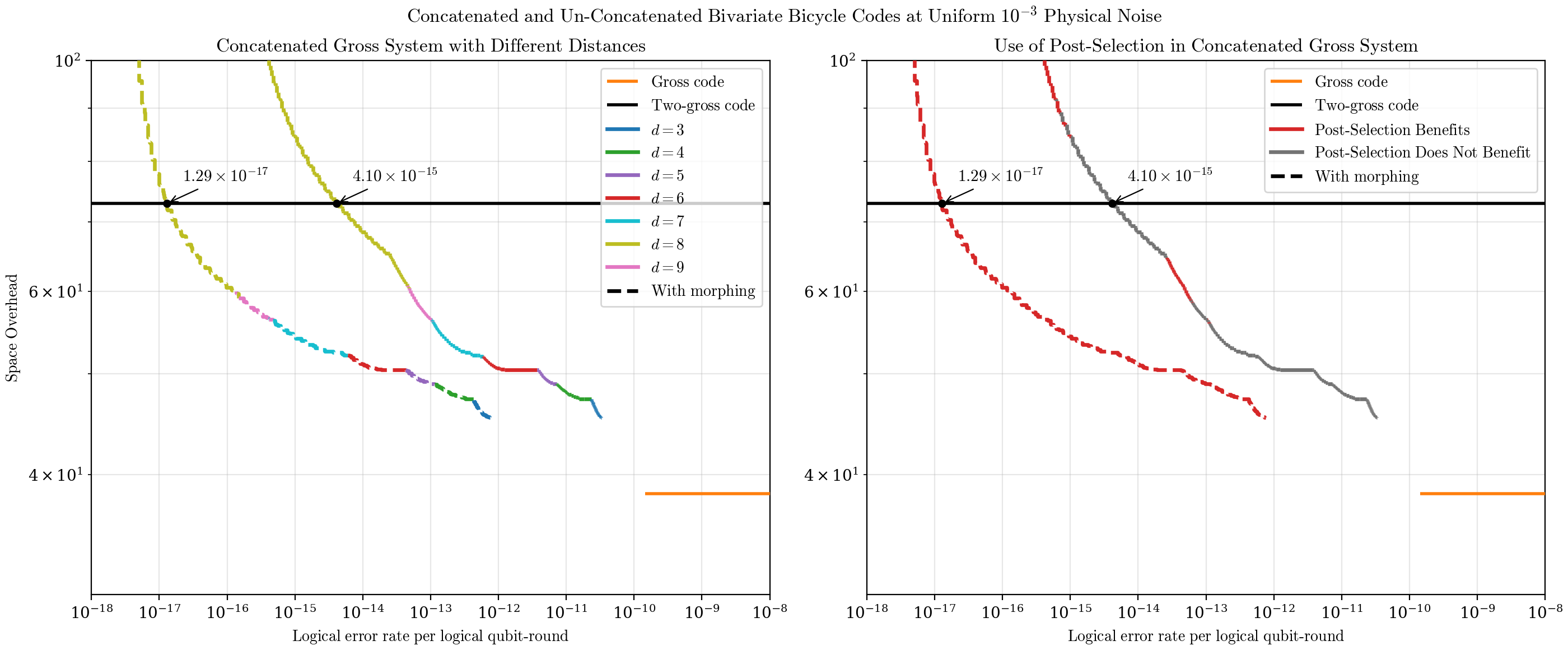}
    \caption{The same plots as Figure~\ref{fig:1e-3_BB_subplots}, but with additional dashed lines showing what is obtained with the logical error rate of the morphing circuit used for the idling operation.}
    \label{fig:1e-3_with_morphing_BB_subplots}
\end{figure}

These are not claimed as the footprints estimates for now, for two reasons. Firstly, our confidence interval for the morphing idling circuit logical error rate is still very wide, see Table~\ref{tab:logical_error_rate_morphing}, although we aim to narrow the confidence interval in future. Additionally, the ancilla systems constructed thus far~\cite{yoder2025tourgrossmodularquantum} would not accommodate the new idling circuit, and would require a different design. Nevertheless, this information is included to show that substantial improvements are anticipated in future. For example, circuits for the surgery instructions also achieving the code distance of $12$ would likely render yet more improvements.

With this said, improvements to the concatenated gross system are unlikely to render a $10\times$ improvement in space overhead over the $2\text{D}$ yoked surface codes. To see this, note that concatenation cannot improve on the space overhead of the underlying inner code; it can only make it higher (in exchange for a lower logical error rate). Referring to the overhead of the regular gross code and the $2\text{D}$ yoked surface codes in Figure~\ref{fig:results_with_morphing_1e-3}, we can see that it is not possible, even in principle, for a concatenated gross code to improve on the space overhead of the $2\text{D}$ yoked surface code, down to a logical error rate of about $10^{-15}$ per logical qubit-round, and this would require the concatenated system to have essentially the same overhead as the un-concatenated system, which seems perhaps unlikely. In principle, one could consider a smaller inner BB code, see~\cite{bravyi2024high} for possibilities.

With this being said, we would note that the layout of the $2\text{D}$ yoked surface code makes reading and writing from and to the memory difficult and slow, since the logical operators are spread in $2$ dimensions; the system is only usable for very cold storage. Our systems have the same layout as the $1\text{D}$ yoked surface codes, and so a comparison between our systems and the $1\text{D}$ yoked surface codes is worth considering.\footnote{Let us further modify this statement by also noting that the outer code cycle time of the concatenated gross system is slow owing to the difficulty in compiling cat states with the current bicycle instructions, and is perhaps more comparable to the outer code cycle time of the $2\text{D}$ yoked surface codes.} Future improvements are likely to yield a $10\times$-improved space overhead here in operationally-relevant regimes, as is already suggested by Figure~\ref{fig:results_with_morphing_1e-3} with the morphing idling circuit.

\section{Details on Reed-Solomon Codes}\label{sec:RS_details}

In this section, we will define Reed-Solomon codes, both classical and quantum. The following will be brief; see~\cite{wills2026review} for full details. For further reading, we recommend~\cite{macwilliams1977theory}.

\subsection{Definition of Classical and Quantum Reed-Solomon Codes}\label{subsec:qrs_def}

With $\mathbb{F}_q$ the field of $q$ elements, for some $q = 2^s$, let $\mathbb{F}_q[x]^{<k}$ be the set of polynomials in one variable $x$ over $\mathbb{F}_q$ of degree less than $k$. Consider $n$ distinct elements $\balpha = (\alpha_1, \alpha_2, \ldots, \alpha_n) \subseteq \mathbb{F}_q$ (note that we must have $n \leq q$). Further consider a vector $\bv = (v_1, v_2, \ldots, v_n) \in \mathbb{F}_q^n$ where $v_i \neq 0$ for all $i$. We can then define a classical code, called a generalised Reed-Solomon code,
\begin{equation}
    \GRS_k(\balpha,\bv) = \left\{(v_1f(\alpha_1), v_2f(\alpha_2), \ldots, v_nf(\alpha_n)):f \in \mathbb{F}_q[x]^{<k}\right\}.
\end{equation}
In words, the code is formed by evaluating all polynomials up to some degree at the points $\balpha$ and component-wise multiplying each codeword by $\bv$. This is an $\mathbb{F}_q$-linear code with parameters
\begin{equation}
    [n,k,n-k+1].
\end{equation}
A generator matrix for the classical code $\mathsf{GRS}_k(\balpha,\bv)$ is given by $H \in \mathbb{F}_q^{k \times n}$, where $H_{ij} = v_j\alpha_j^i$ for $j \in [n]$ and $i \in \{0, 1, \ldots, k-1\}$.
Moreover, we have
\begin{equation}
    \GRS_k(\balpha,\bv)^\perp = \GRS_{n-k}(\balpha,\bu),
\end{equation}
where $\bu = (u_1, \ldots, u_n)$, and
\begin{equation}
    u_i^{-1} = v_i\prod_{j \neq i}(\alpha_i-\alpha_j).\label{eq:orthog_mult_relation}
\end{equation}
By definition, we have
\begin{equation}
    \GRS_{k_1}(\balpha, \bv) \subseteq \GRS_{k_2}(\balpha, \bv)
\end{equation}
for any $k_1 \leq k_2$. Therefore, we may define a qudit CSS code, called a quantum Reed-Solomon code, by defining vector spaces of stabilisers $\mathcal{L}_X$ and $\mathcal{L}_Z$\footnote{A quantum CSS code may be defined via stabiliser spaces $\mathcal{L}_X$ and $\mathcal{L}_Z$ by considering the simultaneous $+1$-eigenspace of all operators $X^x$, for $x \in \mathcal{L}_X$, and $Z^z$, for $z \in \mathcal{L}_Z$. Here, given $x \in \mathbb{F}_q^n$, $X^x \coloneq \bigotimes_{i=1}^nX_i^{x_i}$, and similarly for $Z^z$. For more information, see~\cite{wills2026review}.} via
\begin{align}
    \mathcal{L}_X &= \GRS_{k_1}(\balpha, \bv)\\
    \mathcal{L}_Z^\perp &= \GRS_{k_2}(\balpha, \bv),
\end{align}
so that $\mathcal{L}_Z = \GRS_{n-k_2}(\balpha, \bu)$, where $\bu$ may be found via Equation~\eqref{eq:orthog_mult_relation}. This is a qudit code with length $n$, $k_2-k_1$ logical qudits, $X$-distance $n-k_2+1$ and $Z$-distance $k_1+1$. In this work, when we talk about a length $n$ and distance $d$ quantum Reed-Solomon code, we are considering the above with $k_2 = n-d+1$ and $k_1 = d-1$. The qudit code thus has parameters $[[n,n-2(d-1),d]]_q$. The choice of $\bv$ does not affect the performance of our scheme at all. The choice of $\balpha$ affects the performance of our scheme only very slightly; see the next subsection where we make explicit choices of $\balpha$ for concreteness.

It is possible to write down a qudit stabiliser tableau for a codestate of the qudit code (see~\cite{wills2026review} for the details on Galois qudit stabiliser tableaux):
\begin{equation}\renewcommand{\arraystretch}{1.2}
\left(
\begin{array}{cccc|cccc}
X^{v_1} & X^{v_2} & \cdots & X^{v_n}
  & & & & \\
X^{v_1\alpha_1} & X^{v_2\alpha_2} & \cdots & X^{v_n\alpha_n}
  & & & & \\
\vdots & \vdots & \ddots & \vdots
  & & & & \\
X^{v_1\alpha_1^{d-2}} & X^{v_2\alpha_2^{d-2}} & \cdots & X^{v_n\alpha_n^{d-2}}
  & & & & \\ \hline
& & & &
  Z^{u_1} & Z^{u_2} & \cdots & Z^{u_n} \\
& & & &
  Z^{u_1\alpha_1} & Z^{u_2\alpha_2} & \cdots & Z^{u_n\alpha_n} \\
& & & &
  \vdots & \vdots & \ddots & \vdots \\
& & & &
  Z^{u_1\alpha_1^{d-2}} & Z^{u_2\alpha_2^{d-2}} & \cdots & Z^{u_n\alpha_n^{d-2}} \\ \hline
& & & &
  \multicolumn{4}{c}{\hat{\mathcal{L}}_X} \\
& & & &
  \multicolumn{4}{c}{\hat{\mathcal{L}}_Z}
\end{array}\renewcommand{\arraystretch}{1}
\right),
\end{equation}
where the last two rows denote the logical operators specifying the actual code state in question.

As always, the qudit code may be operated as a qubit code, where the qudit-to-qubit mappings are determined by a tuple $\mathcal{B} = (B_i)_{i=1}^n$, where $B_i$ is a basis for $\mathbb{F}_q$ over $\mathbb{F}_2$; see~\cite{wills2026review} for more details. The resulting qubit code has parameters $[[ns,(n-2(d-1))s, \geq d]]$.

\subsection{Choices of Evaluation Points}\label{subsec:evaluation_points}

We list here the evaluation points $\balpha$ that we use for the definitions of our Reed-Solomon codes for concreteness. Note that our analysis depends only very slightly on the actual choice of $\balpha$. In fact, the choice of $\balpha$ only comes into the analysis of the protocol via the quantities in Subsubsection~\ref{subsubsec:list_decoding}: $F_{2\to k\times 2}(n,d=4,\balpha), F_{3\to k\times 3}(n,d=5,\balpha)$, $F_{3\to k\times 3}(n,d=6,\balpha)$, and $F_{4\to k\times 4}(n,d=6,\balpha)$, and even then we note that these quantities vary only very slightly over the random choice of $\balpha$. Nevertheless, the $\balpha$ we use are listed for concreteness.

For a particular $\balpha$, each quantity can be estimated by a computer by sampling random weight-$2$, $3$ or $4$ errors, as appropriate, and calculating by brute force for each one how many distinct weight-$2$, $3$, or $4$ errors it shares its syndrome with. For every $n$ we are interested in, we choose one $\balpha$ at random (we do not attempt to optimise this choice), and use that value of $\balpha$ for the codes of each distance at that $n$. We take $100,000$ samples of weight-$2$ or $3$ errors to estimate $F_{2\to k \times 2}(n,d=4,\balpha)$, $F_{3\to k \times 3}(n,d=5,\balpha)$ and $F_{3\to k\times 3}(n,d=5,\balpha)$, and $10,000$ samples of weight-$4$ errors to estimate $F_{4\to k\times 4}(n,d=6,\balpha)$.

Field elements in $\mathbb{F}_q$, where $q = 2^{11} = 2048$, are represented as follows. The field is written $\mathbb{F}_{2048} = \mathbb{F}_2(t)$, where $t^{11}+t^2+1 = 0$. Then, the $2048$ elements may be uniquely expressed as $\sum_{i=0}^{10}a_it^i$, where $a_i \in \mathbb{F}_2$. Field elements are then written as integers in $\{0, 1, \ldots, 2047\}$, where their binary expansions are the $(a_i)_{i=0}^{10}$. For example, the integer $1096 = 1024 + 64 + 8 = 2^{10}+2^6+2^3$ represents the field element $t^{10}+t^6+t^3$. $\balpha_n$ are the evaluation points for all codes of length $n$.

\tiny
\allowdisplaybreaks
\begin{align*}
\boldsymbol{\alpha}_{20} &= (108, 549, 179, 1575, 835, 546, 221, 1718, 1846, 1792, 79, 777, 1099, 1152, 681, 698, 1746, 107, 327, 277),\\[0.5em]
\boldsymbol{\alpha}_{21} &= (862, 231, 312, 1460, 818, 542, 868, 1945, 1539, 188, 2039, 822, 195, 1078, 235, 664, 1345, 1347, 563, 630, 1375),\\[0.5em]
\boldsymbol{\alpha}_{22} &= (1557, 176, 877, 1953, 645, 28, 760, 481, 1062, 1700, 1146, 739, 1508, 1510, 471, 1343, 640, 1579, 1026, 1134, 934, 1473),\\[0.5em]
\boldsymbol{\alpha}_{23} &= (771, 883, 89, 1611, 1673, 747, 76, 1626, 1276, 366, 547, 417, 955, 563, 145, 500, 1288, 631, 1793, 1071, 1814, 1664, 298),\\[0.5em]
\boldsymbol{\alpha}_{24} &= (1287, 315, 677, 288, 1151, 1383, 127, 1661, 730, 1160, 682, 259, 1220, 691, 2024, 84, 1793, 1390, 209, 1185, 1274, 1320, 1945, 890),\\[0.5em]
\boldsymbol{\alpha}_{25} &= (1622, 1290, 683, 1297, 1599, 1570, 69, 1623, 1370, 284, 1303, 110, 965, 1969, 441, 377, 893, 1067, 1176, 743, 1983, 1474, 410, 609, 1588),\\[0.5em]
\boldsymbol{\alpha}_{26} &= (758, 1472, 1797, 2040, 1886, 1835, 924, 1454, 1671, 752, 733, 256, 1187, 1547, 147, 1003, 621, 454, 1507, 1849, 1778, 126, 1915, 938, 689,\\
&\qquad 167),\\[0.5em]
\boldsymbol{\alpha}_{27} &= (705, 89, 1009, 238, 361, 1309, 1106, 1697, 897, 2036, 141, 222, 820, 1135, 21, 2024, 1095, 10, 171, 589, 108, 977, 1163, 1820, 1158,\\
&\qquad 671, 1165),\\[0.5em]
\boldsymbol{\alpha}_{28} &= (1205, 2, 71, 747, 667, 150, 1265, 1804, 1969, 794, 1195, 1657, 1638, 1799, 1617, 1388, 1985, 406, 165, 560, 1838, 768, 1602, 1359, 697,\\
&\qquad 1201, 1796, 1065),\\[0.5em]
\boldsymbol{\alpha}_{29} &= (490, 247, 1655, 811, 763, 157, 1633, 1229, 810, 1620, 2032, 8, 1882, 1651, 701, 1613, 1614, 191, 1224, 733, 569, 1685, 1175, 397, 1869,\\
&\qquad 263, 360, 1541, 1013),\\[0.5em]
\boldsymbol{\alpha}_{30} &= (812, 1332, 1341, 1172, 1060, 77, 1879, 1131, 486, 1899, 1238, 750, 1405, 1590, 1536, 1155, 1162, 282, 744, 36, 2014, 1914, 1069, 1821, 1325,\\
&\qquad 788, 50, 1224, 1799, 616),\\[0.5em]
\boldsymbol{\alpha}_{31} &= (60, 741, 215, 1918, 1102, 292, 609, 990, 808, 1032, 906, 556, 1558, 1023, 221, 1154, 598, 359, 565, 749, 1041, 642, 277, 1496, 332,\\
&\qquad 479, 1553, 166, 1843, 1443, 223),\\[0.5em]
\boldsymbol{\alpha}_{32} &= (1147, 1844, 422, 716, 688, 434, 948, 1616, 882, 1547, 201, 190, 654, 364, 1347, 686, 1646, 1725, 1981, 1677, 270, 569, 337, 1552, 531,\\
&\qquad 561, 271, 606, 1917, 331, 997, 1563),\\[0.5em]
\boldsymbol{\alpha}_{33} &= (621, 1386, 632, 675, 1303, 1477, 1579, 111, 1486, 1044, 26, 72, 1210, 1948, 250, 1222, 811, 685, 556, 1822, 312, 1302, 1932, 7, 253,\\
&\qquad 196, 1433, 346, 597, 1987, 1093, 385, 48),\\[0.5em]
\boldsymbol{\alpha}_{34} &= (458, 1357, 1670, 562, 296, 903, 757, 629, 1971, 796, 1466, 1642, 985, 1290, 1378, 2045, 324, 1387, 740, 771, 1992, 805, 81, 523, 143,\\
&\qquad 582, 504, 856, 878, 1442, 444, 1453, 1591, 1520),\\[0.5em]
\boldsymbol{\alpha}_{35} &= (251, 759, 1577, 922, 1942, 1379, 1484, 410, 1449, 1791, 1932, 196, 533, 1511, 1446, 1640, 1260, 1729, 272, 261, 113, 1616, 846, 127, 1756,\\
&\qquad 1677, 752, 1884, 115, 1786, 496, 886, 48, 1512, 274),\\[0.5em]
\boldsymbol{\alpha}_{36} &= (439, 437, 153, 521, 413, 1204, 664, 872, 1179, 463, 2021, 1780, 259, 782, 1431, 1708, 329, 1889, 1060, 525, 508, 422, 43, 40, 1821,\\
&\qquad 409, 176, 92, 154, 1652, 1061, 647, 1143, 1019, 1430, 208),\\[0.5em]
\boldsymbol{\alpha}_{37} &= (1637, 1121, 1966, 334, 2010, 1009, 1632, 1016, 689, 1035, 1291, 423, 1090, 1623, 912, 1875, 142, 303, 1661, 697, 881, 729, 1269, 112, 398,\\
&\qquad 747, 893, 474, 558, 843, 1252, 18, 1963, 1502, 734, 1974, 1283),\\[0.5em]
\boldsymbol{\alpha}_{38} &= (535, 492, 873, 1993, 971, 517, 2034, 656, 1931, 289, 2028, 446, 718, 1411, 222, 93, 227, 505, 1911, 1450, 1948, 407, 1207, 43, 520,\\
&\qquad 1390, 855, 1076, 560, 2036, 359, 572, 176, 14, 173, 1827, 1113, 1410),\\[0.5em]
\boldsymbol{\alpha}_{39} &= (1102, 1064, 769, 1109, 1767, 1980, 1079, 1633, 741, 1083, 1008, 1520, 49, 1439, 1551, 146, 467, 1752, 1871, 367, 895, 1140, 388, 152, 342,\\
&\qquad 1546, 672, 989, 1548, 1066, 1418, 1374, 1437, 36, 265, 1586, 806, 1335, 164),\\[0.5em]
\boldsymbol{\alpha}_{40} &= (1844, 1425, 23, 974, 687, 1866, 1924, 584, 1110, 1063, 1244, 1185, 1755, 1626, 1230, 1892, 1637, 1059, 1990, 786, 1886, 1023, 327, 1096, 1387,\\
&\qquad 797, 306, 864, 1917, 1757, 899, 1877, 1948, 1196, 857, 1655, 1420, 1880, 1035, 710),\\[0.5em]
\boldsymbol{\alpha}_{41} &= (500, 1716, 395, 180, 515, 1750, 294, 455, 990, 203, 759, 791, 1314, 914, 575, 138, 2008, 1207, 451, 574, 1440, 337, 542, 781, 844,\\
&\qquad 1513, 1222, 412, 1628, 597, 1253, 1122, 665, 1573, 1131, 1090, 521, 1010, 61, 1598, 838),\\[0.5em]
\boldsymbol{\alpha}_{42} &= (1846, 579, 1278, 627, 703, 1821, 1128, 1533, 1250, 1097, 1515, 481, 417, 757, 63, 1363, 232, 1512, 2047, 491, 1566, 1949, 1538, 438, 1641,\\
&\qquad 1064, 631, 644, 1695, 1773, 741, 53, 635, 1800, 1176, 1522, 150, 945, 84, 1417, 400, 1117),\\[0.5em]
\boldsymbol{\alpha}_{43} &= (1642, 9, 1577, 1775, 325, 824, 925, 1448, 417, 1514, 692, 34, 1014, 1241, 1013, 1765, 195, 460, 1253, 1717, 268, 1094, 1315, 1179, 446,\\
&\qquad 1433, 820, 1559, 1338, 80, 1074, 841, 148, 1547, 1465, 211, 213, 1145, 1323, 1904, 1, 1958, 293),\\[0.5em]
\boldsymbol{\alpha}_{44} &= (956, 1018, 1525, 58, 63, 1724, 348, 1069, 1029, 341, 1730, 1296, 382, 1132, 522, 442, 625, 958, 1222, 1300, 1214, 462, 1130, 344, 1332,\\
&\qquad 864, 203, 1933, 1096, 606, 1312, 392, 1909, 44, 559, 444, 1623, 1270, 1833, 1176, 692, 1077, 922, 2033),\\[0.5em]
\boldsymbol{\alpha}_{45} &= (109, 404, 1142, 852, 137, 60, 251, 44, 86, 363, 1428, 358, 1496, 560, 200, 219, 946, 699, 1753, 683, 829, 1302, 1097, 569, 2021,\\
&\qquad 1226, 1674, 972, 1977, 1090, 1047, 1020, 768, 1673, 1610, 1750, 1780, 1025, 888, 1130, 1537, 618, 1067, 2023, 845),\\[0.5em]
\boldsymbol{\alpha}_{46} &= (1008, 1359, 1757, 1775, 1755, 357, 1434, 173, 879, 979, 393, 1413, 1887, 65, 135, 1503, 1633, 82, 1779, 443, 67, 1141, 1646, 762, 1773,\\
&\qquad 626, 275, 163, 1423, 2012, 735, 151, 423, 1623, 1725, 533, 1680, 958, 1386, 1720, 1838, 1243, 111, 1739, 236, 1406),\\[0.5em]
\boldsymbol{\alpha}_{47} &= (197, 761, 1426, 1265, 651, 1988, 1118, 467, 1813, 136, 745, 1079, 913, 281, 1013, 1043, 655, 491, 11, 1025, 492, 2021, 785, 538, 1742,\\
&\qquad 611, 1713, 1460, 1569, 1788, 153, 311, 1704, 658, 1653, 1854, 1430, 1565, 6, 1142, 1826, 1686, 154, 1233, 887, 219, 1197),\\[0.5em]
\boldsymbol{\alpha}_{48} &= (1351, 1788, 1198, 437, 550, 292, 183, 1782, 1197, 1650, 1651, 1976, 1643, 537, 845, 194, 700, 549, 1962, 1358, 763, 1135, 1758, 1853, 962,\\
&\qquad 149, 4, 731, 1730, 1720, 1837, 355, 1588, 968, 412, 1242, 370, 1635, 1543, 935, 452, 1560, 691, 1279, 1285, 1840, 288, 466),\\[0.5em]
\boldsymbol{\alpha}_{49} &= (1329, 699, 554, 310, 720, 1568, 444, 325, 603, 988, 82, 107, 1708, 351, 903, 941, 1737, 1389, 914, 1172, 1460, 1543, 537, 1391, 171,\\
&\qquad 445, 1050, 750, 1694, 841, 1346, 625, 375, 124, 243, 61, 160, 1587, 1986, 1957, 586, 1803, 1762, 1292, 1997, 837, 889, 1254, 1023),\\[0.5em]
\boldsymbol{\alpha}_{50} &= (1103, 655, 1453, 912, 77, 81, 1980, 1734, 1814, 1412, 727, 207, 1083, 698, 87, 637, 644, 1211, 1876, 1137, 327, 778, 1209, 1381, 1113,\\
&\qquad 1988, 487, 952, 494, 113, 1662, 19, 1165, 738, 82, 509, 1187, 1654, 909, 653, 78, 1533, 1984, 1498, 805, 543, 1080, 661, 1322, 632),\\[0.5em]
\boldsymbol{\alpha}_{51} &= (1238, 82, 391, 1371, 1609, 294, 1253, 625, 208, 956, 254, 1302, 1720, 1827, 1853, 1392, 1677, 555, 2039, 377, 8, 535, 1760, 13, 1510,\\
&\qquad 1352, 1299, 1914, 1895, 1024, 1648, 524, 668, 865, 1478, 215, 965, 354, 1812, 1370, 1932, 357, 953, 1136, 692, 585, 1519, 1780, 1912, 1872,\\
&\qquad 1374),\\[0.5em]
\boldsymbol{\alpha}_{52} &= (1095, 1510, 859, 321, 976, 1132, 1981, 1787, 1815, 557, 1700, 1481, 679, 399, 1351, 1511, 1525, 1342, 1348, 1455, 1602, 692, 1216, 2022, 98,\\
&\qquad 471, 1914, 924, 881, 1175, 512, 1161, 656, 1850, 262, 1748, 1743, 916, 177, 1328, 325, 1269, 849, 400, 460, 1788, 694, 1296, 466, 1599,\\
&\qquad 576, 618),\\[0.5em]
\boldsymbol{\alpha}_{53} &= (381, 1157, 1443, 1156, 501, 1845, 1197, 30, 397, 1638, 542, 201, 848, 143, 385, 1161, 1404, 910, 1606, 1108, 367, 297, 666, 903, 1037,\\
&\qquad 495, 424, 1345, 836, 1297, 649, 42, 1877, 1554, 2033, 437, 1776, 347, 975, 664, 1849, 1846, 711, 770, 640, 1265, 340, 891, 1422, 572,\\
&\qquad 779, 1632, 504),\\[0.5em]
\boldsymbol{\alpha}_{54} &= (996, 351, 250, 849, 38, 503, 300, 84, 1394, 1224, 752, 1443, 1168, 1001, 1990, 705, 405, 1833, 902, 209, 1463, 934, 443, 773, 1803,\\
&\qquad 1511, 519, 1564, 2039, 54, 786, 1968, 52, 794, 905, 421, 565, 256, 855, 1416, 783, 1523, 61, 441, 1829, 1985, 207, 1972, 1021, 1355,\\
&\qquad 306, 976, 1869, 1093),\\[0.5em]
\boldsymbol{\alpha}_{55} &= (1551, 1699, 1486, 666, 1833, 242, 703, 675, 382, 1301, 1673, 1943, 1829, 1890, 1942, 1368, 825, 244, 1423, 886, 485, 964, 1539, 1979, 1546,\\
&\qquad 1657, 1383, 1457, 874, 736, 288, 1356, 1194, 252, 1877, 1678, 1264, 212, 1392, 176, 1541, 144, 592, 1468, 530, 1926, 1150, 602, 43, 1155,\\
&\qquad 1731, 559, 1663, 503, 584),\\[0.5em]
\boldsymbol{\alpha}_{56} &= (2011, 1476, 799, 1441, 1714, 1493, 904, 778, 1610, 669, 286, 2035, 24, 1252, 1659, 1970, 1079, 1555, 1732, 1165, 403, 1982, 499, 1108, 173,\\
&\qquad 553, 1164, 850, 1344, 1042, 234, 747, 798, 1030, 860, 72, 1755, 248, 1433, 215, 884, 1007, 1190, 541, 1918, 791, 94, 502, 365, 1353,\\
&\qquad 20, 428, 1025, 1393, 2014, 943),\\[0.5em]
\boldsymbol{\alpha}_{57} &= (352, 470, 462, 1987, 90, 386, 1339, 627, 200, 909, 1020, 1584, 1559, 65, 825, 1771, 533, 1175, 1894, 94, 922, 2, 1955, 1689, 1099,\\
&\qquad 1278, 1235, 691, 556, 1468, 1940, 1512, 1283, 1456, 199, 1562, 130, 1595, 1402, 1414, 1931, 384, 1576, 954, 1888, 1772, 445, 41, 914, 1487,\\
&\qquad 1413, 877, 331, 1329, 913, 1723, 1991),\\[0.5em]
\boldsymbol{\alpha}_{58} &= (1824, 354, 423, 1280, 1438, 45, 236, 1601, 899, 466, 1701, 1192, 1879, 706, 888, 412, 362, 1104, 361, 513, 764, 1688, 1614, 1977, 1268,\\
&\qquad 1672, 1727, 606, 1528, 208, 637, 1812, 2041, 46, 1470, 252, 1273, 433, 912, 767, 268, 258, 1863, 1443, 190, 333, 1381, 478, 861, 846,\\
&\qquad 1218, 791, 326, 139, 1825, 1027, 891, 721),\\[0.5em]
\boldsymbol{\alpha}_{59} &= (1736, 1185, 775, 1954, 676, 1564, 97, 1324, 1930, 1113, 55, 117, 1888, 1123, 1998, 1903, 1114, 1216, 981, 518, 934, 784, 1311, 2010, 148,\\
&\qquad 1183, 1599, 525, 1390, 1950, 83, 990, 907, 1868, 1008, 244, 468, 113, 1247, 1030, 1385, 1553, 41, 487, 1560, 139, 805, 1702, 381, 1063,\\
&\qquad 21, 1881, 772, 891, 1435, 1588, 29, 1073, 1414),\\[0.5em]
\boldsymbol{\alpha}_{60} &= (644, 96, 1185, 1494, 1830, 1573, 738, 1611, 1101, 2036, 493, 310, 440, 1920, 1128, 1695, 450, 696, 1192, 261, 759, 514, 655, 1657, 368,\\
&\qquad 461, 505, 897, 592, 176, 288, 714, 1028, 323, 1591, 60, 193, 1824, 1436, 880, 1724, 357, 335, 296, 254, 1072, 1144, 232, 409, 1828,\\
&\qquad 576, 1012, 500, 760, 1129, 334, 1463, 1205, 1782, 33),\\[0.5em]
\boldsymbol{\alpha}_{61} &= (1007, 1132, 1613, 214, 1064, 180, 728, 2042, 158, 45, 887, 1526, 407, 683, 162, 370, 2014, 10, 659, 1890, 2007, 939, 796, 577, 1119,\\
&\qquad 119, 59, 1339, 716, 2005, 1178, 1152, 209, 1635, 851, 751, 1477, 1104, 626, 115, 1880, 1434, 948, 1501, 539, 581, 1651, 1872, 625, 715,\\
&\qquad 593, 1926, 585, 1364, 935, 12, 406, 580, 575, 1687, 708),\\[0.5em]
\boldsymbol{\alpha}_{62} &= (1105, 1862, 1006, 1032, 858, 263, 276, 1460, 1484, 736, 1379, 2006, 197, 1838, 1111, 1808, 543, 1988, 829, 726, 552, 1443, 1704, 961, 1913,\\
&\qquad 681, 92, 161, 1114, 278, 729, 73, 1441, 64, 748, 1146, 653, 642, 915, 1383, 755, 130, 1716, 869, 1445, 378, 600, 1280, 444, 1841,\\
&\qquad 974, 1930, 1368, 205, 489, 1136, 1313, 1626, 851, 336, 1416, 772),\\[0.5em]
\boldsymbol{\alpha}_{63} &= (1721, 1947, 948, 1312, 409, 1843, 1623, 214, 120, 258, 791, 439, 1605, 814, 996, 194, 1159, 253, 640, 581, 562, 306, 877, 1398, 110,\\
&\qquad 1013, 445, 1540, 1230, 2046, 322, 188, 1487, 408, 924, 834, 1679, 958, 1278, 1450, 126, 1775, 1351, 523, 1906, 1380, 1742, 1459, 26, 917,\\
&\qquad 153, 36, 252, 654, 1273, 1505, 883, 486, 1596, 1508, 1444, 307, 891),\\[0.5em]
\boldsymbol{\alpha}_{64} &= (570, 217, 1774, 1926, 1536, 1695, 1525, 1794, 460, 458, 241, 944, 2034, 1222, 1730, 1712, 1280, 624, 398, 794, 1660, 1292, 1806, 1713, 323,\\
&\qquad 2037, 1981, 180, 990, 1286, 1134, 695, 1505, 480, 1542, 1110, 513, 401, 2001, 2009, 1699, 1517, 448, 1724, 796, 108, 1369, 215, 439, 1357,\\
&\qquad 227, 977, 627, 1789, 224, 1653, 740, 11, 1162, 1980, 1677, 1140, 877, 1860),\\[0.5em]
\boldsymbol{\alpha}_{65} &= (2010, 578, 67, 1002, 1419, 1003, 957, 1128, 875, 75, 1173, 1885, 1405, 1818, 1694, 1774, 923, 1197, 1350, 1360, 1469, 922, 66, 971, 782,\\
&\qquad 1641, 1921, 83, 435, 946, 400, 810, 1764, 328, 1227, 1029, 457, 677, 134, 1240, 836, 406, 1732, 379, 729, 1301, 251, 590, 1098, 427,\\
&\qquad 1826, 1066, 2039, 175, 1379, 1182, 586, 811, 373, 1247, 1899, 1992, 611, 643, 1364),\\[0.5em]
\boldsymbol{\alpha}_{66} &= (1714, 794, 192, 249, 699, 1005, 1467, 586, 132, 507, 596, 1729, 1174, 767, 1277, 737, 832, 1918, 1368, 146, 547, 1314, 1167, 252, 1011,\\
&\qquad 1421, 916, 746, 1917, 1740, 1696, 660, 1354, 52, 1282, 1172, 269, 1512, 1858, 37, 71, 1676, 518, 1094, 371, 1797, 1716, 813, 1169, 2011,\\
&\qquad 174, 1528, 1105, 1905, 1856, 913, 1385, 400, 907, 1527, 513, 1266, 1349, 413, 13, 663),\\[0.5em]
\boldsymbol{\alpha}_{67} &= (265, 295, 1106, 1292, 802, 1585, 1804, 704, 1255, 74, 62, 609, 599, 1131, 367, 864, 187, 1470, 1613, 664, 1563, 1697, 1865, 1556, 1407,\\
&\qquad 1247, 980, 806, 1167, 2046, 960, 1083, 70, 1870, 1850, 143, 1551, 840, 1773, 2009, 474, 1415, 1316, 321, 917, 237, 1160, 584, 1243, 592,\\
&\qquad 27, 113, 1884, 1917, 1960, 1143, 247, 424, 570, 1520, 1911, 88, 1657, 1351, 1644, 1228, 398),\\[0.5em]
\boldsymbol{\alpha}_{68} &= (368, 279, 1335, 876, 927, 237, 505, 852, 234, 654, 1991, 119, 2043, 190, 898, 659, 1855, 1008, 1528, 1120, 1831, 415, 1076, 2009, 206,\\
&\qquad 357, 463, 2022, 1657, 586, 1191, 521, 1214, 1748, 1037, 736, 1868, 1194, 544, 75, 798, 372, 23, 1734, 997, 1022, 973, 1372, 1890, 1218,\\
&\qquad 1948, 817, 1283, 1445, 1994, 277, 142, 832, 1200, 12, 1692, 131, 938, 1330, 1789, 303, 214, 1041),\\[0.5em]
\boldsymbol{\alpha}_{69} &= (612, 619, 475, 1980, 1731, 1169, 1950, 1887, 192, 222, 1367, 1150, 2023, 1932, 753, 894, 1194, 569, 1990, 551, 1219, 391, 47, 1554, 1868,\\
&\qquad 915, 1462, 1075, 407, 891, 1271, 1959, 209, 1227, 1425, 1538, 1816, 1629, 2025, 523, 694, 1259, 1370, 1034, 1561, 1937, 277, 1946, 335, 1289,\\
&\qquad 576, 992, 848, 895, 1608, 1260, 83, 920, 79, 886, 1905, 679, 446, 842, 615, 1569, 397, 1571, 181),\\[0.5em]
\boldsymbol{\alpha}_{70} &= (995, 94, 418, 155, 659, 143, 394, 190, 1129, 806, 1375, 1000, 650, 1918, 216, 2021, 1646, 1127, 1398, 1983, 73, 522, 350, 831, 1332,\\
&\qquad 849, 1796, 114, 1899, 440, 699, 1985, 1224, 255, 1797, 499, 370, 475, 848, 96, 1731, 8, 1546, 596, 206, 379, 1082, 1203, 1700, 195,\\
&\qquad 1383, 353, 1494, 572, 895, 1685, 1144, 421, 347, 330, 1288, 429, 128, 1261, 1099, 188, 840, 134, 758, 872),\\[0.5em]
\boldsymbol{\alpha}_{71} &= (1227, 1475, 1038, 1778, 1671, 1098, 369, 803, 1523, 1945, 2022, 1311, 767, 739, 918, 1451, 1230, 1180, 1733, 580, 1326, 1508, 1036, 335, 828,\\
&\qquad 159, 1749, 1377, 946, 1080, 1800, 1924, 378, 932, 1364, 611, 1681, 1369, 1654, 1952, 1308, 575, 1487, 321, 342, 1544, 878, 895, 1504, 603,\\
&\qquad 2, 746, 432, 862, 810, 8, 438, 696, 1798, 1637, 1682, 648, 913, 1265, 983, 784, 63, 1860, 1961, 752, 248),\\[0.5em]
\boldsymbol{\alpha}_{72} &= (1827, 1380, 1209, 1732, 1135, 817, 1396, 151, 1147, 102, 1814, 1556, 604, 1225, 993, 1462, 1153, 833, 1296, 1131, 1575, 1965, 1137, 426, 215,\\
&\qquad 1854, 18, 1987, 520, 1967, 1438, 7, 781, 203, 480, 1488, 411, 627, 1434, 1252, 1770, 500, 1138, 1366, 1936, 1109, 659, 1537, 626, 1562,\\
&\qquad 1664, 1089, 1048, 105, 1896, 6, 719, 638, 170, 1102, 1463, 1513, 1059, 51, 1289, 1775, 384, 1421, 400, 791, 900, 1401),\\[0.5em]
\boldsymbol{\alpha}_{73} &= (496, 720, 1324, 1793, 275, 1347, 659, 1603, 772, 653, 480, 666, 455, 1818, 1205, 276, 603, 1175, 43, 1661, 1695, 1897, 844, 1407, 1756,\\
&\qquad 300, 451, 235, 97, 201, 388, 267, 115, 1039, 2026, 122, 1300, 25, 1421, 439, 1054, 106, 991, 1325, 264, 1040, 1446, 259, 1071, 733,\\
&\qquad 180, 1230, 1301, 1174, 1581, 1753, 358, 696, 1607, 1964, 1289, 2022, 1648, 1673, 1148, 1003, 1848, 1721, 1292, 1561, 1639, 1960, 553),\\[0.5em]
\boldsymbol{\alpha}_{74} &= (1523, 1695, 1774, 1986, 1805, 1798, 1502, 1137, 570, 1183, 950, 574, 1438, 706, 1829, 486, 77, 1652, 1421, 1497, 1692, 18, 1899, 534, 176,\\
&\qquad 616, 1801, 1913, 1836, 1536, 1912, 284, 617, 499, 1381, 1565, 478, 3, 1409, 1367, 74, 1573, 1420, 210, 1733, 1200, 89, 396, 965, 1062,\\
&\qquad 966, 339, 368, 371, 152, 2020, 91, 1265, 209, 476, 801, 795, 1309, 1074, 1983, 350, 256, 141, 1403, 395, 311, 1070, 1570, 2005),\\[0.5em]
\boldsymbol{\alpha}_{75} &= (173, 1880, 1581, 816, 1006, 215, 166, 1369, 630, 954, 21, 104, 232, 1505, 372, 1149, 262, 949, 836, 1, 1561, 223, 680, 390, 929,\\
&\qquad 880, 133, 1435, 1879, 838, 1064, 1090, 874, 868, 990, 1791, 1306, 1229, 1443, 1769, 381, 1215, 817, 1201, 1518, 1427, 986, 1050, 1930, 1759,\\
&\qquad 1713, 66, 520, 1243, 1331, 833, 1461, 1402, 29, 500, 1725, 328, 1693, 1365, 335, 776, 198, 770, 1885, 599, 2009, 1189, 708, 4, 1758),\\[0.5em]
\boldsymbol{\alpha}_{76} &= (669, 1719, 1583, 1499, 487, 1689, 7, 1434, 1475, 958, 341, 105, 507, 937, 72, 32, 1536, 1705, 197, 1621, 1270, 556, 1115, 1567, 509,\\
&\qquad 1137, 1884, 445, 1637, 1623, 1799, 1653, 373, 1596, 140, 787, 395, 278, 1038, 1356, 717, 137, 1140, 1995, 180, 326, 831, 1593, 524, 1763,\\
&\qquad 1338, 1133, 1066, 192, 1773, 972, 1944, 786, 1415, 58, 2045, 2008, 864, 465, 298, 906, 2001, 1252, 789, 1353, 1410, 1755, 974, 405, 815,\\
&\qquad 1321),\\[0.5em]
\boldsymbol{\alpha}_{77} &= (1902, 662, 1184, 1049, 2036, 707, 360, 1921, 222, 651, 170, 1527, 2023, 51, 1516, 788, 52, 786, 7, 1383, 753, 404, 703, 350, 412,\\
&\qquad 1005, 96, 880, 735, 1473, 432, 1811, 1139, 734, 947, 496, 268, 1492, 1750, 1347, 1901, 1852, 827, 1080, 1501, 1586, 206, 1185, 862, 298,\\
&\qquad 1723, 1338, 730, 118, 1354, 1966, 1179, 1693, 1666, 680, 1506, 797, 1930, 1434, 1643, 590, 1518, 789, 1076, 1859, 1314, 1426, 215, 1745, 1504,\\
&\qquad 1098, 1897),\\[0.5em]
\boldsymbol{\alpha}_{78} &= (318, 1837, 456, 1714, 272, 2003, 438, 1010, 434, 1326, 946, 1256, 359, 233, 376, 658, 364, 1170, 1790, 476, 573, 844, 1851, 1690, 328,\\
&\qquad 1656, 1725, 1904, 1495, 1741, 551, 1158, 688, 585, 1786, 393, 429, 1878, 800, 1253, 1988, 739, 323, 595, 751, 362, 2008, 332, 264, 667,\\
&\qquad 1973, 1953, 1427, 1202, 1949, 72, 165, 269, 426, 872, 162, 591, 421, 779, 951, 780, 809, 2040, 73, 1576, 1279, 1724, 147, 69, 128,\\
&\qquad 383, 931, 48),\\[0.5em]
\boldsymbol{\alpha}_{79} &= (1858, 952, 2010, 1731, 1281, 1142, 699, 373, 758, 1986, 991, 1838, 203, 1385, 218, 792, 372, 1031, 1304, 272, 655, 1784, 1713, 98, 1498,\\
&\qquad 1023, 1335, 280, 1172, 88, 428, 207, 1714, 1869, 159, 1701, 332, 374, 497, 802, 72, 1273, 1587, 919, 1440, 1615, 981, 2019, 1700, 1793,\\
&\qquad 162, 1871, 1311, 1364, 265, 1294, 1198, 1645, 519, 267, 1984, 846, 1654, 1688, 863, 761, 1551, 90, 992, 1363, 1212, 950, 727, 955, 389,\\
&\qquad 1673, 1396, 1882, 476),\\[0.5em]
\boldsymbol{\alpha}_{80} &= (767, 550, 1194, 1374, 1532, 688, 431, 771, 1885, 766, 70, 510, 1979, 947, 1294, 293, 1941, 1939, 675, 171, 1801, 1039, 151, 1266, 1595,\\
&\qquad 638, 693, 1143, 471, 761, 691, 359, 68, 1853, 984, 354, 1124, 1305, 133, 85, 1340, 524, 473, 20, 1183, 1208, 406, 208, 1812, 1878,\\
&\qquad 775, 234, 226, 212, 1954, 1195, 1028, 1238, 1145, 1333, 244, 408, 826, 1953, 1828, 405, 442, 902, 1957, 225, 316, 544, 1693, 1935, 356,\\
&\qquad 1391, 1761, 1554, 1212, 963).
\end{align*}
\normalsize

\section{Deferred Proofs for Time-Like Failures}\label{sec:tlf_proofs}
We collect here deferred proofs of lemmas used in the discussion of time-like failures in Subsection~\ref{subsec:time_like_failures}.
\begin{proof}[Proof of Lemma~\ref{lem:tlf_case_1}]
    We have
    \begin{align}
        N_1(n,d,\balpha,\hat{\beta}^{(X)}*\nu) &= \left|\left\{(k,i) \in [d-2]\times[n]:\Pi_{\leq k}h_i \in \ker(\hat{\beta}^{(X)}*\nu)\right\}\right|\\
        &= \left|\left\{(k,i) \in [d-2]\times[n]:\nu^{-1}*\Pi_{\leq k}h_i \in \ker(\hat{\beta}^{(X)})\right\}\right|.
    \end{align}
    Here, given two vectors $v,w$, we have written $v*w$ for the componentwise multiplication of those vectors, and $\nu^{-1}$ denotes the componentwise inversion of $\nu$, noting it has all non-zero entries. Next, we note that, for any $i$, $\Pi_{\leq k}h_i$ is non-zero exactly on the coordinates $\{1, \ldots, k\}$. Moreover, via a uniformly random choice of $\nu$, $\nu^{-1}*\Pi_{\leq k}h_i$ becomes a uniformly random vector which is non-zero exactly on the coordinates $\{1, \ldots, k\}$. Since there are $(q-1)^k$ of these, we have
    \begin{equation}
        \mathbb{E}_{\nu}\left\{N_1(n,d,\balpha,\hat{\beta}^{(X)}*\nu)\right\} = n\cdot \sum_{k=1}^{d-2}\frac{D_k}{(q-1)^k},
    \end{equation}
    where $D_k$ is the number of vectors in $\ker(\hat{\beta}^{(X)})$ that are non-zero exactly on the first $k$ coordinates.

    It remains to compute $D_k$. We notice that, because $\begin{pmatrix}
        \hat{\beta}^{(X)}&I
    \end{pmatrix}$ is a parity-check matrix for a classical code of distance $M+1$, no set of $M$ columns of $\hat{\beta}^{(X)}$ can be linearly dependent. This means that $D_k = 0$ for $k \leq M$, and we have
    \begin{equation}
        \mathbb{E}_{\nu}\left\{N_1(n,d,\balpha,\hat{\beta}^{(X)}*\nu)\right\} = n\cdot \sum_{k=M+1}^{d-2}\frac{D_k}{(q-1)^k}.
    \end{equation}
    Moreover, because $\begin{pmatrix}
        \hat{\beta}^{(X)}&I
    \end{pmatrix}$ is an $M \times (d-1+M)$ parity-check matrix for a code of distance $M+1$, for any $k$ with $M+1 \leq k \leq d-2$, the first $k$ columns of $\hat{\beta}^{(X)}$ form a parity-check matrix for a code with parameters $[k,k-M,M+1]$. This is therefore an MDS (maximum distance separable), whose weight distributions may be characterised exactly. In particular, the number of weight-$k$ codewords of this code, which is then also $D_k$, is (see~\cite{macwilliams1977theory})
    \begin{equation}
        D_k = \sum_{t=0}^{k-M-1}(-1)^t\begin{pmatrix}
            k\\t
        \end{pmatrix}(q^{k-M-t}-1),
    \end{equation}
    and the result follows.
\end{proof}

\begin{proof}[Proof of Lemma~\ref{lem:tlf_case_2}]
    For $\mu \in \mathbb{F}_q^*$, it is useful to write
    \begin{equation}
        x_\mu \coloneq \Pi_{\leq k}h_i + \mu e_k,
    \end{equation}
    so that
    \begin{equation}
        \mathsf{TLF}_2^{(X)}[k,i](\hat{\beta}^{(X)}*\nu) = \left\{\mu \in \mathbb{F}_q^*:x_\mu \in \ker(\hat{\beta}^{(X)}*\nu)\setminus\{0\}\right\}.
    \end{equation}
    We then have
    \begin{multline}
        \sum_{\mu \in \mathsf{TLF}_2^{(X)}[k,i](\hat{\beta}^{(X)}*\nu)}\mathbb{P}\left[\text{Mixed}^{(k)},i,Z^\lambda,\mu\lambda\;|\;\mathcal{B}\right] =\\ \sum_{\mu \in \mathbb{F}_q^*}\mathbb{I}\left[x_\mu \in \ker(\hat{\beta}^{(X)}*\nu)\setminus\{0\}\right]\mathbb{P}\left[\text{Mixed}^{(k)},i,Z^\lambda,\mu\lambda\;|\;\mathcal{B}\right],
    \end{multline}
    where $\mathbb{I}\left[A\right]$ is the indicator function for the event $A$. We then obtain
    \begin{multline}
        \mathbb{E}_\nu\left\{\sum_{\mu \in \mathsf{TLF}_2^{(X)}[k,i](\hat{\beta}^{(X)}*\nu)}\mathbb{P}\left[\text{Mixed}^{(k)},i,Z^\lambda,\mu\lambda\;|\;\mathcal{B}\right]\right\}=\\
        \sum_{\mu \in \mathbb{F}_q^*}\mathbb{P}_\nu \left[x_\mu \in \ker(\hat{\beta}^{(X)}*\nu)\setminus\{0\}\right]\cdot\mathbb{P}\left[\text{Mixed}^{(k)},i,Z^\lambda,\mu\lambda\;|\;\mathcal{B}\right].
    \end{multline}
    In this equation, it is worth emphasising the the first probability on the right-hand side is a probability over the randomness of $\nu$, whereas the second is the probability of some quantum error. The next thing to do is focus on $\mathbb{P}_\nu \left[x_\mu \in \ker(\hat{\beta}^{(X)}*\nu)\setminus\{0\}\right]$.

    First, notice that if $\mu \neq (h_i)_k$, then $x_\mu$ is non-zero exactly on the first $k$ coordinates. We have that
    \begin{equation}
        x_\mu \in \ker(\hat{\beta}^{(X)}*\nu)\setminus\{0\} \iff \nu^{-1} * x_\mu \in \ker(\hat{\beta}^{(X)}) \setminus\{0\},
    \end{equation}
    where given vectors $v$ and $w$, $v*w$ denotes the componentwise multiplication of $v$ and $w$, and $\nu^{-1}$ denotes the componentwise inversion of $\nu$. Over a uniformly random choice of $\nu$, we have that $\nu^{-1} * x_\mu$ is uniformly random over the $(q-1)^k$ vectors that are non-zero on the first $k$ coordinates. By the same reasoning as in the proof of Lemma~\ref{lem:tlf_case_1}, the number of non-zero vectors that are in $\ker(\hat{\beta}^{(X)})$ and non-zero on exactly the first $k$ coordinates, for $k \in [d-1]$, is
    \begin{equation}
        D_k \coloneq \begin{cases}
            0 &\text{ if }k \leq M\\
            \sum_{t=0}^{k-M-1}(-1)^t\begin{pmatrix}
                k\\t
            \end{pmatrix}(q^{k-M-t}-1)&\text{ if } k \geq M+1.
        \end{cases}
    \end{equation}
    For $\mu \neq (h_i)_k$, we therefore have
    \begin{equation}
        \mathbb{P}_\nu\left[x_\mu \in \ker(\hat{\beta}^{(X)}*\nu)\setminus\{0\}\right] = \frac{D_k}{(q-1)^k}.
    \end{equation}
    On the other hand, suppose $\mu = (h_i)_k$. Then $x_\mu = 0$ for $k=1$, and is non-zero on exactly the first $k-1$ coordinates for $k\geq 2$. In this case, we thus have
    \begin{equation}
        \mathbb{P}_\nu\left[x_\mu \in \ker(\hat{\beta}^{(X)}*\nu)\setminus\{0\}\right] = \frac{D_{k-1}}{(q-1)^{k-1}},
    \end{equation}
    where it is convenient to write $D_0 = 0$.
    Returning to what we ultimately want to calculate, we have that
    \begin{multline}
        \mathbb{E}_\nu\left\{\sum_{\mu \in \mathsf{TLF}_2^{(X)}[k,i](\hat{\beta}^{(X)}*\nu)}\mathbb{P}\left[\text{Mixed}^{(k)},i,Z^\lambda,\mu\lambda\;|\;\mathcal{B}\right]\right\}\\=
        \frac{D_k}{(q-1)^k}\sum_{\mu \in \mathbb{F}_q^*\setminus\{(h_i)_k\}}\left\{\mathbb{P}\left[\text{Mixed}^{(k)},i,Z^\lambda,\mu\lambda\;|\;\mathcal{B}\right]\right\} \\+ \frac{D_{k-1}}{(q-1)^{k-1}}\cdot\mathbb{P}\left[\text{Mixed}^{(k)},i,Z^\lambda,(h_i)_k\lambda\;|\;\mathcal{B}\right],
    \end{multline}
    for every $k \in [d-1]$, using the definition $D_0 = 0$.
    We now use the result of the following lemma, whose proof is deferred until after this one.
    \begin{lemma}\label{lem:recurrence}For every $k \in [d-1]$,
        \begin{equation}
        \frac{D_{k-1}}{(q-1)^{k-1}} = \frac{D_k}{(q-1)^k} + (-1)^{k-M}\frac{\begin{pmatrix}
            k-2\\M-1
        \end{pmatrix}}{(q-1)^{k-1}}.\footnote{We use the convention throughout that $\begin{pmatrix}
            A\\B
        \end{pmatrix} = 0$ if $B>A$.}
    \end{equation}
    \end{lemma}
    \noindent Using this lemma, we get
    \begin{multline}
        \mathbb{E}_\nu\left\{\sum_{\mu \in \mathsf{TLF}_2^{(X)}[k,i](\hat{\beta}^{(X)}*\nu)}\mathbb{P}\left[\text{Mixed}^{(k)},i,Z^\lambda,\mu\lambda\;|\;\mathcal{B}\right]\right\}\\=\frac{D_k}{(q-1)^k}\sum_{\mu \in \mathbb{F}_q^*}\mathbb{P}\left[\text{Mixed}^{(k)},i,Z^\lambda,\mu\lambda\;|\;\mathcal{B}\right] + (-1)^{k-M}\frac{\begin{pmatrix}
            k-2\\M-1
        \end{pmatrix}}{(q-1)^{k-1}}\mathbb{P}\left[\text{Mixed}^{(k)},i,Z^\lambda,(h_i)_k\lambda\;|\;\mathcal{B}\right]\\
        \leq \frac{D_k}{(q-1)^k}\sum_{\mu \in \mathbb{F}_q^*}\mathbb{P}\left[\text{Mixed}^{(k)},i,Z^\lambda,\mu\lambda\;|\;\mathcal{B}\right] + \frac{\begin{pmatrix}
            k-2\\M-1
        \end{pmatrix}}{(q-1)^{k-1}}\mathbb{P}\left[\text{Mixed}^{(k)},i,Z^\lambda,(h_i)_k\lambda\;|\;\mathcal{B}\right].
    \end{multline}
    Let us now denote by
    \begin{equation}
        \mathbb{P}\left[\text{Mixed}^{(k)},i,Z^\lambda,*\;|\;\mathcal{B}\right]
    \end{equation}
    the probability that the $k$'th check creates a $Z$ error $\lambda e_i \in \mathbb{F}_q^n$, and has any non-trivial measurement error. We have
    \begin{align}
        \sum_{\mu \in \mathbb{F}_q^*}\mathbb{P}\left[\text{Mixed}^{(k)},i,Z^\lambda,\mu\lambda\;|\;\mathcal{B}\right] &= \mathbb{P}\left[\text{Mixed}^{(k)},i,Z^\lambda,*\;|\;\mathcal{B}\right]\\
        \mathbb{P}\left[\text{Mixed}^{(k)},i,Z^\lambda, (h_i)_k\lambda\;|\;\mathcal{B}\right] &\leq \mathbb{P}\left[\text{Mixed}^{(k)},i,Z^\lambda,*\;|\;\mathcal{B}\right],
    \end{align}
    and so
    \begin{multline}
        \mathbb{E}_\nu\left\{\sum_{\mu \in \mathsf{TLF}_2^{(X)}[k,i](\hat{\beta}^{(X)}*\nu)}\mathbb{P}\left[\text{Mixed}^{(k)},i,Z^\lambda,\mu\lambda\;|\;\mathcal{B}\right]\right\}\\\leq\frac{D_k}{(q-1)^k}\mathbb{P}\left[\text{Mixed}^{(k)},i,Z^\lambda,*\;|\;\mathcal{B}\right] + \frac{\begin{pmatrix}
            k-2\\M-1
        \end{pmatrix}}{(q-1)^{k-1}}\mathbb{P}\left[\text{Mixed}^{(k)},i,Z^\lambda,*\;|\;\mathcal{B}\right]
    \end{multline}
    It remains to take the expectation over $\mathcal{B}$, which we do as follows. The event corresponding to the probability $\mathbb{P}\left[\text{Mixed}^{(k)},i,Z^\lambda,*\;|\;\mathcal{B}\right]$ requires the qudit $XX$ measurement on the $i$'th qudit consuming the $k$'th cat state to fail in such a way that an error $Z^\lambda$ is deposited on the data qudit, and any non-trivial measurement error results on that check. We want to consider
    \begin{equation}
        \mathbb{E}_{\mathcal{B}}\left\{\mathbb{P}\left[\text{Mixed}^{(k)},i,Z^\lambda,*\;|\;\mathcal{B}\right]\right\}.
    \end{equation}
    When we change $\mathcal{B}$, specifically $B_i$ (the qudit-to-qubit mapping on the $i$'th qudit), the qubit error that must be created on the data to correspond to $Z^\lambda$ must also change, and every qubit error is visited equally. Whenever $\mathcal{B}$ is changed, however, any non-trivial measurement error gives the required event. We thus have
    \begin{equation}
        \mathbb{E}_{\mathcal{B}}\left\{\mathbb{P}\left[\text{Mixed}^{(k)},i,Z^\lambda,*\;|\;\mathcal{B}\right]\right\} = \frac{1}{q-1}\sum_{v \in \mathbb{F}_2^s\setminus\{0\}}\mathbb{P}\left[XX, Z^v, *\right],
    \end{equation}
    where $\mathbb{P}\left[XX, Z^v,*\right]$ denotes the probability that the qudit $XX$ measurement fails in such a way to leave the $Z$ error $Z^v$ on the data qudit, and creates any non-trivial measurement error on the corresponding check. These are all distinct ways the qudit $XX$ measurement can fail, and so we have
    \begin{equation}
        \sum_{v \in \mathbb{F}_2^s\setminus\{0\}}\mathbb{P}\left[XX, Z^v,*\right] \leq \mathbb{P}\left[XX\right],
    \end{equation}
    where the right-hand side is the probability of any failure in the qudit $XX$ measurement. In summary,
    \begin{multline}
        \mathbb{E}_{\mathcal{B},\nu}\left\{\sum_{\mu \in \mathsf{TLF}_2^{(X)}[k,i](\hat{\beta}^{(X)}*\nu)}\mathbb{P}\left[\text{Mixed}^{(k)},i,Z^\lambda,\mu\lambda\;|\;\mathcal{B}\right]\right\}\\\leq\left(\frac{D_k}{(q-1)^{k+1}}+\frac{1}{(q-1)^k}\begin{pmatrix}
            k-2\\M-1
        \end{pmatrix}\right)\mathbb{P}\left[XX\right],
    \end{multline}
    as required.
\end{proof}
\begin{proof}[Proof of Lemma~\ref{lem:recurrence}]
    First, note that the statement is trivial for $k\leq M$, since we defined $D_0 = 0$. Next, let us restrict to $k \geq M+1$. We have
    \begin{align}
        D_k&=\sum_{t=0}^{k-M-1}(-1)^t\begin{pmatrix}
            k\\t
        \end{pmatrix}(q^{k-M-t}-1)\\
        &=\sum_{t=0}^{k-M-1}(-1)^t\begin{pmatrix}
            k-1\\t
        \end{pmatrix}(q^{k-M-t}-1)+\sum_{t=0}^{k-M-1}(-1)^t\begin{pmatrix}
            k-1\\t-1
        \end{pmatrix}(q^{k-M-t}-1),
    \end{align}
    where we have used Pascal's identity, with the convention $\begin{pmatrix}
        A\\B
    \end{pmatrix} = 0$ for $B>A$. Then,
    \begin{align}
        D_k &= \sum_{t=0}^{k-M-1}(-1)^t\begin{pmatrix}
            k-1\\t
        \end{pmatrix}(q^{k-M-t}-1)+\sum_{t=-1}^{k-M-2}(-1)^{t+1}\begin{pmatrix}
            k-1\\t
        \end{pmatrix}(q^{k-M-t-1}-1)\\
        &= \sum_{t=0}^{k-M-1}(-1)^t\begin{pmatrix}
            k-1\\t
        \end{pmatrix}(q^{k-M-t}-1)-\sum_{t=0}^{k-M-2}(-1)^{t}\begin{pmatrix}
            k-1\\t
        \end{pmatrix}(q^{k-M-t-1}-1)\\
        &= (-1)^{k-M-1}\begin{pmatrix}
            k-1\\k-M-1
        \end{pmatrix}(q-1) + \sum_{t=0}^{k-M-2}(-1)^t\begin{pmatrix}
            k-1\\t
        \end{pmatrix}(q^{k-M-t}-q^{k-M-t-1})\\
        &= (-1)^{k-M-1}\begin{pmatrix}
            k-1\\k-M-1
        \end{pmatrix}(q-1) + (q-1)\sum_{t=0}^{k-M-2}(-1)^t\begin{pmatrix}
            k-1\\t
        \end{pmatrix}q^{k-M-t-1},
    \end{align}
    where the third line follows by combining the sums.
    Then, the second term is equal to
    \begin{multline}
        (q-1)\sum_{t=0}^{k-M-2}(-1)^t\begin{pmatrix}
            k-1\\t
        \end{pmatrix}\left(q^{k-M-t-1}-1\right) + (q-1)\sum_{t=0}^{k-M-2}(-1)^t\begin{pmatrix}
            k-1\\t
        \end{pmatrix} \\= (q-1)D_{k-1} + (q-1)(-1)^{k-M}\begin{pmatrix}
            k-2\\k-M-2
        \end{pmatrix},
    \end{multline}
    where we have used that $\sum_{t=0}^m(-1)^t\begin{pmatrix}
        n\\t
    \end{pmatrix}=(-1)^m\begin{pmatrix}
        n-1\\m
    \end{pmatrix}$ for positive integers $n$. We thus have
    \begin{align}
        D_k &= (q-1)D_{k-1} + (-1)^{k-M-1}(q-1)\left[\begin{pmatrix}
            k-1\\k-M-1
        \end{pmatrix}-\begin{pmatrix}
            k-2\\k-M-2
        \end{pmatrix}\right]\\
        &=(q-1)D_{k-1} + (-1)^{k-M-1}(q-1)\begin{pmatrix}
            k-2\\k-M-1
        \end{pmatrix}\\
        &=(q-1)D_{k-1} + (-1)^{k-M-1}(q-1)\begin{pmatrix}
            k-2\\M-1
        \end{pmatrix},
    \end{align}
    and the result follows.
\end{proof}

\section{Non-CSS Outer Codes}\label{sec:non_css_outer_codes}

We include here a brief comment on a technique which we believe would significantly improve the scheme, although we do not include in this work for simplicity and brevity, and leave it to the future. Namely, our codes are formed by taking quantum Reed-Solomon codes for an alphabet of size $q = 2^{11}$, and then using qudit-to-qubit mappings to make them into qubit codes~\cite{wills2026review}. As discussed in Section~\ref{sec:ler_calc}, we imagine taking random qudit-to-qubit mappings $\mathcal{B}$, and show that our scheme performs well in expectation over this choice, regardless of the underlying logical error distribution. This amounts to performing a random $\mathsf{CNOT}$ circuit on the $s$ qubits forming each Galois qudit.

One could, instead, imagine performing a random \textit{Clifford} operation on each set of $s$ qubits forming every qudit. This would, of course, have the effect of forming a non-CSS outer code, and one's (Shor) error correction scheme would have to handle both $X$ and $Z$ checks in one larger round, not in two smaller rounds. However, we believe doing this would have a significant benefit. It should be the case that the resultant code behaves like a quantum Reed-Solomon code with alphabet size $q^2 = 4194304$. In general, the list decoding techniques improve as $q$ gets much larger than $n$, the length of the code; as one example, recall from Subsubsection~\ref{subsubsec:list_decoding} that, for the distance $6$ code, the fraction of weight-$3$ errors that are uncorrectable (in the sense that they share a syndrome with some other weight-$3$ error) is
\begin{equation}
    \frac{(n-3)(n-4)(n-5)}{6(q-1)^2}.
\end{equation}
Since we consider lengths like $n \approx 40$, increasing $q$ from $2^{11}$ to $4^{11}$ would have a significant impact. However, we are not claiming this formally; one would have to work out the details, which we leave to future work.

\section{5-Qudit Code Example}\label{sec:5_qudit_code}

Below we work through an example of a length $5$ and distance $3$ quantum Reed-Solomon code over $\mathbb{F}_q$, where $q = 2^{11}$ as usual, which may be binarised to produce a $[[55,11,\geq 3]]$ qubit code. The field $\mathbb{F}_q$ is represented using the irreducible polynomial $t^{11} + t^2 + 1 = 0$; the same irreducible polynomial used in Appendix~\ref{subsec:evaluation_points}. 

Referring to the definitions in~\cite{wills2026review}, we choose evaluation points $\balpha = (t, t^2, t^3, t^4, t^5)$ and $\bv = (1,1,1,1,1)$. This gives the parity-check matrix for the $X$-stabilisers $(H_X)_{ab} = \alpha_b^a$ for $a \in \{0,1\}$ and $b \in \{1,2,3,4,5\}$ i.e.,
\begin{equation}
    H_X = \begin{pmatrix}
        1 & 1 & 1 & 1 & 1\\
        t & t^2 & t^3 & t^4 & t^5
    \end{pmatrix}.
\end{equation}
For the $Z$-stabilisers, one then calculates the $\bu$, which satisfy $u_i^{-1} = \prod_{j:j \neq i}(\alpha_i-\alpha_j)$, so that we have the $Z$ parity-check matrix $(H_Z)_{ab} = u_b\alpha_b^a$. One may calculate
\begin{equation}
H_Z =
\begin{pmatrix}
\begin{aligned}[c]
t^3 + t^6 + t^7 + \\
t^{10}
\end{aligned}
&
\begin{aligned}[c]
1 + t + t^2 + \\
t^8 + t^9 + t^{10}
\end{aligned}
&
\begin{aligned}[c]
1 + t + t^2 + t^3 + \\
t^4 + t^5 + t^8 + t^{10}
\end{aligned}
&
\begin{aligned}[c]
1 + t^5 + t^6 + \\
t^7 + t^8 + t^9
\end{aligned}
&
\begin{aligned}[c]
1 + t^4 + \\
t^8 + t^{10}
\end{aligned}
\\[1em]
\begin{aligned}[c]
1 + t^2 + t^4 + \\
t^7 + t^8
\end{aligned}
&
\begin{aligned}[c]
1 + t + t^4 + \\
t^{10}
\end{aligned}
&
\begin{aligned}[c]
1 + t^3 + t^5 + \\
t^6 + t^7 + t^8
\end{aligned}
&
\begin{aligned}[c]
1 + t + t^3 + \\
t^9 + t^{10}
\end{aligned}
&
\begin{aligned}[c]
t^2 + t^5 + \\
t^6 + t^9
\end{aligned}
\end{pmatrix}
\end{equation}
and check that $H_X\cdot H_Z^T = 0$, so that these matrices define a CSS code for qudits. Moreover, one checks that the parity-check matrices $H_X$ and $H_Z$ define classical codes of distance $3$, so the resulting quantum code has distance $3$. Since the code has $4$ linearly independent parity checks, the qudit code has parameters $[[5,1,3]]_q$.

To binarise the code, it is convenient to consider a single basis. This is a self-dual basis~\cite{wills2026review}, and we may use the same basis to expand $X$-stabiliser elements and $Z$-stabiliser elements. We use the same basis on all $5$ qudits. This basis is
\begin{equation}
    B = \begin{pmatrix}
        1 + t^5 + t^6\\
        1 + t + t^3 + t^{10}\\
        t^6 + t^9\\
        t + t^3 + t^7 + t^9\\
        t^2 + t^3 + t^5 + t^6 + t^7 + t^9\\
        t + t^4 + t^5 + t^6 + t^7 + t^9 + t^{10}\\
        1 + t + t^5 + t^7 + t^8 + t^10\\
        t^3 + t^7 + t^9 + t^{10}\\
        1 + t^2 + t^3 + t^5 + t^6 + t^7\\
        1 + t + t^4 + t^5 + t^6 + t^{10}\\
        t + t^3 + t^8 + t^9 + t^{10}
    \end{pmatrix}.
\end{equation}
As a point of interest, this is in fact a self-dual normal basis~\cite{mullen2013handbook}, although the normality is not important here.

The same basis is chosen for each qudit-to-qubit mapping. Since this is a self-dual normal basis, the expansion for the $X$ and $Z$-stabilisers may be performed with the same map. For the $r$'th binary $X$-stabiliser, where $r = 11a+i$ for $a \in \{0,1\}$ and $i \in \{0, \ldots, 10\}$, its $c$'th entry, where $c = 11(b-1)+k$ for $b \in \{1,2,3,4,5\}$ and $k \in \{0, \ldots, 10\}$, is $\Tr(B_i(H_X)_{ab}B_k)$, where $B_i$ is the $i$'th entry of the basis $B$. 

Binarising the $X$-stabilisers yields
\begin{verbatim}
    1: [10000000000 10000000000 10000000000 10000000000 10000000000]
    2: [01000000000 01000000000 01000000000 01000000000 01000000000]
    3: [00100000000 00100000000 00100000000 00100000000 00100000000]
    4: [00010000000 00010000000 00010000000 00010000000 00010000000]
    5: [00001000000 00001000000 00001000000 00001000000 00001000000]
    6: [00000100000 00000100000 00000100000 00000100000 00000100000]
    7: [00000010000 00000010000 00000010000 00000010000 00000010000]
    8: [00000001000 00000001000 00000001000 00000001000 00000001000]
    9: [00000000100 00000000100 00000000100 00000000100 00000000100]
    10: [00000000010 00000000010 00000000010 00000000010 00000000010]
    11: [00000000001 00000000001 00000000001 00000000001 00000000001]
    12: [10011110111 01100110110 00010001001 00110010011 00100001100]
    13: [00101011111 11001111011 01011011001 00110011011 01110001011]
    14: [01011000100 10010101111 00011100110 11100111101 11011011001]
    15: [10101111000 00101100010 11111001110 11001010111 01110010000]
    16: [11111000111 01010111100 01111111110 00010110001 00100000101]
    17: [10010000001 11111100011 00101110010 00101011110 00000010111]
    18: [11010010000 11001000000 01001101111 11111110001 00110101110]
    19: [01010000111 01101001000 11011011011 01100100000 11100011010]
    20: [11101001111 10101000011 00111010100 00110100100 10001110010]
    21: [11001001100 11110100111 00111111010 11010100001 01000111101]
    22: [11001101100 01100100110 11000011001 11111010011 01101100011].
\end{verbatim}
One similarly binarises the $Z$-stabilisers to obtain
\begin{verbatim}
    1: [00010010000 10000100000 00110001111 01011110100 11111001011]
    2: [01001100100 01111001000 00110100010 11111011010 11111010100]
    3: [00011110010 01110001111 11101100001 01100101001 11100110101]
    4: [10100100111 01101101000 11001100001 11000001011 11000100101]
    5: [01100011011 01010010010 00110010010 11000001001 11000010010]
    6: [01110001101 10010111111 01110100011 10100110101 00110100100]
    7: [10101011111 00001110101 00001010100 11000111011 01101000101]
    8: [00001110110 01110101101 10000001001 01111010101 10000000111]
    9: [01010111010 00100111001 10000010111 10000101001 01110111101]
    10: [00111011100 00101100001 11001100100 01010010001 10001001000]
    11: [00011110000 00100111110 10110101100 00111111110 10110011100]
    12: [01111101000 10011010101 11111000100 00001110110 00010001111]
    13: [10101010110 01001000010 11110100000 01110100011 01100010111]
    14: [11011100010 00100111100 11000100011 01100001011 01011110110]
    15: [10111011101 10011000011 11011000101 01001011010 10110000001]
    16: [11110101100 11011111110 10010111101 10011111011 00100010100]
    17: [10101000111 00101111101 01101100000 11001110110 00100101100]
    18: [01010001100 10101111010 00001000101 10011110101 01101000110]
    19: [10011010101 00101111000 00001001011 00111000110 10000100000]
    20: [01011111101 10101100010 10011010001 10000111110 11101110000]
    21: [01100100011 01011010110 00100001000 11111101110 11100010011]
    22: [00010101110 10010100000 00111011100 01101010001 11010000011].
\end{verbatim}
One checks that the binarised $X$ stabilisers are orthogonal to the binarised $Z$ stabilisers, as expected.

\end{document}